%% file: main.tex
\documentclass[11pt]{article}

\usepackage{xspace}
\usepackage{graphicx}
\usepackage{waingarten}
\usepackage{thmtools}
\usepackage{thm-restate}
\usepackage{color}
\usepackage{setspace}

\def\FullBox{\hbox{\vrule width 8pt height 8pt depth 0pt}}
\newcommand{\QED}{\;\;\;\FullBox}
\renewenvironment{proof}{\noindent{{\textbf{Proof:}~}}} {\hfill\QED}
\providecommand{\email}[1]{\href{mailto:#1}{\nolinkurl{#1}\xspace}}

\def\FullBox{\hbox{\vrule width 8pt height 8pt depth 0pt}}

\title{Streaming and Massively Parallel Algorithms for Euclidean Max-Cut}
    \author { 
      Nicolas Menand\thanks{University of Pennsylvania.
      \texttt{email:~nmenand@cis.upenn.edu}. Supported by an AWS AI Gift for Research on Trustworthy AI and by the National Science Foundation (NSF) under Grant No. CCF-2337993.}
      \and
      Erik Waingarten\thanks{University of Pennsylvania.
      \texttt{email:~ewaingar@cis.upenn.edu}. Supported by the National Science Foundation (NSF) under Grant No. CCF-2337993.}
    }
    
\begin{document}
\maketitle

\begin{abstract}
Given a set of vectors $X = \{ x_1,\dots, x_n \} \subset \R^d$, the Euclidean max-cut problem asks to partition the vectors into two parts so as to maximize the sum of Euclidean distances which cross the partition. We design new algorithms for Euclidean max-cut in models for massive datasets:
\begin{itemize}
    \item We give a fully-scalable constant-round MPC algorithm using $O(nd) + n \cdot \poly( \log(n) / \eps)$ total space which gives a $(1+\eps)$-approximate Euclidean max-cut. 
    \item We give a dynamic streaming algorithm using $\poly(d \log \Delta / \eps)$ space when $X \subseteq [\Delta]^d$, which provides oracle access to a $(1+\eps)$-approximate Euclidean max-cut. 
\end{itemize}
Recently, Chen, Jiang, and Krauthgamer [{STOC}~'23] gave a dynamic streaming algorithm with space $\poly(d\log\Delta/\eps)$ to approximate the value of the Euclidean max-cut, but could not provide oracle access to an approximately optimal cut. This was left open in that work, and we resolve it here. 
Both algorithms follow from the same framework, which analyzes a ``parallel'' and ``subsampled'' (Euclidean) version of a greedy algorithm of Mathieu and Schudy [{SODA}~'08] for dense max-cut. 
\end{abstract}

\thispagestyle{empty}
\newpage
\begin{spacing}{0.75}
\tableofcontents
\end{spacing}
\thispagestyle{empty}

\newpage

\pagenumbering{arabic}
\setcounter{page}{1}

\input{intro.tex}
\input{non-streaming.tex}
\input{mpc.tex}
\input{insertion-only.tex}
\input{analysis.tex}
\input{turnstile.tex}
\input{dynamic-analysis.tex}

\newpage

\appendix

\input{algorithmic-prelims.tex}

\ignore{
\input{simple-greedy.tex}
}
\bibliographystyle{alpha}
\bibliography{waingarten}

\end{document}

%% file: intro.tex

\section{Introduction}\label{sec:introduction}

\newcommand{\cutval}{\text{cut-val}}

This paper is on algorithms for Euclidean max-cut in sublinear models of computation, namely, massively parallel computation and streaming. Given a set of vectors, we consider the geometric graph consisting of $n$ vertices (one for each vector) and all edges connecting two vectors with weights given by Euclidean distances. Euclidean max-cut asks to partition the vectors into two parts so as to maximize the sum of weights of edges cut by the partition. It is the natural geometric version of the classical max-cut problem on graphs, which has been well-studied from the approximation algorithms perspective. In general graphs, there is polynomial-time $0.878$-approximation~\cite{GW95} which is optimal assuming the Unique Games Conjecture~\cite{KKMO07}. For dense or metric instances of max-cut, there are known PTAS~\cite{V96, FK01, MS08}, which in particular, will play a prominent role in this work. 


From the sublinear algorithms perspective, there has been recent progress on Euclidean max-cut in dynamic geometric streams~\cite{I04}, where a set of input vectors from $[\Delta]^d$ are presented as a stream of insertions and deletions, and the algorithm must approximate (up to a multiplicative $(1\pm\eps)$-factor) the Euclidean max-cut of the underlying set of vectors. In~\cite{CJK23}, the authors present a dynamic streaming algorithm which uses $\poly(d\log(\Delta) / \eps)$ space, and can approximate the \emph{value} of the max-cut of the underlying set of vectors up to a $(1\pm \eps)$-factor. This algorithm improved on the prior work of~\cite{FS05}, who showed this was possible for constant dimension $d$ (since space complexity scaled exponentially in $d$). As we  detail in~Section~\ref{sec:tech-overview}, the algorithm of~\cite{CJK23} builds on a line-of-work on random sampling for max-CSPs~\cite{AVKK03, RV07}. Importantly, the algorithm approximates the value and not the cut itself; while a low-space algorithm cannot explicitly output a cut,\footnote{Since it is inherently an $\Theta(n)$-sized object, consisting of a partition of the vectors.} a low-space streaming algorithm could provide, after processing a set of vectors, oracle access of the form ``which side of the cut is $x$ in?'' which produces an approximate max-cut. This was discussed in Appendix~B of~\cite{CJK23} and left open, where they noted the algorithm of ~\cite{FS05} along with~\cite{LSS09, L11} provides such oracle access, albeit with quantitatively much weaker $\exp\left( \poly(1/\eps) \right) \cdot \poly(\log \Delta)$ space.
\begin{quote}
    Do there exists low-space streaming algorithms for high-dimensions (more generally, algorithms in sublinear models) which can, not only approximate the value of the Euclidean max-cut, but also provide oracle access to an approximately optimal cut?
\end{quote}
The above question motivated this work. Here, we provide such algorithms which give oracle access to an approximate max-cut in the massively parallel model and streaming model of computation.

\paragraph{Approximating the value versus the object in high-dimensions.} Beyond Euclidean max-cut, multiple geometric optimization problems currently exhibit gaps in the complexity of approximating values and producing approximate solutions. Understanding the circumstances under which such gaps exists (or do not) is an interesting open problem, and one which has received recent attention. For example, such a gap (provably) exists for computing the Chamfer distance~\cite{BIJSW23}, and may also exist for the Earth Mover's Distance~\cite{R19} and Euclidean minimum spanning tree~\cite{CS09,HIM12}. A potential gap may also exist in sublinear models of computation, such as in massively parallel computation and streaming. It has been recently studied for clustering and the facility location problem~\cite{CJKVY22, CGJKV24}, as well as the Euclidean minimum spanning tree~\cite{CJLW22, CCJLW23,JMNZ24}.

\subsection{Our Contribution}

In this work, we study algorithms in massively parallel and streaming models of computation which can provide oracle access to an approximately optimal Euclidean max-cut without paying exponentially in the accuracy parameter $\eps$. In the model of massively parallel computation (MPC), a dataset of $n$ vectors in $\R^d$ is partitioned across $m$ different machines of space $s$ (where $m \cdot s$ is the total space used). The computation proceeds in rounds where machines communicate with each other. Upon termination of the algorithm, each point in the dataset should report which side of an approximately optimal cut it lies in. The following theorem appears formally in Section~\ref{sec:mpc} as Theorem~\ref{thm:mpc-full}.

\begin{theorem}[MPC (Informal)]\label{thm:mpc-intro}
There is a $O(1)$-round fully-scalable MPC algorithm which outputs a $(1+\eps)$-approximate Euclidean max-cut using $O(nd) + n \cdot \poly(\log n / \eps)$ total space.
\end{theorem}

In contrast with~\cite{CJK23} (which was designed for dynamic streaming, so applies in MPC), it can report an approximately optimal max-cut instead of solely the value. Alternatively,~\cite{FS05} does approximate the Euclidean max-cut, but requires the total space at least $n \cdot \poly(\log n) \cdot \exp(\poly(1/\eps))$. As a point of comparison, the recent work of~\cite{CGJKV24} provides a constant-round $O_{\eps}(1)$-approximation to Euclidean facility location in fully-scalable MPC with total space $O(n^{1+\eps})\cdot \poly(d)$, and \cite{JMNZ24} provides a $\tilde{O}(\log\log n)$-round, constant-factor approximation to the Euclidean minimum spanning tree in fully scalable MPC with total space $O(nd + n^{1+\eps})$. 

As we detail in Section~\ref{sec:tech-overview}, Theorem~\ref{thm:mpc-intro} will come from designing a ``parallel'' and ``subsampled'' version of an elegant greedy algorithm for max-cut in dense instances~\cite{MS08}. This approach is a unifying theme throughout this paper, and Theorem~\ref{thm:mpc-intro} is the first (and simplest) application. Then, we show that this same framework gives rise to a low-space insertion-only streaming algorithm which gives oracle access to an approximate max-cut. More specifically, in insertion-only streaming, a dataset of vectors $X = \{ x_1, \dots, x_n \} \subset [\Delta]^d$ is presented as a sequence of vector insertions; then, the algorithm provides oracle access to a function which determines, for each point $x_i$, which side of the cut it lies in. The following theorem appears formally in Section~\ref{sec:insertion-only} as Theorem~\ref{thm:insertion-only}. 

\begin{theorem}[Insertion-Only Streaming (Informal)]\label{thm:insertion-intro}
There is an insertion-only streaming algorithm using $\poly(d \log \Delta / \eps)$ space which provides oracle access to a $(1+\eps)$-approximate Euclidean max-cut. 
\end{theorem}

We emphasize that, conceptually, the proof of Theorem~\ref{thm:insertion-intro} is another application of the same framework used for Theorem~\ref{thm:mpc-intro}. We then turn our attention to the dynamic streaming model, where vectors in $[\Delta]^d$ may be both inserted and deleted, and the streaming algorithm must provide oracle access to an approximate max-cut of the set of vectors which remain. Even though~\cite{CJK23}, could not provide access to the cut itself, we show that their main algorithmic primitive, ``geometric sampling,'' can be combined with our framework to provide access to the underlying cut.

\begin{theorem}[Dynamic Streaming (Informal)]\label{thm:dynamic-intro}
There is a dynamic streaming algorithm using $\poly(d\log \Delta /\eps)$ space which provides oracle access to a $(1+\eps)$-approximate Euclidean max-cut.
\end{theorem}

Theorem~\ref{thm:dynamic-intro} is strictly stronger than Theorem~\ref{thm:insertion-intro} and Theorem~\ref{thm:mpc-intro}, since the algorithm from Theorem~\ref{thm:dynamic-intro} may be used (in a black-box fashion) in insertion-only streaming and massively parallel computation. The benefit to differentiating them is that the proof of Theorem~\ref{thm:dynamic-intro} adds several layers of complexity in the analysis. As we describe in more detail, in massively parallel computation and insertion-only streams, we are able to simulate an algorithm where each point from the input dataset behaves independently---this independence is crucial to both Theorem~\ref{thm:mpc-intro} and Theorem~\ref{thm:insertion-intro}. The algorithm for Theorem~\ref{thm:dynamic-intro} will forego the independence of the dataset points, and this necessitates re-analyzing our framework.

\subsection{Technical Overview}\label{sec:tech-overview}

As mentioned above, the streaming algorithm of~\cite{CJK23} builds on a line-of-work on random sampling for max-CSPs~\cite{AVKK03, RV07}. These works (tailored to max-cut) show that in any graph on $n$ nodes, the max-cut value of 
 the subgraph induced by a random sample of $s = O(1/\eps^4)$ nodes (re-scaled by $n^2/s^2$) gives an additive $\pm\eps n^2$ approximation to the max-cut of the graph with high (constant) probability.~\cite{CJK23} show the analogous result for geometric graphs; the $\pm \eps n^2$ approximation becomes $\pm \eps T$, where $T$ is the total sum of edge weights and is a $(1\pm\eps)$-multiplicative approximation since the max-cut is always at least $T/2$. Instead of (uniform) random sampling of nodes as in~\cite{AVKK03,RV07}, one samples a node with probability proportional to the weighted degree (re-scaling is as in ``importance sampling''). Their streaming algorithm takes ``geometric samples''~(developed in~\cite{CJK23} to sample points with probability proportional to the weighted degree, or \emph{weight}), and finds the max-cut on the sample. The problem is the following: a max-cut on a sample does not necessarily extend beyond the sample. In this paper, we provide a method to do just that in streaming and massively parallel models: given ``geometric samples'' and a cut of those samples, provides query access to an approximately optimal cut of the entire dataset. Given a dataset point, we can determine whether it lies inside or outside a cut---the guarantee is that the resulting cut is a $(1+\eps)$-approximate max-cut.  

 \paragraph{A Parallel and Subsampled Greedy Algorithm.} Our starting point is an elegant analysis of the following greedy algorithm of~\cite{MS08} for non-geometric graphs: given a graph $G$ on $n$ nodes, process nodes in a uniformly random order; for the first $t_0 = O(1/\eps^2)$ nodes, try all possible cut assignments, and for a fixed assignment of the first $t_0$ nodes, the remaining nodes are iteratively assigned greedily; a node picks a side maximizing the contribution to already-assigned points. The fact that such extensions achieve an additive $\pm \eps n^2$-approximation in (non-geometric) graphs is already highly non-trivial in classical models of computation; for sublinear models (like massively parallel and streaming algorithms) there are two immediate challenges to overcome:
 \begin{itemize}
     \item \textbf{Challenge \#1: Random Order}. Processing the nodes in a uniformly random order is central to the analysis in~\cite{MS08} and directly at odds with the streaming model, whose inputs are received in potentially adversarial order. 
     \item \textbf{Challenge \#2: Space Complexity.} The assignment of a node depends on the assignment of ``previously-assigned'' nodes. This would seemingly require large space, or access to the input the moment a node is being assigned. In streaming, providing query access after fully processing the input means that neither large space nor access to the input is available. 
 \end{itemize}
 As we describe next, we show how a ``parallel'' and ``subsampled'' algorithm can overcome both challenges in Euclidean instances. The parallel nature of the assignment rule will mean that, irrespective of the order in which points are processed, we are able to simulate a greedy algorithm. Furthermore, a greedy rule for assigning a point will pick the side of the cut which maximizes the sum of edge weights crossing the cut; since we cannot implement this greedy rule exactly, we use a sub-sample to estimate the greedy assignment. However, sub-samples cannot be independently drawn for each point assignment---if each point chose an independent sample, the overall space would still be too large. We proceed as follows: consider a fixed dataset $X = \{ x_1,\dots, x_n\}$ of $n$ points in Euclidean space.
 \begin{itemize}
     \item A (discrete) time axis begins at $t = 0$ toward an end time $t_e$. Each point $x_i$ generates a timeline $\bA_i \in \{0,1\}^{t_e}$: at each time $t = 1,\dots, t_e$, it  ``activates'' by sampling $\bA_{i,t} \sim \Ber(\bw_i^t)$ with probability $\bw_i^t$, which is proportional to its weight.\footnote{Throughout the paper we use bolded variables to denote random variables. As we discuss later, we will have to decrease the weight as time proceeds, and hence the dependence on $t$ on the weight $\bw_i^t$ in Definition~\ref{def:activate}.} A point assigns itself to inside/outside the cut the moment it is first activated (see Definition~\ref{def:activate}). As in~\cite{MS08}, points activated by time $t_0$ try all cut assignments.
     \item Points activated after $t_0$ are assigned greedily, but not by maximizing the contribution to all already-activated points. Alongside a timeline, each $x_i$ samples a mask $\bK_i \in \{0,1\}^{t_e}$: at each time $t$, $x_i$ is ``kept'' by sampling $\bK_{i,t} \sim \Ber(\gamma_t)$. Then, a point $x_i$ activated after $t_0$ is assigned greedily by maximizing a weighted contribution to already-assigned points which were simultaneously ``activated'' and ``kept'' before time $t$, i.e., by considering points $x_j$ where $\bA_{j,\ell} \cdot \bK_{j,\ell} = 1$ for some $\ell < t$ (see Definition~\ref{def:mask}).
 \end{itemize}
 The above is meant to address both challenges simultaneously (see Figure~\ref{fig:timelines-im}). Because each point ``activates'' itself in a parallel manner, the order points are processed is irrelevant---each point behaves independently and checks its timeline and mask to determine whether or not (and when) it was activated and kept. The sum of all weights $\sum_{i=1}^n \bw_i^t$ will always be at most $1$, so we expect one point to activate at each time $t$, and $\gamma_t$ points (which is less than $1$) to simultaneously be ``activated and kept'' at time $t$. Importantly, we only store points which are simultaneously activated and kept (which suffices, since our assignment only considers these points). Roughly speaking, we let $\gamma_t$ be $1$ at $t \leq t_0$ and decay as $1/t$ for $t > t_0$. As we will see, it will suffice for us to set an end time of $t_e = O(n/\eps)$, so only $O(t_0 + \log(n/\eps))$ points are stored in expectation. It remains to show that, even though we modified the assignment, a $(1+\eps)$-approximation is recovered. Theorem~\ref{thm:main-structural} is the main structural theorem concluding this fact; assuming Theorem~\ref{thm:main-structural} (whose proof we overview shortly), the application to MPC and insertion-only streaming setup the data structures to simulate the assignment of Theorem~\ref{thm:main-structural}.

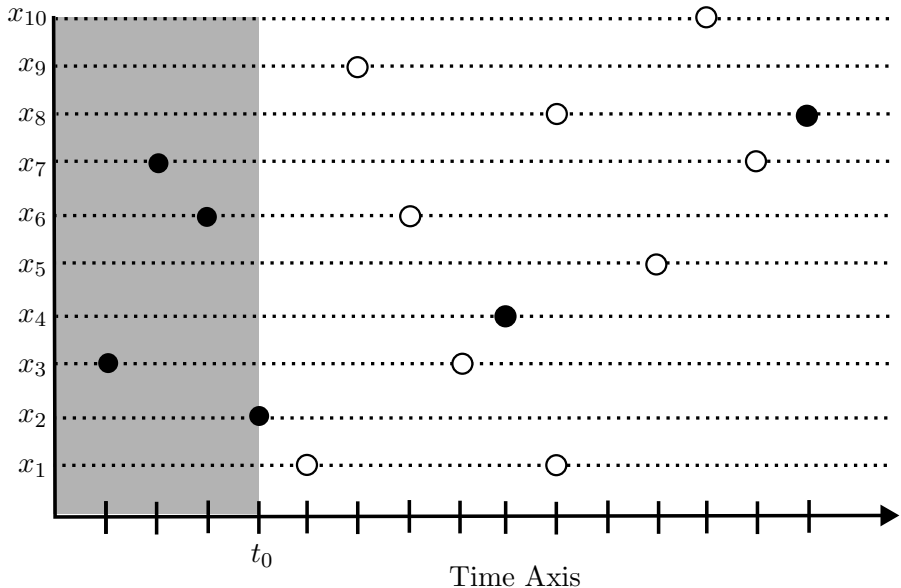
\begin{figure}[h]
    \centering
    \begin{picture}(400, 200)
        \put(0,0){\includegraphics[width=0.8\linewidth]{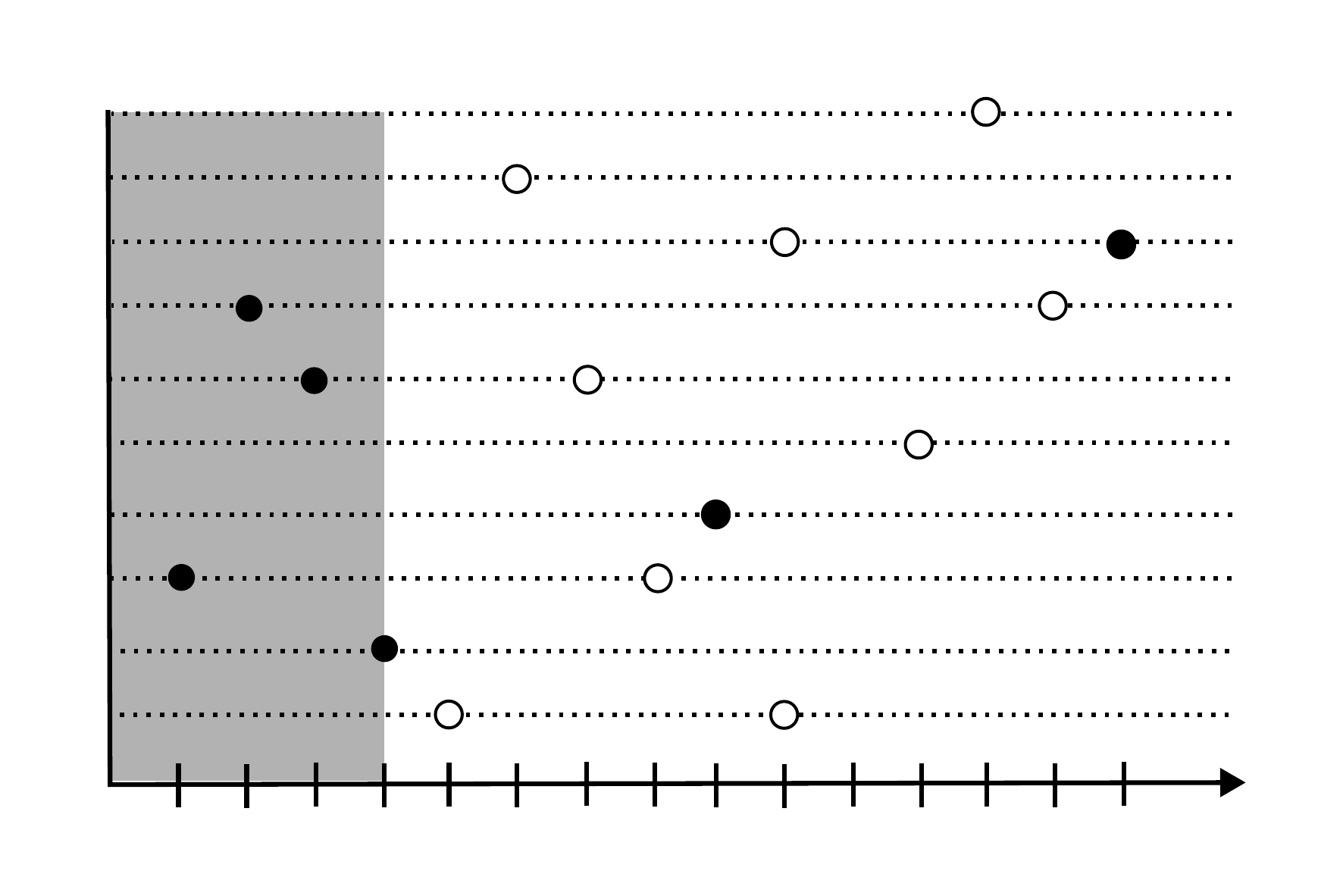}}
        \put(180, 5){Time Axis}
        \put(105, 14){$t_0$}
        \put(17, 48){$x_1$}
        \put(17, 67){$x_2$}
        \put(17, 86){$x_3$}
        \put(17, 105){$x_4$}
        \put(17, 124){$x_5$}
        \put(17, 143){$x_6$}
        \put(17, 162){$x_7$}
        \put(17, 181){$x_8$}
        \put(17, 200){$x_9$}
        \put(13, 219){$x_{10}$}
    \end{picture}
    \caption{Representation of the timelines for each point. The time axis is on the bottom, each each horizontal dotted line corresponds to the timeline of each point $x_1, \dots, x_{10}$. Time $t_0$ appears in the timeline with times up to $t_0$ in gray. For a point $x_i$, a dot on the timeline corresponds to an activation, and a solid dot to an activation which is also kept. Up to time $t_0$, every point which is activated is also kept, but after that, points which are activated at time $t$ are kept with probability $\gamma_t$. A point is assigned greedily whenever it is first activated; the assignment depends on the edge weights to the previously ``activated'' and ``kept'' points (i.e., those with solid black dots before the time of first activation). Hence, it suffices for algorithm to store points which are simultaneously activated and kept (i.e., those with solid dots).}\label{fig:timelines-im}
\end{figure}

 \paragraph{Application to MPC.} The application to the MPC model is in Section~\ref{sec:mpc}, where we show that one can simulate the parallel and subsampled greedy algorithm for Euclidean max-cut. In MPC, points are distributed across machines. First, each point computes its own weight (as well as the sum-of-weights) in $O(1)$ rounds using $\ell_1(\ell_p)$-sketches (Lemma~\ref{lem:mpc-weights}). Since each point computes its own weight, it can generate a timeline and mask according to Theorem~\ref{thm:main-structural}. In particular, each point $x_i$ can generate the horizontal dotted line for $x_i$ in Figure~\ref{fig:timelines-im} and know which times have activated/kept (i.e., the solid/non-solid dots). If there exists a time $t \leq t_e$ where a point is simultaneously activated and kept, it is communicated across all machines along with its weight (recall, there will be few of these). Letting $\gamma_t=1$ for $t \leq t_0$, points activated before $t_0$ are always kept. By taking a few additional samples, we give a procedure for the algorithm to try all possible assignments of the points activated by $t_0$, and to agree on which assignment to use (Lemma~\ref{lem:seed-check}). This procedure for agreeing on an assignment to points activated by $t_0$ is shared in all of our algorithms. Once an assignment of the points activated by $t_0$ is agreed upon, the remaining points assign themselves by maximizing a weighted contribution (as in importance sampling) to the activated and kept points (which were communicated to all machines).

 \paragraph{Application to Insertion-Only Streaming.} The application to insertion-only streaming, Theorem~\ref{thm:insertion-only}, follows a conceptually similar framework. We maintain a sketch which allows a point to compute its weight (and sum of weights) with respect to points inserted so far (Lemma~\ref{lem:computing-weights}). With this weight, a point generates its own timeline and mask, and checks whether there exists a time $t$ such that it is simultaneously activated and kept; if so, it remembers the point and the weight, and if not, then it ``forgets'' the point (Lemma~\ref{lem:timeline-mask}). Once more, one should keep Figure~\ref{fig:timelines-im} in mind; the goal is for each point $x_i$ to know its own dotted horizontal line, so it knows when it was assigned, and whether other point's assignments depend on it (in case it is both activated and kept). An important caveat is that weights change: upon inserting a new point $x$, the weight of a prior point $y$ changes. To address this caveat, we maintain timelines with respect to minimum weights so far, and thus, once we ``forget'' a point, it will not be simultaneously activated and kept when the weight changes in the future (because it always decreases). It turns out, the minimum weights under insertions cannot be smaller than a constant factor of the final weight (Claim~\ref{cl:final-min}). At the end of the stream, we update the timeline of point remembered to a constant-fraction of the final weight (which still gives the desired $(1+\eps)$-approximation guarantee). In order to provide query access to the cut assignment of a point, we determine its (final) weight, and check its activation time in a timeline with respect to the constant-fraction of the final weight. We use Nisan's pseudorandom generator, so that the algorithm can reconstruct the timeline and perfectly simulate the parallel and subsampled greedy algorithm.

 \paragraph{Proof of the Theorem~\ref{thm:main-structural}: Parallel and Subsampled Greedy Algorithm.} For the application to MPC and insertion-only streaming, it remains to analyze the result of the greedy assignment produced by the above process. Our analysis will use a technique of~\cite{MS08} in defining a ``fictitious cut'' and adapting it to our parallel and subsampled process. Recall, the major differences are that (i) there is no underlying random permutation, but rather the $n$ timelines generated by independently sampling activations from $\Ber(\bw_i^t)$ with different weights, and (ii) one cannot execute the greedy rule exactly, and one cannot even use independent sampling to estimate the greedy process; rather, there is one subsampling done (for the entire process) which is dictated by the mask $\bK$. 
 
 The analysis begins in Section~\ref{sec:greedy-process}, where we describe a (random) process, \textsc{GreedyProcess}, which iterates through $t = 1, \dots, t_e$ and generates the timelines and masks for each point, as well as produces the parallel and subsampled greedy assignment from Theorem~\ref{thm:main-structural}. We encode (partial) cuts of $n$ points as $[n] \times \{0,1\}$ matrices $z$ with entries in $[0,1]$; the $i$-th row $z_i$ is $(1,0)$ in case it is assigned to the $0$-side, $(0, 1)$ in case it is assigned to the $1$-side, or $(0, 0)$ if unassigned; we let $f(z)$ measure the sum of uncut edges and approximately minimize $f$.
As in~\cite{MS08}, we track a ``fictitious cut'', is a  $[n]\times \{0,1\}$ matrix, $\hat{\bz}^t$ for $t = 0, \dots t_e$, which evolves over time (Definition~\ref{def:fictitious}). The fictitious cut will be defined so as the satisfy the following guarantees:
\begin{enumerate}
    \item The fictitious cut $\hat{\bz}^0$ begins by encoding the (optimal) max-cut $z^*$, and as \textsc{GreedyProcess} executes, becomes the output cut $\hat{\bz}^{t_e}$. Each $x_i$ and time $t$ has $\hat{\bz}^{t}_i$ set (and forever fixed) the first time $\bA_{i,t} \sim \Ber(\bw_i^t)$ is set to $1$ (i.e., when it first activates). Hence, it suffices to upper bound $\Ex[f(\hat{\bz}^t) - f(\hat{\bz}^{t-1})]$ for each $t$, since the sum telescopes to $\Ex[f(\hat{\bz}^{t_e})] - f(z^*)$.
    \item When a point $x_i$ is first activated at time $t > t_0$, we set $\hat{\bz}_i^t$ (and it is forever fixed) as follows. First, we evaluate
    \begin{align} 
    \sum_{j=1}^n d(x_i, x_j) \cdot \left(\frac{1}{t-1} \sum_{\ell = 1}^{t-1} \dfrac{\bA_{j,\ell} \cdot \bK_{j,\ell}}{\bw_j^{\ell} \cdot \gamma_{\ell}} \right) \cdot \hat{\bz}^{t-1}_j, \label{eq:contrib-funct}
    \end{align}
    which gives the two numbers (recall $\hat{\bz}^{t-1}_j$ is the $n$-th row consisting of two numbers) and are interpreted as estimated contribution to either side of the cut. We set $\hat{\bz}_i^t$ greedily to $(0,1)$ if the first entry is larger, and to $(1,0)$ if the second is larger. 
    \item If $x_i$ is not activated at time $t > t_0$, we also update $\hat{\bz}_i^t$. This update is solely for the analysis for a reason we explain shortly. It also appears in~\cite{MS08} in a simpler form (because points here are weighted), and we will see the weights add a few technical complications.
\end{enumerate}

Since our goal is to upper bound $\Ex[f(\hat{\bz}^t) - f(\hat{\bz}^{t-1})]$ for each $t = 1, \dots, t_e$, a natural approach is to fix an execution up to time $t-1$, and then to randomize (and take the expectation) over the randomness in $\bA_{i, t}$ and $\bK_{i, t}$ for all $i$. By a simple manipulation, the change $f(\hat{\bz}^t) - f(\hat{\bz}^{t-1})$ can be expressed as a sum of two types of terms: type-A terms sum over $i$ of $(\hat{\bz}_i^t - \hat{\bz}_i^{t-1}) \sum_{j=1}^{n} d(x_i, x_j) \hat{\bz}_j^{t-1}$ which measures, roughly speaking, the total weight change when a reassignment of $x_i$ incurs, and type-B terms sum over all $i, j$ of $d(x_i, x_j) \cdot (\hat{\bz}_i^{t} - \hat{\bz}_i^{t-1})(\hat{\bz}_j^t - \hat{\bz}_j^{t+1})$ which accounts for points $x_i$ and $x_j$ which change simultaneously (expression starting in (\ref{eq:linear-term})). Type-B terms are easy to handle; roughly speaking, a single term $i, j$ of $(\hat{\bz}_i^t - \hat{\bz}_i^{t-1})(\hat{\bz}_i^t - \hat{\bz}_i^{t-1})$ is large when $x_i$ and $x_j$ activate simultaneously, which occurs with probability $\bw_i^t \bw_j^t$ and is small enough to upper bound the expectation (Lemma~\ref{lem:non-linear}). 

The challenging terms are of type-A, and we informally refer to $(\hat{\bz}_i^t - \hat{\bz}_i^{t-1})$ as the $i$-th change, and $\sum_{j=1}^n d(x_i, x_j) \hat{\bz}_j^{t-1}$ as the $i$-th contribution. When $x_i$ is activated, the new value of \smash{$\hat{\bz}_i^t$} is set greedily according to (\ref{eq:contrib-funct}). Importantly, in order to evaluate (\ref{eq:contrib-funct}), it suffices to consider those $j \in [n]$ where $\bA_{j,\ell} \cdot \bK_{j,\ell} = 1$ for $\ell < t$ (i.e., those which have been simultaneously activated and kept before time $t$), and this is the reason we only keep activated and kept points. On first glance, (\ref{eq:contrib-funct}) appears to be an unbiased estimator for the $i$-th contribution (recall $\bA_{j,\ell} \sim \Ber(\bw_j^{\ell})$ and $\bK_{j,\ell} \sim \Ber(\gamma_{\ell})$). If so, a greedy setting of $\hat{\bz}_i^t$ with respect to the $i$-th contribution guarantees type-A terms are all negative (Lemma~\ref{lem:greedy-choice-opt}, whose analogous statement appears in~\cite{MS08}), and gives the required upper bound. 

The important caveat is that (\ref{eq:contrib-funct}) may not be an unbiased estimator for the $i$-th contribution because the fictitious cut $\hat{\bz}^{t-1}_j$ and the timelines and masks are not independent. 
The fix involves updating \smash{$\hat{\bz}_j^{\ell}$} when $j$ is not activated at time $\ell$---this way, one can ensure (\ref{eq:contrib-funct}) becomes an unbiased estimator (the second case of (\ref{eq:shadow-def}) in Definition~\ref{def:fictitious}). This also appears in~\cite{MS08}, albeit a different update which depends on the weight.\footnote{In the weighted setting, a technical issue arises for times when $1/\ell$ is larger than the weight. In these cases, $\hat{\bz}_i^{t-1}$ may have entries which grow beyond $[0,1]$, which is a problem for Lemma~\ref{lem:greedy-choice-opt}. Thus, we enforce $\bw_i^{\ell} = \min\{ w_i, 1/\ell\}$ when $x_i$ has not yet activated.} In fact, a stronger statement holds showing that for any fixed execution of the first $t-1$ steps, the expected error of (\ref{eq:contrib-funct}) at time $t$, denoted $\berr_i^{t+1}$, is smaller than the error of (\ref{eq:contrib-funct}) at time $t-1$, denoted $\berr_i^t$ (Lemma~\ref{lem:martingale}), which also occurs in~\cite{MS08} as the ``martingale argument''. Thus, we account for type-A terms by upper bounding the expected deviation $|\berr_i^{t+1}|$ of (\ref{eq:contrib-funct}) in Lemma~\ref{lem:error-bound}, which is where the fact $\gamma_t$ cannot decay too rapidly comes in. 

\paragraph{Incorporating the Mask.} If $\gamma_t$ is always one and there is no sub-sampling (as in~\cite{MS08}), the additional error $|\berr_i^{t+1}|$ at each time $t$ becomes $T / t^{3/2}$ (in (\ref{eq:deviation-exp}), set $\gamma^{\ell} = 1$ and the expression involving distances simplifies to $O(T)$ by our setting of weights,  Definition~\ref{def:metric-compat}). Once we sum over the $i$-th change, $\hat{\bz}^t_i - \hat{\bz}^{t-1}_i$ (Lemma~\ref{lem:non-linear}) and sum over all times larger than $t_0$, the total error becomes the sum over all $t > t_0$ of $T / t^{3/2}$ which is on the order of $T / \sqrt{t_0}$ (before $t_0$, we try all possible assignments and incur no error). Notice that this is exactly analogous to~\cite{MS08}---the parameter $T$, which is the sum of edge weights corresponds to $O(n^2)$ in dense graphs, and $t_0 = O(1/\eps^2)$ would work. From (\ref{eq:deviation-exp}), we can see that a smaller setting of $\gamma_{t}$ would increase the variance (and hence the error incurred), but one can still control the error as long as $\gamma_t$ does not decay too rapidly. For $\gamma_{\ell} = \gamma / \ell$ in (\ref{eq:deviation-exp}), we obtain an error per time step of $T (\sqrt{\gamma} / t^2 + 1 / (t\sqrt{\gamma}))$ (see Lemma~\ref{lem:error-bound}). This final expression is also $\eps T$ when $\gamma = \log^2(t_e) / \eps^2$ and $t_0 \geq \sqrt{\gamma} / \eps = \log(t_e) / \eps^2$. 

\paragraph{Dynamic Streaming Algorithm.} So far, our analysis proceeded by generating the timelines per timestep; in particular, Figure~\ref{fig:timelines-im} considered the horizontal dotted lines: each point $x_i$ independently samples the horizontal row $x_i$ (according to the weight), and the space complexity is dominated by the number of simultaneously activated and kept points (the solid dots in Figure~\ref{fig:timelines-im}). In a dynamic streaming algorithm, the underlying set of vectors (and hence, which  horizontal dotted lines will appear) is only determined at the end of the stream. Instead, we focus on generating the columns of Figure~\ref{fig:timelines-im} which will hold solid black dots (see Figure~\ref{fig:timelines-dyn}). Specifically, we will use the ``geometric sampling'' technique of~\cite{CJK23}. A geometric sampling sketch is a small-space data structure which processes a dynamic dataset of vectors $X \subset [\Delta]^d$ and can produce a sample $\bx \sim X$ with probability proportional to the weighted degree of $\bx$. In addition, it returns an accurate estimate of the probability $\bp^*$ that $\bx$ was sampled (for ``importance sampling''). We use the following approach (Definition~\ref{def:corr-timeline-mask}):
\begin{itemize}
    \item We sample one mask $\bK \in \{0,1\}^{t_e}$ to determine which times will have points simultaneously activated and kept (instead of independent masks for each point). As in Definition~\ref{def:mask}, $\bK_t \sim \Ber(\gamma_t)$, where $\gamma_t = 1$ for all $t \leq t_0$, and $\gamma_t = \gamma / t$ for $t > t_0$. 
    \item If $\bK_t = 0$, no point is activated and kept at that time and we set $\bA_{i,t} \cdot \bK_t = 0$ for all $i$. If $\bK_t = 1$, we use a geometric sampling sketch to sample $\bx_i \sim X$, and set $\bA_{i,t} \cdot \bK_t = 1$ and for $j \neq i$, $\bA_{j,t} \cdot \bK_t = 0$.
\end{itemize}
Each (individual) entry $\bA_{i,t} \cdot \bK_t$ is distributed as desired (although no longer independent), $\bA_{i,t} \sim \Ber(\bw_i^{t})$ and $\bK_t \sim \Ber(\gamma_t)$. The points which are simultaneously ``activated'' and ``kept'' are recovered by the algorithm (since these are generated by the geometric samples). Because the assignment of a point $x_j$ also depends on the first time it is activated, i.e., the first $t_j$ where $\bA_{j,t_j} = 1$, we maintain a sketch to provide query access to the weight of a point---this allows a point to generate the remaining part of the timeline (the horizontal dotted lines in Figure~\ref{fig:timelines-dyn}). With that, all ingredients are set (Subsection~\ref{sec:dynamic-stream-alg}): while processing the stream, we prepare $\poly(d\log \Delta/\eps)$ geometric sampling sketches, and a sketch to later query weights; once the stream is processed, we generate a single mask $\bK \sim \calK(t_0, \gamma)$ and whenever $\bK_t = 1$, we use a geometric sample to determine which point is simultaneously ``activated'' and ``kept'' at time $t$. On a query $x_j$, we use the sketch to estimate the weight $\bw_j$ and generate the remainder of the timeline $\bA_{j}$; this specifies the activation time $\bt_j$, which is used to simulate the ``parallel'' and ``subsampled'' greedy assignment. The final piece is ensuring the cut output is still a $(1+\eps)$-approximation (Theorem~\ref{thm:main-structural-dynamic}), even though timelines are no longer independent, and the probabilities that geometric samples output are estimates, as opposed to the exact probabilities used. The proof of this final theorem appears in Section~\ref{sec:dynamic-greedy-process}.

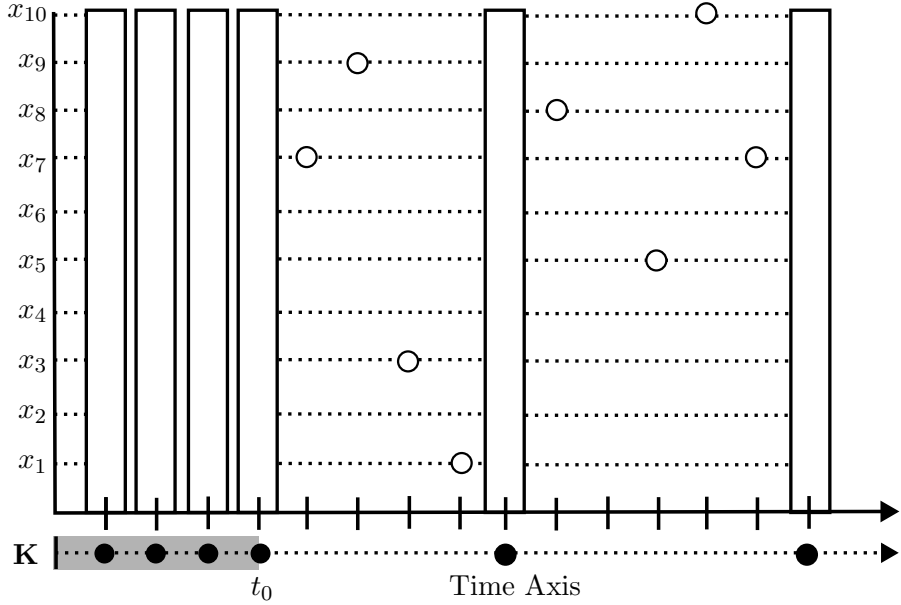
\begin{figure}[h]
    \centering
    \begin{picture}(400, 200)
        \put(0,0){\includegraphics[width=0.8\linewidth]{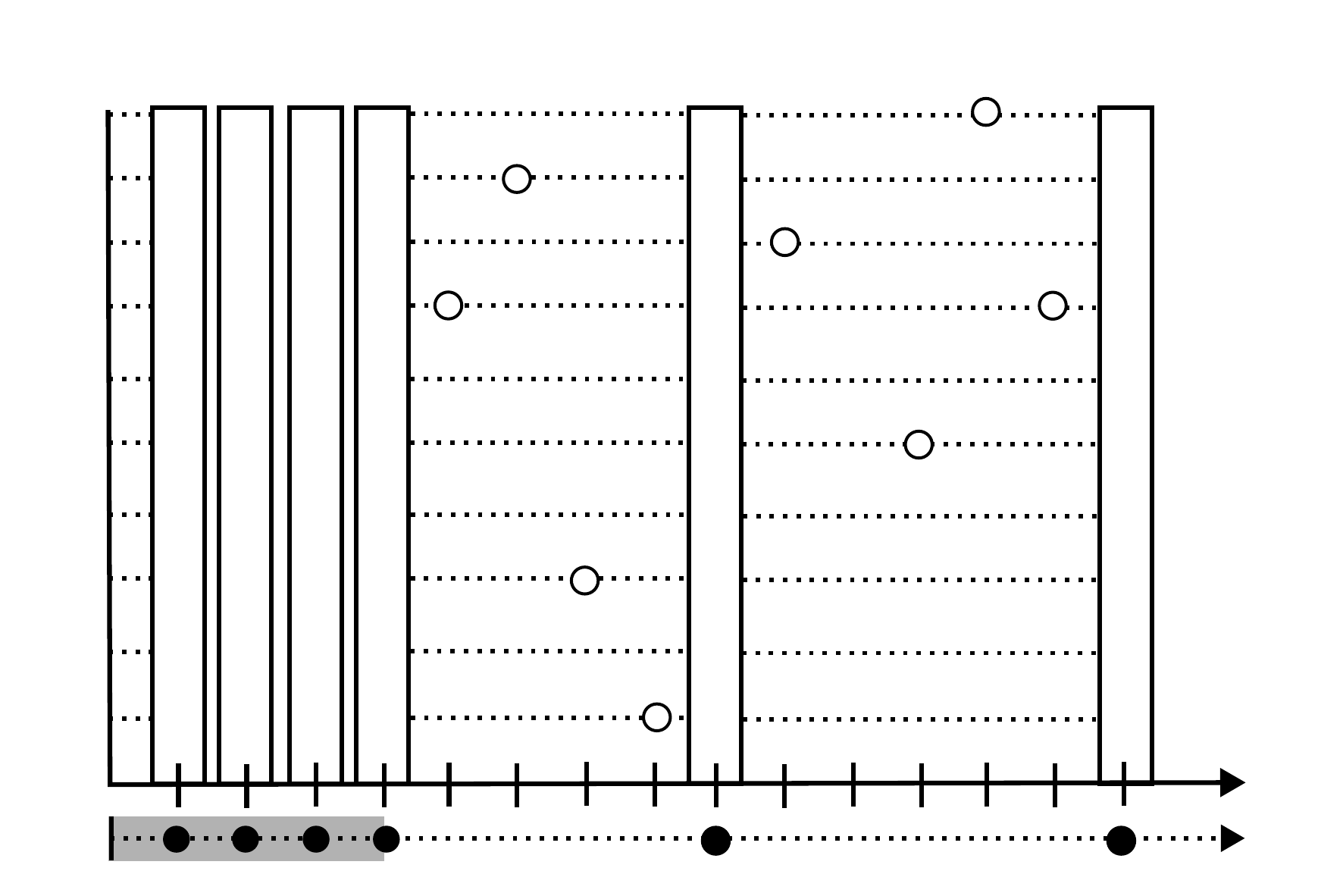}}
        \put(180, 0){Time Axis}
        \put(15, 12){$\bK$}
        \put(105, 0){$t_0$}
        \put(17, 48){$x_1$}
        \put(17, 67){$x_2$}
        \put(17, 86){$x_3$}
        \put(17, 105){$x_4$}
        \put(17, 124){$x_5$}
        \put(17, 143){$x_6$}
        \put(17, 162){$x_7$}
        \put(17, 181){$x_8$}
        \put(17, 200){$x_9$}
        \put(13, 219){$x_{10}$}
    \end{picture}
    \caption{Representation of the timelines for the dynamic streaming algorithm. Right above the time axis on the bottom, we sample one mask $\bK$ which determines which times will contain points which are simultaneously activated and kept---these are the solid black dots on $\bK$. Times before $t_0$ activate and keep a single point, and after $t_0$, points are kept with probability $\gamma_t$. Whenever $\bK_t = 1$ (with a solid black dot), the rectangle above it represents a geometric sampling sketch used to determine which point is activated and kept. After geometric sampling sketches are generated, the horizontal dotted lines represents the (remaining) timelines to be generated for each point.}\label{fig:timelines-dyn}
\end{figure}

\subsection{Related Work}

In relation to our streaming algorithm, this paper pursues a line of work initiated by~\cite{I04} in designing dynamic streaming algorithms for high-dimensional geometric problems. These have been extensively studied especially in the last few years. These include, in addition to Euclidean max-cut~\cite{CJK23}, the Earth Mover's distance~\cite{AIK08, ABIW09, CJLW21}, the Euclidean minimum spanning tree~\cite{CCJLW23}, facility location~\cite{CJKVY22}, and clustering~\cite{BFLSY17}. More broadly, the development of smaller coresets have also given rise significantly more efficient streaming algorithms for geometric problems (see, for example, \cite{WY22, CW23, BCPSS24} and references therein). For massively parallel computation, a line of work initiated by~\cite{ANOY14} on geometric problems has resulted in algorithms for Euclidean MST~\cite{JMNZ24}, Metric MST~\cite{ABJLMZ25}, Facility Location and Clustering~\cite{CGJKV24}, as well as probabilistic tree embeddings~\cite{AAHKZ23}.

The greedy algorithm of~\cite{MS08} was initially developed for dense max-cut, and found other applications in Gale-Berlekamp games \cite{KS09}, and graph-based semi-supervised learning~\cite{WJC13}. In the context of sublinear algorithms, \cite{Y14} has also used random sampling to speed-up the execution of~\cite{MS08} for dense CSPs; the algorithm and analysis differs from this work (the algorithm is still sequential, and sampling done independently for each step). 


%% file: non-streaming.tex

\section{Parallel and Subsampled Greedy Assignment}\label{sec:structural-1}

We will first analyze the following ``parallel'' and ``subsampled'' greedy assignment and phrase the main structural theorem. The analysis is based on the greedy algorithm of~\cite{MS08} with two important differences; we can both ``parallelize'' and ``subsample'' the assignment rule used in~\cite{MS08} while still controlling the quality of the cut. This will allow us to maintain all the necessary information needed to simulate the greedy algorithm in various models of sublinear algorithms. An important aspect of the subsequent theorem is being able to sample points from the metric with certain probabilities which are compatible with that metric (essentially, they will need to result in certain low-variance estimates, as in~\cite{FK01}). Toward that, we first define the following two definitions.
\begin{definition}[Metric-Compatible Weights]\label{def:metric-compat}
For a metric space $(X = \{ x_1,\dots, x_n\}, d)$, we say that a set of weights $w = (w_1,\dots, w_n) \in (0, 1/2]^n$ is $\lambda$-compatible if, for every $i, j \in [n]$, 
\begin{align*}
\frac{d(x_i, x_j)}{w_j} \leq \lambda \sum_{k=1}^n d(x_i, x_k). 
\end{align*}
\end{definition}

\begin{definition}[Activation Timeline]\label{def:activate}
Given an end time $t_e \in \N$ and a weight $w \in (0, 1/2]$, an activation timeline $\calT(w)$ is a distribution supported on strings $\{0,1\}^{t_e}$ sampled according to the following procedure:
\begin{itemize}
\item In order to generate $\ba \sim \calT(w)$, we iterate through all $t = 1, \dots, t_e$.
\item For each $t$, we generate $\ba_t \in \{0,1\}$ by first setting 
\[ \bw^t \eqdef \left\{ \begin{array}{cc} \min\{ w, 1/t\} & \forall \ell < t, \ba_{\ell} = 0 \\
								w & \text{otherwise} \end{array} \right. ,\]
and then sampling $\ba_t \sim \Ber(\bw^t)$.
\end{itemize}
For a draw $\ba \sim \calT(w)$, the smallest $\bt$ with $\ba_{\bt} = 1$ is denoted the activation time.
\end{definition}

We will use Definitions~\ref{def:metric-compat} and~\ref{def:activate} as a means of ``parallelizing'' the cut assignment rule. In particular, we will consider a metric $(X, d)$ consisting of dataset points, and every point $x_i$ in the metric $(X, d)$ will be able to compute a metric-compatible weight $w_i$. Since $t_e$ will be a fixed parameter, $x_i$ can sample an activation timeline $\bA_i \sim \calT(w_i)$, and the activation timeline will determine the ``time'' it is assigned to either side of the cut---namely, this will be the activation time of $\bA_i$. If $\bA_{i, t}$ is never set to $1$, then the assignment will be a default value, which we will show occurs with small enough probability so as to not affect the final quality of the cut. At the activation time $\bt_i$ for $x_i$, it will determine the assignment by considering other points in the metric which had already been activated by time $\bt_i$. The last ingredient for our algorithm will be the way to ``subsample'' the assignment rule so that, even though many points in the metric may be activated, the assignment rule depends only on a few of them.

\begin{definition}[Masking]\label{def:mask}
For an end time $t_e \in \N$, as well as a keeping time $t_0\in \N$ and a parameter $\gamma \geq 1$, the mask $\calK(t_0, \gamma)$ is a distribution supported on strings $\{0,1\}^{t_e}$ which is sampled according to the following procedure:
\begin{itemize}
\item In order to generate $\bk \sim \calK(t_0, \gamma)$, we iterate through all $t = 1, \dots, t_e$.
\item For each $t$, we generate $\bk_t \in \{0,1\}$ by first setting
\[ \gamma^t \eqdef \left\{ \begin{array}{cc} 1 & t \leq t_0 \\
							\min\{ \gamma / t, 1 \} & t > t_0 \end{array} \right. ,\]
and we sample $\bk_t \sim \Ber(\gamma^t)$.  
\end{itemize}
\end{definition}

As sketched above, each point $x_i$ will compute a $\lambda$-metric-compatible weight $w_i$ and generate an activation timeline $\bA_i \sim \calT(w_i)$. In addition, each point will also sample a mask $\bK_i \sim \calK(t_0, \gamma)$. We will incorporate the mask in the following way: a point $x_i$ is assigned to a side of the cut at its activation time $\bt_i$ by considering other points in the metric which have been activated before time $\bt_i$, and in addition, were also ``kept'' according to the mask. Namely, the assignment of $x_i$ will depend on point $x_j$ if there exists $\ell < \bt_i$ where $\bA_{j, \ell} \cdot \bK_{j,\ell} = 1$. As a result, even though many points $x_j$ will be activated by time $\bt_i$, only a few will be activated and kept for the same time $\ell < \bt_i$. There is one final component needed in order to describe the assignment rule.

\begin{definition}[Seed]\label{def:seed}
For the $n$ timelines $\bA_1, \dots, \bA_n$ sampled from $\calT(w_1),\dots,\calT(w_n)$, we let:
\begin{itemize}
\item The seed length $\bm \in \N$ be given by $\bm = |\bS^{t_0}|$, where $\bS^{t_0} \eqdef \left\{ j \in [n] : \exists \ell \leq t_0 \text{ where } \bA_{j, \ell} = 1 \right\}.$
\item The seed $\sigma \in \{0,1\}^{\bm}$ assigns the points $x_j$ for $j \in \bS^{t_0}$ to the $0$-side or the $1$-side of the cut. It does this by:
\begin{itemize}
\item Considering a lexicographic order $\pi$ on the set of points $\{ x_j : j \in \bS^{t_0}\}$,
\item Letting $x_j$ be assigned to the $\sigma_{\pi(x_j)}$-side of the cut.
\end{itemize}
\end{itemize}
\end{definition}

\subsection{The $\textsc{Assign}_{\sigma,\bP}(\cdot,\cdot)$ Sub-routine}\label{sec:online-assignment}

As in~\cite{MS08}, we encode (partial) cuts by $n \times \{0,1\}$ Boolean matrices $z \in \{0,1\}^{n \times \{0,1\}}$, where the $i$-th row of the matrix will be $(0,0)$, in the case that point $i$ is currently not assigned, $(1, 0)$, in the case that it is assigned to the $0$-side of the cut, or $(0, 1)$, in the case that it is assigned to the $1$-side of the cut. The $n$ rows of a matrix $z$ encoding a partial cut are naturally associated with the $n$ points of the metric. Toward that end, we describe below an assignment procedure which, under our setup (specified below) and a small amount of additional information, can consider a point $x_i$ and its activation time $t_i$, and output $z_i \in \{(0,1), (1, 0)\}$. The necessary global information is captured in the following, ``timeline-mask'' summary, which will maintain the necessary information from a collection of activation timelines and masks $(\bA_j, \bK_j)$ for $j \in [n]$ needed to perform an assignment.

\begin{definition}[Timeline-Mask Summary]\label{def:timeline-mask-sum}
Consider a metric $(X =\{x_1,\dots, x_n \}, d)$, a set of weights $w_1,\dots, w_n \in [0, 1/2)$, as well as parameters $t_e, t_0\in \N$ and $\gamma \geq 1$.
\begin{itemize}
\item Sample activation timelines and masks $(\bA_j, \bK_j)$ where $\bA_j \sim \calT(w_j)$ and $\bK_j \sim \calK(t_0, \gamma)$ for every $j \in [n]$.
\item The summary of $(\bA_1, \bK_1), \dots, (\bA_n, \bK_n)$ is given by the set of tuples
\begin{align*}
\bP &= \textsc{Summ}(\bA_j, \bK_j ; j \in [n]) \\
	&\eqdef \left\{ (x_j, \ell; w_j, t_j) : \bA_{j,\ell} \cdot \bK_{j, \ell} = 1 \text{ and } t_j \text{ is the activation time of $\bA_j$}\right\}.
 \end{align*}
\end{itemize}
\end{definition}

We now turn to describing the main online assignment sub-routine; this is the algorithm which decides given a point (as well as access to some additional information), which side of the cut the point lies in. We assume that there is an underlying metric $(X = \{ x_1,\dots, x_n \}, d)$, and that the timelines and masks $(\bA_i, \bK_i)$ for the metric have already been sampled, where $\bA_i \sim \calT(w_i)$ and $\bK_i \sim \calK(t_0, \gamma)$. The online assignment rule relies on the following ``global information:''
\begin{itemize}
\item A timeline-mask summary $\bP = \textsc{Summ}(\bA_j, \bK_j : j \in [n])$, as specified in Definition~\ref{def:timeline-mask-sum}, which summarizes information of the activation timelines and masks $(\bA_i, \bK_i)$ succinctly. In particular, we will show that $|\bP|$ is small, and hence the summary $\bP$ may be easily stored and communicated.
\item A seed $\sigma \in \{0,1\}^{\bm}$ for the draw $\bA_1,\dots, \bA_n$, as specified in Definition~\ref{def:seed}. We note that the length of the seed, $\bm$, is of size $|\bS^{t_0}|$, corresponding to points $x_i$ where $\bA_{i, \ell} = 1$ for some $\ell \leq t_0$. Since $\bK_{i,\ell} = 1$ for all $\ell \leq t_0$, the points participating in $\bS^{t_0}$ are stored in $\bP$; so $\bm = |\bS^{t_0}| \leq |\bP|$, and $\bS^{t_0}$ may be re-constructed from $\bP$. 
\end{itemize}

\begin{figure}
\begin{framed}
\textbf{Assignment Rule} \textsc{Assign}$_{\sigma,\bP}(x_i, t_i)$. The activation timelines and masks $(\bA_j, \bK_j)$ for all $j \in [n]$ have been sampled, and $\bP$ is a timeline-mask summary for the activation timelines and masks, and that $\sigma$ is a seed for $\bA_1,\dots, \bA_n$. Since the timeline-mask summary $\bP$ is completely determined by $(\bA_j, \bK_j)$ for $j \in [n]$, we refer to the rule as $\textsc{Assign}_{\sigma}(x_i, t_i)$ for convenience. \\

\textbf{Input}: The point $x_i$ to be assigned and the activation time $t_i$ for the timeline $\bA_i$.  

\textbf{Inductive Hypothesis}: We assume the algorithm has access to a partial cut $\bz = z(\sigma, \bP, t_i) \in \{0,1\}^{[n] \times \{0,1\}}$ which assigns every point $x_j$ which participates in a tuple of $\bP$ and whose activation time $t_j$ is smaller than $t_i$.

\textbf{Output}: The assignment, either $(1, 0)$ or  $(0,1)$ corresponding to the 0-side or 1-side of the cut, for the $i$-th point $x_i$. 

\begin{itemize}
\item \textbf{Base Case}: If $t_i \leq t_0$, the point $x_i$ has index $i \in \bS^{t_0}$.\footnotemark~Note that the set of points $\{ x_j : j\in \bS^{t_0}\}$ can be re-constructed from $\bP$, and hence, the index $\pi(x_i) \in [\bm]$ can be determined. Output $(1, 0)$ if $\sigma_{\pi(x_i)} = 0$ and $(0,1)$ if $\sigma_{\pi(x_i)}= 1$. 
\item \textbf{Incorporating Estimates}: Otherwise, $t_i > t_0$ and we proceed as follows:
\begin{enumerate}
\item As per Definition~\ref{def:activate} for every $(x_j, \ell; w_j, t_j) \in \bP$, we let $\bw^{\ell}_j = \min\{ w_j, 1/\ell \}$ if $\ell = t_j$ and $w_j$ if $\ell \geq t_j$. Furthermore, as per Definition~\ref{def:mask}, we let $\gamma^{\ell} = \min\{ \gamma / \ell, 1\}$. We consider the subset of pairs 
\[ \bP_{t_i} = \left\{ (x_j, \ell) : (x_j, \ell; w_j, t_j) \in \bP \text{ and } \ell < t_{i}\right\} \]
which have been assigned by the inductive hypothesis (by recursively executing $\textsc{Assign}_{\sigma, \bP}(x_j, t_j)$). 
\item We compute the estimated contributions, from the $0$-side and $1$-side of the cut, at the time $t_i$:
\[ \bC_{0} = \sum_{(x_j, \ell) \in \bP_{t_i}} \dfrac{d(x_i, x_j)}{\bw_{j}^{\ell} \cdot \gamma^{\ell}} \cdot \bz_{j,0} \qquad \text{and}\qquad \bC_{1} = \sum_{(x_j,\ell) \in \bP_{t_i}} \dfrac{d(x_i, x_j)}{\bw_j^{\ell} \cdot \gamma^{\ell}} \cdot \bz_{j,1}. \]
Output $(0, 1)$ if $\bC_{0} > \bC_{1}$ (to assign $x_i$ to the $1$-side of the cut), and $(1, 0)$ otherwise (to assign $x_i$ to the $0$-side of the cut). 
\end{enumerate}
\end{itemize}
\end{framed}
\caption{The assignment procedure $\textsc{Assign}_{\sigma}(\cdot,\cdot)$} \label{fig:assign}
\end{figure}
\footnotetext{There is a minor ambiguity here as to whether the algorithm which receives the point $x_i$ also knows the index $i$. On the first reading, it is helpful to consider an algorithm which knows the index. This makes it easier notationally to index assignments of a point $x_i$ by the $i$-th entry of $\bz$. However, we note that this is only for notational convenience, as we could store the partial cut $\bz$ as a function from a subset of the points $X$ to $\{(1,0), (0,1)\}$. We keep the indices $i$, as this helps connect the assignment procedure to the analysis, however, in the algorithm any uses of the indices may be replaced by indexing according to the point.}

The assignment rule $\textsc{Assign}_{\sigma,\bP}(x_i, t_i)$ appears in Figure~\ref{fig:assign}. It takes in a point $x_i$ and its activation time $t_i$, as well as access to the seed $\sigma$ and the timeline-mask summary $\bP$. We note that the algorithm assumes, by inductive hypothesis, that all points $x_j$ participating in tuples of $\bP$ with activation times before $t_i$ have been assigned. This is done by a recursive application of $\textsc{Assign}_{\sigma, \bP}(x_j, t_j)$. Given a draw $(\bA_1, \bK_1)$, \dots, $(\bA_n, \bK_n)$, the timeline-mask summary $\bP$ is completely determined (and hence $\bm = |\bS^{t_0}|$ is determined as well), but there are $2^{\bm}$ choices of $\sigma \in \{0,1\}^{\bm}$. Hence, we will oftentimes drop the notation $\bP$ from the assignment rule, and simply refer to $\textsc{Assign}_{\sigma}(\cdot, \cdot)$. 

\subsection{Main Structural Theorem}\label{sec:main-structural}

We may now state our main structural theorem, which shows that the cut obtained from executing the assignment rule $\textsc{Assign}_{\sigma}(x_i,t_i)$ in parallel (for the best choice of $\sigma$) gives an approximately optimal max-cut for certain setting of parameters. Similarly to~\cite{MS08}, we will show that there exists a setting of $\sigma$ for which $\textsc{Assign}_{\sigma}(\cdot, \cdot)$ produces a cut $\bz \in \{0, 1\}^{n \times \{0,1\}}$ which approximately minimizes the sum of \emph{internal distances}
\begin{align}
f(z) \eqdef \frac{1}{2} \sum_{i=1}^n \sum_{j=1}^n d(x_i, x_j) \left( z_{i,0} z_{j,0} + z_{i,1} z_{j,1} \right), \label{eq:internal-f-def}
\end{align}
up to an additive factor of $O(\eps) \sum_{i=1}^n \sum_{j=1}^n d(x_i, x_j)$. Since every pair $(i, j) \in [n]\times[n]$ is either on the same side of the cut, or across the cut, minimizing $f$ corresponds to maximizing the sum of cut edges, up to the additive error of $O(\eps) \sum_{i=1}^n \sum_{j=1}^n d(x_i, x_j)$. 

\begin{theorem}
\label{thm:main-structural}
Fix an $n$-point metric $(X = \{x_1,\dots, x_n\}, d)$, a collection of weights $w_1,\dots, w_n \in (0, 1/2]$ which is $\lambda$-compatible, and the setting of parameters
\begin{align*}
t_e \geq \frac{n \cdot \lambda}{\eps} \qquad \gamma \geq \frac{(\ln(t_e) + 1)^2 \cdot \lambda}{\eps^2} \qquad \text{and}\qquad t_0 \geq \max\left\{ \frac{\sqrt{\gamma \cdot \lambda}}{\eps}, \frac{1}{\eps} \right\}.
\end{align*}
Consider the following setup:
\begin{itemize}
\item Draw $n \times t_e$ matrices $\bA$ and $\bK$, where the $i$-th rows $\bA_i \sim \calT(w_i)$ and $\bK_i \sim \calK(t_0, \gamma)$.
\item For any seed $\sigma$ of $\bA_1,\dots, \bA_n$, let $\bz(\sigma)$ denote the following cut:
\[ \bz(\sigma)_i = \left\{\begin{array}{ll}  \textsc{Assign}_{\sigma}(x_i, \bt_i) & \bt_i \text{ is the activation time for $\bA_i$} \\
								(1, 0) & \text{$\bA_{i,t} = 0$ for all $t \in [t_e]$} \end{array} \right. .\]
\end{itemize}
Then,
\begin{align*}
\Ex_{\bA, \bK}\left[ \min_{\sigma} f(\bz(\sigma)) \right] \leq \min_{\substack{z\in \{0,1\}^{n \times \{0,1\}}\\ z_i \neq (0, 0)}} f(z) + O(\eps) \sum_{i=1}^n \sum_{j=1}^n d(x_i, x_j). 
\end{align*}
\end{theorem}

We defer the proof of Theorem~\ref{thm:main-structural} to Section~\ref{sec:greedy-process}, and show how to use the theorem in massively parallel computation and insertion-only streaming.

%% file: mpc.tex

\section{O(1)-Round Fully Scalable MPC Algorithm}\label{sec:mpc}
In the massively parallel computation model (MPC), a cluster of machines is given as input a set of points $X \subset \R^d$ partitioned arbitrarily across all machines. Each machine $\calM_j$ needs to compute an approximate max-cut assignment $z_i \in \{(0,1), (1,0)\}$ for each point $x_i$ in it's partition $X_j$. The key constraint of the MPC model is the memory of each machine. We consider a fully scalable (also known as strongly sublinear) regime, where there are $m=\Theta(\frac{nd}{s})$ machines which partitions the points, each with local memory $s = \Theta(n^{\alpha}d)$ for some $\alpha \in (0,1]$. Note that the amount of communication per round between machines is also bounded by $s$; no machine can send or receive more than $s$ words of communication.\footnote{A word contains $O(\log(n d \Delta/\eps))$ bits, where $\Delta$ is the aspect ratio of the points.} We consider duplicate points to be separate and distinct points, which if desired can easily be removed in $O(1)$-rounds \cite{G19}. 
We show how to apply the parallel and sparsified assignment rule to get an $O(1)$-round, fully-scalable algorithm for approximating the max-cut of $X \subset \R^d$ in any $\ell_p$ norm, for $p \in [1,2]$. 

\begin{theorem}\label{thm:mpc-full} For any $\eps \in (0,1)$ and $p \in [1,2]$, there exists a randomized fully scalable MPC algorithm with the following guarantees:
\begin{itemize}
\item The algorithm receives as input a multi-set $X = \{ x_1, \dots, x_n \} \subset \mathbb{R}^d$ distributed across $m$ machines each with local memory $s$ of at least $dn^\alpha\cdot\poly(1/\eps)$ words for some $\alpha \in (0,1]$. 
\item After $O_{\alpha}(1)$-rounds of communication, each machine $\calM_j$ for $j \in [m]$ can compute for all points $x_i \in X$ in its memory, a (randomized) assignment $\bz_i \in \{ (1,0), (0,1)\}$ such that 
\begin{align}
\Ex\left[ f(\bz) \right] &\leq  \min_{\substack{z \in\{0,1\}^{n \times \{0,1\}} \\ z_i \neq (0,0)}} f(z) + \eps \sum_{i=1}^n \sum_{j=1}^n d_{\ell_p}(x_i, x_j), \label{eq:thm-mpc-exp}
\end{align} 
where the expectation is taken over the internal randomness of the algorithm, and $f$ is in (\ref{eq:internal-f-def}) with $\ell_p$-metric.
\end{itemize}
The total space used by the algorithm is $O(nd) \cdot \poly(\log(nd)/\eps)$.
\end{theorem}

\begin{remark}
    When instantiating Theorem~\ref{thm:mpc-full} to Euclidean instances (i.e., $p = 2$) one may apply an MPC implementation of the Johnson-Lindenstrauss lemma (Theorem~3 of~\cite{AAHKZ23}) which runs in constant rounds and is fully-scalable with total space $O(nd + n\log^3 n/\eps^2)$. This means one may consider the above theorem with $d = O(\log n /\eps^2)$ and obtain Theorem~\ref{thm:mpc-full} with total space $O(nd) + n \cdot \poly(\log(n)/\eps)$.
\end{remark}

At a high level, the algorithm consists of generating the global information necessary to execute $\textsc{Assign}_{\sigma}(\cdot,\cdot)$ for the best seed $\sigma$ (see Figure~\ref{fig:assign}) and apply Theorem~\ref{thm:main-structural}. This first involves each machines computing metric-compatible weights of its points (Definition~\ref{def:metric-compat}). This allows each machine to independently sample the activation timelines (Definition~\ref{def:activate}) and masks (Definition~\ref{def:mask}) for each of the points held in each machine, and to compute the set of tuples $\bP$ of the global information needed for \textsc{Assign}$_{\sigma,\bP}(\cdot,\cdot)$. Elements of $\bP$ are then shared across all machines, alongside the setting of the seed $\sigma$ which the root of a converge-cast primitive can compute. Now each machine has access to the global information necessary to compute $\textsc{Assign}_{\sigma,\bP}(\cdot,\cdot)$ for all of it's points. Before we prove Theorem~\ref{thm:mpc-full}, we first outline the important primitives used as subroutines. 

\subsection{MPC Preliminaries}
Our algorithm makes use of the following standard MPC primitives:  
\begin{lemma}[Lemma 2.2 in \cite{GSZ11}]\label{lem:mpc-sort} There exists an $O(\log_s n)$-round deterministic sorting algorithm for $m=\Theta(\frac{n}{s})$ machines each with $s=n^\alpha$ words of memory for any $\alpha \in(0,1]$, such that on input $\{p_1,\dots,p_n \} = P$ where $p_1 \leq \dots \leq p_n$ when the algorithm terminates, if machine $\calM$ holds element $p_i$, it knows the rank $i$. 
\end{lemma}

\begin{lemma}[Broadcast \& Converge-cast \cite{G19}]\label{lem:mpc-broadcast} Given a cluster of machines each with $s$ words of  memory, one can define a communication tree with a branching factor of $\sqrt{s}$, where each node is a machine. The tree has total depth $O(\log_{s}(n))$. Using this tree, the root machine $\calM_0$ can broadcast $\sqrt{s}$ words to all machines in $O(\log_{s}(n))$ rounds. Likewise the machines can communicate $\sqrt{s}$ total words to the root machine $\calM_0$ in $O(\log_{s}(n))$ rounds, provided each machine never receives more than $s$ total words or sends more than $\sqrt{s}$ words. These procedures are referred to as broadcast and converge-cast, respectively. 
\end{lemma}

\subsection{Computing Geometric Weights}

The first main ingredient in our algorithm is a method to compute  metric-compatible weights (Definition~\ref{def:metric-compat}). We use the $\ell_1(\ell_p)$-matrix sketches from Theorem \ref{lem:cascaded} in order to determine the weight. The algorithm is a direct application of the sketches with some additional bookkeeping. 

\begin{lemma}[Geometric weights in MPC]\label{lem:mpc-weights} For any $\eps,\delta \in (0,1)$ and $p \in [1,2]$, there exists a randomized fully scalable MPC algorithm satisfying the following:
\begin{itemize}
\item Given as input a multi-set $X = \{ x_1, \dots, x_n \} \subset \mathbb{R}^d$ distributed across machines $m$ whose local memory $s$ is at least $O(dn^{\alpha}) \cdot \poly(\log(nd/\delta)/\eps)$ for $\alpha \in (0,1]$.
\item After $O_{\alpha}(1)$ rounds of communication, each machine $\calM_j$ for $j \in [m]$ computes and stores a weight $\bw_i \in [0, 1]$ for all points $x_i$ in its memory.
\end{itemize}
With probability at least $1-\delta$ over the internal randomness of the algorithm, every $i \in [n]$ satisfies
 $$\frac{\sum_{j=1}^n d_{\ell_p}(x_i, x_j)}{ \sum_{k=1}^n \sum_{j=1}^n d_{\ell_p}(x_k, x_j)} \leq \bw_i \leq (1+\eps) \cdot \frac{\sum_{j=1}^n d_{\ell_p}(x_i, x_j)}{ \sum_{k=1}^n \sum_{j=1}^n d_{\ell_p}(x_k, x_j)}.$$
\end{lemma}

\begin{figure}
\begin{framed}
\textbf{MPC Algorithm} \textsc{Compute-Weights}. The algorithm is used as a sub-routine within the MPC algorithm to compute metric-compatible weights. \\

\textbf{Input}: The parameters $\eps, \delta \in (0,1)$ and a multi-set $X = \{x_1,\dots, x_n \} \subset \mathbb{R}^d$ distributed across machines each with local memory $s$ which can fit, in addition to the points store, $dn^{\alpha} \cdot \poly(\log(nd/\delta)/\eps)$ words for $\alpha \in (0,1]$. The subset of $X$ distributed to machine $\calM_j$ for $j \in [m]$ is denoted $X_j$. \\

\textbf{Output}: Each machine $\calM_j$ for $j \in [m]$ which holds the set $X_j$ can determine, for each point $x_i \in X_j$, the index $i \in [n]$ and a weight $\bw_i$. 

\begin{enumerate}
\item The machines jointly compute $n$, the total number of points (using Lemma~\ref{lem:mpc-broadcast}), as well as a rank (i.e., labels) $i \in [n]$ for each point in $X$ (using Lemma~\ref{lem:mpc-sort}). We assume, henceforth, that each machine $\calM_j$ which holds $x_i$ in its memory ``knows'' the index $i \in [n]$. 
\item Each machine $\calM_j$ creates an $n \times d$ matrix $A^j \in \R^{n\times d}$ such that the $i$-th row $A_{i}^j$ is set to 
\begin{align*}
    A^j_{i} =
    \begin{cases}
    x_i &\ \text{if } x_i \in X_j \\
    \smash{\underbrace{(0, 0, \dots, 0)}_{\text{$d$ times}} } &\ \text{o.w.}
    \end{cases} ,
\end{align*}
and we let $A = \sum_{j=1}^m A^j$ be the $n \times d$ matrix which stacks points $x_1,\dots, x_n \in \R^d$ as the rows of the matrix.
\item Using shared (or communicated) randomness, the machines draw a single $\poly(\log(nd/\delta) / \eps) \times nd$ matrix $\bS\sim \calD$, from Theorem~\ref{lem:cascaded} with accuracy  $\eps / 3$ and failure probability $\delta/n$, $p$ as received and $k=1$. Each machine $\calM_j$ can independently compute $\smash{\sk(A^j) = \bS \cdot A^j}$, and by broadcasting and converge-cast, can compute $\sk(A) = \bS \cdot A$. 
\item Each machine $\calM_j$ iterates through each input $x_i \in X_j$ and creates the $n \times d$ matrix $x_i^{\circ} \in \R^{n\times d}$ which places $x_i$ on all $n$ rows. Each machine executes the evaluation algorithm from Theorem~\ref{lem:cascaded} on $\sk(A)-\sk(x_i^\circ)$ to compute $\bd_i$. The algorithm broadcasts and converge-casts $\bD = \sum_{i=1}^n \bd_i$.
\item A machine which holds the point $x_i \in X_j$ sets the weight $\bw_i$ to the minimum of $(1+\eps /3) \cdot \bd_i / \bD$ and $1$.
\end{enumerate}
\end{framed}
\caption{MPC Algorithm for Computing Geometric Weights}\label{fig:mpc-weights}
\end{figure}

\begin{proof}
We present the algorithm in Figure~\ref{fig:mpc-weights}, and we refer to elements of the algorithm in this proof. First, observe that that $\sum_{j=1}^{m} A^j = A$ where $A\in \R^{n\times d}$ is the matrix of each $x_i$ stacked together, such that $x_i$ is the $i$-th row of $A$. Also, note that for all $i \in [n]$, $\|A -x_i^{\circ}\|_{\ell_1(\ell_p)} = \sum_{j=1}^n d_{\ell_p}(x_i,x_j)$. Hence by the properties of Theorem  \ref{lem:cascaded} and a union bound, we have that for any $i \in [n]$, we compute a $\bd_i$ such that with probability at least $1-\delta$, 
\[ \sum_{j=1}^n d_{\ell_p}(x_j,x_i) \leq \bd_i \leq \left(1+\frac{\eps}{3} \right)\sum_{j=1}^n d_{\ell_p}(x_j,x_i).  \]
By summing over all $\bd_i$, we get an estimate  $\bD$ which is a $(1 + \eps/3)$-approximation of $\sum_{i=1}^n \sum_{j=1}^n d_{\ell_p}(x_i, x_j)$. Then, letting $\bw_i = \min\{ (1+\frac{\eps}{3}) \cdot \bd_i / \bD, 1\}$, which by a simple computation satisfies the following property: 
\begin{align*}
\frac{\sum_{j=1}^n d_{\ell_p}(x_i, x_j)}{ \sum_{k=1}^n \sum_{j=1}^n d_{\ell_p}(x_k, x_j)} \leq \bw_i \leq (1+\eps) \cdot \frac{\sum_{j=1}^n d_{\ell_p}(x_i, x_j)}{ \sum_{k=1}^n \sum_{j=1}^n d_{\ell_p}(x_k, x_j)}.
\end{align*}

All we now need is to bound the space and communication used. Each machine $\calM_j$ store $n$, the ranks $i$ for each $x_i \in X_j$ in local memory. The matrices $A^j$ only contain non-zero entries for the points stored in machine $\calM_i$ and hence fits into local memory.  The shared randomness used to draw $\bS\sim \calD$ from Theorem~\ref{lem:cascaded}, is bounded by $\poly(\log(nd/\delta)/\eps)$ and hence also fits into local memory. Each $\sk(A^j)$ takes no more than $\poly(\log(nd/\delta)/\eps)$ words to store. The partial sums also take at no more than $\poly(\log(nd/\delta)/\eps)$ words to store, by the properties of the sketch. The sums  $\bD$ and weights $\bw_i$ are stored in one word each. Each phase of communication (sorting, converge-cast, broadcast) each take $O(\log_s(n))$ rounds for a total $O(\log_s(n))$ rounds of communication.
\end{proof}

\subsection{Main MPC Algorithm and Proof of Theorem~\ref{thm:mpc-full}}

The description of the main MPC Algorithm for Euclidean Max-Cut appears in Figure~\ref{fig:mpc-alg}. We will use the algorithm to prove Theorem~\ref{thm:mpc-full}, whose proof is divided within the two subsequent subsections.

\begin{figure}
\begin{framed}
\textbf{MPC Algorithm} \textsc{E-Max-Cut}. \\

\textbf{Input}: An accuracy parameter $\eps \in (0,1)$, a multi-set $X = \{x_1, \dots, x_n \} \subset \mathbb{R}^d$ distributed across machines with at least one non-equal point (otherwise, all cuts have value $0$). Each machine has local memory $s \geq \poly(d\log(n)/\eps)$ and there are $m = \Theta(n/s)$ machines. For each $j \in [m]$, we denote the subset of $X$ in machine $\calM_j$ as $X_j$. \\

\textbf{Output}: For each $j \in [m]$, the machine $\calM_j$ can determine a (random) assignment $\bz_i \in \{(1, 0), (0,1) \}$ for each $x_i \in X_j$. 

\begin{enumerate}
\item\label{en:mpc-1}We utilize the sub-routine \textsc{Compute-Weights} so that each machine $\calM_j$ may compute a weight $\bw_i$ for each $x_i \in X_j$ accuracy $\eps$, $p$ and failure probability $\delta\leftarrow\eps$ (see Lemma~\ref{lem:mpc-weights}). 
\item\label{en:mpc-2} The machines determine the parameter $n = |X|$, and set parameters
\begin{align*}
\lambda &= 8 \qquad t_e = \frac{n \cdot \lambda}{\eps} \qquad \gamma = \frac{(\ln(t_e) + 1)^2 \cdot \lambda}{\eps^2} \\
t_0 &= \max\left\{ \frac{\sqrt{\gamma \cdot \lambda}}{\eps}, \frac{1}{\eps} \right\}, \qquad \xi = \poly \left (\frac{\lambda \log(n)}{\eps} \right )
\end{align*}
so as to align with the parameter settings of Theorem~\ref{thm:main-structural}, Lemma~\ref{lem:weight-to-compatible} and Lemma~\ref{lem:seed-check}.\footnotemark~ Using these parameter settings, each $\calM_j$ samples for each  $x_i \in X_j$, an activation timeline $\bA_i \sim \calT(\bw_i')$ where $\bw_i' = \bw_i / 2$, a mask $\bK_i\sim \calK(t_0, \gamma)$, and a bit $\bc_i \sim \Ber(\min\{\xi \cdot \bw_i', 1\})$.
\item\label{en:mpc-3} For each $j \in [m]$, machine $\calM_j$ computes 
\begin{align*}
\bP_j &= \left\{ (x_i, \ell; \bw_i, t_i) : x_i \in X_j, \ell \in [t_e], \bA_{i,\ell} \cdot \bK_{i,\ell} = 1, t_i \text{ is activation time of $\bA_i$} \right\},\\
\bC_j &= \left\{ (x_i, \xi \cdot \bw_i, t_i) :  x_i \in X_j, \bc_i = 1, t_i \text{ is activation time of $\bA_i$} \right\}.
\end{align*} 
Machines converge-cast the union of $\bP_j$ to compute $\bP$, and the union of $\bC_j$ to compute $\bC$. The root machine of converge-cast uses Lemma~\ref{lem:seed-check} to decide on a seed $\bsigma^*$. The root then broadcasts $\bP$ and $\bsigma^*$. 
\item\label{en:mpc-4} Each machine $\calM_j$, using $\bP$ and $\bsigma^*$ can then compute $\bz_i \leftarrow \textsc{Assign}_{\bsigma^*,\bP}(x_i, t_i)$ for all points $x_i \in X_j$.
\end{enumerate}

\end{framed}
\caption{MPC Algorithm for Euclidean Max-Cut.}\label{fig:mpc-alg}
\end{figure}
\footnotetext{As we will see, the value of $\log |\calF|$ when applying Lemma~\ref{lem:seed-check} will be $\poly(\log n/\eps)$.}

\subsubsection{Correctness Guarantee of Theorem~\ref{thm:mpc-full}} 

We let $\bz \in \{0,1\}^{[n] \times \{0,1\}}$ denote the cut assignment produced by the MPC algorithm, \textsc{E-Max-Cut}, and we will use Theorem~\ref{thm:main-structural} to show (\ref{eq:thm-mpc-exp}) is satisfied. First, we apply Lemma~\ref{lem:mpc-weights} to conclude that, with probability $1-\eps$, all weights $\bw_1, \dots, \bw_n \in (0, 1]$ computed in Line~\ref{en:mpc-1} satisfy the conditions of Lemma~\ref{lem:weight-to-compatible} with $\sfD = 1 + \eps$, where weights $\bw_i > 0$ since there are at least two non-equal points. Assume, from now on, that this is the case; with probability at most $\eps$, the above fails and we will later bound the expected contribution from such ``catastrophic'' failures. By Lemma~\ref{lem:weight-to-compatible}, the weights $\bw_i' = \bw_i / 2$ are $8$-compatible and satisfy $\|\bw'\|_1 \leq (1+\eps)/2$. We now apply Theorem~\ref{thm:main-structural} with the weights $\bw_1', \dots, \bw_n' \in (0, 1/2]$ with parameter settings according to Line~\ref{en:mpc-2}; each point $x_i$ computes an activation timeline $\bA_i \sim \calT(\bw_i')$ and masks $\bK_i \sim \calK(t_0, \gamma)$. In Subsection~\ref{subsec:mpc-space}, we will show the expectation of $|\bP|$ is at most $\poly(\log n / \eps)$; with probability $1 - \eps$, $|\bP| \leq \poly(\log n / \eps)$. Assume this is the case (otherwise, with probability $\eps$, we obtain another ``catastropic failure), and note this implies a seed for $\bA_1, \dots, \bA_n$ has length at most $\poly(\log n / \eps)$.

After the converge-cast of $\bP$ and $\bC$ in Line~\ref{en:mpc-3}, the root machine contains the necessary components for Lemma~\ref{lem:seed-check} (because $\log |\calF|$ is at most the seed length $\poly(\log n /\eps)$). In particular, the root machine has access to $\bP$, so it can compute, for any seed $\sigma$, the value $\textsc{Assign}_{\sigma,\bP}(x_i, t_i)$ for every $(x_i, \xi \cdot w_i, t_i) \in \bC$. This allows the algorithm to determine a setting of $\sigma^*$, which can be broadcasted, along with $\bP$ to every machine. In Line~\ref{en:mpc-4}, each machine $\calM_j$ now has all the global information needed to compute $\textsc{Assign}_{\sigma^*, \bP}(x_i, t_i)$ for all $x_i \in X^j$. In summary, the output vector $\bz \in \{(1,0), (0,1)\}^{[n] \times \{0,1\}}$ satisfies, using $f$ from (\ref{eq:internal-f-def}) with $\ell_p$-metric,
\begin{align*}
\Ex_{(\bA, \bK)}\left[ f(\bz) \right] &\mathop{\leq}^{\text{Lem.~\ref{lem:seed-check}}} \Ex_{(\bA, \bK)}\left[ \min_{\sigma} f(\bz(\sigma)) \right] + \eps \sum_{i=1}^n \sum_{j=1}^n d_{\ell_p}(x_i, x_j) \\
			&\mathop{\leq}^{\text{Thm.~\ref{thm:main-structural}}} \min_{\substack{z \in \{0,1\}^{[n] \times \{0,1\}} \\ z_i \neq (0,0)}} f(z) + O(\eps) \sum_{i=1}^n \sum_{j=1}^n d_{\ell_p}(x_i, x_j). 
\end{align*}
We now upper bound the contribution of ``catastrophic'' failures, when the weights $\bw_1,\dots, \bw_n$ fail to satisfy the conditions to Lemma~\ref{lem:mpc-weights}, or when the seed length (which is bounded by $|\bP|$) is more than $\poly(\log n / \eps)$. These cases occurs with probability at most $2\eps$, and the maximum internal distance may occur when all inputs are assigned to the same side of the cut and contribute $\sum_{i=1}^n \sum_{j=1}^n d_{\ell_p}(x_i, x_j)$. Hence, this incurs another factor of $\eps \sum_{i=1}^n \sum_{j=1}^n d_{\ell_p}(x_i, x_j)$.

\subsubsection{Round and Space Guarantees for Theorem~\ref{thm:mpc-full}}\label{subsec:mpc-space}

The total round complexity of this algorithm is $O(\log_s(n))$. It consists of the \textsc{Compute-Weights} subroutines of Line~\ref{en:mpc-1} (Lemma~\ref{lem:mpc-weights}), and the converge-casts and broadcasts of $\bP$ and $\bC$ in Line~\ref{en:mpc-3}. It only remains to show that $|\bP|$ and $|\bC|$ is small; to see this, note that for a fixed setting of weights $\bw_1,\dots, \bw_n$, the definition of $\bP_j$ (and hence $\bP$) in Line~\ref{en:mpc-3} is
\begin{align*}
\Ex_{(\bA, \bK)}\left[ |\bP| \right] &\leq \sum_{i=1}^n \sum_{\ell=1}^{t_e} \Prx\left[ \bA_{i,\ell} \cdot \bK_{i,\ell} = 1 \right] \leq (t_0 + \gamma) \sum_{i=1}^n \frac{\bw_i}{2} + \sum_{\ell=t_{0}+1}^{t_e} \sum_{i=1}^n \frac{\bw_i}{2} \cdot \frac{\gamma}{t}  \\
		&\leq O(t_0 + \gamma) + O(\gamma \ln(t_e+1)) \leq \poly(\log n / \eps). \\
\Ex_{\bc_1,\dots, \bc_n}\left[ |\bC| \right] &\leq \sum_{i=1}^n \Prx\left[ \bc_i = 1 \right] \leq O(\xi).
\end{align*}
Storing each entry in $\bP$ takes $O(d)$ words for the point, the weight, time and activation time. Furthermore, each entry of $\bC$ similarly takes $O(d)$ words. 
The size of the seed $\sigma$ is bounded by $\bP$. 
Since all of these values are smaller than $\sqrt{s}$ words, the broadcast and converge-cast operations can be run in $O(\log_s(n))$ rounds. There are a constant number of such subroutines, which each take the total round complexity is $O(\log_s(n))$, which is $O(1)$ in the fully scalable regime.

\ignore{\subsubsection{Seed Check Length}
{\color{red} Not sure where in the paper to put this as we reference it a few times and if the bound is good enough. I assume poly($t_0$) is fine...}

The length $\bm$ of the seed $\bsigma$ is the final component needed to prove the correctness of Theorem~\ref{thm:mpc-full}. We bound the expected length $\bm = |\bS^{t_0}|$ as follows: \[\Ex_{(\bA, \bK)}\left[ |\bS^{t_0}| \right] \leq  \sum_{i=1}^n \sum_{\ell=1}^{t_0} \Prx\left[ \bA_{i,\ell} \cdot \bK_{i,\ell} = 1 \right] = t_0 \sum_{i=1}^n \frac{\bw_i}{2} = t_0\] and then using independence to calculate the the variance 
\begin{align*}
\Varx_{(\bA, \bK)}\left[ |\bS^{t_0}| \right] 
\leq \sum_{i=1}^n \sum_{\ell=1}^{t_0} \Varx_{(\bA, \bK)}\left[ \bA_{i,\ell} \cdot \bK_{i,\ell} \right] 
= t_0 \sum_{i=1}^n  \frac{\bw_i}{2}(1-\frac{\bw_i}{2}) \leq t_0
\end{align*}
Using Chebyshev's inequality, we can bound the probability of the seed length exceeding $t_0$ by more than some $t_0/\eps$ factor, giving the following inequality 
\begin{align*}
    \Pr[|\bm-t_0|\geq \sqrt{\frac{t_0}{\eps}}] \leq \eps
\end{align*}
This allows us to control the size of the seed as with probability $1-\eps$, $\bm=\poly(t_0\eps)$. }

%% file: insertion-only.tex

\section{Insertion-Only Streaming Algorithm}\label{sec:insertion-only}

In this section, we show how to apply Theorem~\ref{thm:main-structural} to design an insertion-only streaming algorithm which makes one-pass over the data and provides query access to an approximately optimal max-cut in $\ell_p$ space, for all $p \in [1,2]$. An insertion-only streaming algorithm is a small-space randomized algorithm which computes in the following manner:
\begin{itemize}
\item Throughout the entire execution, the algorithm maintains a memory of $s$ ``words'', which contains enough bits to specify (i) indices for individual points, (ii) individual coordinates of points, and (iii) intermediate numbers involved in the computation. There will be $n$ points encoded as vectors in $[\Delta]^d$, so a word will consist of $\Theta(\log(nd\Delta/\eps))$ bits. We assume that there are no duplicate points, so that $n$ is always at most $[\Delta]^d$. 
\item An input dataset $X = \{ x_1, \dots, x_n \}$ of points in $[\Delta]^d$ is streamed to the algorithm one-at-a-time and in an arbitrary order. We will use indices $1, \dots, n$ to specify the order the algorithm sees them, so that $x_1$ is presented first, followed by $x_2$, \dots, $x_n$. Upon receiving each point $x_i \in X$, the algorithm may update its memory. Once the entire dataset has been inserted, the algorithm provides query access to a function which specifies, for each point $x_i \in X$, whether it belongs to the 0-side of a cut or the 1-side of a cut.
\end{itemize}

\begin{remark}[Use and Storage of Randomness]
\emph{We present algorithms in the ``random oracle model:'' we assume the algorithm has query access to a function $\boldr \colon \{0,1\}^* \to \{0, 1\}$ which outputs, for any string in $\{0,1\}^*$, an independent and uniformly distributed random bit. The assumption is impractical, and it is fairly standard to use Nisan's pseudorandom generator~\cite{N92} (see Section 4.3.3 in~\cite{N20}). In particular, the number of uniform random bits ever referenced is at most $\exp((d \log \Delta/\eps)^{O(1)})$; furthermore, for any fixed set $X \subset [\Delta]^d$, there is a $\poly(d\log\Delta/\eps)$-space read-once branching program whose input are the $\exp((d \log \Delta/\eps)^{O(1)})$ random bits referenced and which computes the cost of the cut output. Hence, we let the algorithm use $\boldr$ to (i) sample independent random variables (which will be used for Theorem~\ref{thm:main-structural}), and (ii) refer back to previously-used randomness without explicitly storing it. }
\end{remark}

\begin{theorem}\label{thm:insertion-only}
For any $p \in [1, 2]$, $d, \Delta \in \N$, and $\eps \in (0, 1)$, there exists an insertion-only randomized streaming algorithm with the following guarantees:
\begin{itemize}
\item \emph{\textbf{Maintenance}}: The streaming algorithm implicitly maintains a set of (distinct) points $X = \{ x_1, \dots, x_n \} \subset [\Delta]^d$, as well as a random seed $\boldr$, using $\poly(d\log \Delta/ \eps)$ words of space.
\item \emph{\textbf{Updates}}: For a point $x \in [\Delta]^d$, the sub-routine $\textsc{Add-Point}(x)$ updates the streaming algorithm so as to maintain $X \cup \{ x \}$. 
\item \emph{\textbf{Query}}: For a point $x \in [\Delta]^d$, the sub-routine $\textsc{Assign}(x)$ outputs $(1, 0)$ or $(0,1)$, corresponding to the 0-side or 1-side of the cut.  
\end{itemize}
For any set $X = \{ x_1, \dots, x_n \} \subset [\Delta]^d$ (indexed so that $x_i$ is the $i$-th point inserted), consider executing $\textsc{Add-Point}(x_i)$ for all $i \in [n]$ in order; once done, let $\bz \in \{0,1\}^{n \times \{0,1\}}$ be
\[ \bz_i \leftarrow \textsc{Assign}(x_i).\]
Then, 
\begin{align*}
&\Ex\left[ f(\bz) \right] \leq  \min_{\substack{z \in\{0,1\}^{n \times \{0,1\}} \\ z_i \neq (0,0)}} f(z) + \eps \sum_{i=1}^n \sum_{j=1}^n d_{\ell_p}(x_i, x_j),
\end{align*}
using $f$ from (\ref{eq:internal-f-def}) with $\ell_p$-metric.
\end{theorem}

At a high level, the insertion-only streaming algorithm follows a pattern similar to the MPC algorithm. There is a sub-routine (a part of the streaming algorithm) which is dedicated to computing the weights using an $\ell_1(\ell_k)$ cascaded sketch, a sub-routine which maintains activation timelines and masks for individual points, and a sub-routine dedicated to selecting the final seed for the assignment procedure. Then, the streaming algorithm computes weights and maintains the points needed to execute the procedure of Theorem~\ref{thm:main-structural}. We incorporate a new component, a ``cleanup'' procedure, where the algorithm ensures that its space complexity will not exceed some bound by ``forgetting'' certain points.

\subsection{Computing the Weights}

We begin by specifying the streaming algorithm to compute metric-compatible weights. As in the MPC algorithm, we will also make use of the $\ell_1(\ell_p)$-matrix sketches in order to determine the weight (see Lemma~\ref{lem:cascaded}). The algorithm is straight-forward, and uses an $\ell_1(\ell_p)$ sketch which are prepared in order to compute sum of all distances. As in Lemma~\ref{lem:mpc-weights}, rows of the matrix will be encoded with dataset points. We will also maintain a counter, which allows us to determine how many non-zero rows the matrix has. 

\begin{lemma}\label{lem:computing-weights}
For any $p \in [1,2]$, $ d,\Delta \in \N$, and $\eta,\delta \in (0, 1)$, there exists an insertion-only randomized streaming algorithm with the following guarantees:
\begin{itemize}
\item \emph{\textbf{Maintenance}}: The streaming algorithm implicitly maintains a set of points $X = \{ x_1, \dots, x_n \} \subset [\Delta]^d$ using $\poly(d \log (\Delta/\delta) /\eta)$ words of space.
\item \emph{\textbf{Updates}}: For a point $x \in [\Delta]^d$, the sub-routine $\textsc{Add-Point}(x)$ updates the streaming algorithm so as to maintain $X \cup \{x\}$.
\item \emph{\textbf{Query}}: For a point $x \in [\Delta]^d$, the sub-routine $\textsc{Weight}(x)$ outputs a weight in $[0,1]$.
\end{itemize}
For any set $X = \{ x_1, \dots, x_n \} \subset [\Delta]^d$ and any $i \in [n]$, consider executing $\textsc{Add-Point}(x_j)$ for all $j\leq i$ in order. After inserting points $x_1, \dots, x_i$, we let 
\[ \bw_i(x_j) \leftarrow \textsc{Weight}(x_j),\]
for all $j \leq i$. Then, with probability $1 - \delta$, every $\bw_{i}(x_j) \in [0, 1]$ satisfies
\begin{align*}
\dfrac{\sum_{\ell=1}^i \|x_{\ell} - x_j\|_p}{\sum_{\ell=1}^i \sum_{k=1}^i \|x_{\ell} - x_{k}\|_p} \leq \bw_i(x_j) \leq (1+\eta) \cdot \dfrac{\sum_{\ell=1}^i \|x_{\ell} - x_j\|_p}{\sum_{\ell=1}^i \sum_{k=1}^i \|x_{\ell} - x_{k}\|_p},
\end{align*}
where we default $0/0$ as $1$ (which occurs at $i=1$).
\end{lemma}

\begin{figure}
\begin{framed}
\textbf{Streaming Algorithm} \textsc{Compute-Weights}. The streaming algorithm is initialized with parameters $d, \Delta \in \N$ and $\eta,\delta \in (0,1)$ and $p \in [1,2]$, and implicitly maintains a set of points $X \subset [\Delta]^d$. At any moment, the algorithm may query a weight of a point $x$ which we assume was previously inserted in $X$; the algorithm outputs a $(1+\eta)$-approximation to the sum of $\ell_p$-distances from each $x' \in X$ to $x$ divided by the sum of all pairwise $\ell_p$-distances in $X$ with probability at least $1-\delta$.  \\

\textbf{Maintenance}. The streaming algorithm maintains the following information:
\begin{itemize}
\item A counter $m$ (initialized to zero), which counts the number of points inserted. 
\item A sketch $\sk(A)$ for the $\ell_1(\ell_p)$ norms of an implicit $\Delta^d \times d$ matrix $A$ (initially all zero), which contains the $m$ inserted points as the first $m$ rows (see Lemma~\ref{lem:cascaded}). The sketch is initialized by storing $\bS \sim \calD$ from Lemma~\ref{lem:cascaded} with accuracy $\eps \leftarrow \eta/3$ and failure probability $\delta' \leftarrow \delta/\Delta^{2d}$.
\item A total weight $\bW$, initialized to $0$.
\end{itemize}

\textbf{Update}. \textsc{Add-Point}$(x)$.
\begin{enumerate}
\item We increment the counter $m$.
\item\label{ln:2} We update the sketch $\sk(A)$ by adding the point $x$ as the $m$-th row of the matrix. Notice this is done implicitly with the fact the sketch is linear: we consider the $\Delta^d \times d$ matrix $A_x$ which has $x$ as the $m$-th row and is otherwise 0. The (implicit) matrix $A$ is updated to $A + A_x$ and the sketch updated to $\sk(A) + \sk(A_x)$.
\item\label{ln:3} We build the $\Delta^d \times d$ matrix $x^{\circ}$ which contains $x$ as the first $m$ rows, and we evaluate the answer on $\sk(A - x^{\circ})$ the obtaining a number $\bd$ and update $\bW \leftarrow \bW + \bd$. 
\end{enumerate}

\textbf{Query}. \textsc{Weight}$(x)$.
\begin{enumerate}
\item\label{en:ins-q-1} Evaluate the sketch $\sk(A - x^{\circ})$ (see above), to obtain a number $\bd$. Note that this number may be different to the value $\bd$ obtained during \textsc{Add-Point}$(x)$, since other points may have been inserted. 
\item Output $(1+\eta/3) \cdot \bd / (2W)$ if $W \neq 0$ and the value is below $1$, and $1$ otherwise. 
\end{enumerate}
\end{framed}
\caption{The \textsc{Compute-Weights} Streaming Algorithm.}\label{fig:compute-weights}
\end{figure}

\paragraph{Proof of Lemma~\ref{lem:computing-weights}: Space Complexity.}
The streaming algorithm \textsc{Compute-Weights} is described in Figure~\ref{fig:compute-weights}, and it is clear that it uses $O(1) + \poly(d\log( \Delta/\delta)/\eta)$ words of space: the counter $m$ and total weight $W$ are stored in one word each (note $W \leq n^2 \Delta d$, which means it fits within a word), and the space complexity of the $\ell_1(\ell_p)$-sketch comes directly from Lemma~\ref{lem:cascaded} for $\Delta^d \times d$ matrices, with accuracy $\eta/3$ and failure probability $\delta' = \delta/\Delta^{2d}$.

\paragraph{Proof of Lemma~\ref{lem:computing-weights}: Correctness.} We show correctness of the algorithm by maintaining the following invariant. Right after the execution of $\textsc{Add-Point}(x_1), \dots, \textsc{Add-Point}(x_i)$,
\begin{itemize}
\item The counter $m$ is set to $i$.
\item The implicit matrix $A$ consists of a $\Delta^d \times d$ matrix whose first $i$ rows are $x_1,\dots, x_i$, and whose remaining rows are zero.
\item Let $\bd_i$ denote the random variable output by Line~\ref{ln:3} in the final execution of $\textsc{Add-Point}(x_i)$. Then, with probability at least $1-\delta'$ over the randomness in $\bS \sim \calD$, 
\[ \sum_{\ell=1}^i d_{\ell_p}(x_{\ell}, x_i) \leq \bd_i \leq (1+\eta/3) \sum_{\ell=1}^i d_{\ell_p}(x_{\ell}, x_i). \]
\item For all $j \leq i$, let $\bW_j$ denote the random variable $\bW$ right after the execution of $\textsc{Add-Point}(x_j)$. Then, with probability at least $1-i \delta'$ over the randomness in initializing the sketch $\bS \sim \calD$, every $j \leq i$ satisfies
\begin{align} 
\sum_{k=1}^{j} \sum_{\ell=1}^k d_{\ell_p}(x_{\ell}, x_k) \leq \bW_j \leq (1+\eta/3) \sum_{k=1}^{j} \sum_{\ell=1}^k  d_{\ell_p}(x_{\ell}, x_k).  \label{eq:approx-W}
\end{align}
\end{itemize}
We prove that the invariant holds by induction on $i$. For $i = 1$, a single call of $\textsc{Add-Point}(x_1)$ increments the counter once and updates the implicit matrix $A$ in Line~\ref{ln:2} to $A_{x_1}$, which is the matrix with a single non-zero first row set to $x_1$. The matrix $x_1^{\circ}$ in Line~\ref{ln:3} is also $A_{x_1}$, so by Lemma~\ref{lem:cascaded}, $\bd_i = 0$ with probability at least $1-\delta'$ and, hence, $\bW_j = 0$ with probability $1-\delta'$ for all $j \leq i$. Suppose by induction that the invariant holds for index $i \in [n]$; we now execute $\textsc{Add-Point}(x_{i+1})$ and verify the invariant for $i+1$. The counter increments from $i$ to $i+1$ and the implicit matrix $A$, adds the $i+1$-th row consisting of $x_{i+1}$ to $A$. The random variable $\bd_{i+1}$ becomes the evaluation of the sketch $\sk(A - x_{i+1}^{\circ})$, which by Lemma~\ref{lem:cascaded}, satisfies
\begin{align} 
\sum_{\ell=1}^{i+1} d_{\ell_p}(x_{\ell}, x_{i+1}) \leq \bd_{i+1} \leq (1+\eta/3) \sum_{\ell=1}^{i+1} d_{\ell_p}(x_{\ell}, x_{i+1}) \label{eq:approx-d-i+1}
\end{align}
with probability $1-\delta'$. By a union bound over the invariant after $\textsc{Add-Point}(x_i)$, with probability at least $1 - (i+1) \delta'$, the random variable $\bd_{i+1}$ during $\textsc{Add-Point}(x_{i+1})$ satisfies (\ref{eq:approx-d-i+1}) and $\bW_j$ satisfies (\ref{eq:approx-W}) for all $j \leq i$. Thus, the updated value $\bW_{i+1}$ which is set to $\bW_{i} + \bd_{i+1}$ after $\textsc{Add-Point}(x_{i+1})$ satisfies
\begin{align*}
&\sum_{k=1}^{i+1} \sum_{\ell=1}^k d_{\ell_p}(x_{\ell}, x_k) = \sum_{k=1}^i \sum_{\ell=1}^k d_{\ell_p}(x_{\ell}, x_k) + \sum_{\ell=1}^{i+1}d_{\ell_p}(x_{\ell}, x_{i+1}) \leq \bW_{i} + \bd_{i+1} = \bW_{i+1} \\
			&\qquad\qquad \leq (1+\eta/3) \left(\sum_{k=1}^i \sum_{\ell=1}^k d_{\ell_p}(x_{\ell}, x_k) + \sum_{\ell=1}^{i+1}d_{\ell_p}(x_{\ell}, x_{i+1}) \right) = (1+\eta/3) \sum_{k=1}^{i+1} \sum_{\ell=1}^k d_{\ell_p}(x_{\ell}, x_k).
\end{align*}
This concludes the inductive proof that the invariant holds for all $i \in [n]$. Hence, after an execution of $\textsc{Add-Point}(x_1)$, \dots, $\textsc{Add-Point}(x_i)$, the counter is $i$, the matrix $A$ contains the points $x_1, \dots, x_i$ stacked as rows and remaining rows are zero, and the random variable $\bW$ stored in \textsc{Compute-Weights} is set to $\bW_i$ and satisfies (\ref{eq:approx-W}) with $j=i$. We now execute $\textsc{Weight}(x_j)$ for all $j \leq i$ and denote the random variable $\bd_j'$ output by Line~\ref{en:ins-q-1} in $\textsc{Weight}(x_j)$. By Lemma~\ref{lem:cascaded} once again, each $\bd_j'$ individually satisfies
\begin{align} 
\sum_{\ell=1}^i d_{\ell_p}(x_{\ell}, x_j) \leq \bd_j' \leq (1+\eta/3) \sum_{\ell=1}^i d_{\ell_p}(x_{\ell}, x_j) \label{eq:d-j-new}
\end{align}
with probability at least $1- \delta'$ over the randomness in $\bS \sim \calD$, so by a union bound, $\bW_i$ satisfies (\ref{eq:approx-W}) and all $j \leq i$ satisfy (\ref{eq:d-j-new}) with probability $1 - 2 i \delta'$. Notice that, $\bw_i(x_j)$ is set to $(1+\eta/3) \bd_j' / (2\bW_i)$. The numerator $\bd_j'$ satisfies (\ref{eq:d-j-new}) as desired, and $2 \bW_i$ is a $(1+\eta/3)$-approximation to $\sum_{k=1}^i \sum_{\ell=1}^i d_{\ell_p}(x_{\ell}, x_k)$ since $d_{\ell_p}(x_{\ell}, x_{\ell}) = 0$. Hence, $(1+\eta/3) \bd_j / (2\bW_i)$ is at least the desired quantity (which occurs when $\bW_i$ maximally over-estimates and is accounted by the additional $(1+\eta/3)$-factor), and is at most the desired quantity (which occurs when $\bd_j'$ maximally over-estimates and obtains approximation $(1+\eta/3)^2 \leq (1+\eta)$). By a union bound over all $i \in [n]$, all such $\bw_i(x_j)$ satisfy the required guarantees with probability $1 - 2n^2 \delta'$. Since $n \leq \Delta^d$ and $\delta' = \delta / \Delta^{2\delta}$, we obtain the desired $1 - \delta$.

\ignore{\begin{figure}
\begin{framed}
\textbf{Streaming Algorithm.} \textsc{Compute-Weights}. \\

\textbf{Maintenance}. The streaming algorithm maintains the following information:
\begin{itemize}
\item A counter $m$ (initialized to zero), which counts the number of points inserted. 
\item A parameter $H$ which is always $\lceil \log_2 m \rceil$ (and $\log_2 0$ defaults to 0). A collection of $H$ sketches for the $\ell_1(\ell_p)$ norms of matrices (see Theorem~\ref{thm:cascaded-sketch}):
\begin{itemize}
\item For $h \in \{0, \dots, H\}$, the $h$-th sketch is for $2^h \times d$ matrices with accuracy parameter $\eps \leftarrow \eta$ and with failure probability $\delta \leftarrow 1/n^3$.
\item Initially, the $2^h \times d$ matrix is zero, and points will be inserted as rows of the matrix.
\item The $h$-th sketch is ``full'' when all $2^h$ rows have been updated with points.
\end{itemize}
\item A total weight $W$, initialized to $0$.
\end{itemize}

\textbf{Update}. \textsc{Add-Point}$(x_i)$.
\begin{enumerate}
\item We increment the counter $m$.
\item We find the first sketch $h \in H$ which is not yet full (or initialize a new sketch if all are full), and we update the $h$-th sketch by inserting the point $x_i$ as the first row which has not yet been updated. (Note that the first ``un''-updated row index can be derived from $c$, and the sketch updated since it is linear).
\item We evaluate all $H$ sketches with the matrices containing $x_i$ in all (updated) rows, obtaining a number $\bd_i$, which is a $(1+\eta)$-approximation to
\[ \sum_{\ell=1}^i \|x_{\ell} - x_i \|_p,\]
and update $W \leftarrow W + \bd_i$. 
\end{enumerate}

\textbf{Query}. \textsc{Weight}$(x_i)$.
\begin{enumerate}
\item Evaluate all $H$ sketches with the matrices containing $x_i$ in all (updated) rows, obtaining a number $\bd_i$ which is a $(1+\eta)$-approximation to
\[ \sum_{x \in X} \|x - x_i \|_p.\]
Note that this number may be different to the value $\bd_i$ obtained during \textsc{Add-Point}$(x_i)$, since we are now evaluating distances with respect to all points inserted, and not only those inserted before $x_i$. 
\item Output $\bd_i / (2W)$.
\end{enumerate}
\end{framed}
\caption{The \textsc{Compute-Weights} Streaming Algorithm.}\label{fig:compute-weights}
\end{figure}}

\subsection{Maintaining Activation and Kept Timelines}

The second component is a streaming algorithm for maintaining an activation timeline and a mask for a particular point (recall Definition~\ref{def:activate} and Definition~\ref{def:mask}). The crucial aspect is that our streaming algorithm will only need to maintain information on points which were both activated and kept. As the streaming algorithm progresses and points enter the (implicit) set $X$, the metric-compatible weights change, and so do the set of points which are activated and kept.

\begin{lemma}\label{lem:timeline-mask}
For any $t_0 , t_e \in \N$ with $t_0 \leq t_e$ and $\gamma \geq 0$, there exists a data structure, \textsc{Timeline-Mask} (Figure~\ref{fig:timeline-mask}), which satisfies the following:
\begin{itemize}
\item \emph{\textbf{Initialization}}: The algorithm is initialized by $\textsc{Init}(x, \boldr)$, where $x \in [\Delta]^d$ and $\boldr \colon \{0,1\}^* \to \{0,1\}$ provides access to a random oracle. 
\item \emph{\textbf{Maintenance}}: It stores a weight $w$ (initialized to $1/2$) and implicitly maintains draws to an activation timeline $\bA \sim \calT(w)$ and a mask $\bK \sim \calK(t_0, \gamma)$. 
\item \emph{\textbf{Updates}}: The weight may be updated with the sub-routine $\textsc{Mod-Min-Weight}(\tilde{w})$, which replaces $w\leftarrow \min\{w, \tilde{w}\}$. Furthermore, if $\bA_t \cdot \bK_t = 0$ for all $t \in [t_e]$ before the update, then after the update, $\bA_t \cdot \bK_t$ is still $0$ for all $t \in [t_e]$.
\item \emph{\textbf{Query}}: One can query \textsc{Activation-Time}$()$, which returns the activation time of the timeline $\bA$, and \textsc{Active-Kept}$()$, which returns all indices $t \in [t_e]$ where $\bA_{t} \cdot \bK_t = 1$.
\end{itemize}
The space complexity is that of storing the parameters $t_0, t_e, \gamma$, $w$, a pointer to $\boldr$, and the point $x$.  
\end{lemma}

\begin{remark}[Initialization Procedure $\textsc{Init}(x_, \boldr)$]\label{rem:init}
\emph{We remark on an important aspect of the initialization above. Upon initialization for a point $x_i$ with random oracle $\boldr$, the data structure generates a timeline which only depends on $x_i$ and $\boldr$ and the weight $w_i$ held. This has the following consequence. Suppose the algorithm receives as input the point $x_i$ and initializes a $\textsc{Timeline-Mask}$ data structure to maintain the activation timeline and mask $(\bA_i, \bK_i)$. If there exists an index $t \in [t_e]$ where $\bA_{i,t} \cdot \bK_{i,t} = 1$, then point $x_i$ is activated and kept. In the description of $\textsc{Assign}_{\sigma}(\cdot,\cdot)$ in Figure~\ref{fig:assign}, the point $x_i$ must be kept in the timeline-mask summary $\bP$ to assign other points $x_j$. If there is no time $t \in [t_e]$ where $\bA_{i,t} \cdot \bK_{i,t} = 1$, the algorithm may completely ``forget'' about $x_i$, since $x_i$ is not used to assign other points. The benefit of forgetting is that it decreases the total amount of space when assigning points $x_j$. However, when assigning the point $x_i$, the assignment rule $\textsc{Assign}_{\sigma}(x_i, t_i)$ needs to know $t_i$. Recall, $t_i$ is the activation time of the timeline $\bA_{i}$, and we must argue that the algorithm may re-construct $t_i$. The guarantee that the timeline and mask depend solely on $x_i, \boldr$ and the weight $w_i$ will allow us to re-create the activation timeline for $x_i$ when generating the assignment for $x_i$.}
\end{remark}

\begin{figure}
\begin{framed}
\textbf{Streaming Algorithm} \textsc{Timeline-Mask}. The algorithm is initialized with parameters $t_0, t_e \in \N$ and $\gamma > 0$, and maintains a weight and implicitly an activation timeline $\bA \sim \calT(w)$ (Definition~\ref{def:activate}) and a mask $\bK \sim \calK(t_0, \gamma)$ (Definition~\ref{def:mask}). The algorithm can support updates to weights (so long as they do not increase significantly), and can report the activation timeline and all indices $t$ which were both activated and kept for that particular timeline and mask.\\

\textbf{Maintenance}. The streaming algorithm is initialized with $\textsc{Init}(x, \boldr)$, where $x$ is a point and $\boldr$ is a random oracle. It maintains the following information:
\begin{itemize}
\item A weight $w$, initialized to $1/2$.
\item The point $x$, which will index to a starting position of the random oracle $\boldr$. 
\item We use $\boldr$ and $x$ to generate a sequence of $2\cdot t_e$ independent random variables
\[ \bA^{(r)}_{1}, \dots, \bA^{(r)}_{t_e} \sim [0, 1] \qquad\text{and}\qquad \bK^{(r)}_{1}, \dots, \bK^{(r)}_{t_e} \sim [0, 1]. \]
\item The (implicit) activation timeline and mask are given by iterating through $t=1, \dots, t_e$, and letting
\begin{align}
\bA_{t}(w) &= \left\{ \begin{array}{ll} \ind\left\{ \bA^{(r)}_t \leq \min\{w, \frac{1}{t}\}  \right\} &\text{if }\forall \ell < t, \bA_{\ell}(w) = 0 \\
						    \ind\left\{ \bA^{(r)}_t \leq w \right\} & \text{if } \exists \ell < t, \bA_{\ell}(w) = 1 \end{array} \right. , \label{eq:a-t-w}\\ 
 \bK_{t} &= \ind\left\{ \bK^{(r)}_t \leq \gamma^t \right\}, \label{eq:k-t-w}
 \end{align}
 where recall that $\gamma^t = 1$ if $t \leq t_0$ and $\min\{ \gamma / t, 1\}$ if $t > t_0$.
\end{itemize}

\textbf{Update}. \textsc{Mod-Min-Weight}$(\tilde{w})$.
\begin{enumerate}
\item We update $w\leftarrow \min\{ w, \tilde{w}\}$.
\item If for all $t \in [t_e]$, $\bA_{t}(w) \cdot \bK_{t} = 0$, the data structure reports ``not activated and kept.''
\end{enumerate}

\textbf{Queries}.
\begin{itemize}
\item For $\textsc{Activation-Time}()$, we find the smallest $t \in [t_e]$ where $\bA_t(w) = 1$ and output $t$. 
\item For $\textsc{Active-Kept}()$, we find all $t \in [t_e]$ where $\bA_t(w) \cdot \bK_t= 1$. 
\end{itemize}
\end{framed}
\caption{The \textsc{Timeline-Mask} Data Structure.}\label{fig:timeline-mask}
\end{figure}

\begin{remark}[Using the Random Oracle $\boldr$ and Finite Precision of Random Variables]\label{rem:generating-randomness} \emph{In Figure~\ref{fig:timeline-mask}, the data structure samples uniform random variables $\bA_t^{(r)}$ and $\bK_t^{(r)}$ for all $t \in [t_e]$ in $[0,1]$, and it does this using the random oracle $\boldr \colon \{0,1\}^* \to \{0,1\}$. The claim that random variables are ``uniform in $[0,1]$'' is an abuse in notation, since the algorithm can only maintain numbers up to a finite bit precision. Notice, however, that the random variables $\bA^{(r)}_t$ and $\bK^{(r)}_t$ are auxiliary, and the random variables affecting the execution of $\textsc{Mod-Min-Weight}(\tilde{w})$, $\textsc{Activation-Time}()$, and $\textsc{Active-Kept}()$ are $\bA_t(w)$ and $\bK_t$, which are thresholded values of $\bA_t^{(r)}$ and $\bK_t^{(r)}$ (see (\ref{eq:a-t-w}) and (\ref{eq:k-t-w})). Hence, it suffices to sample $\bA_t^{(r)}$ and $\bK_t^{(r)}$ up to the bit-precision of $\min\{ w, 1/t\}$ (for the case of $\bA_t(w)$), and $\gamma^t$ (for the case of $\bK_t$). In particular, it suffices to use bit complexity $B$, which is $O(\log t_e)$ plus the largest bit complexity $w$ and $\gamma^t$, in order to sample $\bA_t^{(r)}$ and $\bK_t^{(r)}$. We now expand on how the algorithm uses $\boldr \colon \{0,1\}^* \to \{0,1\}$ to sample these random variables in $\textsc{Init}(x, \boldr)$. Given access to $\boldr \colon \{0,1\}^* \to \{0,1\}$ and $x$:
\begin{itemize}
\item The algorithm encodes the point $x \in [\Delta]^d$ as $\Enc(x) \in \{0,1\}^{d \log \Delta}$.
\item In order to sample the random variable $\bA_t^{(r)}$, it reads off each bit using $\boldr$. Namely, it considers the string $\tau$ given by the concatenation of $\Enc(x)$, $\Enc(t) \in \{0,1\}^{\log t_e}$, ``0'' (specifying ``$\bA$''), and $\Enc(k) \in \{0,1\}^{\log B}$, and sets the $k$-th bit of $\bA_t^{(r)}$ to $\boldr(\tau)$. 
\item Similarly, the random variable $\bK_t^{(r)}$ is also sampled by reading off each bit using $\boldr$. The string $\tau$ now is given by the concatenation of $\Enc(x)$, $\Enc(t) \in \{0,1\}^{\log t_e}$, ``1'' (specifying ``$\bK$''), and $\Enc(k) \in \{0,1\}^{\log B}$, and sets the $k$-th bit of $\bK_t^{(r)}$ to $\boldr(\tau)$. 
\end{itemize}
Note the above upper bounds the number of uniform random bits referenced within $\textsc{Timeline-Mask}$ data structures by $\exp((d\log \Delta / \eps)^{O(1)})$, since we will set $t_e \leq O(\Delta^d / \eps)$ and $B$ will be $\poly(d\log\Delta/\eps)$.
}
\end{remark}

\paragraph{Proof of Lemma~\ref{lem:timeline-mask}.} We note that, since the weight is initialized to $1/2$ and can only decrease (via the update $\textsc{Mod-Min-Weight}(\tilde{w})$), the weight is always between $0$ and $1/2$. We will show that, for any point $x \in [\Delta]^d$ and any fixed weight $w$, the (implicit) random variables $\bA_t(w)$ for $t \in [t_e]$ and $\bK_t$ (as a function of the random oracle $\boldr$) are distributed as a timeline drawn from $\calT(w)$, and a mask drawn from $\calK(t_0, \gamma)$. This is almost immediate---the fixed setting of $x$ means that, over the randomness of $\boldr \colon \{0,1\}^* \to \{0,1\}$, the random variables $\bA_t^{(r)}$ and $\bK_t^{(r)}$ are independent and uniformly distributed in $[0,1]$. Hence, the variables $\bK_t$ are independent and distributed according to $\Ber(\gamma^t)$, as in Definition~\ref{def:mask}. Similarly, the random variables $\bA_t(w)$ are generated by iterating $t = 1, \dots, t_e$, $\bA_t(w)$ drawn by thresholding $\bA_t^{(r)}$ with $\min\{w, 1/t\}$, which is distributed as $\Ber(\min\{ w, 1/t\})$ if $\bA_{\ell}(w) =0$ for all $\ell < t$, or by thresholding $\bA_t^{(r)}$ by $w$, which is distributed as $\Ber(w)$ if $\bA_{\ell}(w) =1$ for some $\ell < t$. This procedure is exactly that which generates $\calT(w)$ in Definition~\ref{def:activate}. Finally, whenever $\textsc{Mod-Min-Weight}(\tilde{w})$ is executed, the weight $w$ can only decrease, and hence $\bA_{t}(\min\{w, \tilde{w}\})$ (which would be the new activation timeline) will be below $\bA_t(w)$. 

\subsection{Description of Insertion-Only Streaming Algorithm}

We now show how to combine the components in Figure~\ref{fig:compute-weights} and Figure~\ref{fig:timeline-mask}, as well as Lemma~\ref{lem:seed-check}, in order to execute the assignment rule from Theorem~\ref{thm:main-structural}. This will complete the description of the streaming algorithm and prove Theorem~\ref{thm:insertion-only}. We first describe what the streaming algorithm maintains in its memory and the update procedure \textsc{Add-Point}$(x)$. These definitions appear in Figure~\ref{fig:insertion-only-maintenance}. Then, we show how to implement the procedure \textsc{Assign}$(x)$ after performing a preprocessing step, $\textsc{Preprocess}()$. 

\begin{figure}
\begin{framed}
\textbf{Maintenance and Updates for} \textsc{IOS-E-Max-Cut}. The algorithm is initialized with parameters $\Delta, d \in \N$, $p \in [1,2]$, and an accuracy parameter $\eps \in (0, 1]$. We describe the maintenance and update procedure \textsc{Add-Point}$(\cdot)$, where for a set of points $X = \{ x_1, \dots, x_n \} \subset [\Delta]^d$ presented as a stream (where necessarily $n \leq \Delta^d$), the stream is processed by executing the sub-routine $\textsc{Add-Point}(x_i)$ for every $i \in [n]$. \\

\textbf{Maintenance}. The streaming algorithm initializes and stores parameters
\begin{align*} 
\lambda &= 60 \qquad  t_e = \frac{\Delta^d \cdot \lambda}{\eps} \qquad \gamma = \frac{(\ln(t_e) + 1)^2 \cdot \lambda}{\eps^2}\\
t_0 &= \max\left\{ \frac{\sqrt{\gamma \cdot \lambda}}{\eps} , \frac{1}{\eps} \right\} \qquad \xi = \poly(t_0/\eps \cdot d\log \Delta), 
\end{align*}
a random oracle $\boldr \colon \{0,1\}^* \to \{0,1\}$, and maintains the following information:
\begin{itemize}
\item An instance of the \textsc{Compute-Weights} streaming algorithm from Figure~\ref{fig:compute-weights}, initialized with $d, \Delta$, parameter $p$, $\eta = 1$ and $\delta = \eps$. 
\item A (random) set $\bC$ which is initially empty, and will consist of tuples of the form $(x, w)$, where $x$ is a point in $[\Delta]^d$ and $w \in (0, 1/2]$.  
\item A (random) set $\bT$ of instances of \textsc{Timeline-Mask} data structures from Figure~\ref{fig:timeline-mask}, initialized by $\textsc{Init}(x, \boldr)$ where $x \in [\Delta]^d$, with $t_e, t_0$ and $\gamma$. The data structures in $\bT$ are indexed by points $x$; given $x$, we can retrieve its data structure in $\bT$ if it exists.
\end{itemize}

\textbf{Update}. \textsc{Add-Point}$(x)$.
\begin{enumerate}
\item\label{en:ios-add-1} Execute the \textsc{Add-Point}$(x)$ update for the \textsc{Compute-Weights} streaming algorithm, and execute \textsc{Weight}$(x)$ and let $\bw$ denote the output and $\bw' = \bw / 2$.
\item\label{en:ios-add-2} Sample a Bernoulli random variable $\bZ \sim \Ber(\min\{\xi \cdot \bw',1\})$. If $\bZ = 1$, let $\bC \leftarrow \bC \cup \{ (x, \bw') \}$.
\item\label{en:ios-add-3} Initialize an instance of \textsc{Timeline-Mask} with $\textsc{Init}(x, \boldr)$ and update \textsc{Mod-Min-Weight}$(\bw')$. If the modified weight produced an output ``not activated and kept,'' delete the data structure. Otherwise, store the data structure in $\bT$, indexed by the point $x$.
\item\label{en:ios-add-4} Execute the following ``clean-up'' procedure:
\begin{enumerate}
\item\label{en:ios-add-4-1} For every tuple $(y, \sigma) \in \bC$ except possibly $(x, \bw')$ from Line~\ref{en:ios-add-2}, execute \textsc{Weight}$(y)$ in \textsc{Compute-Weights} and let $\bsigma'$ be the output divided by two. Sample $\bZ' \sim \Ber\left(\min\{\bsigma' / \sigma, 1\}\right)$; if $\bZ' = 1$, update $(y, \sigma)$ to $(y, \min\{ \sigma, \bsigma'\}) \in \bC$, and if $\bZ' = 0$, delete $(y, \sigma)$ from $\bC$.
\item\label{en:ios-add-4-2} For every \textsc{Timeline-Mask} data structure in $\bT$ for point $y$ except possibly $x$, execute \textsc{Weight}$(y)$ in \textsc{Compute-Weights} and let $\bsigma'$ be the output divided by two. Execute \textsc{Mod-Min-Weight}$(\bsigma')$ for the \textsc{Timeline-Mask} data structure and delete if it outputs ``not activated and kept.''
\end{enumerate}
\end{enumerate}
\end{framed}
\caption{The maintenance and update procedure for $\textsc{IOS-E-Max-Cut}$.}\label{fig:insertion-only-maintenance}
\end{figure}

\paragraph{Proof of Theorem~\ref{thm:insertion-only}: Invariants of the Maintenance and Update.} We now show some important invariants that are satisfied by our maintenance and update procedure from Figure~\ref{fig:insertion-only-maintenance} when we initialize and add points to \textsc{IOS-E-Max-Cut}. These invariants will be used in the preprocessing step, as well as in the $\textsc{Assign}(\cdot)$ sub-routine to connect with Theorem~\ref{thm:main-structural}. Specifically, we consider the following process:
\begin{itemize}
\item First, initialize an instance of $\textsc{IOS-E-Max-Cut}$ with the parameters $\Delta, d \in \N$, $p \in [1,2]$, and accuracy parameter $\eps \in (0, 1)$. 
\item For a set $X = \{ x_1, \dots, x_n \} \subset [\Delta]^d$, we execute $\textsc{Add-Point}(x_1)$, $\textsc{Add-Point}(x_2)$, \dots, $\textsc{Add-Point}(x_n)$. 
\end{itemize}
Note that, in each execution of \textsc{Add-Point}$(x_i)$ of \textsc{IOS-E-Max-Cut}, Line~\ref{en:ios-add-1} executes \textsc{Add-Point}$(x_i)$ in a \textsc{Compute-Weights} streaming algorithm initialized with $\eta = 1$ and $\delta = \eps$ (the parameters $d, \Delta$ and $p$ are passed from the initialization). We apply Lemma~\ref{lem:computing-weights} directly to obtain the following guarantee: with probability at least $1 - \eps$ over the randomness in $\textsc{Compute-Weights}$, every random variable $\bw_i(x_j) \in [0, 1]$ given by the output of $\textsc{Weight}(x_j)$ after inserting points $x_1,\dots, x_i$ (for $j \leq i$) satisfies
\begin{align} 
\dfrac{\sum_{\ell=1}^i d_{\ell_p}(x_{\ell}, x_j)}{\sum_{\ell=1}^i \sum_{k=1}^i d_{\ell_p}(x_{\ell}, x_k)} \leq \bw_i(x_j) \leq 2 \cdot \dfrac{\sum_{\ell=1}^i d_{\ell_p}(x_{\ell}, x_j)}{\sum_{\ell=1}^i \sum_{k=1}^i d_{\ell_p}(x_{\ell}, x_k)}, \label{eq:weight-est} 
\end{align}
and $\bw_1(x_1) = 1$ (as expression (\ref{eq:weight-est}) would be $0/0$). We will assume throughout the remainder of the analysis the above event occurs, and henceforth assume (\ref{eq:weight-est}). The cases in which the above does not occur is considered a ``catastrophic'' failure; this occurs with probability at most $\eps$. We will then similarly upper bound the effect of these failures as in the proof of Theorem~\ref{thm:mpc-full}. We define the random variables, for every $i \in [n]$ and $j \leq i$,
\begin{align} 
\bomega_i(x_j) = \min\left\{ \bw_{\ell}(x_j) : j \leq \ell \leq i\right\}. \label{eq:omega-def}
\end{align}
Note that $\bomega_i(x_j)$ is always at most $1$ (from the fact $\bw_{\ell}(x_j) \leq 1$) and they are greater than $0$ since points are distinct.
Furthermore, we consider (solely for the analysis) keeping track of $n$ many \textsc{Timeline-Mask} data structures, $\DS_1,\dots, \DS_n$, where each data structure will correspond to a point $x_i$ and is initialized with 
\[ \DS_i \leftarrow \textsc{Init}(x_i, \boldr). \]
As the calls for \textsc{Add-Point}$(x_1)$, \dots, \textsc{Add-Point}$(x_i)$ execute, the data structures $\DS_1, \dots, \DS_n$ will be the ones initialized in Line~\ref{en:ios-add-3}. Whenever Line~\ref{en:ios-add-3} and Line~\ref{en:ios-add-4-2} perform an update to a \textsc{Timeline-Mask} data structure for $x_i$ via $\textsc{Mod-Min-Weight}(\cdot)$, we update the corresponding data structures $\DS_i$ as well. We apply Lemma~\ref{lem:timeline-mask} directly, and conclude that for each $i \in [n]$, if $w_i$ is the weight held in the data structure $\DS_i$, it maintains a draw to an activation timeline and mask $(\bA_i, \bK_i)$ where $\bA_i \sim \calT(w_i)$ with the weight stored in $\DS_i$ and $\bK_i \sim \calK(t_0, \gamma)$. This way, the set of data structures $\bT$ is always a subset of $\DS_1,\dots, \DS_n$. 

\begin{claim}\label{cl:invariant}
For any $i \in [n]$, right after the execution of \textsc{Add-Point}$(x_1)$, \dots, \textsc{Add-Point}$(x_i)$, for any fixing of the randomness in \textsc{Compute-Weights}, we satisfy:
\begin{itemize}
\item The set $\bC$ is a random subset of tuples, for every $j \leq i$, $(x_j, \bomega_i(x_j)/2)$ is included in $\bC$ independently and with probability $\min\{\xi \cdot  \bomega_i(x_j) / 2,1\}$, and randomness is over the samples in Line~\ref{en:ios-add-2} and Line~\ref{en:ios-add-4-1} of \textsc{Add-Point}$(x_1)$, \dots, \textsc{Add-Point}$(x_i)$.
\item For $j \leq i$, the \textsc{Timeline-Mask} data structure \emph{$\DS_j$} for point $x_j$ has weight $\bomega_i(x_j)/2$, and $\bT$ contains the data structure \emph{$\DS_j$} if and only if the corresponding timeline and mask $(\bA_j, \bK_j)$, where $\bA_j \sim \calT(\bomega_i(x_j)/2)$ and $\bK_j \sim \calK(t_0,\gamma)$, have some $t \in [t_e]$ with $\bA_{j,t} \cdot \bK_{j,t} = 1$. 
\end{itemize}
\end{claim}

\begin{proof}
We prove the above claim by induction on $i$. For $i = 1$, only \textsc{Add-Point}$(x_i)$ is executed; $\bw$ and $\bw'$ in Line~\ref{en:ios-add-1} are set to $\bw_1(x_1)=1$ and $1/2$, respectively. Hence, $\bomega_1(x_1) = \bw_1(x_1)$. Line~\ref{en:ios-add-2} samples $\bZ \sim \Ber(\min\{\xi \cdot \bomega_1(x_1) / 2,1\})$, and hence, $\bC$ contains $(x_1, \bomega_1(x_1)/2)$ independently with probability $\min\{\xi \cdot \bomega_1(x_1)/2,1\}$ right after Line~\ref{en:ios-add-2} and Line~\ref{en:ios-add-4-1} does not modify it. For the second item, Line~\ref{en:ios-add-3} initializes the data structure $\DS_1$ with $\textsc{Init}(x_1, \boldr)$ and sets the weight to $\bomega_1(x_1) / 2$. If $\DS_1$ holds the activation timeline and mask $(\bA_1, \bK_1)$ where $\bA_{1,t} \cdot \bK_{1,t} = 0$ for all $t \in [t_e]$, the data structure, Line~\ref{en:ios-add-3} deletes it, and Line~\ref{en:ios-add-4-2} does not modify it. Hence, $\bT$ contains $\DS_1$ if and only if there exists $t \in [t_e]$ where $\bA_{1,t} \cdot \bK_{1,t} = 1$.

We assume, for inductive hypothesis that Claim~\ref{cl:invariant} holds for $i$, and we execute $\textsc{Add-Point}(x_{i+1})$. Similarly to above, the set $\bC$ contains the tuple $(x_{i+1}, \bomega_{i+1}(x_{i+1}) / 2)$ independently with probability $\min\{\xi \cdot \bomega_{i+1}(x_{i+1})/2,1\}$ (from Line~\ref{en:ios-add-2}); $\DS_{i+1}$ is initialized in Line~\ref{en:ios-add-3} and has weight $\bomega_{i+1}(x_{i+1})/2$ and is stored in $\bT$ if and only if $\bA_{i+1, t} \cdot \bK_{i+1,t} = 1$ for some $t \in [t_e]$. It remains to check the condition for $j \leq i$. For the first item, the inductive claim is that right after the execution of \textsc{Add-Point}$(x_i)$ (thus, right before \textsc{Add-Point}$(x_{i+1})$), $(x_j, \bomega_i(x_j)/2)$ is included in $\bC$ independently with probability $\min\{\xi \cdot \bomega_i(x_j)/2,1\}$. Upon executing \textsc{Add-Point}$(x_{i+1})$, 
\begin{itemize}
\item Line~\ref{en:ios-add-4-1} will eventually consider the tuple $(x_j, \bomega_i(x_j)/2)$. When it does, notice $\bsigma'$ becomes $\bw_{i+1}(x_j)/2$, and also $\min\{ \bomega_i(x_j), \bw_{i+1}(x_j)\} = \bomega_{i+1}(x_j)$. Hence, $(x_j, \bomega_{i+1}(x_j)/2)$ is in $\bC$ independently with probability 
\begin{align*} 
\underbrace{\min\left\{ \frac{\xi \cdot \bomega_i(x_j)}{2}, 1\right\}}_{\text{induction}} \times \underbrace{\min\left\{ \frac{\bw_{i+1}(x_j)}{2} \cdot \frac{2}{\bomega_i(x_j)}, 1\right\}}_{\text{case $\bZ'=1$}} &= \min\left\{\frac{\xi \bw_{i+1}(x_j)}{2}, \frac{\xi\bomega_{i}(x_j)}{2}, 1 \right\} \\
&= \min\left\{ \frac{\xi \cdot \bomega_{i+1}(x_j)}{2}, 1\right\}
\end{align*}
\item Line~\ref{en:ios-add-4-2} may consider the point $x_j$ if $\DS_j$ is in $\bT$. Suppose $\DS_j$ is not in $\bT$; by induction, this is because $\DS_j$ (which had weight $\bomega_{i}(x_{j})$) held an activation timeline and mask $(\bA_{j}, \bK_j)$ where $\bA_{j,t} \cdot \bK_{j,t} = 0$ for all $t \in [t_e]$. For the analysis's sake, we update $\DS_j$ with $\textsc{Mod-Min-Weight}(\bw_{i+1}(x_j))$, so $\bA_{j,t} \cdot \bK_{j,t}$ is still $0$ for all $t \in [t_e]$ (by Lemma~\ref{lem:timeline-mask}). If $\DS_j$ is in $\bT$, then $\DS_j$ is eventually processed by Line~\ref{en:ios-add-4-2} and the weight updated via \textsc{Mod-Min-Weight}$(\bsigma')$ where $\bsigma' = \bw_{i+1}(x_j)/2$. Hence, the new weight becomes $\min\{ \bomega_i(x_j)/2, \bw_{i+1}(x_j)/2\} = \bomega_{i+1}(x_j) / 2$. 
\end{itemize}
\end{proof}

\paragraph{Proof of Theorem~\ref{thm:insertion-only}: \textsc{Preprocess}$()$.} We present a ``preprocessing'' sub-routine \textsc{Preprocess}$()$ in Figure~\ref{fig:insertion-only-preprocess}. The sub-routine is used to prepare the timeline-mask summary $\bP$, as well as find the best seed in Figure~\ref{fig:assign}. As discussed for Claim~\ref{cl:invariant}, we consider a fixed setting of \textsc{Compute-Weights} whose weights $\bw_i(x_j)$ satisfy (\ref{eq:weight-est}); recall, also, the definition of $\bomega_i(x_j)$ as the minimum of $\bw_{\ell}(x_j)$ for $j \leq \ell \leq i$. The $n$ \textsc{Timeline-Mask} data structures $\DS_1,\dots, \DS_n$ hold draws of the activation timeline and masks $(\bA_i, \bK_i)$, where $\bA_i \sim \calT(\bomega_n(x_i)/2)$ and $\bK \sim \calK(t_0,\gamma)$. The first step of $\textsc{Preprocess}()$ (Line~\ref{en:prepro-1} in Figure~\ref{fig:insertion-only-preprocess}), is making the weights of time $\textsc{Timeline-Mask}$ for $x_j$ be $\bw_n(x_j) / 30$, we do this by executing $\textsc{Mod-Min-Weight}(\bw_n(x_j)/30)$, which sets the new weight to $\min\{ \bw_n(x_j)/30, \bomega_n(x_j)/2 \}$. The following claim guarantees that the minimum of the two weights is $\bw_n(x_j) / 30$, so we can assume henceforth, that the weight of \textsc{Timeline-Mask} for $x_j$ is $\bw_n(x_j)/30$ after Line~\ref{en:prepro-1}. 
\begin{claim}\label{cl:final-min}
    Suppose all $j \leq i$ have weights $\bw_i(x_j)$ satisfying (\ref{eq:weight-est}). Then, for all $j \in [n]$, $\bw_n(x_j) \leq 15 \cdot \bomega_n(x_j)$. 
\end{claim}
\begin{proof}
    Let $j \leq h < n$ be that which satisfies $\bw_h(x_j) = \bomega_n(x_j)$ (if $h=n$, then we are automatically done), and consider $\bc \sim \{ x_1, \dots, x_n \}$. We first expand using the lower bound on (\ref{eq:weight-est}) and the triangle inequality to $\bc$ in order to say
    \begin{align}
        \sum_{\ell=1}^n d_{\ell_p}(x_{\ell}, x_j) &\leq \bw_h(x_j) \left( \sum_{\ell=1}^h \sum_{i=1}^h d_{\ell_p}(x_i, x_{\ell}) \right) +\sum_{\ell=h+1}^n \Ex_{\bc}\left[  d_{\ell_p}(x_{\ell}, \bc) \right]  \label{eq:bound} \\
        &\qquad+ (n-h-1) \Ex_{\bc}\left[ d_{\ell_p}(\bc,x_j)\right]. \nonumber
    \end{align}
    We simply the right-hand side of (\ref{eq:bound}) using the following two facts,
    \begin{align*}
        \sum_{\ell=h+1}^{n} \Ex_{\bc}\left[ d_{\ell_p}(x_{\ell}, \bc)\right] &\leq \frac{1}{n} \sum_{i=1}^n \sum_{\ell=1}^n d_{\ell_p}(x_i, x_{\ell}), \\
        \Ex_{\bc}\left[ d_{\ell_p}(\bc, x_j)\right] &= \bw_h(x_j) \Ex_{\bc}\left[ \dfrac{d_{\ell_p}(\bc, x_j)}{\bw_h(x_j)} \right] \leq 4\bw_h(x_j) \Ex_{\bc}\left[ \sum_{\ell=1}^h d_{\ell_p}(\bc, x_{\ell}) \right] \\
        &\leq \frac{4\bw_h(x_j)}{n} \sum_{i=1}^n \sum_{\ell=1}^n d_{\ell_p}(x_i,x_{\ell}),
    \end{align*}
    where the second inequality applies Lemma~\ref{lem:weight-to-compatible} with the weights $\bw_h(x_1), \dots, \bw_h(x_h)$ and the point $\bc$. Thus, we simplify (\ref{eq:bound}) by
    \begin{align*}
    \sum_{\ell=1}^n d_{\ell_p}(x_{\ell}, x_j) &\leq \bw_h(x_j) \left(\sum_{\ell=1}^h \sum_{i=1}^h d_{\ell_p}(x_i, x_{\ell}) + \frac{4(n-h-1)}{n} \sum_{i=1}^n \sum_{\ell=1}^n d_{\ell_p}(x_i, x_{\ell}) \right) \\
        &\qquad + \frac{1}{n} \sum_{i=1}^n \sum_{\ell=1}^n d_{\ell_p}(x_i, x_{\ell}) \\
        &\leq \left( 5 \cdot \bw_h(x_j) + \frac{1}{n} \right) \sum_{i=1}^n \sum_{\ell=1}^n d_{\ell_p}(x_i, x_{\ell}),
    \end{align*}
    which implies $\bw_n(x_j) \leq 10 \bw_h(x_j) + 2/n$ using (\ref{eq:weight-est}). Finally, applying Lemma~\ref{lem:weight-to-compatible} for $\bw_h(x_1), \dots, \bw_h(x_h)$, we have $\bw_h(x_j) \geq 1/(2h)$ and $h \leq n$, $10 \bw_h(x_j) + 2/n \leq 15 \bw_h(x_j) - 5 /(2h) + 2/n \leq 15 \bw_h(x_j)$.
\end{proof}

For the point $x_i$, we let $t_i$ be its activation time in $\bA_i$. We first state and prove the following lemma, which will help us apply Lemma~\ref{lem:seed-check} (and later, Theorem~\ref{thm:main-structural}) by showing that the sampling weights guaranteed from Claim~\ref{cl:invariant} are metric compatible. 

\begin{lemma}\label{lem:min-metric-compat}
Suppose that, for $X = \{ x_1,\dots, x_n \} \subset \R^d$, the weights $\bw_{i}(x_j)$ satisfy (\ref{eq:weight-est}) for all $j \leq i$. Then, the weights $w_1', \dots, w_n'$ given by $w_j' = \bomega_n(x_j)/2$ from (\ref{eq:omega-def}) are $8$-compatible, and the weights $\bw_n(x_1)/30, \dots, \bw_n(x_n)/30$ are $60$-compatible.
\end{lemma}

\begin{proof}
Fix a setting of $j\in[n]$ and suppose that $h \geq j$ is that which minimizes $\bomega_n(x_j) = \bw_h(x_j)$. We show that Definition~\ref{def:metric-compat} holds by applying Lemma~\ref{lem:weight-to-compatible} with the points $\{ x_1,\dots, x_h\}$, the weight $w_j = \bw_h(x_j)$, and the point $c = x_i$. Letting $w_j' = \bw_h(x_j) /2$, Lemma~\ref{lem:weight-to-compatible} implies
\begin{align*}
\dfrac{d_{\ell_p}(x_{i}, x_j)}{w_j'} \leq 8 \sum_{k=1}^h d_{\ell_p}(x_{i}, x_k),
\end{align*}
and we use the fact distances are non-negative to conclude $\sum_{k=1}^h d_{\ell_p}(x_i, x_k) \leq \sum_{k=1}^n d_{\ell_p}(x_i, x_j)$. The claim that $\bw_n(x_1)/30, \dots, \bw_n(x_n)/30$ are $60$-compatible by Lemma~\ref{lem:weight-to-compatible} as well.
\end{proof}

\begin{claim}\label{cl:preprocess-sum}
After executing \textsc{Preprocess}$()$ from Figure~\ref{fig:insertion-only-preprocess}, the memory contents of \textsc{IOS-E-Max-Cut} satisfy the following guarantees:
\begin{itemize}
\item The set $\bP$ is the timeline-mask summary $\textsc{Summ}(\bA_j, \bK_j : j \in [n])$ where $\bA_j \sim \calT(\bw_n(x_j) / 30)$ and $\bK_{j} \cdot \calK(t_0, \gamma)$ (see Definition~\ref{def:timeline-mask-sum}).
\item Let $\bm = |\bS^{t_0}|$, and consider the family of complete cuts $z(\sigma) \in \{0,1\}^{[n] \times \{0,1\}}$ which are parametrized by $\sigma \in \{0,1\}^{\bm}$ given by $z(\sigma)_i = \textsc{Assign}_{\sigma}(x_i, t_i)$.  Then, with probability $1-\eps$, 
\[ f(z(\bsigma^*)) \leq \min_{\sigma \in \{0,1\}^m} f(z(\sigma)) + \eps \sum_{i=1}^n \sum_{j=1}^n d_{\ell_p}(x_i,x_j). \]
\end{itemize}
\end{claim}

\begin{proof}
We refer to Figure~\ref{fig:insertion-only-preprocess} and apply the guarantees of Claim~\ref{cl:invariant} with $i = n$, as well as Lemma~\ref{lem:seed-check}. Note that, from Claim~\ref{cl:invariant} and Claim~\ref{cl:final-min}, the \textsc{Timeline-Mask} data structures in $\bT$ provide access to $(\bA_j, \bK_j)$ with $\bA_j \sim \calT(\bw(x_j)/30)$ and $\bK_j \sim \calK(t_0, \gamma)$ with $\bA_{j,t} \cdot \bK_{j,t} = 1$ for some  $t \in [t_e]$, and Line~\ref{en:prepro-1} assembles these into $\bP = \textsc{Summ}(\bA_j, \bK_j : j \in [n])$. Recall that $\bS^{t_0}$ consists of all points $x_j$ whose activation time $t_j$ in $\bA_j$ occurs before $t_0$; since $\bK_{j,t_j} = 1$ if $t_j \leq t_0$, every $x_j \in \bS^{t_0}$ must is included in $\bP$, and the algorithm may recover $\bS^{t_0}$. Thus, for fixed $\bP$, the algorithm \textsc{Assign}$_{\sigma}(x_i, t_i)$ defines an implicit family $\calF$ parameterized by $\sigma \in \{0,1\}^{\bm}$. Furthermore, the expectation of $\bm = |\bS^{t_0}|$ is at most $O(t_0)$ since it consists of the points which are activated by time $t_0$, and we have $\|\bw_n\|_1 = O(1)$, and by Markov's inequality, it is at most $O(t_0 / \eps)$ with probability at least $1-\eps$. If that is the case, then $\log|\calF| = \bm \leq O(t_0 / \eps)$, which is why we set $\xi = \poly(t_0 / \eps \cdot d \log \Delta).$

We verify that we can apply Lemma~\ref{lem:seed-check}. From Claim~\ref{cl:invariant}, $\bC$ is a random set which contains tuple $(x_j, \bomega_n(x_j)/2)$ independently with probability $\min\{\xi \cdot \bomega_n(x_i)/2,1\}$. Since the algorithm stores $(x_j, \bomega_n(x_j)/2) \in \bC$ and $\xi$, it knows the sampling probability $\min\{\xi \bomega_n(x_j)/2,1\}$. Finally, we note the algorithm may execute $\textsc{Assign}_{\sigma}(x_i, t_i)$ for each $x_i$ which participates in a tuple $(x_i, \bomega_n(x_i)/2) \in \bC$. The reason is that it may recover the activation time $t_i$ (as discussed in Remark~\ref{rem:init}); it computes $\textsc{Weight}(x_i)$ to recover $\bw_n(x_i)$, initializes the \textsc{Timeline-Mask} data structure with $\textsc{Init}(x_i, \boldr)$, executes \textsc{Mod-Min-Weight}$(\bw_n(x_i)/2)$, and then sets $t_i \leftarrow \textsc{Activation-Time}()$. By setting the $\lambda$ in Lemma~\ref{lem:seed-check} to $8$ the weights $\bomega_n(x_1), \dots, \bomega_n(x_n)$ are $\lambda$-compatible (by Lemma~\ref{lem:min-metric-compat}), the family $\calF$ has size $\log|\calF| \leq O(t_0 / \eps)$, so $\xi$ in Figure~\ref{fig:insertion-only-maintenance} is set accordingly as $\log n \leq d \log \Delta$, Lemma~\ref{lem:seed-check} implies $\bsigma^*$ satisfies the approximately optimal condition.
\end{proof}

\begin{figure}
\begin{framed}
\textbf{Preprocess After Insertions}. We assume that for a set $X = \{ x_1,\dots, x_n \} \subset [\Delta]^d$, the sub-routine $\textsc{Add-Point}(x_i)$ has been executed for every $i \in [n]$ (see Figure~\ref{fig:insertion-only-maintenance}). The \textsc{Preprocess}$()$ sub-routine prepares the following memory contents:
\begin{itemize}
\item The set $\bP = \textsc{Summ}(\bA_j, \bK_j : j \in [n])$, as per Definition~\ref{def:timeline-mask-sum}.
\item A seed $\bsigma^* \in \{0,1\}^{\bm}$ for activation timelines $\bA_1,\dots ,\bA_n$, as per Definition~\ref{def:seed}.
\end{itemize}

\textsc{Preprocess}$()$.
\begin{enumerate}
\item\label{en:prepro-1} Initialize $\bP = \emptyset$. Each \textsc{Timeline-Mask} for $x_j$ in $\bT$, stores the point $x_j$ and executes $\textsc{Weight}(x_j)$ to obtain $\bw_n(x_j)$. We execute $\textsc{Mod-Min-Weight}(\bw_n(x_j)/30)$ for $x_j$'s timeline in order to make the final weight $\bw_n(x_j)/30$ (from Claim~\ref{cl:invariant} and Claim~\ref{cl:final-min}).
\item The $\textsc{Timeline-Mask}$ for $x_j$ in $\bT$ then executes $\textsc{Activation-Time}()$ to obtain $t_j$, and $\textsc{Active-Kept}()$, letting $\ell_1, \dots, \ell_k$ be the output, update for each $\kappa \in [k]$, $\bP \leftarrow \bP \cup \left\{ (x_j, \ell_{\kappa}; w_j, t_j)  \right\}$. 
\item\label{en:prepro-2} Let $\calF$ denote the family of complete cuts $z(\sigma) \in \{0,1\}^{[n] \times \{0,1\}}$ (referenced implicitly and parametrized by $\sigma \in \{0,1\}^{\bm}$) given by $z(\sigma)_i = \textsc{Assign}_{\sigma}(x_i, t_i)$ where $t_i$ is the activation time of $x_i$, and the set $\bP$ is used as the timeline-mask summary (see Figure~\ref{fig:assign}).
\item\label{en:prepro-3} Execute the algorithm from Lemma~\ref{lem:seed-check} with accuracy $\eps$ and failure probability $\eps$ to select $\bsigma^*$, where we use $\bC$ for the random set (note that the activation time $t_i$ for $x_i$ can be recovered since $(x_i, w_i) \in \bC$, see Remark~\ref{rem:init}). 
\end{enumerate}

\end{framed}
\caption{Preprocessing for the Insertion-only Algorithm.}\label{fig:insertion-only-preprocess}
\end{figure}

\paragraph{Proof of Theorem~\ref{thm:insertion-only}: Space Complexity.} Notice from the maintenance procedure in Figure~\ref{fig:insertion-only-maintenance} and preprocessing procedure~\ref{fig:insertion-only-preprocess}, the streaming algorithm maintains:
\begin{itemize}
\item A pointer to the random oracle $\boldr$, the parameters $\lambda, t_e, \gamma, t_0$ and $\xi$. Assuming that $\boldr$ can be stored in small space (as we will use Nisan's pseudorandom generator), the parameters use $O(d)$ words of space. 
\item A \textsc{Compute-Weights} streaming algorithm initialized with $d,\Delta$, and $p, \eta = 1$ and $\delta = \eps$, using $\poly(d \log(\Delta / \eps))$ words of space (Lemma~\ref{lem:computing-weights}). 
\item The set $\bC$ uses $O(d)$ words per element. After executing \textsc{Add-Point}$(x_1)$, \dots, \textsc{Add-Point}$(x_i)$, $|\bC|$ is a sum of $i$ independent random variables in $[0,1]$ whose expectation is $O(\xi)$. By Hoeffding's inequality, $|\bC| \leq O(d \xi \log \Delta)$ with probability $1 - \Delta^{-2d}$.
\item The set $\bT$ of \textsc{Timeline-Mask} data structures each use $O(d)$ words of memory. After executing \textsc{Add-Point}$(x_1)$, \dots, \textsc{Add-Point}$(x_i)$, $|\bT|$ is at most the number of pairs $(j, \ell) \in [i] \times [t_e]$ such that the activation timelines and masks $\bA_{j, \ell} \cdot \bK_{j,\ell} = 1$, and furthermore, by Claim~\ref{cl:invariant},
\begin{align*}
\Ex\left[ \sum_{j=1}^i \sum_{\ell=1}^{t_e} \bA_{j,\ell} \cdot \bK_{j,\ell} \right] \leq \sum_{\ell=1}^{t_e} \frac{\gamma}{\ell} \sum_{j=1}^i \frac{\bomega_i(x_j)}{2},
\end{align*}
and $\bomega_i(x_j) \leq \bw_i(x_j)$, so $\sum_{j=1}^i \bomega_i(x_j) \leq 1$, giving us the bound $O(\gamma \log t_e)$ and simplifies to $\poly(d \log \Delta / \eps)$. By Hoeffding's inequality once more, $|\bT|$ is at most $\poly(d \log \Delta / \eps)$ with probability $1 - \Delta^{-2d}$.  
\end{itemize}
Finally, the summary set $\bP$ has space complexity dominated by that of $\bT$, and the seed $\bsigma^*$ uses $m$ bits, where $|m| \leq |\bP|$. We may now take a union bound over at most $\Delta^d$ steps of the algorithm and conclude that the space complexity does not exceed $\poly(d\log \Delta / \eps)$ with high probability throughout the entire execution.

\paragraph{Proof of Theorem~\ref{thm:insertion-only}: Query and Output Guarantees.} The query \textsc{Assign}$(\cdot)$ for $\textsc{IOS-E-Max-Cut}$ is described in Figure~\ref{fig:insertion-only-query}. We will show that, after executing $\textsc{Add-Point}(x_1)$, \dots, $\textsc{Add-Point}(x_n)$, we may run $\textsc{Preprocess}()$ to assemble the timeline-mask summary $\bP$ and the approximately optimal seed $\bsigma^*$. For the final quality of the cut (and hence the proof of Theorem~\ref{thm:insertion-only}), we apply Theorem~\ref{thm:main-structural}. We let $X = \{x_1,\dots, x_n\}$ with $\ell_p$ metric, and we let $w_1,\dots, w_n$ be $\bw_n(x_1) / 30, \dots, \bw_n(x_n)/30\in (0, 1/2]$, which are $60$-compatible by Lemma~\ref{lem:min-metric-compat} whenever (\ref{eq:weight-est}) holds, which occurs with probability $1-\eps$ over the randomness in \textsc{Compute-Weights} (see also Lemma~\ref{lem:computing-weights}). Since $n$ is at most $\Delta^d$, the parameter settings $t_e, \gamma$ and $t_0$ are according to those in Figure~\ref{fig:insertion-only-maintenance}, the parameters satisfy the guarantees of Theorem~\ref{thm:main-structural}.

By Claim~\ref{cl:invariant}, Claim~\ref{cl:final-min}, and Claim~\ref{cl:preprocess-sum}, the set $\bP$ is a timeline-mask summary for the activation timelines and masks $(\bA_i, \bK_i)$, where $\bA_i \sim \calT(\bw_n(x_i)/30)$ and $\bK_i \sim \calK(t_0,\gamma)$. Thus, any seed $\sigma$ for $\bA_1,\dots, \bA_n$ gives rise to the complete cut $\bz(\sigma)$ specified from $\textsc{Assign}_{\sigma}(x_i, t_i)$, and the algorithm \textsc{Assign}$(x_i)$ provides query access to $\bz(\bsigma^*)$. We apply the second item in Claim~\ref{cl:preprocess-sum} to say that except with probability $1 - \eps$ over the randomness in $\bC$, we satisfy the guarantees $f(\bz(\bsigma^*))$ is approximately optimal; so by a union bound over the randomness in $\textsc{Compute-Weights}$ and the random set $\bC$, we have
\begin{align*}
f(\bz(\bsigma^*)) \leq \min_{\sigma \in \{0,1\}^m} f(\bz(\sigma)) + \eps \sum_{i=1}^n \sum_{j=1}^n d_{\ell_p}(x_i,x_j)
\end{align*}
with probability at least $1-2\eps$. We take expectations over $\bA$ and $\bK$ and applying Theorem~\ref{thm:main-structural}, so, except with probability $1 - 2\eps$ over randomness in $\textsc{Compute-Weights}$ and $\bC$,
\begin{align*}
\Ex_{\bA, \bK}\left[ f(\bz(\bsigma^*)) \right] \leq \min_{\substack{z \in \{0,1\}^{n\times\{0,1\}} \\ z_i \neq (0,0)}} f(z) + 2\eps \sum_{i=1}^n \sum_{j=1}^n d_{\ell_p}(x_i, x_j). 
\end{align*}
Finally, note that the maximum contribution that a cut may have is always at most $\sum_{i=1}^n \sum_{j=1}^{n} d_{\ell_p}(x_i, x_j)$, so when we take expectation of the cut quality (even after averaging over randomness in \textsc{Compute-Weights} and $\bC$), we satisfy
\begin{align*}
\Ex\left[ f(\bz(\bsigma^*)) \right] \leq \min_{\substack{z \in \{0,1\}^{n\times\{0,1\}} \\ z_i \neq (0,0)}} f(z) + 4\eps \sum_{i=1}^n \sum_{j=1}^n d_{\ell_p}(x_i, x_j),
\end{align*}
which concludes Theorem~\ref{thm:insertion-only}.

\begin{figure}
\begin{framed}
\textbf{Query Algorithm for }\textsc{IOS-E-Max-Cut}. We assume that the sub-routines \textsc{Add-Point}$(x_1)$, \dots, $\textsc{Add-Point}(x_n)$ have been executed, and \textsc{Preprocess}$()$ has been run to store the timeline-mask summary $\bP$ and seed $\bsigma^*$ in memory. We specify the \textsc{Assign}$(x)$ sub-routine, which receives as input a point $x$, promised to belong to $\{x_1,\dots, x_n\}$, and outputs $(1,0)$ or $(0,1)$, corresponding to whether the point is assigned to the 0-side or 1-side of the cut. \\

\textbf{Query}. \textsc{Assign}$(x)$.
\begin{enumerate}
\item First, we obtain the activation time for $x$ with the following procedure:
\begin{itemize}
\item If there is a \textsc{Timeline-Mask} data structure for point $x$ in $\bT$, then query the data structure with \textsc{Activation-Time}$()$ to obtain the activation time $t$. 
\item If there is not: query $w \leftarrow \textsc{Weight}(x)$ in \textsc{Compute-Weights} (which is set to $\bw_n(x)$), initialize a \textsc{Timeline-Mask} with $\textsc{Init}(x, \boldr)$ and run \textsc{Mod-Min-Weight}$(w/30)$; then, let $t\leftarrow\textsc{Activation-Time}()$ (see Remark~\ref{rem:init}).
\end{itemize}
\item Since $\bP$ and $\bsigma^*$ have been stored by $\textsc{Preproces}()$, we execute \textsc{Assign}$_{\bsigma^*}(x, t)$, and output whatever it outputs.
\end{enumerate}
\end{framed}
\caption{Query for the Insertion-only Algorithm.}\label{fig:insertion-only-query}
\end{figure}

%% file: analysis.tex

\section{Proof of Theorem~\ref{thm:main-structural}}\label{sec:greedy-process}

To prove Theorem~\ref{thm:main-structural}, we describe a randomized algorithm, which receives as input access to the metric space, certain parameters which describe the random process, and (perhaps counter-intuitively,) a description of a minimizing cut $z^*$ of the $f$. The random process then produces a (random) partial cut $\bz$, as well as a string $\bsigma^*$ which (i) will be short, and (ii) will encode the necessary information needed to simulate the process without access to $z^*$. We will then show that the cut $\bz$ which the process outputs is exactly the cut which the assignment rule $\textsc{Assign}_{\bsigma^*}(\cdot, \cdot)$ (in Figure~\ref{fig:assign}) can recreate (without knowing $z^*$).

\paragraph{Process.} We call the random process $\textsc{GreedyProcess}$, which we describe as a randomized algorithm with the following input/output behavior:
\begin{itemize}
\item \textbf{Input}: access to the metric space $(X = \{ x_1,\dots, x_n \}, d)$, a vector $w \in (0, 1/2]^{n}$ encoding weights, times $t_0, t_e \in \N$, a parameter $\gamma \geq 1$, and a matrix $z^* \in \{0,1\}^{n \times \{0,1\}}$ which will be a minimizer of $f$ encoding cuts:
\[ z^* = \argmin\left\{ f(z) : z \in \{0,1\}^{n \times \{0,1\}} \text{ where every $i \in [n]$ has $z_i \in \{ (1, 0) , (0, 1) \}$} \right\}. \]
\item \textbf{Output}: the process outputs a random (partial) cut $\bz \in \{0,1\}^{n\times \{0,1\}}$ where $\bz_i \in \{ (1, 0), (0, 1), (0, 0)\}$ for all $i\in[n]$. In addition, a number $\bm \in \N$, and a string $\bsigma^* \in \{0,1\}^{\bm}$.
\end{itemize}
The algorithm is described below.

\paragraph{Description.} \textsc{GreedyProcess}$((X = \{ x_1,\dots, x_n\}, d), w, \gamma, t_0, t_e, z^*)$ proceeds in rounds $t=1, 2, \dots, t_e$, and will 
generate the following sequences of random variables:
\begin{align*}
    &\bz^t \text{: a matrix $\{0,1\}^{n\times \{0,1\}}$ which encodes a partial cut at time $t$.} \\
    &\bw^t \text{: a vector of $n$ weights at time $t$.}\\
    &\bS^t \text{: a set of assigned points at time $t$.}\\
    &\bg^t \text{: a binary $n\times \{0,1\}$ matrix which encodes the  greedy assignment at time $t$.}\\
    &\bA^t \text{: an $n \times t$ binary matrix.} \\
    &\bK^t \text{: an $n \times t$ binary matrix.}
\end{align*}

We let $\bz^0$ be the all zero's matrix, $\bA^0$ and $\bK^0$ be the empty $n \times 0$ matrices, and $\bS^0 = \emptyset$. For time $1 \leq t \leq t_e$, we instantiate $\bA^t$ and $\bK^t$ (which adds the $t$-th column to $\bA^{t-1}$ and $\bK^{t-1}$):
\begin{enumerate}
\item\label{en:greedy-ln1} For each $i \in [n]$, we set $\bw_i^t$ to $\min\{ w_i, 1/t\}$ if $i \notin \bS^{t-1}$, and $w_i$ otherwise. Furthermore, we let $\gamma^{t}$ be $1$ if $t \leq t_0$, and $\min(\gamma/t, 1)$ otherwise.

\item\label{en:greedy-ln2} For each $i \in [n]$, if $t \leq t_0$ let $ \bg^t_i = z^*_i$,
else for $t > t_0$, we let 
\begin{align*}
\tilde{\bc}_{i, 0}^{t} &= \sum_{j=1}^n d(x_i, x_j) \left( \frac{1}{t-1} \sum_{\ell=1}^{t-1} \dfrac{\bA_{j,\ell}^{t-1} \cdot \bK_{j,\ell}^{t-1}}{\bw_{j}^\ell \cdot \gamma^\ell} \right) \bz_{j, 0}^{t-1} \\
\tilde{\bc}_{i,1}^t &= \sum_{j=1}^n d(x_i, x_j) \left( \frac{1}{t-1} \sum_{\ell=1}^{t-1} \dfrac{\bA_{j,\ell}^{t-1} \cdot \bK_{j,\ell}^{t-1}}{\bw_{j}^\ell \cdot \gamma^\ell} \right) \bz_{j, 1}^{t-1}
\end{align*}
and let $\bg^t_i$ be set to $(1, 0)$ if $\tilde{\bc}_{i, 0}^{t} < \tilde{\bc}_{i, 1}^t$ or $(0, 1)$ otherwise. 
\item\label{en:greedy-ln3} Update the matrices $\bA^t$ and $\bK^t$ by setting
\begin{align*}
    &\bA^t_{i,\ell} = \begin{cases}
        \bA_{i,\ell}^{t-1} & \text{ if $\ell \leq t-1$} \\
        \Ber(\bw_{i}^t) & \ell = t \\
    \end{cases}, \qquad 
    \bK^t_{i,\ell} = \begin{cases}
        \bK_{i,\ell}^{t-1} & \text{ if $\ell \leq t-1$} \\
        \Ber(\gamma^{\ell}) & \ell = t
    \end{cases}, \\
    &\bS^t = \bS^{t-1} \cup \{i \text{ for which } \bA^t_{i,t}=1 \}
\end{align*}
\item\label{en:greedy-ln4} For each $i \in [n]$, define \begin{align*}
    \bz^t_i = \begin{cases}
        \bz^{t-1}_i &  i \in \bS^{t-1} \\
        \bg_i^t &  i \in \bS^{t} \land i \notin \bS^{t-1} \\
        (0,0) &  \text{else} \\
    \end{cases}
\end{align*}
\item Let $\bm = |\bS^{t_0}|$, and consider a natural ordering $\pi$ where $\pi_j$ is the $j$-th element of $\bS^{t_0}$. Let
\begin{align*}
    \bsigma_{i}^* = \bz^*_{\pi_i}
\end{align*}
\end{enumerate}
When $t > t_e$, the process completes and we output $(\bz, \bm, \bsigma^*)$, and we will refer to $\bA^{t_e}$ and $\bK^{t_e}$ as $\bA$ and $\bK$, respectively.

\ignore{\begin{definition}[Greedy Process Distribution] 
For a metric space $(X = \{ x_1,\dots, x_n \}, d)$, we consider a weight vector $w \in \R_{\geq 0}^n$, and parameters $\gamma \geq 1$ and $t_0,t_e \in \N$. We consider the following two distributions:
\begin{itemize} 
\item $\calD^*(w, \gamma, t_0, t_e)$ is the distribution from outputs to the process $\textsc{GreedyProcess}((X, d), w, \gamma, t_0, t_e, z^*)$.
\item $\calD(w, \gamma, t_0, t_e)$ is the distribution over pairs of matrices given by sampling $(\bz, \bA, \bK, \bm, \bsigma^*)$ from $\calD^*(w, \gamma, t_0,t_e)$ and outputting $(\bA, \bK)$.
\end{itemize}
\end{definition}

We state a few simple claims which emphasize the main aspects of the greedy process above which will lead to a local simulation procedure, as well as connect with the guarantees of Theorem~\ref{thm:distribution-for-matrix}. In the subsequent claims, we let $(\bz, \bA, \bK, \bm, \bsigma^*)$ be the output of \textsc{GreedyProcess}$((X, d), w, \gamma, t_0, t_e, z^*)$.  }

\subsection{From \textsc{GreedyProcess} to \textsc{Assign}}

We collect the following observations, which allows us to connect the analysis of \textsc{GreedyProcess} to the assignment rule from Theorem~\ref{thm:main-structural}. In particular, we will show that the distribution over matrices $(\bA, \bK)$ at the end of the above greedy process is exactly that which is defined for the statement of Theorem~\ref{thm:main-structural}. This allows us to prove Theorem~\ref{thm:main-structural} by showing that the partial cut produced by the greedy process is approximately optimal. 

\begin{observation}\label{obs:1}
Since $\gamma^{\ell} = 1$ for all $\ell \leq t_0$, $\bK_{i,\ell}^{t} = 1$ for all $\ell \leq t_0$ and $i \in [n]$. Therefore,
\[ \bS^{t_0} = \left\{ j\in [n] : \exists \ell \leq t_0 \text{ s.t }\bA_{j,\ell} = 1 \right\} = \left\{ j \in [n] : \exists \ell \leq t_0 \text{ s.t } \bA_{j,\ell} \cdot \bK_{j,\ell} = 1\right\}.\]
Therefore, by having access to pairs $(j, \ell)$ and point $x_j$ with $\bA_{j,\ell} \cdot \bK_{j,\ell} = 1$ and knowledge of $t_0$, one can construct the set $\bS^{t_0}$, identify the parameter $\bm$, and determine the ordering $\pi$ so that $\pi_j$ denotes the $j$-th element of $\bS^{t_0}$.  
\end{observation}

\begin{observation}\label{obs:2}
Since the row $\bg_i^t = z_i^*$ for every $t \leq t_0$, letting $\pi$ denote the natural ordering where $\pi_j$ is the $j$-th element of $\bS^{t_0}$, for every $j \in [\bm]$,
\[ \bz_{\pi(j)} = z_{\pi(j)}^* = \left\{\begin{array}{cc} (1, 0) & \bsigma^*_{j} = 0 \\
						 			(0, 1) & \bsigma^*_j = 1 \end{array} \right. .\]
\end{observation}

\begin{observation}\label{obs:3}
For every $i \in [n]$, let $\bt_i$ denote the smallest number where $\bA_{i, \bt_i} = 1$. Then, $\bz_i = \bg_i^{\bt_i}$, and $\bg_i^{\bt_i}$ depends on the values $\tilde{\bc}_{i,0}^{\bt_i}$ and $\tilde{\bc}_{i, 1}^{\bt_i}$, which may be computed from:
\begin{itemize}
\item The point $x_i$ and $\bt_i$,
\item The pairs $(j, \ell)$ where $\bA_{j,\ell} \cdot \bK_{j,\ell} = 1$ and $\ell \leq \bt_i -1$.
\item For every pair $(j, \ell)$ as above, the point $x_j$, the weight $w_j$, and the smallest $\bt_j$ where $\bA_{j, \bt_j} = 1$, and $\bz_j$ (which may, itself, be recursively computed or from $\bsigma^*$).
\end{itemize}
In fact, the assignment rule is exactly that which is set by $\textsc{Assign}_{\bsigma^*}(x_i, \bt_i)$ in Figure~\ref{fig:assign}.
\end{observation}

\subsection{Proof of Theorem~\ref{thm:main-structural} assuming Lemma~\ref{lemma:main}}

Our main structural claim is showing that \textsc{GreedyProcess} produces an approximately optimal cut. We state the main lemma below and defer its proof. We show next that there is a local algorithm, such that, by receiving only limited information, can simulate the decisions of \textsc{GreedyProcess} and produce an approximately optimal cut.

\begin{lemma}[Main Lemma] \label{lemma:main}
Let $(\bz, \bm, \bsigma^*)$ be the output of \textsc{GreedyProcess}$((X, d), w, \gamma, t_0, t_e, z^*)$, where the weights $w$ are $\lambda$-compatible. For any $\eps > 0$, if 
\[  \gamma \geq \dfrac{(\ln(t_e) + 1)^2 \cdot \lambda}{\eps^2} \qquad \text{and}\qquad t_0 \geq \max\left\{ \frac{\sqrt{\gamma \cdot \lambda}}{\eps}, \frac{1}{\eps} \right\}. \] 
Then, 
\[ \Ex[f(\bz) - f(z^*)] \lsim \eps \sum_{i=1}^n \sum_{j=1}^n d(x_i, x_j).  \]
\end{lemma}

We defer the proof of Lemma~\ref{lemma:main} to Section~\ref{sec:greedy-analysis}, but we are now ready to conclude Theorem~\ref{thm:main-structural}, which combines the above observations, as well as the fact that, by setting $t_e \geq n \lambda / \eps$, we may place any point $x_i$ in the $0$-side of the cut (by letting $\bz_i = (1, 0)$) as a default whenever $\bA_{i, t}$ is always 0 without degrading the quality of the cut. 

\begin{proof}[Proof of Theorem~\ref{thm:main-structural} assuming Lemma~\ref{lemma:main}]
We first notice that, by construction of the greedy process, we have defined each row $\bA_i$ to behave independently and distributed from $\calT(w_i)$. In particular, note that $\bA_{i, t}$ is always sampled from $\Ber(\bw_i^{t})$, where the setting of $\bw_i^t$ depends solely on whether $i$ is or is not in $\bS^{t-1}$; if it is not in $\bS^{t-1}$, it is set to $\min \{ w_i, 1/t\}$ and otherwise $w_i$. Note that $i \in \bS^{t-1}$ if and only if there exists $\ell \leq t-1$ with $\bA_{i,\ell} = 1$, coinciding with $\calT(w_i)$. The case of $\bK_{i,t}$ is simpler, as it is sampled from $\Ber(\gamma^t)$ which is $1$ if $t \leq t_0$ and $\min\{ \gamma /t, 1\}$ otherwise; similarly, from $\calK(t_0, \gamma)$. 

If we let $\bz$ denote the cut which was output by greedy process, we partition the coordinates into
\[ \bG = \left\{ i \in [n] : \exists t \leq t_e \text{ s.t } \bA_{i,t} = 1\right\} \qquad\text{and}\qquad \bB = [n] \setminus \bG.  \]
Importantly, we have defined \textsc{GreedyProcess} and \textsc{Assign} so that, whenever $i \in \bG$, if we let $\bt_i$ denote the activation time of point $x_i$, we have $\bz(\bsigma^*)_i$, which is the output of $\textsc{Assign}_{\bsigma^*}(x_i, \bt_i)$ is the same as $\bz_i$. Therefore, we have
\begin{align}
\Ex_{\bA, \bK}\left[ \min_{\sigma} f(\bz(\sigma)) \right] \leq \Ex_{\bA, \bK}\left[ f(\bz(\bsigma^*)) \right] &\leq \Ex_{\bA, \bK}\left[ f(\bz)\right] + \sum_{i =1}^n \left( \sum_{j=1}^n d(x_i, x_j)\right) \cdot \Prx\left[ i \in \bB \right]. \label{eq:almost}
\end{align}
We finally may apply Lemma~\ref{lemma:main} as well as
\begin{align*}
\Prx\left[ i \in \bB \right] = \prod_{t=1}^{t_e} \left( 1 - \min\left\{ w_i , 1/ t\right\} \right) \leq \prod_{t=\lceil 1 / w_i \rceil}^{t_e} \left(1 - \frac{1}{t} \right) = \prod_{t = \lceil 1/w_i \rceil}^{t_e} \frac{t - 1}{t} = \dfrac{\lceil 1/w_i \rceil - 1}{t_e}.
\end{align*}
Since $w_i \leq 1/2$, we have $\lceil 1 / w_i \rceil - 1 \leq 1 / w_i$, so the left-most term in (\ref{eq:almost}) is at most
\begin{align*}
\frac{1}{t_e} \sum_{i=1}^n \sum_{j=1}^n \dfrac{d(x_i, x_j)}{w_i} \leq \frac{\lambda \cdot n}{t_e} \sum_{i=1}^n \sum_{k=1}^n d(x_i, x_k). 
\end{align*}
\end{proof}

\ignore{
\begin{lemma}\label{lem:min-over-sigma}
Letting $\sigma$ vary over all strings $\{0,1\}^m$ where $m$ is the number of non-zero entries in the first $t_0$ columns of $\bA \circ \bK$, we have
\begin{align*}
&\Ex_{(\bA, \bK)\sim \calD(w,\gamma,t_0, t_e)}\left[\min_{\sigma}\left( \frac{1}{2} \sum_{i=1}^n \sum_{j=1}^n d(x_i, x_j) \cdot \ind\left\{ g_{\bA, \bK, \sigma}(i) = g_{\bA, \bK, \sigma}(j) \right\} \right) \right] \\
&\qquad \leq \Ex_{(\bz, \bA, \bK, \bm, \bsigma^*)}\left[ f(\bz)\right] + \sum_{i=1}^n \frac{1}{w_i \cdot t_e}\left( \sum_{j=1}^n d(x_i, x_j)\right).
\end{align*}
If $\lambda > 1$ and weights are $\lambda$-compatible, for $t_e \geq n \cdot \lambda / \eps$, 
\[ \sum_{i=1}^n \frac{1}{w_i\cdot t_e} \left(\sum_{j=1}^n d(x_i, x_j) \right) \leq \eps \sum_{i=1}^n \sum_{j=1}^n d(x_i, x_j). \]
\end{lemma}

\begin{proof}
First, note that 
\begin{align*}
&\Ex_{(\bA, \bK)}\left[ \min_{\sigma}\left(\frac{1}{2} \sum_{i=1}^n \sum_{j=1}^n d(x_i, x_j) \ind\{ g_{\bA, \bK, \sigma}(i) = g_{\bA, \bK, \sigma}(j)\} \right)\right] \\
&\qquad\qquad \leq \Ex_{(\bz,\bA, \bK, \bm, \bsigma^*)}\left[ \frac{1}{2} \sum_{i=1}^n \sum_{j=1}^n d(x_i, x_j) \ind\{ g_{\bA, \bK, \bsigma^*}(i) = g_{\bA, \bK, \bsigma^*}(j)\} \right],
\end{align*}
so it suffices to analyze what happens on the ``correct'' string $\bsigma^*$. Consider the partition $[n]$ into two sets $\bG$ and $\bB$, according to whether or not there is a non-zero entry in the $i$-th row of $\bA$; if there is, then $i \in \bG$ as $t_i \leq t_e$, and $g_{\bA, \bK, \bsigma^*}(i)$ is set according to \textsc{Local-Sim}, and if not, we place $i \in \bB$ and $g_{\bA, \bK, \bsigma^*}(i)$ is defaulted to zero. We may upper bound 
\begin{align*}
&\frac{1}{2} \sum_{i = 1}^n \sum_{j=1}^n d(x_i, x_j) \cdot \ind\left\{ g_{\bA, \bK, \bsigma^*}(i) = g_{\bA, \bK, \bsigma^*}(j) \right\} \\
&\qquad \leq \frac{1}{2} \sum_{i \in \bG}^n \sum_{j\in \bG}^n d(x_i, x_j) \cdot \ind\left\{ g_{\bA, \bK, \bsigma^*}(i) = g_{\bA, \bK, \bsigma^*}(j) \right\} +  \sum_{i \in \bB} \left(\sum_{j=1}^n d(x_i, x_j) \right) \\
&\qquad \leq f(\bz) + \sum_{i \in \bB} \left(\sum_{j=1}^n d(x_i, x_j) \right).
\end{align*}
Finally, we note that
\begin{align*}
\Prx\left[ i \in \bB \right] &= \prod_{t=1}^{t_e} \left( 1 - \min\{w_i, 1/t \}\right) \leq \left(1 - w_i \right)^{\lfloor 1/w_i \rfloor} \dfrac{\lceil 1/w_i \rceil - 1}{t_e} \leq \frac{1}{w_i \cdot t_e},
\end{align*}
whenever giving the desired bound.
\end{proof}}

\ignore{This lemma connects the output of our process to $\bz^t$ to $g_{\bA,\bK, \sigma}$  from \ref{thm:greedy-is-opt}. Now let's analyze this process, first showing that it is computable given access to the weights $w_1, \dots, w_n$ and $\gamma$, and requires very low space. In Section \ref{sec:greedy-analysis}, we will show that the cut output is near optimal. }

\ignore{
\begin{claim} Points are activated independently of each other and can be computed for any time $t$ without information other than $w_1,\dots,w_n$ and $\gamma$.
\end{claim} 
\begin{proof}
Consider any point $x_i$ in the input and it's associated weight $w_i$. We shall show that $\bA^t_i$ and $\bK^t_i$ can be computed without any information other than the weight $w_i$ and a parameter $\gamma$. Let $\bA_i^t$, $\bK_i^t$ represent the entries corresponding to $x_i$ at time step $t$. At time $t=0$, $\bA^t$ and $\bK^t$ are empty, thus are trivially independent and requires no information to compute. Now we proceed by induction, assuming for some $t> 0$, that $\bA^{t-1}$ and $\bK^{t-1}$ can be computed independently and only with access to $w_i$ and $\gamma$. Now consider a single round of the process at time step $t > 0$. The process first computes $\bw_i^t$ which only relies on $w_i$ and $t$. First the process computes $\bA^{t-1}_{i, \ell}$ and $\bK^{t-1}_{i, \ell}=1$ if $x_i$ has already been activated, if $\bA^{t-1}_{i,t-1}=1$ which by the I.H. can be computed independently and only with $w_i$ and $t$. If $x_i$ was activated then $\bw_i^t = \min(w_i, \frac{1}{t})$  else $\bw_i^t = w_i$. Then the process will compute two Bernoulli random variables with probability $\bw_i^t$ and the other with $\min(\frac{\gamma}{t},1)$, independently of any other element $x_j$. The output of these two Bernoulli's and the prior matrices $\bA^{t-1}_i$ and $\bK^{t-1}_i$ are combined to compute $\bA^{t}_i$ and $\bK^{t}_i$. Since all terms are computed independently and only using $w_i$ and $\gamma$ as input, then we can conclude that the $\bA^t_i$ and $\bK^t_{i}$ are computed independently and only using $w_i$ and $\gamma$. 
\end{proof}

\begin{claim} Space needed to compute $Y^t$ is just nonzero's in $\bA^t \cdot \bK^t$ which is small e.g. $\textsc{poly}(d\log(n))$.
\end{claim} 
\begin{proof}
From the definition of the greedy process, the only information needed to compute $\bY^t$ other than $w_1,\dots,w_n$ and $\gamma$ is just the points $j$ which have nonzero entries in $\bA^t_{j,\cdot}\cdot \bK^t_{j,\cdot}$. Thus the expected space required for our algorithm after $\text{poly}(n)$ time steps can be computed by applying linearity of expectation \begin{align*}
    E[\text{space at time $t$}] &\leq \sum_{j=0}^n d\log(\Delta) \sum_{\ell=0}^t \E[\bA^t_{j,\ell}\cdot\bK^t_{j,\ell}]\\
    &\leq \sum_{j=0}^n d\log(\Delta) \sum_{\ell=0}^t \min(\frac{1}{\ell},w_j)\min(\frac{\gamma}{\ell},1)\\
    &\leq \sum_{j=0}^n d\log(\Delta) \sum_{\ell=0}^{t}w_j\frac{\gamma}{\ell}\\
    &\leq \gamma d\log(\Delta)\log(t)\\
\end{align*}
where the last step follows from $\sum_{j=0}^n w_j \leq 1$. Since $t=\text{poly}(n)$, the total space is bounded by $O(\gamma\log^2(n))$
\end{proof}}

\ignore{
\subsection{Proof of Theorem~\ref{thm:distribution-for-matrix} assuming Lemma~\ref{lemma:main}}

We let $\calD$ be the distribution $\calD(w, \gamma, t_0,t_e)$. Recall that a draw from this distribution may be generated by sampling $(\bz, \bA, \bK, \bm, \bsigma^*) \sim \calD^*(w, \gamma, t_0,t_e)$ from the output of \textsc{GreedyProcess}$((X,d), w, \gamma, t_0, t_e, z^*)$ and returning on the matrices $(\bA, \bK)$. For the first two items,
\begin{itemize}
\item The $(i,t)$-th entry of $\bA$ is generated by sampling from $\Ber(\bw_i^t)$, where $\bw_i^t$ is $\min\{ w_i, 1/t\}$ if $i \notin \bS^{t-1}$ and $w_i$ otherwise. Similarly, $\bK_{i,t}$ is drawn from $\Ber(\gamma^t)$ which is $1$ if $t \leq t_0$ and $\min(\gamma/t, 1)$ otherwise. Note that, $i \in \bS^{t-1}$ is equivalent to the event that there exists $\ell < t$ with $\bA_{i, \ell} = 1$. Hence, every row acts independently.  
\item Note that an entry in $\bA \circ \bK$ is non-zero only if both $\bA$ and $\bK$ are non-zero at that entry. We upper bound the expectation by:
\begin{align*}
\sum_{i=1}^{n} \sum_{t=1}^{t_e}  \Prx\left[ (\bA \circ \bK)_{i, t} = 1 \right] \leq \sum_{i=1}^n w_i \left( t_0 + \gamma \sum_{t=t_0+1}^{t_e} \frac{1}{t} \right) \lsim \|w\|_1 \left( t_0 + \gamma \log\left( t_e\right) \right). 
\end{align*}
\end{itemize}
For the description of the function $g_{\bA, \bK, \sigma}(\cdot)$, we use the local simulation procedure of Lemma~\ref{lem:local-sim}, which specifies the algorithm for computing the function from the desired inputs. The last point of Theorem~\ref{thm:distribution-for-matrix} follows from combining Lemma~\ref{lem:min-over-sigma} and Lemma~\ref{lemma:main}. In particular, the maximizing $\sigma$ of sum of distances among $(i, j)$ with  $g_{\bA, \bK, \sigma}(i) \neq g_{\bA, \bK, \sigma}(j)$ is the same as minimizing $\sigma$ of sum of distances among $(i,j)$ where $g_{\bA, \bK, \sigma}(i) = g_{\bA, \bK, \sigma}(j)$; note that we always incur the additive error
\[ O(\eps) \sum_{i=1}^n \sum_{j=1}^n d(x_i, x_j). \]}

\section{Proof of Lemma~\ref{lemma:main}} 
\label{sec:greedy-analysis}

As with the analysis of~\cite{MS08}, the analysis of the greedy process relies on a ``fictitious cut'' vector. The fictitious cut vector begins by encoding an optimal cut, and transforms into the cut that is output by the greedy process once it terminates. Even though the structure of the argument (in particular, tracking a fictitious cut) remains the same as~\cite{MS08}, the specification of the fictitious cut itself will differ, in order to correspond to our greedy process. 

\paragraph{Notation.} For the remainder of this section, we will consider a fixed metric space $(X = \{ x_1,\dots, x_n\}, d)$ consisting of $n$ points, as well as settings of the weight vector $w \in (0,1/2]^n$, parameters $\gamma \geq 1$, $t_0, t_e \in \N$ and an optimal cut $z^* \in \{0,1\}^{n \times \{0,1\}}$. In the definitions below, we will reference random variables (such as $\bz^t, \bw^t, \bS^t$, \dots, etc.) referring to an execution of $\textsc{GreedyProcess}((X = \{ x_1,\dots, x_n \} , d) , w, \gamma, t_0, t_e, z^*)$ and a draw $(\bz, \bA, \bK, \bm, \bsigma^*) \sim \calD^*(w, \gamma, t_0, t_e)$. 

\begin{definition}[Fictitious cut]\label{def:fictitious}
The fictitious cut is specified by a sequence of random variables $\hat{\bz}^t \in [0, 1]^{n \times \{0,1\}}$ for each $t \in \{0, \dots, t_e\}$. We let:
\begin{itemize}
\item For $t = 0$, we define $\hat{\bz}^0 = z^*$. 
\item For $t > 0$, we denote for each $i \in [n]$, the vector $\hat{\bz}_i^t \in [0,1]^{\{0,1\}}$ by\footnote{Even though it is not yet clear, we will show that $\hat{\bz}_i \in [0,1]^{\{0,1\}}$.}
\begin{align}
\hat{\bz}_i^t &= \left\{ \begin{array}{cc} \bz_i^t & i \in \bS^t \\
							\dfrac{1}{1 - \bw_i^t} \left(\dfrac{t-1}{t} \cdot \hat{\bz}_i^{t-1} + \frac{1}{t} \cdot \bg_i^t -  \bw_i^t \cdot \bg_i^t \right) & \text{o.w.} \end{array} \right. . \label{eq:shadow-def}
\end{align}
\end{itemize}
\end{definition}

\begin{claim}\label{cl:simple-fictitious}
The following are simple facts about the fictitious cut, which follow from Definition~\ref{def:fictitious}:
\begin{itemize}
\item For every $t \in \{0, \dots, t_0\}$, we have $\hat{\bz}^t = z^*$.
\item For every $i \in [n]$, we have $\hat{\bz}^{t_e}_{i, 1} \geq \bz_{i,1}$ and $\hat{\bz}^{t_e}_{i,0} \geq \bz_{i, 0}$.
\end{itemize}
\end{claim}

\begin{proof}
The first item follows from the definition of \textsc{GreedyProcess}, as well as a simple inductive argument. For $t = 0$, set $\hat{\bz}^0 = z^*$ by Definition~\ref{def:fictitious}, so assume such is the case for $t \in \{0 ,\dots, \tilde{t}\}$ for $\tilde{t} < t_0$, and we will establish it for $t = \tilde{t} + 1$.  Recall that, for every $t \in \{0, \dots, t_0\}$, Line~\ref{en:greedy-ln2} of \textsc{GreedyProcess} sets $\bg_i^t = z^*_i$. So consider the two cases which may happen in (\ref{eq:shadow-def}); if $i \in \bS^t$, then $\hat{\bz}_i^{t} = \bz_i^t$, which is equal to $\bz_i = z^*_i$ for $i \in \bS_{t} \subset \bS^{t_0}$ (by Observation~\ref{obs:2}). If $i \notin \bS^t$, then (\ref{eq:shadow-def}) gives
\[ \hat{\bz}_i^{t} = \frac{1}{1 - \bw_i^t} \left( \frac{t-1}{t} \cdot z_i^* + \frac{1}{t} \cdot z_i^* - \bw_i^t \cdot z_i^* \right) = z_i^*. \]
For the second item, notice that $\hat{\bz}_i^{t_e} = \bz_i$ if $i \in \bS^t$, and otherwise lies in $[0,1]^{\{0,1\}}$, whereas $\bz_i = (0, 0)$.
\end{proof}

Recall that, for the proof of Lemma~\ref{lemma:main}, we seek to upper bound the expected difference between $f(\bz) - f(z^*)$. Thus, the definition of the fictitious cut, as well as Claim~\ref{cl:simple-fictitious},
\begin{align}
\Ex\left[ f(\bz) - f(z^*) \right] &\leq \Ex\left[ f(\hat{\bz}^{t_e}) - f(\hat{\bz}^0)\right] = \sum_{t = t_0+1}^{t_e} \Ex\left[ f(\hat{\bz}^{t}) - f(\hat{\bz}^{t-1}) \right],\label{eq:goal-1}
\end{align} 
where in the first inequality, the fact $f(\bz) \leq f(\hat{\bz}^{t_e})$ follows from definition of $f$, the fact distances are non-negative, and Claim~\ref{cl:simple-fictitious}. We will focus on showing an upper bound for each $t \in \{ t_0+1, \dots, t_e\}$, and by symmetry of $d(\cdot, \cdot)$ and definition of $f$, we may write
\begin{align}
f(\hat{\bz}^{t}) - f(\hat{\bz}^{t-1}) &= \sum_{i=1}^n \left(\hat{\bz}_{i,0}^{t} - \hat{\bz}_{i,0}^{t-1}  \right) \sum_{j=1}^n d(x_i, x_j) \hat{\bz}_{j,0}^{t-1} +  \sum_{i=1}^n \left(\hat{\bz}_{i,1}^{t} - \hat{\bz}_{i,1}^{t-1}  \right) \sum_{j=1}^n d(x_i, x_j) \hat{\bz}_{j,1}^{t-1} \label{eq:linear-term}\\
		&\qquad + \frac{1}{2} \sum_{i=1}^n \sum_{j=1}^n d(x_i, x_j) \left(\hat{\bz}_{i,0}^{t} - \hat{\bz}_{i,0}^{t-1} \right) \left( \hat{\bz}_{j,0}^{t} - \hat{\bz}_{j,0}^{t-1} \right) \label{eq:non-linear-1}\\
		&\qquad + \frac{1}{2} \sum_{i=1}^n \sum_{j=1}^n d(x_i, x_j) \left(\hat{\bz}_{i,1}^{t} - \hat{\bz}_{i,1}^{t-1} \right) \left( \hat{\bz}_{j,1}^{t} - \hat{\bz}_{j,1}^{t-1} \right), \label{eq:non-linear-2}
\end{align}

\begin{lemma}\label{lem:non-linear}
For every $t \in \{0, \dots, t_e\}$ and any fixed execution of the first $t-1$ steps of \textsc{GreedyProcess}, then for every $i \notin S^{t-1}$ and any $b \in \{0,1\}$,
\[ \Ex\left[ \hat{\bz}^{t}_{i,b} - \hat{z}^{t-1}_{i,b} \right] = \frac{1}{t} \cdot \left( g_{i,b}^{t} - \hat{z}_{i,b}^{t-1} \right), \]
where the expectation above is only over the randomness which defines $\hat{\bz}_{i,b}^{t}$ after fixing the first $t-1$ steps of the process (thereby fixing $\hat{z}^{t-1}_{i,b}$ and $g_{i,b}^{t}$). Furthermore, the expectation of expressions (\ref{eq:non-linear-1}) and (\ref{eq:non-linear-2}) is at most
\begin{align*}
\frac{1}{2t^2} \sum_{i=1}^n \sum_{j=1}^n d(x_i, x_j).
\end{align*}
\end{lemma}

\begin{proof}
Consider any fixed execution of the first $t-1$ steps of \textsc{GreedyProcess} which fixes the setting of $\hat{z}_{i,b}^{t-1}$, the setting of $w_{i}^{t}$ (in Line~\ref{en:greedy-ln1}), and $g_{i,b}^{t}$ (in Line~\ref{en:greedy-ln2}), as these depend only on the randomness which occurred up to step $t-1$. We now consider what happens to \smash{$\hat{\bz}_{i,b}^{t}$} at step $t$ whenever $i \notin S^{t-1}$:
\begin{itemize}
\item With probability $w_i^{t}$, we have $i \in \bS^{t}$, so (\ref{eq:shadow-def}) sets $\hat{\bz}_i^{t} = z_i^{t-1}$, which is equal to $g_{i}^{t}$ according to Line~\ref{en:greedy-ln4}:
\[ \hat{\bz}^{t}_{i,b} - \hat{z}^{t-1}_{i,b} = g_{i,b}^{t} - \hat{z}^{t-1}_{i,b}. \]
\item With probability $1 - w_i^{t}$, we have $i \notin \bS^{t}$, and 
\[ \hat{\bz}^{t}_{i,b} - \hat{z}^{t-1}_{i,b} = \frac{1}{1 - w_{i}^{t}} \left( \frac{t-1}{t} \cdot \hat{z}_{i,b}^{t-1} + \frac{1}{t} \cdot g_{i,b}^{t} - w_i^{t}  g_{i,b}^{t}\right) - \hat{z}^{t-1}_{i,b}. \]
\end{itemize}
Thus, taking expectation with respect to the randomness in this step $t$ (namely, the samples in Line~\ref{en:greedy-ln3}), after some cancellations:
\begin{align*}
\Ex\left[ \hat{\bz}^{t}_{i,b} - \hat{z}^{t-1}_{i,b}\right] &= \frac{1}{t} \cdot \left(g_{i,b}^{t} - \hat{z}_{i,b}^{t-1}\right).
\end{align*}
This gives the first part of the lemma. For the second part, we note that for any execution of the first $t-1$ steps, any $i \in S^{t-1}$, will have $i \in \bS^{t}$ in Line~\ref{en:greedy-ln3}, and therefore $\hat{\bz}^{t}_i = \hat{z}_i^{t-1} = z_i^{t-1}$. In other words, may re-write the expression (\ref{eq:non-linear-1}) and (\ref{eq:non-linear-2}) as
\begin{align*}
\frac{1}{2} \sum_{i=1}^n \sum_{j=1}^n d(x_i, x_j) \left( \hat{\bz}^{t}_{i,b} - \hat{z}^{t-1}_{i,b} \right) \left( \hat{\bz}_{j,b}^{t} - \hat{z}_{j,b}^{t-1} \right) = \frac{1}{2} \sum_{i \notin S^{t-1}} \sum_{j \notin S^{t-1}} d(x_i, x_j) \left( \hat{\bz}^{t}_{i,b} - \hat{z}^{t-1}_{i,b} \right) \left( \hat{\bz}_{j,b}^{t} - \hat{z}_{j,b}^{t-1} \right).
\end{align*}
Since $d(x_i, x_j) = 0$ if $i = j$, we have, by independence when $i \neq j$, we have that the expectation over the step $t$ is 
\begin{align*}
\Ex\left[\left( \hat{\bz}_{i,b}^{t} - \hat{z}_{i,b}^{t-1}\right)\left( \hat{\bz}_{j,b}^{t} - \hat{z}_{j,b}^{t-1}\right)\right] &= \dfrac{g_{i,b}^{t} - \hat{z}_{i,b}^{t-1}}{t} \cdot \frac{g_{j,b}^{t} - \hat{z}_{j,b}}{t} \leq \frac{1}{t^2},
\end{align*}
since $g_{\cdot,b}^{t}, \hat{z}_{\cdot, b}^{t-1} \in [0,1]$.
\end{proof}

Lemma~\ref{lem:non-linear} will handle the upper bound for (\ref{eq:non-linear-1}) and (\ref{eq:non-linear-2}), and we now work on upper bounding (\ref{eq:linear-term}). We will now follow the approach of~\cite{MS08} and consider a contribution matrix and an estimated contribution; this will help us relate (\ref{eq:linear-term}) to how the \textsc{GreedyProcess} evolves. 
\begin{definition}[Contribution Matrix and Estimated Contribution]
For time $t \in \{ 1,\dots, t_e\}$, the contribution matrix $\bc^{t} \in \R_{\geq 0}^{n\times \{0,1\}}$ is given by letting, 
\[ \bc^{t}_{i, 0} = \sum_{j=1}^n d(x_i, x_j) \cdot \hat{\bz}^{t-1}_{j,0} \qquad\text{and}\qquad \bc^{t}_{i,1} = \sum_{j=1}^n d(x_i, x_j) \cdot \hat{\bz}^{t-1}_{j, 1}.  \]
For a time $t \in \{ 1, \dots, t_e\}$ and an index $j \in [n]$, we define $\bR_j^{0} = 1$ and for $t \geq 2$
\[ \bR_{j}^{t-1} = \frac{1}{t-1} \sum_{\ell=1}^{t-1} \dfrac{\bA_{j,\ell}^{t-1} \cdot \bK_{j, \ell}^{t-1}}{\bw_j^{\ell} \cdot \gamma^{\ell}} ,\]
and the estimated contribution matrix  $\tilde{\bc}^t \in \R_{\geq 0}^{n \times \{0,1\}}$ be given by
\[ \tilde{\bc}^{t}_{i,0} = \sum_{j=1}^n d(x_i, x_j) \cdot \bR_j^{t-1} \cdot \hat{\bz}_{j,0}^{t-1} \qquad \text{and}\qquad  \tilde{\bc}^{t}_{i,1} = \sum_{j=1}^n d(x_i, x_j) \cdot \bR_j^{t-1} \cdot \hat{\bz}_{j,1}^{t-1}, \]
and notice that, since $\hat{\bz}_{j,0}^{t-1} = \bz_{j,0}^{t-1}$ and $\hat{\bz}_{j,1}^{t-1} = \bz_{j,1}^{t-1}$ when $j \in \bS^{t-1}$ (and equivalently, when $\bR_{j}^{t-1}$ is non-zero), we defined $\tilde{\bc}_{i,0}^{t}$ and $\tilde{\bc}_{i,1}^{t}$ to align with Line~\ref{en:greedy-ln2} in $\textsc{GreedyProcess}$.
\end{definition}

The purpose of defining the contribution matrix and the expected contribution matrix is in order to upper bound (\ref{eq:linear-term}). In fact, we continuing on the upper bound, we may write (focusing on the $0$-side of the cut),
\begin{align*}
\sum_{i=1}^n \left( \hat{\bz}_{i,0}^{t} - \hat{\bz}_{i,0}^{t-1} \right) \bc_{i,0}^t &= \sum_{i=1}^n \left(\hat{\bz}_{i,0}^{t} - \hat{\bz}_{i,0}^{t-1} \right) \left( \bc_{i,0}^t - \tilde{\bc}_{i,0}^t \right) + \sum_{i=1}^n \left( \hat{\bz}_{i,0}^{t} - \hat{\bz}_{i,0}^{t-1} \right) \cdot \tilde{\bc}_{i,0}^{t} \\
&= \sum_{i \notin \bS^{t-1}} \left( \hat{\bz}_{i,0}^{t} - \hat{\bz}_{i,0}^{t-1} \right) \left( \bc_{i,0}^t - \tilde{\bc}_{i, 0}^{t} \right) + \sum_{i \notin \bS^{t-1}} \left(\hat{\bz}_{i,0}^{t} - \hat{\bz}^{t-1}_{i,0} \right) \cdot \tilde{\bc}_{i,0}^t,
\end{align*}
where we used the fact that, in \textsc{GreedyProcess} and (\ref{eq:shadow-def}), $i \in \bS^{t-1}$ will mean $\bz_i^t = \bz_i^{t-1}$, $\hat{\bz}_i^{t-1} = \bz_i^{t-1}$, and $\hat{\bz}_{i}^{t} = \bz_i^{t}$, so $\hat{\bz}_{i}^{t} = \hat{\bz}_i^{t-1}$. By the first part of Lemma~\ref{lem:non-linear}, we can simplify 
\begin{align}
&\Ex\left[ \sum_{i=1}^n \left( \hat{\bz}_{i,0}^{t} - \hat{\bz}_{i,0}^{t-1}\right) \bc_{i,0}^{t} \right] \nonumber \\
&\qquad\qquad = \Ex\left[ \frac{1}{t} \sum_{i \notin \bS^{t-1}} \left(\bg_{i,0}^{t} - \hat{\bz}_{i,0}^{t-1} \right) \left( \bc_{i,0}^{t} - \tilde{\bc}_{i,0}^{t} \right)  \right]  + \Ex\left[ \frac{1}{t} \sum_{i \notin \bS^{t-1}} \left(\bg_{i,0}^{t} - \hat{\bz}_{i,0}^{t-1} \right) \cdot \tilde{\bc}_{i,0}^t \right] \nonumber \\
&\qquad\qquad \leq \frac{1}{t} \Ex\left[ \sum_{i \notin \bS^{t-1}} |\bc_{i,0}^{t} - \tilde{\bc}_{i,0}^t| \right] + \frac{1}{t} \Ex\left[ \sum_{i \notin \bS^{t-1}} \left( \bg_{i,0}^{t} - \hat{\bz}_{i,0}^{t-1} \right) \cdot \tilde{\bc}_{i,0}^t \right] \label{eq:goal-2}
\end{align}
where have used the first part of Lemma~\ref{lem:non-linear} and took expectation over the randomness of step $t$, and the fact that $\bg_{i,0}^{t}, \hat{\bz}_{i,0}^{t-1} \in [0,1]$ to place the absolute values. We conclude the proof of Lemma~\ref{lemma:main} by stating the following two lemmas (whose proof we defer shortly).

\begin{lemma}\label{lem:greedy-choice-opt}
For any execution of \textsc{GreedyProcess}, any $t \in \{t_0+1, \dots, t_e\}$ and any $i \notin \bS^{t-1}$, 
\begin{align*}
\left( \bg_{i,0}^{t} - \hat{\bz}^{t-1}_{i,0}\right) \tilde{\bc}_{i,0}^{t} + \left( \bg_{i,1}^{t} - \hat{\bz}^{t-1}_{i,1}\right) \tilde{\bc}_{i,1}^t \leq 0. 
\end{align*}
\end{lemma}

\begin{lemma}\label{lem:error-bound}
For any $t \in \{2, \dots, t_e\}$ and any $b \in \{0,1\}$, whenever $\|w\|_{\infty} \leq 1/2$,
\begin{align*}
\Ex\left[ \sum_{i \notin \bS^{t-1}} \left| \bc_{i,b}^t - \tilde{\bc}_{i,b}^{t}\right| \right] \leq 10 \left(\frac{2\sqrt{\gamma}}{t-1} + \frac{1}{\sqrt{\gamma}} \right) \sum_{i=1}^n \left( \sum_{j=1}^n \frac{d(x_i, x_j)^2}{w_j} \right)^{1/2}.
\end{align*}
\end{lemma}

\begin{proof}[Proof of Lemma~\ref{lemma:main} assuming Lemma~\ref{lem:greedy-choice-opt} and~\ref{lem:error-bound}]
We put all the elements together, where we seek to upper bound:
\begin{align*}
&\Ex\left[ f(\bz) - f(z^*) \right] \leq \sum_{t=t_0+1}^{t_e} \Ex\left[ f(\hat{\bz}^{t}) - f(\hat{\bz}^{t-1}) \right] \\
	&\qquad \leq \sum_{t=t_0+1}^{t_e} \Ex\left[ \sum_{i=1}^n \left( \hat{\bz}_{i,0}^{t} - \hat{\bz}_{i,0}^{t-1} \right) \bc_{i,0}^t + \sum_{i=1}^n \left( \hat{\bz}_{i,1}^{t} - \hat{\bz}_{i,1}^{t-1} \right) \bc_{i,1}^t \right] + \left(\sum_{i=1}^n \sum_{j=1}^n d(x_i, x_j)  \right) \sum_{t=t_0+1}^{t_e} \frac{1}{t^2} ,
\end{align*}
where we first followed (\ref{eq:goal-1}), and then simplified (\ref{eq:non-linear-1}) and (\ref{eq:non-linear-2}) using the second part of Lemma~\ref{lem:non-linear}. Then, using (\ref{eq:goal-2}) and Lemma~\ref{lem:greedy-choice-opt}, we obtain 
\begin{align*}
\Ex\left[ f(\bz) - f(z^*) \right] &\leq \sum_{t=t_0+1}^{t_e}\left( \frac{1}{t} \Ex\left[ \sum_{i \notin \bS^{t-1}} |\bc_{i,0}^t - \tilde{\bc}_{i,0}| \right] + \frac{1}{t} \Ex\left[ \sum_{i \notin \bS^{t-1}} |\bc_{i,1}^t - \tilde{\bc}_{i,1}| \right] \right) \\
		&\qquad \qquad + O(1/t_0) \left( \sum_{i=1}^n \sum_{j=1}^n d(x_i, x_j) \right) \\
		&\lsim \frac{\ln(t_e) + 1}{\sqrt{\gamma}} \sum_{i=1}^n \left( \sum_{j=1}^n \frac{d(x_i, x_j)^2}{w_j} \right)^{1/2} + \dfrac{\sqrt{ \gamma}}{t_0} \sum_{i=1}^n \left( \sum_{j=1}^n \frac{d(x_i, x_j)^2}{w_j}\right)^{1/2} \\
		&\qquad \qquad + \frac{1}{t_0}  \left( \sum_{i=1}^n \sum_{j=1}^n d(x_i, x_j) \right).
\end{align*}
We note that the assumptions on our weight vector imply that
\begin{align*}
\sum_{i=1}^n \left(\sum_{j=1}^n \frac{d(x_i, x_j)^2}{w_j}\right)^{1/2} \leq \sqrt{\lambda} \sum_{i=1}^n \sum_{j=1}^n d(x_i, x_j).
\end{align*}
so the above error bound simplifies to
\[ \left( \frac{(\ln(t_e) + 1) \sqrt{\lambda} }{\sqrt{\gamma}} + \frac{\sqrt{\gamma \cdot \lambda}}{t_0} + \frac{1}{t_0} \right) \sum_{i=1}^n \sum_{j=1}^n d(x_i, x_j). \]
that when 
\[ \gamma \geq \frac{(\ln(t_e) + 1)^2 \cdot \lambda}{\eps^2} \qquad\text{and}\qquad t_0 \geq \max\left\{ \frac{\sqrt{\gamma \cdot \lambda}}{\eps}, \frac{1}{\eps} \right\}, \]
we obtain our desired bound.
\end{proof}

\subsection{Proof of Lemma~\ref{lem:greedy-choice-opt}}

The proof of Lemma~\ref{lem:greedy-choice-opt} proceeds by first showing that the fictitious cut $\hat{\bz}^t \in [0,1]^{n \times \{0,1\}}$. The proof is a straight-forward induction. This allows us to say that the greedy choices $\bg^t$ are chosen so as to make the expression in Lemma~\ref{lem:greedy-choice-opt} negative.

\begin{lemma}\label{lemma:fic-cut-bounded} For any $t \in \{0, \dots t_e\}$ and any $i \in [n]$, we have
\begin{align*}
    \hat{\bz}_{i,0}^t + \hat{\bz}_{i,1}^t = 1 &&     \left |\hat{\bz}_{i,0}^t - \hat{\bz}_{i,1}^t \right | \leq 1 ,
\end{align*}
which implies $\hat{\bz}^t \in [0,1]^{n \times \{0,1\}}$.
\end{lemma}

\begin{proof}
The proof proceeds by induction on $t$. We note that $\hat{\bz}^0 = z^* \in \{0,1\}^{n\times \{0,1\}}$, so the condition trivially holds. We assume for induction that the condition holds for all $t \leq \tilde{t}$ and show the condition for $t=\tilde{t}+1$. First, note that if $i \in \bS^{t}$, then $\hat{\bz}_i^{t} = \bz_{i}^t \in \{0,1\}^{\{0,1\}}$ and satisfies the condition. So, suppose $i \notin \bS^{t}$.
\begin{align*}
\hat{\bz}_{i,0}^{t} + \hat{\bz}_{i,1}^{t} &= \dfrac{1}{1 - \bw_{i}^{t}} \left(\dfrac{t-1}{t} \cdot \left(\hat{\bz}_{i,0}^{t-1} + \bz_{i,1}^{t-1}\right) + \left(\frac{1}{t} - \bw_{i}^t \right) \cdot \left(\bg_{i,0}^{t} + \bg_{i,1}^{t} \right) \right) \\
				&= \dfrac{1}{1 - \bw_i^t}\left( \frac{t-1}{t} + \frac{1}{t} - \bw_i^t \right) = 1,
\end{align*}
where we used the inductive hypothesis and the fact $\bg_{i,0}^t + \bg_{i,1}^t = 1$. For the difference, we have:
\begin{align*}
\left| \hat{\bz}^{t}_{i,0} - \hat{\bz}^t_{i,1}\right| &= \left|\dfrac{1}{1 - \bw_{i}^{t}} \left(\dfrac{t-1}{t} \cdot \left(\hat{\bz}_{i,0}^{t-1} - \bz_{i,1}^{t-1}\right) + \left(\frac{1}{t} - \bw_{i}^t \right) \cdot \left(\bg_{i,0}^{t} - \bg_{i,1}^{t} \right) \right) \right| \\
			&\leq \left| \frac{1}{1 - \bw_i^{t}}\right| \cdot \left( \frac{t-1}{t} + \left| \frac{1}{t} - \bw_i^t \right| \right).
\end{align*}
Finally, we note that \textsc{GreedyProcess} enforces that $\bw_i^t \leq 1/t$ when $i \notin \bS^{t-1}$, so $(t-1) / t + |1/t - \bw_i^t| = 1 - \bw_i^t$. Giving the upper bound on the difference.
\end{proof}    

To conclude the proof of Lemma~\ref{lem:greedy-choice-opt}, recall that $\bg_{i,0}^{t}$ and $\bg_{i,1}^{t}$ are chosen so as to minimize
\begin{align*}
\bg_{i, 0}^{t} \cdot \tilde{\bc}_{i,0}^{t} + \bg_{i,1}^{t} \cdot \tilde{\bc}_{i,1}^{t} = \min_{\substack{\alpha, \beta \in [0,1] \\ \alpha + \beta = 1}} \alpha \cdot \tilde{\bc}_{i,0}^{t} + \beta \cdot \tilde{\bc}_{i,1}^{t} \leq \hat{\bz}_{i,0}^{t-1} \cdot \tilde{\bc}_{i,0} + \hat{\bz}_{i,1}^{t-1} \cdot \tilde{\bc}_{i,1}^t,
\end{align*}
which gives the desired bound.

\subsection{Proof of Lemma~\ref{lem:error-bound}}

It will be useful for the proof of Lemma~\ref{lem:error-bound} to re-index the time so that we upper bound, for all $t \in \{ 1, \dots, t_e-1\}$, the quantity
\[ \Ex\left[ \sum_{i \notin \bS^{t}} \left|\bc_{i,b}^{t+1} - \tilde{\bc}_{i,b}^{t+1} \right|\right],\]
because the definition of $\bc_{i,b}^{t+1}$ depends solely on the randomness up to step $t$, i.e., those from $\bA^{t}_{j,\ell}, \bK^t_{j,\ell}$ for all $j \neq i$ and $\ell \leq t$. We interpret $\bc_{i,b}^{t+1}$ as the contribution of point $i$ to the $b$-th side of the cut once all points activated by time $t$ have been assigned according to $\hat{\bz}^t$, and $\tilde{\bc}_{i,b}^{t+1}$ is meant to be an estimate of that contribution. It is natural to define the error in the estimation coming from point $j$, 
\begin{align*}
\berr_{i,b}^{t+1}(j) &\eqdef d(x_i, x_j) \left(1 - \bR_j^{t} \right) \cdot \hat{\bz}_{j,b}^{t} =  \sum_{j=1}^{n} d(x_i, x_j) \left(1 - \frac{1}{t} \sum_{\ell=1}^{t} \dfrac{\bA_{j,\ell}^{t} \cdot \bK_{j, \ell}^{t}}{\bw_{j}^{\ell} \cdot \gamma^{\ell}} \right) \cdot \hat{\bz}_{j,b}^{t} \\
\berr_{i,b}^{t+1} &\eqdef \sum_{j=1}^n \berr_{i,b}^{t+1}(j) = \bc_{i,b}^{t+1} - \tilde{\bc}_{i,b}^{t+1}.
\end{align*}
Note that the case $t = 0$, $\berr_{i,b}^1(j) = 0$ by definition of $\bR_j^0 = 1$ for all $j$. 
By Jensen's inequality and the definition of $\berr_{i,b}^{t+1}$,
\begin{align}
\Ex\left[ \sum_{i \notin \bS^{t}} |\bc_{i,b}^{t+1} - \tilde{\bc}_{i,b}^{t+1}| \right]  
   \leq \sum_{i=1}^n \left( \Ex\left[(\berr_{i,b}^{t+1})^2 \right] \right)^{1/2}. \label{eq:expression}
\end{align}

\begin{lemma}\label{lem:martingale}
Consider any $t \geq 1$ and any execution of the first $t-1$ steps of \textsc{GreedyProcess}, fixing the random variables $A^{t-1}, K^{t-1}, \hat{z}^{t-1}$, \dots, etc. Then, considering only the $t$-th step of \textsc{GreedyProcess}, 
\[ \Ex\left[ \berr_{i,b}^{t+1}(j) \right] =  \frac{t-1}{t} \cdot \err_{i,b}^t(j),\]
where the expectation is taken only with respected to the randomness in the $t$-th step. This implies that
\begin{align*}
\Ex\left[ \berr_{i,b}^{t+1} - \frac{t-1}{t} \cdot \berr_{i,b}^{t}\right] = 0.
\end{align*}
\end{lemma}

\begin{proof}
We simply go through the calculation, where we write $\berr_{i,b}^{t+1}(j)$ in terms of $R_j^{t-1}$ as 
\begin{align*}
\berr_{i,b}^{t+1}(j) = d(x_i, x_j) \left(1 - \frac{1}{t} \cdot \dfrac{\bA_{j,t}^{t} \cdot \bK_{j,t}^t}{w_j^{t} \cdot \gamma^{t}}  - \frac{t-1}{t} \cdot R_j^{t-1} \right) \hat{\bz}^t_{j,b}.
\end{align*}
Note that, $\bA_{j,t}^{t}$ and $\bK_{j,t}^{t}$, for all $j$ is the only source of randomness in the $t$-th time step; in particular, $R_j^{t-1}$ is fixed, and $w_j^t$ is fixed (and hence unbolded); however, $\hat{\bz}^{t}_j$ does depend on $\bA_{j,t}^t$ if $j \notin S^{t-1}$. We now consider two cases, depending on whether or not $j \in S^{t-1}$. In the case that $j \in S^{t-1}$, the fictitious cut $\hat{\bz}^t_{j,b}$ will always be set to $\hat{z}^{t-1}_{j,b}$ (as it remains fixed), and draw $\bA_{j,t}^{t} \sim \Ber(w_i^t)$ and $\bK_{j,t}^t \sim \Ber(\gamma^t)$. Therefore, the expectation over the  draw of $\bA_{j,t}^t$ and $\bK_{j,t}^t$ results in
\begin{align*}
\Ex\left[ d(x_i, x_j) \left(1 - \frac{1}{t} \cdot \dfrac{\bA_{j,t}^{t} \cdot \bK_{j,t}^t}{w_j^{t} \cdot \gamma^{t}}  - \frac{t-1}{t} \cdot R_j^{t-1} \right) \hat{\bz}^t_{j,b} \right] &= d(x_i, x_j) \left(1 - \frac{1}{t} - \frac{t-1}{t} \cdot R_j^{t-1} \right) \hat{z}^{t-1}_{j,b} \\
&= \frac{t-1}{t} \cdot d(x_i,x_j) \left(1 - R_j^{t-1} \right) \hat{z}_{j,b}^{t-1} \\
&= \frac{t-1}{t} \cdot \err_{i,b}^{t}(j).
\end{align*}
On the other hand, if $j \notin S^{t-1}$, then $R_j^{t-1} = 0$. The expectation calculation, however, is more complicated than two Bernoulli's (as above), since the setting of $\hat{\bz}^t_{j,b}$ depends on the realization of $\bA_{j,t}^t$ and $\bK_{j,t}^t$. We go through the two cases:
\begin{itemize}
\item With probability $w_j^t$, we have $\bA_{j,t}^t = 1$. This results in a setting of $\hat{\bz}^t_{j,b} = \bz^t_{j,b}$ according to the first case of (\ref{eq:shadow-def}), and $\bz^t_{j,b} = g_{j,b}^t$ from Line~\ref{en:greedy-ln4}. Taking expectation over $\bK_{j,t}^t \sim \Ber(\gamma^t)$, we have that $\berr_{i,b}^{t+1}(j)$ becomes
\begin{align*}
d(x_i, x_j) \left(1 - \frac{1}{t} \cdot \frac{1}{w_j^t} \right) \cdot g_{j,b}^t.
\end{align*}
\item With probability $1 - w_j^t$, we have $\bA_{j,t}^t = 0$, in which case $\hat{\bz}_{j,b}^{t}$ is updated according to the second case of (\ref{eq:shadow-def}), and $\berr_{i,b}^{t+1}(j)$ becomes:
\begin{align*}
d(x_i,x_j) \cdot \frac{1}{1 - w_j^t} \left(\frac{t-1}{t} \cdot \hat{z}_{j,b}^{t-1} + \frac{1}{t} \cdot g_{j,b}^t - w_j^t \cdot g_{j,b}^t \right).
\end{align*}
\end{itemize}
All together, this gives us an expectation of $\berr_{i,b}^{t+1}(j)$ is
\begin{align*}
d(x_i, x_j) \left( w_j^t \cdot g_{j,b}^t - \frac{1}{t} \cdot g_{j,b}^t + \frac{t-1}{t} \cdot \hat{z}_{j,b}^{t-1} + \frac{1}{t} \cdot g_{j,b}^{t} - w_j^t \cdot g_{j,b}^{t} \right) &= \frac{t-1}{t} \cdot d(x_i, x_j) \cdot \hat{z}_{j,b}^{t-1} \\
	&= \frac{t-1}{t} \cdot \err_{i,b}^{t}(j),
\end{align*}
since $R_j^{t-1} = 0$. Taking the sum over terms concludes the lemma.
\end{proof}

\begin{lemma}\label{lem:reduce-to-error}
For any $t \in \{2, \dots, t_e-1\}$, 
\begin{align*}
\Ex\left[\sum_{i \notin \bS^t} \left|\bc_{i,b}^{t+1} - \tilde{\bc}_{i,b}^{t+1} \right| \right] \leq \frac{1}{t} \sum_{i=1}^n \left( \sum_{\ell=1}^{t} \ell^2 \sum_{j=1}^n \Ex\left[\left( \berr_{i,b}^{\ell+1}(j) - \frac{\ell-1}{\ell} \cdot \berr_{i,b}^{\ell}(j)\right)^2 \right] \right)^{1/2}.
\end{align*}
\end{lemma}

\begin{proof}
We will upper bound (\ref{eq:expression}), where we first prove by induction on $t \in \{0, \dots, t_e-1\}$ that we always satisfy
\begin{align*}
\Ex\left[ \left( \berr_{i,b}^{t+1} \right)^2\right] = \sum_{\ell=1}^{t} \frac{\ell^2}{t^2} \Ex\left[ \left(\berr_{i,b}^{\ell+1} - \frac{\ell-1}{\ell} \cdot \berr_{i,b}^{\ell} \right)^2\right].
\end{align*}
Note that the base case of $t = 0$ corresponds to $\berr_{i,b}^1$, which is deterministically $0$ by $\bR_j^0 = 1$. Thus, assume for the sake of induction that
\begin{align*}
\Ex\left[ \left( \berr_{i,b}^{t} \right)^2 \right] = \sum_{\ell=1}^{t-1} \frac{\ell^2}{(t-1)^2} \cdot \Ex\left[ \left(\berr_{i,b}^{\ell+1} - \frac{\ell-1}{\ell} \cdot \berr_{i,b}^{\ell}\right)^2 \right].
\end{align*}
Then, we have 
\begin{align*}
\Ex\left[ \left( \berr_{i,b}^{t+1} \right)^2\right] &= \Ex\left[ \left( \berr_{i,b}^{t+1} - \frac{t-1}{t} \cdot \berr_{i,b}^{t} + \frac{t-1}{t} \cdot \berr_{i,b}^t \right)^2\right] \\
        &= \Ex\left[ \left( \berr_{i,b}^{t+1} - \frac{t-1}{t} \cdot \berr_{i,b}^{t} \right)^2\right] + \left(\frac{t-1}{t} \right)^2 \Ex\left[ (\berr_{i,b}^{t})^2\right],
\end{align*}
where we expand and use Lemma~\ref{lem:martingale} to say
\begin{align*} 
\Ex\left[ \left( \berr_{i,b}^{t+1} - \frac{t-1}{t} \cdot \berr_{i,b}^{t} \right) \cdot \left(\frac{t-1}{t} \cdot \berr_{i,b}^t \right)\right] = 0, 
\end{align*}
since in particular, it holds for every setting of $\err_{i,b}^{t}$. Applying the inductive hypothesis:
\begin{align*}
    &\Ex\left[ \left(\berr_{i,b}^{t+1}\right)^2 \right] \\
    &\qquad = \Ex\left[ \left( \berr_{i,b}^{t+1} - \frac{t-1}{t} \cdot \berr_{i,b}^{t} \right)^2\right] + \left(\frac{t-1}{t}\right)^2 \sum_{\ell=1}^{t-1} \frac{\ell^2}{(t-1)^2} \Ex\left[\left( \berr_{i,b}^{\ell+1} - \frac{\ell-1}{\ell} \cdot \berr_{i,b}^{\ell} \right)^2 \right] \\
    &\qquad = \sum_{\ell=1}^{t} \frac{\ell^2}{t^2} \Ex\left[ \left( \berr_{i,b}^{\ell+1} - \frac{\ell-1}{\ell} \cdot \berr_{i,b}^{\ell} \right)^2 \right].
\end{align*}
We now expand $\berr_{i,b}^{\ell+1} = \sum_{j=1}^n \berr_{i,b}^{\ell+1}(j)$, and the fact that $\berr_{i,b}^{\ell+1}(j)$ and $\berr_{i,b}^{\ell+1}(j')$ are independent to conclude, by Lemma~\ref{lem:martingale} once more, the desired inequality.

\end{proof}

\begin{lemma}\label{lem:bound-on-error}
Fix any $\ell > 1$ and any $i, j \in [n]$ with $w_j \leq 1/2$,
\begin{align*}
\Ex\left[\left( \berr_{i,b}^{\ell+1}(j) - \frac{\ell-1}{\ell} \cdot \berr_{i,b}^{\ell}(j) \right)^2 \right] &\leq 61 \cdot d(x_i, x_j)^2 \cdot \dfrac{1}{\ell^2 \cdot w_j \cdot \gamma^{\ell}}.
\end{align*}
\end{lemma}

\begin{proof}
We will compute the expectation by first considering any fixed execution of the first $\ell-1$ steps, and then take the expectation over the randomness in the $\ell$-th step. Note that, by fixing the first $\ell-1$ steps, we write:
\begin{align*}
\berr_{i,b}^{\ell+1}(j) &= d(x_i, x_j) \cdot \left(1 - \frac{1}{\ell} \cdot \frac{\bA_{j,\ell}^{\ell} \cdot \bK_{j,\ell}^{\ell}}{w_j^{\ell} \cdot \gamma^{\ell}} - \frac{\ell-1}{\ell} \cdot R_j^{\ell-1} \right) \hat{\bz}_{j,b}^{\ell}, \\
\err_{i,b}^{\ell}(j) &= d(x_i, x_j) \cdot \left(1 - R_j^{\ell-1} \right) \hat{z}_{j,b}^{\ell-1}.
\end{align*}
In the case that $j \in S^{\ell-1}$, we have fixed the setting of $\hat{\bz}_{j,b}^{\ell} = \hat{z}_{j,b}^{\ell-1}$, by Line~\ref{en:greedy-ln1}, $w_j^{\ell} = w_j$. The expectation, solely over the randomness in the draw of $\bA_{j,\ell}^{\ell} \sim \Ber(w_j)$ and $\bK_{j,\ell}^{\ell} \sim \Ber(\gamma^{\ell})$, is  
\begin{align} 
&\Ex\left[ \ind\{ j \in S^{\ell-1} \} \left(\berr_{i,b}^{\ell+1}(j) - \frac{\ell-1}{\ell} \cdot \err_{i,b}^{\ell}(j)\right)^2 \right] \nonumber \\
&\qquad = \ind\{ j \in S^{\ell-1} \} \cdot d(x_i, x_j)^2 \cdot \frac{(\hat{z}_{j,b}^{\ell-1})^2}{\ell^2} \cdot w_j \cdot \gamma^{\ell} \left(1 - \frac{1}{w_j \cdot \gamma^{\ell}}\right)^2 \leq \dfrac{d(x_i, x_j)^2}{\ell^2 \cdot w_j \cdot \gamma^{\ell}}. \label{eq:activated-bound}
\end{align}
So, we now consider the case that $j \notin S^{\ell-1}$, where $\bA_{j,\ell}^{\ell} \sim \Ber(w_j^{\ell})$ for $w_j^{\ell} = \min\{ w_j, 1/\ell\}$ and $\bK_{j,\ell}^{\ell} \sim \Ber(\gamma^{\ell})$. Similarly to the proof of Lemma~\ref{lem:martingale}, the expectation is slightly more complicated, as $\hat{\bz}_{j,b}^{\ell}$ depends on the realization of $\bA_{j,\ell}^{\ell}$. However, the fact $j \notin S^{\ell-1}$ means $R_j^{\ell-1} = 0$, which simplifies the expression as well. Consider the three cases:
\begin{itemize}
\item With probability $w_{j}^{\ell} \cdot \gamma^{\ell}$, both $\bA_{j,\ell}^{\ell}$ and $\bK_{j,\ell}^{\ell}$ are both set to $1$, and $\hat{\bz}_{j,b}^{\ell}$ is set to $g_{j,b}^{\ell}$. So,
\begin{align*}
\left( \berr_{i,b}^{\ell+1}(j) - \frac{\ell-1}{\ell} \cdot \err_{i,b}^{\ell}(j) \right)^2 &= d(x_i, x_j)^2 \cdot \left( \left(1 - \frac{1}{\ell \cdot w_j^{\ell} \cdot \gamma^{\ell}} \right) g_{j,b}^{\ell} - \frac{\ell-1}{\ell} \cdot \hat{z}_{j,b}^{\ell-1} \right)^2 \\
	&\leq d(x_i, x_j)^2 \cdot \left( \frac{3}{\ell \cdot w_j^{\ell} \cdot \gamma^{\ell}} \right)^2,
\end{align*}
where we used the fact $w_j^{\ell} \leq 1/\ell$, as well as the fact $g_{j,b}^{\ell}, \hat{z}_{j,b}^{\ell-1} \in [0,1]$.
\item With probability $w_j^{\ell} \cdot (1 - \gamma^{\ell})$, $\bA_{j,\ell}^{\ell}$ is set to $1$, but $\bK_{j,\ell}^{\ell}$ is set to $0$. In this case, $\hat{\bz}_{j,b}^{\ell}$ is set to $g_{j,b}^{\ell}$, and we obtain
\begin{align*}
\left( \berr_{i,b}^{\ell+1}(j) - \frac{\ell-1}{\ell} \cdot \err_{i,b}^{\ell}(j)\right)^2 &= d(x_i, x_j)^2 \cdot \left( g_{j,b}^{\ell} - \frac{\ell-1}{\ell} \cdot \hat{z}_{j,b}^{\ell-1} \right)^2 \leq 4 \cdot d(x_i, x_j)^2
\end{align*}
\item With probability $1 - w_j^{\ell}$, we set $\bA_{j,\ell}^{\ell}$ to $0$ (irrespective of the draw of $\bK_{j,\ell}^{\ell}$), in which case $\hat{\bz}_{j,b}^{\ell}$ is updated according to the second case in Definition~\ref{def:fictitious}. Thus, we obtain
\begin{align*}
\left( \berr_{i,b}^{\ell+1}(j) - \frac{\ell-1}{\ell} \cdot \err_{i,b}^{\ell}(j)\right)^2 &= d(x_i, x_j)^2 \cdot \left( \frac{w_j^{\ell}}{1 - w_{j}^{\ell}} \cdot \frac{\ell-1}{\ell} \cdot \hat{z}^{\ell-1}_{j,b} + \frac{1}{\ell} \cdot g_{j,b}^{\ell} - w_j^{\ell} \cdot g_{j,b}^{\ell} \right)^2 \\
		&\leq d(x_i, x_j)^2 \cdot \left( \dfrac{4}{\ell} \right)^2,
\end{align*}
where we use the fact $w_j^{\ell} \leq \min\{ 1/\ell, 1/2\}$ to obtain our desired bound.
\end{itemize}
Once we combine the above three cases (while multiplying by their respective probabilities), 
the first case, when $\bA_{j,\ell}^{\ell} \cdot \bK_{j, \ell}^{\ell} = 1$, dominates the calculation. In particular, when we compute the expectation solely over the randomness in $\bA_{j,\ell}^{\ell}$ and $\bK_{j,\ell}^{\ell}$, 
\begin{align*}
\Ex\left[ \ind\{ j \notin S^{\ell-1} \} \left( \berr_{i,b}^{\ell+1}(j) - \frac{\ell-1}{\ell} \cdot \err_{i,b}^{\ell}(j) \right)^2 \right] &\leq \ind\{ j \notin S^{\ell-1} \} \cdot \dfrac{30 \cdot d(x_i, x_j)^2}{\ell^2 \cdot w_j^{\ell} \cdot \gamma^{\ell}}.
\end{align*}
The above bounds the expectations over the randomness solely over step $\ell$. However, for large $\ell$, $w_{j}^{\ell} = 1/\ell$ which weakens the bound (contrast this with (\ref{eq:activated-bound}) where $w_{j}^{\ell}$ is replaced by $w_{j}$). However, we exploit the fact $j \notin \bS^{\ell-1}$ with probability which decays for large $\ell$. In particular, whenever $\ell-1 \geq 1/w_j$, we may obtain the following bound on the probability $j \notin \bS^{\ell-1}$,
\begin{align}
\Prx\left[ j \notin \bS^{\ell-1} \right] = \prod_{\ell'=1}^{\ell-1} \left(1 - \min\{ w_j, 1/\ell' \} \right) \leq \left(1 - w_j \right)^{\lfloor 1/w_j \rfloor} \dfrac{\lceil 1/w_j \rceil - 1}{\ell-1} \leq \frac{1}{w_j \cdot (\ell-1)}. \label{eq:prob-bound}
\end{align}
Applying this to the expectation over the time steps up to $\ell-1$, 
\begin{align*}
\Ex\left[ \left(\berr_{i,b}^{\ell+1}(j) - \frac{\ell-1}{\ell} \cdot \berr_{i,b}^{\ell}(j) \right)^{2}\right] &\leq \frac{d(x_i,x_j)^2}{\ell^2 \cdot w_j \cdot \gamma^{\ell}} + \Prx\left[ j \notin \bS^{\ell-1} \right] \cdot \dfrac{30 \cdot d(x_i, x_j)^2}{\ell^2 \cdot \min\{w_j, 1/\ell\} \cdot \gamma^{\ell}},\end{align*}
and this value is at most 
\[ \frac{31 \cdot d(x_i, x_j)^2}{\ell^2 \cdot w_j \cdot \gamma^{\ell}} \]
whenever $w_j^{\ell} = \min\{ 1/\ell, w_j \} = w_j$. Otherwise, we have $w_j^{\ell} = 1/\ell$ and this occurs when $\ell > 1$ since $w_j \leq 1/2$. The above expectation becomes at most
\begin{align*}
\frac{d(x_i, x_j)^2}{\ell \cdot w_j \cdot \gamma^{\ell}} + \dfrac{1}{(\ell - 1) \cdot w_j} \cdot \frac{30 \cdot d(x_i, x_j)^2}{\ell \cdot \gamma^{\ell}} \leq \frac{61 \cdot d(x_i, x_j)^2}{\ell^2 \cdot w_j \cdot \gamma^{\ell}}.
\end{align*}

\end{proof}

Combining Lemma~\ref{lem:reduce-to-error} and~\ref{lem:bound-on-error}, we obtain:
\begin{align}
\Ex\left[ \sum_{i \notin \bS^t} \left| \bc_{i,b}^{t} - \tilde{\bc}_{i,b}^t \right| \right] &\leq  \frac{1}{t} \sum_{i=1}^n \left( 61 \sum_{\ell=1}^{t} \sum_{j=1}^n \dfrac{d(x_i, x_j)^2}{w_j \cdot \gamma^{\ell}}  \right)^{1/2}  \leq \frac{10}{t} \sum_{i=1}^n \left( \sum_{j=1}^n \frac{d(x_i, x_j)^2}{w_j} \right)^{1/2} \left(\sum_{\ell=1}^t \frac{1}{\gamma^{\ell}} \right)^{1/2} \label{eq:deviation-exp}
\end{align}
and we have,
\begin{align*}
\left( \sum_{\ell=1}^{t} \frac{1}{\gamma^{\ell}}\right)^{1/2} \leq \left( 1 + \gamma + \frac{1}{\gamma} \sum_{\ell=1}^t \ell \right)^{1/2} \leq 2\sqrt{\gamma}  + \frac{t}{\sqrt{\gamma}} ,
\end{align*}
which means that plugging in to the above bound gives us our desired bound of
\begin{align*}
10 \left(\frac{2\sqrt{\gamma}}{t} + \frac{1}{\sqrt{\gamma}} \right) \sum_{i=1}^n \left( \sum_{j=1}^n \frac{d(x_i, x_j)^2}{w_j} \right)^{1/2}.
\end{align*}

\ignore{ \paragraph{Proof Lemma \ref{lemma:ficticous-approximation}}
\begin{proof}
Now we have the tools to finish the proof of Lemma \ref{lemma:ficticous-approximation}. We start by expanding out the cut value and breaking it into smaller pieces. 

We first expand the change in cut value over each of the time steps using linearity of expectation. 
\begin{align*}
    \E[f(\hat{\bz}) - f(\hat{\bz}^{0})] 
    &= \sum_{t=t_0}^{t_e-1} \E[f(\hat{\bz}^{t+1}) - f(\hat{\bz}^{t})] \\
    &=  \sum_{t=t_0}^{t_e-1} \E[ \frac{1}{2}\sum_{i=1}^{n} \sum_{j=1}^{n} d(x_i, x_j) (\hat{\bz}^{t+1}_{i,0} \hat{\bz}^{t+1}_{j,0} -  \hat{\bz}^t_{i,0} \hat{\bz}^t_{j,0}) + \frac{1}{2}\sum_{i=1}^{n} \sum_{j=1}^{n} d(x_i, x_j) (\hat{\bz}^{t+1}_{i,1} \hat{\bz}^{t+1}_{j,1} -  \hat{\bz}^t_{i,1} \hat{\bz}^t_{j,1})]\\
    &=  \sum_{t=t_0}^{t_e-1} \E[ \frac{1}{2}\sum_{i=1}^{n} \sum_{j=1}^{n} d(x_i, x_j) (\hat{\bz}^{t+1}_{i,0} -  \hat{\bz}^t_{i,0})  \hat{\bz}^{t}_{j,0} + \frac{1}{2}\sum_{i=1}^{n} \sum_{j=1}^{n} d(x_i, x_j) (\hat{\bz}^{t+1}_{i,1} -  \hat{\bz}^t_{i,1}) \hat{\bz}^{t}_{j,1}\\
    &+\frac{1}{2}\sum_{i=1}^{n} \sum_{j=1}^{n} d(x_i, x_j) (\hat{\bz}^{t+1}_{j,0} -  \hat{\bz}^t_{j,0})  \hat{\bz}^{t+1}_{i,0} + \frac{1}{2}\sum_{i=1}^{n} \sum_{j=1}^{n} d(x_i, x_j) (\hat{\bz}^{t+1}_{j,1} -  \hat{\bz}^t_{j,1}) \hat{\bz}^{t+1}_{i,1}]
    \end{align*}
    Expanding and then simplifying via symmetry of $d(x_i,x_j)$ the terms can be rewritten as
    \begin{align*}
    &\sum_{t=t_0}^{t_e-1} \E[ \frac{1}{2}\sum_{i=1}^{n} (\hat{\bz}^{t+1}_{i,0} -  \hat{\bz}^t_{i,0}) \sum_{j=1}^{n} d(x_i, x_j) \hat{\bz}^{t}_{j,0} + \frac{1}{2}\sum_{i=1}^{n} (\hat{\bz}^{t+1}_{i,1} -  \hat{\bz}^t_{i,1}) \sum_{j=1}^{n} d(x_i, x_j) \hat{\bz}^{t}_{j,1}\\
    &+\sum_{i=1}^{n} \sum_{j=1}^{n} d(x_i, x_j) (\hat{\bz}^{t+1}_{j,0} -  \hat{\bz}^t_{j,0})  \hat{\bz}^{t+1}_{i,0} +\sum_{i=1}^{n} \sum_{j=1}^{n} d(x_i, x_j) (\hat{\bz}^{t+1}_{j,1} -  \hat{\bz}^t_{j,1}) \hat{\bz}^{t+1}_{i,1}\\
    &+\frac{1}{2}\sum_{i=1}^{n} \sum_{j=1}^{n} d(x_i, x_j) (\hat{\bz}^{t+1}_{i,0} -  \hat{\bz}^t_{i,0}) (\hat{\bz}^{t+1}_{j,0} -  \hat{\bz}^t_{j,0}) + \frac{1}{2}\sum_{i=1}^{n} \sum_{j=1}^{n} d(x_i, x_j) (\hat{\bz}^{t+1}_{i,1} -  \hat{\bz}^t_{i,1}) (\hat{\bz}^{t+1}_{j,1} -  \hat{\bz}^t_{j,1})]\\
    &=\sum_{t=t_0}^{t_e-1} \E[\sum_{i=1}^{n} (\hat{\bz}^{t+1}_{i,0} -  \hat{\bz}^t_{i,0}) \sum_{j=1}^{n} d(x_i, x_j) \hat{\bz}^{t}_{j,0} + \sum_{i=1}^{n} (\hat{\bz}^{t+1}_{i,1} -  \hat{\bz}^t_{i,1}) \sum_{j=1}^{n} d(x_i, x_j) \hat{\bz}^{t}_{j,1}\\
    &+\frac{1}{2}\sum_{i=1}^{n} \sum_{j=1}^{n} d(x_i, x_j) (\hat{\bz}^{t+1}_{i,0} -  \hat{\bz}^t_{i,0}) (\hat{\bz}^{t+1}_{j,0} -  \hat{\bz}^t_{j,0}) + \frac{1}{2}\sum_{i=1}^{n} \sum_{j=1}^{n} d(x_i, x_j) (\hat{\bz}^{t+1}_{i,1} -  \hat{\bz}^t_{i,1}) (\hat{\bz}^{t+1}_{j,1} -  \hat{\bz}^t_{j,1})]
\end{align*}
Now consider splitting the first two and last two terms in the summation by linearity of expectation. Then we can compute the expectation of the last two terms which is bounded by the expected activating at time $t+1$. This gives
\begin{align*}
    &\sum_{t=t_0}^{t_e-1} \E[ \frac{1}{2}\sum_{i=1}^{n} \sum_{j=1}^{n} d(x_i, x_j) (\hat{\bz}^{t+1}_{i,0} -  \hat{\bz}^t_{i,0}) (\hat{\bz}^{t+1}_{j,0} -  \hat{\bz}^t_{j,0}) + \frac{1}{2}\sum_{i=1}^{n} \sum_{j=1}^{n} d(x_i, x_j) (\hat{\bz}^{t+1}_{i,1} -  \hat{\bz}^t_{i,1}) (\hat{\bz}^{t+1}_{j,1} -  \hat{\bz}^t_{j,1})]\\
    &\leq \sum_{t=t_0}^{t_e-1} \frac{1}{2(t+1)^2}  \sum_{i \notin \bS^t}\sum_{j \notin \bS^t} d(x_i,x_j)  \prod_{l=1}^t(1-\bw^t_{i})(1-\bw^t_{j}) \big ((\bg^{t+1}_{i,0} - \hat{z}^t_{i,0})(\bg^{t+1}_{j,0} - \hat{z}^t_{j,0}) + (\bg^{t+1}_{i,1} - \hat{z}^t_{i,1})(\bg^{t+1}_{j,1} - \hat{z}^t_{j,1}) \big)\\
    &\leq \sum_{t=t_0}^{t_e-1} \frac{1}{2(t+1)^2}  \sum_{i \notin \bS^t}\sum_{j \notin \bS^t} d(x_i,x_j)\prod_{l=1}^t(1-\bw^t_{i})(1-\bw^t_{j})) \big ((1 + |\hat{z}^t_{i,0}|)(1 + |\hat{z}^t_{j,0}|) + (1 + |\hat{z}^t_{i,1}|)(1 + |\hat{z}^t_{j,1}|) \big)\\
\end{align*}
Applying Lemma \ref{lemma:fic-cut-bounded} further reduces this to 
\begin{align*}
    &\leq \sum_{t=t_0}^{t_e-1} \frac{4}{(t+1)^2}  \sum_{i \notin \bS^t}\sum_{j \notin \bS^t} d(x_i,x_j) \\
\end{align*}
Now consider the expectation of the remaining terms in the summation. First, take the expectation over a single time step to simply split these terms into two parts to reduce the calculation into a calculation of the error terms. 
\begin{align*}
    &\sum_{t=t_0}^{t_e-1} \underset{\bA^{t+1}, \bK^{t+1}, \bw^{t+1}}\E[\sum_{i=1}^{n} (\hat{\bz}^{t+1}_{i,0} -  \hat{\bz}^t_{i,0}) \sum_{j=1}^{n} d(x_i, x_j) \hat{\bz}^{t}_{j,0} + \sum_{i=1}^{n} (\hat{\bz}^{t+1}_{i,1} -  \hat{\bz}^t_{i,1}) \sum_{j=1}^{n} d(x_i, x_j) \hat{\bz}^{t}_{j,1}]\\
    &= \sum_{t=t_0}^{t_e-1} \underset{\bA^{t}, \bK^{t}, \bw^{t}}\E[ \frac{1}{t}\sum_{i=1}^{n} (\bg^{t+1}_{i,0} -  \hat{\bz}^t_{i,0}) \sum_{j=1}^{n} d(x_i, x_j) \hat{\bz}^{t}_{j,0} 
    +  \frac{1}{t} \sum_{i=1}^{n} (\bg^{t+1}_{i,1} -  \hat{\bz}^t_{i,1}) \sum_{j=1}^{n} d(x_i, x_j) \hat{\bz}^{t}_{j,1}]\\
    &= \sum_{t=t_0}^{t_e-1} \underset{\bA^{t}, \bK^{t}, \bw^{t}}\E[ \frac{1}{t}\sum_{i=1}^{n} (\bg^{t+1}_{i,0} -  \hat{\bz}^t_{i,0}) \sum_{j=1}^{n} d(x_i, x_j) (1-\bR_j^t) \hat{\bz}^{t}_{j,0} 
    +  \frac{1}{t}\sum_{i=1}^{n} (\bg^{t+1}_{i,0} -  \hat{\bz}^t_{i,0}) \sum_{j=1}^{n} d(x_i, x_j) \bR_j^t\hat{\bz}^{t}_{j,0} \\
    &+  \frac{1}{t} \sum_{i=1}^{n} (\bg^{t+1}_{i,1} -  \hat{\bz}^t_{i,1}) \sum_{j=1}^{n} d(x_i, x_j) (1-\bR_j^t)\hat{\bz}^{t}_{j,1}
    +  \frac{1}{t} \sum_{i=1}^{n} (\bg^{t+1}_{i,1} -  \hat{\bz}^t_{i,1}) \sum_{j=1}^{n} d(x_i, x_j) \bR_j^t\hat{\bz}^{t}_{j,1}]
\end{align*}
Applying Lemma \ref{lemma:greedy-is-optimal}, We are now left with:
\begin{align*}
    &\leq \sum_{t=t_0}^{t_e-1} \underset{\bA^{t}, \bK^{t}, \bw^{t}}\E[ \frac{1}{t}\sum_{i=1}^{n} (\bg^{t+1}_{i,0} -  \hat{\bz}^t_{i,0}) \sum_{j=1}^{n} d(x_i, x_j) (1-\bR_j^t) \hat{\bz}^{t}_{j,0} +  \frac{1}{t} \sum_{i=1}^{n} (\bg^{t+1}_{i,1} -  \hat{\bz}^t_{i,1}) \sum_{j=1}^{n} d(x_i, x_j) (1-\bR_j^t)\hat{\bz}^{t}_{j,1}]
\end{align*}

Note that $\sum_{j=1}^n d(x_i, x_j) (1-\bR_j^{t})\hat{\bz}_{j,1}^t $ is just error term used to define the shadow. We will first refine the bound to be in terms of the error at time t. WOLOG consider only one side of the cut.  Apply Lemma \ref{lemma:fic-cut-bounded} and
\begin{align*}
    &\sum_{t=t_0}^{t_e-1} \underset{\bA^{t}, \bK^{t}, \bw^{t}}\E[ \frac{1}{t}\sum_{i\notin \bS_t} (g_{i,0}^{t+1} - \hat{\bz}_{i,0}^t) \left (\sum_{j=1}^n d(x_i, x_j) (1-\bR_j^{t})\hat{\bz}_{j,0}^t \right) + (\bg^{t+1}_{i,1} -  \hat{\bz}^t_{i,1}) \sum_{j=1}^{n} d(x_i, x_j) (1-\bR_j^t)\hat{\bz}^{t}_{j,1}]\\
    &\leq   \sum_{t=t_0}^{t_e-1} \underset{\bA^{t}, \bK^{t}, \bw^{t}}\E\left [ \frac{1}{t}\sum_{i\notin \bS_t} \lvert g_{i,0}^{t+1} - \hat{\bz}_{i,0}^t \rvert \left \lvert \sum_{j=1}^n d(x_i, x_j) (1-\bR_j^{t})\hat{\bz}_{j,0}^t \right \rvert + \lvert g_{i,1}^{t+1} - \hat{\bz}_{i,1}^t \rvert \left \lvert \sum_{j=1}^n d(x_i, x_j) (1-\bR_j^{t})\hat{\bz}_{j,1}^t \right \rvert\right ]\\
    &\leq  \sum_{t=t_0}^{t_e-1} \sum_{s\in \{0,1\}}\underset{\bA^{t}, \bK^{t}, \bw^{t}}\E\left [  \frac{4}{t}\sum_{i\notin \bS_t} \left \lvert \sum_{j=1}^n d(x_i, x_j) (1-\bR_j^{t})\hat{\bz}_{j,s}^t \right \rvert \right ]\\
    &=  \sum_{t=t_0}^{t_e-1}  \sum_{s\in \{0,1\}}\underset{\bA^{t}, \bK^{t}, \bw^{t}}\E\left [  \frac{4}{t}\sum_{i\notin \bS_t} \lvert \mathbf{err}^t_{i,s} \rvert \right ]\\
    &=  \sum_{t=t_0}^{t_e-1}  \sum_{s\in \{0,1\}}\underset{\bA^{t}, \bK^{t}, \bw^{t}}\E\left [  \frac{4}{t}\sum_{i\notin \bS_t} \sqrt{(\mathbf{err}^t_{i,s})^2}\right ]\\
\end{align*}
Now taking the expectation, we have via linearity of expectation, Jensen's inequality and independence:
\begin{align*}
        \underset{\bA^{t}, \bK^{t}, \bw^{t}}\E[\frac{4}{t}\sum_{i\notin \bS_t} \sqrt{(\mathbf{err}^t_{i,s})^2}]
        &= \frac{4}{t}\sum_{i\notin \bS_t} \underset{\bA^{t}, \bK^{t}, \bw^{t}}\E[\sqrt{(\mathbf{err}^t_{i,s})^2}] \\
        &\leq \frac{4}{t}\sum_{i\notin \bS_t} \sqrt{\underset{\bA^{t}, \bK^{t}, \bw^{t}}\E[(\mathbf{err}^t_{i,s})^2]} \\ 
        &= \frac{4}{t}\sum_{i\notin \bS_t} \sqrt{\underset{\bA^{t}, \bK^{t}, \bw^{t}}\E[(\mathbf{err}^t_{i,s} -  \frac{t-1}{t}\mathbf{err}^{t-1}_{i,s} + \frac{t-1}{t}\mathbf{err}^{t-1}_{i,s})^2 ]} \\
\end{align*}
Note that we can break up the expectation over all time steps up to $t$ into the expectation over all time steps up to $t-1$ and the time step at just $t$, and compute the randomness step by step.  
\begin{align*}
        &\leq \frac{4}{t}\sum_{i\notin \bS_t} \sqrt{\underset{\bA^{t-1}, \bK^{t-1}, \bw^{t-1}}\E \left [\underset{\bA^{t}_{t}, \bK^{t}_{t}, \bw^{t}_{t}}\E[(\mathbf{err}^t_{i,s} - \frac{t-1}{t}\mathbf{err}^{t-1}_{i,s})^2] +  (\frac{t-1}{t}\mathbf{err}^{t-1}_{i,s})^2  \right ]} \\
        &\leq \frac{4}{t}\sum_{i\notin \bS_t} \sqrt{\underset{\bA^{t}, \bK^{t}, \bw^{t}}\E[(\mathbf{err}^t_{i,s} - \frac{t-1}{t}\mathbf{err}^{t-1}_{i,s})^2] + \underset{\bA^{t-1}, \bK^{t-1}, \bw^{t-1}}\E[(\frac{t-1}{t}\mathbf{err}^{t-1}_{i,s})^2]} \\
\end{align*}
We can recursively break apart each term and apply the bound on $ \mathbf{Var}[(\mathbf{err}^t_{i,s} - \frac{t-1}{t}\mathbf{err}^{t-1}_{i,s})]$ from observation \ref{lemma:error-variance}. 
Note that by construction $\E[\mathbf{err}^t_{i,s} -\frac{t-1}{t}\mathbf{err}^{t-1}_{i,s}]=0$, and hence $\E[(\mathbf{err}^t_{i,s} - \frac{t-1}{t}\mathbf{err}^{t-1}_{i,s})^2]=\textbf{Var}[\mathbf{err}^t_{i,s} - \frac{t-1}{t}\mathbf{err}^{t-1}_{i,s}]$
\begin{align*}
    &\frac{4}{t}\sum_{i\notin \bS_t} \sqrt{\underset{\bA^{t}, \bK^{t}, \bw^{t}}\E[(\mathbf{err}^t_{i,s} - \frac{t-1}{t}\mathbf{err}^{t-1}_{i,s})^2] + \underset{\bA^{t-1}, \bK^{t-1}, \bw^{t-1}}\E[(\frac{t-1}{t}\mathbf{err}^{t-1}_{i,s})^2]} \\
    &\leq \frac{4}{t}\sum_{i\notin \bS_t} \sqrt{\sum_{\ell=1}^{t} \frac{\ell^2}{t^2}\underset{\bA^{\ell}, \bK^{\ell}, \bw^{\ell}}\E[(\mathbf{err}^\ell_{i,s} - \frac{\ell-1}{\ell}\mathbf{err}^{\ell-1}_{i,s})^2]} \\
    &\leq \frac{4}{t}\sum_{i\notin \bS_t} \sqrt{\sum_{\ell=1}^{t} \frac{\ell^2}{t^2} \underset{\bA^{\ell}, \bK^{\ell}, \bw^{\ell}}{\textbf{Var}}[(\mathbf{err}^\ell_{i,s}- \frac{\ell-1}{\ell}\mathbf{err}^{\ell-1}_{i,s})]} \\
\end{align*}
Plugging in our bound on the variance, claim \ref{claim:importance-sample-weights} and setting $\bgamma = \frac{\log^2(t_e)}{t\epsilon^2}$ gives
\begin{align*}
    \frac{4}{t}\sum_{i\notin \bS_t} \sqrt{\sum_{\ell=1}^{t} \frac{\ell^2}{t^2} \underset{\bA^{\ell}, \bK^{\ell}, \bw^{\ell}}{\textbf{Var}}[(\mathbf{err}^\ell_{i,s} - \frac{\ell-1}{\ell}\mathbf{err}^{\ell-1}_{i,s})]} 
    &\leq \frac{4}{t}\sum_{i\notin \bS_t} \sqrt{\sum_{l=1}^{t}\sum_{j=1}^n d(x_i, x_j)\frac{2}{t^2w_{j,t}\gamma_{t}}} \\
    &\leq \frac{16\epsilon}{t\log(t_e)} \sum_{i=1}^n\sum_{j=1}^n d(x_i,x_j) \\
\end{align*}
Bringing everything together, we can bound the entire sum by noting that the $t_0\geq\frac{1}{\epsilon^2}$ \\
\begin{align*}
    &\leq \sum_{t=t_0}^{t_e-1} \left ( \frac{4}{(t+1)^2} +  \frac{16 \epsilon}{t \log(t_e)} \right ) \sum_{i=1}^n\sum_{j=1}^n d(x_i,x_j) \\
    &= O( \epsilon  \sum_{i=1}^n\sum_{j=1}^n  d(x_i,x_j) )\\
\end{align*}
\end{proof}

\paragraph{Two Important Lemmas} 

\begin{lemma}\label{lemma:greedy-is-optimal}
For any valid cut vector $\bz^t$, including $\hat{\bz}^t$ and cut estimate $\bY^t$, the greedy choice minimizes the cut value over all possible cuts with respect to the estimate. 
\begin{align*}
    \sum_{i \notin \bS^t} \left(\bg^{t+1}_{i,0} - \bz^{t}_{i,0}\right)\bY^t_{0} + \sum_{i \notin \bS^t} \left(\bg^{t+1}_{i,1} - \bz^{t}_{i,1}\right)\bY^t_{1} \leq 0\end{align*}
\end{lemma}

\begin{lemma}\label{lemma:ficticous-approximation} 
Let $\hat{\bz}$ be the fictitious cut specified by an execution of  \textsc{GreedyProcess} $((X, d), w, \gamma, t_0, t_e, z^*)$. Then as long as $t_0 \geq ??$ and $t_e \geq ??$, \[ \E[f(\hat{\bz})-f(\hat{\bz}^0)]\leq \epsilon \sum_{i=1}^n \sum_{j=1}^n d(x_i, x_j) \]
\end{lemma}

\paragraph{Proof of Lemma \ref{lemma:main}}
\begin{proof}
    At time $t=0$, for any element $i$, $\hat{\bz}^{-1}_i$ is all zero's, while $\bg^0_i=z^*_i$, thus $\hat{\bz}_i^{0}=\frac{1}{(1-\bw_i)}(z^*_i - \bw_i z^*_i) = z^*_i$ which sets the fictitious cut to the optimal cut value. Once we reach time $t_e$, with high probability {\color{blue} TODO Account for $t_e$ which should just require $t_e\geq\min_i(\frac{1}{w_i})$}, all $\bA^{{t_e}}_i$ have a nonzero entry and for each $i\in S^{t_e}$, $\hat{\bz}_i = \bz_i$, thus the fictitious cut matches the greedy cut. Combining this with Lemma \ref{lemma:ficticous-approximation}, we get $\E[f(\bz) - f(\bz^*)] = \E[f(\hat{\bz})-f(\hat{\bz}^0)] \leq \epsilon T $. 
\end{proof}

\paragraph{Proof Lemma \ref{lemma:greedy-is-optimal}}
\begin{proof}
The definition of the greedy choices ensures that we always minimize the internal distance, thus 
\[\bg^{t+1}_{i,0} \cdot \bY_0 + \bg^{t+1}_{i,1} \cdot \bY_1 = \min(\bY_0,\bY_1)\]
Recall for any cut vector $\bz^t$, including the fictitious cut $\hat{\bz}^t$ by \ref{lemma:fic-cut-bounded}, we have $|\bz_{i,0}^t-\bz_{i,1}^t|\leq 1$. Applying this second bound to the sum we see
\begin{align*}
    \bz^t_{i,0}\bY_0 + \bz^t_{i,1}\bY_1 &= \frac{\bY_0+\bY_1}{2} \cdot (\bz^t_{i,0}+\bz^t_{i,1}) + \frac{\bY_0-\bY_1}{2}\cdot(\bz^t_{i,0}-\bz^t_{i,1})\\
    &\leq \min_{\alpha\in [-1,1]}  \left \{ \frac{\bY_0+\bY_1}{2} + \alpha \cdot \frac{\bY_0-\bY_1}{2} \right \} \\
    &= \min\{\bY_0, \bY_1\}
\end{align*}
Hence 
\begin{align*}
    \sum_{i \notin \bS^t} \left(\bg^{t+1}_{i,0} - \bz^{t}_{i,0}\right)\bY^t_{0} + \sum_{i \notin \bS^t} \left(\bg^{t+1}_{i,1} - \bz^{t}_{i,1}\right)\bY^t_{1} \leq 0
\end{align*}
\end{proof}

\subsection{Fictitious Error}
In order to prove Lemma \ref{lemma:ficticous-approximation}, we need to understand the change of the fictitious cut at a given time step. Let's start by defining the ``error'' of the fictitious cut.  
\begin{definition}[Error] The error of row $i$ of the fictitious cut at time step $t+1$ is given by 
    \begin{align*} \textbf{err}^t_{i,0} &= \sum_{j=1}^n d(x_i, x_j) \left( \frac{1}{t-1} \sum_{\ell=1}^{t-1} \dfrac{\bA_{j,\ell}^{t-1} \cdot \bK_{j,\ell}^{t-1}}{\bw_{j}^\ell \cdot \bgamma^\ell} -1 \right) \hat{\bz}_{j, 0}^{t-1} \\
    \textbf{err}^t_{i,1} &= \sum_{j=1}^n d(x_i, x_j) \left( \frac{1}{t-1} \sum^{t-1}_{\ell=1} \dfrac{\bA_{j,\ell}^{t-1} \cdot \bK_{j,\ell}^{t-1}}{\bw_{j}^\ell \cdot \bgamma^\ell} -1 \right) \hat{\bz}_{j, 1}^{t-1}
    \end{align*}
\end{definition}
The most important property of the error is that it forms a martingale, where current error is exactly a scaled down version of the error at the prior time step. We also compute the variance of error of the fictitious cut.
\begin{observation} One can re-write the inner sum in terms of the current time step and a scaled version of the prior time step.
\[ \frac{1}{t+1}\sum_{l=1}^{t+1}\frac{\bA^{t+1}_{j,\ell} \cdot \bK^{t+1}_{j,\ell}}{\bw^\ell_{j} \cdot \bgamma^\ell} = 
\frac{t}{t+1}\left (\frac{1}{t}\sum_{l=1}^{t}\frac{\bA^{t}_{j,\ell} \cdot \bK^{t}_{j,\ell}}{\bw^\ell_{j}  \cdot\bgamma^\ell} \right ) + \frac{1}{t+1} \frac{\bA^t_{j,t+1} \cdot \bK^t_{j,t+1}}{\bw^{t+1}_{j} \cdot \bgamma^{t+1}}\] 
which will be useful when considering the expectation taken over a single time step. Now observe that in expectation the error at a time step $t+1$ is exactly a scaled version of the error at the prior time step $t$. For convenience, define the following sum of random variable to be \[\bR_j^{t+1} = \frac{1}{t+1}\sum_{l=1}^{t+1}\frac{\bA^{t+1}_{j,\ell} \cdot \bK^{t+1}_{j,\ell}}{\bw^\ell_{j} \cdot \bgamma^\ell}\]
\end{observation}

\begin{lemma}[Estimator is a Martingale] \label{lemma:martingale} Note that in expectation over the t+1-th time step, the scaled error for any row $i$ in the fictitious cut forms a martingale:  \[\underset{\bA^{t+1}, \bK^{t+1}, \bw^{t+1}}\E[\textbf{err}_{i,s}^{t+1} - \frac{t}{t+1} \textnormal{err}_{i,s}^{t}] = 0\].\end{lemma}
\begin{proof}
Simply expand the definitions and take linearity of expectation over the rows $j$ to get 
\begin{align*}
     \underset{\bA^{t+1}, \bK^{t+1}, \bw^{t+1}}\E[\textbf{err}_i^{t+1} - \frac{t}{t+1} \textnormal{err}_i^{t}]
    &=  \E[\sum_{j=1}^n d(x_i, x_j)(\bR_j^{t+1} - 1)\mathbf{\hat{\bz}}_{j,s}^{t+1} - \frac{t}{t+1} \sum_{j=1}^n d(x_i, x_j)(R_j^t - 1)\hat{\bz}_{j,s}^{t}]\\
    &= \sum_{j=1}^n d(x_i, x_j) \E[(\frac{\bA^{t+1}_{j,t+1} \cdot \bK^{t+1}_{j,t+1}}{\bw^{t+1}_{j} \cdot \bgamma^{t+1}}
    +\frac{t}{t+1}R_j^t - 1)\mathbf{\hat{\bz}}_{j,s}^{t+1} - \frac{t}{t+1} (R_j^t - 1)\hat{\bz}_{j,s}^{t}]
\end{align*}
Now focusing on the inner expectation, consider the different possible settings of $\hat{\bz}^{t+1}_j$. We have two cases to consider, either if $j \in S^{t}$ or $j \notin S^{t}$, In both cases, a simple computation of the expectation suffices:\\
\textbf{Case 1:}  $j \in S^{t}$ which fixes $\mathbf{\hat{\bz}}_{j,s}^{t+1} = \hat{\bz}_{j,s}^t$. Then 
\begin{align*}
    &  \underset{\bA^{t+1}, \bK^{t+1}, \bw^{t+1}}\E[(\frac{\bA^{t+1}_{j,t+1} \cdot \bK^{t+1}_{j,t+1}}{(t+1) \bw^{t+1}_{i} \cdot \bgamma^{t+1}}
    +\frac{t}{t+1}R_j^t - 1)\hat{\bz}_{j,s}^{t} - \frac{t}{t+1} (R_j^t - 1)\hat{\bz}_{j,s}^{t}] \\
    &=\frac{1}{(t+1)}\E[(\frac{\bA^{t+1}_{j,t+1} \cdot \bK^{t+1}_{j,t+1}}{\bw^{t+1}_{j} \cdot \bgamma^{t+1}}\hat{\bz}_{j,s}^{t} - \hat{\bz}_{j,s}^{t}] \\
    &= 0
\end{align*}
\textbf{Case 2:} $j \notin S^{t}$. This fixes $R_j^t = 0$. Then
\begin{align*}
    & \underset{\bA^{t+1}, \bK^{t+1}, \bw^{t+1}}\E[(\frac{\bA^{t+1}_{j,t+1} \cdot \bK^{t+1}_{j,t+1}}{(t+1) \bw^{t+1}_{j} \cdot \bgamma^{t+1}}
    +\frac{t}{t+1}R_j^t - 1)\mathbf{\hat{\bz}}_{j,s}^{t+1} - \frac{t}{t+1} (R_j^t - 1)\hat{\bz}_{j,s}^{t}] \\
    & = \E[(\frac{\bA^{t+1}_{j,t+1} \cdot \bK^{t+1}_{j,t+1}}{(t+1) \bw^{t+1}_{j} \cdot \bgamma^{t+1}}
    - 1)\mathbf{\hat{\bz}}_{j,s}^{t+1} + \frac{t}{t+1}\hat{\bz}_{j,s}^{t}] \\
    &= 0
\end{align*}
\end{proof}

We also bound the variance of the change in error over prior time steps. 
\begin{claim}[Variance] \label{lemma:error-variance}The variance of the change in error up to time step $t$ is bounded:  \begin{align*}  \underset{\bA^{t}, \bK^{t}, \bw^{t}_j}\Var[\textbf{err}_{i,s}^{t+1} - \frac{t}{t+1} \textnormal{err}_{i,s}^{t+1}] \leq \sum_{j=1}^n d(x_i, x_j)^2\frac{1}{t^2 w_{j}^{t+1}\gamma_{j}^{t+1}}
\end{align*}. 
\end{claim}
\begin{proof}
Recall that each row in $\bA^{t+1}, \bK^{t+1}, \bw^{t+1}$ is independent and hence the variance of the sum is the sum of the variances, since the covariance is 0.
\begin{align*}
\underset{\bA^{t+1}, \bK^{t+1}, \bw^{t+1}}{\textbf{Var}} \left [ \mathbf{err}^{t+1}_{i,s} - \frac{t}{t+1} \cdot \textnormal{err}^t_{i,s} \right ] = \sum_{j=1}^n d(x_i, x_j)^2 \underset{\bA^{t+1}, \bK^{t+1}, \bw^{t+1}}{\textbf{Var}}\left [(\bR_j^{t+1} -1 )\mathbf{\hat{\bz}}_{j,s}^{t+1} - \frac{t}{t+1}(R_j^{t+1}-1)\hat{\bz}_{j,s}^{t} \right ]
\end{align*}
Now by lemma \ref{lemma:martingale}, since each of the r.v.s has expectation 0, $\textbf{Var}$ is just the expectation of the squared r.v.
\begin{align*}
= \sum_{j=1}^n d(x_i, x_j)^2 \underset{\bA^{t+1}, \bK^{t+1}, \bw^{t+1}}\E\left [ \left (\bR_j^{t+1} -1 )\mathbf{\hat{\bz}}_{j,s}^{t+1} - \frac{t}{t+1}(R_j^{t+1}-1)\hat{\bz}_{j,s}^{t} \right )^2 \right ]
\end{align*}
Now as with the proof of lemma \ref{lemma:martingale} we once again split into cases: \\
\textbf{Case 1:} $j \in S_{t}$ so $\mathbf{\hat{\bz}}_{j,s}^{t+1} = \hat{\bz}_{j,s}^t$.
Then we have 
\begin{align*}
    & =  \E\left [ \left ((\frac{\bA^{t+1}_{j,t+1} \cdot \bK^{t+1}_{j,t+1}}{(t+1) \bw^{t+1}_{j} \cdot \bgamma^{t+1}} +\frac{t}{t+1}R_j^{t} - 1)\hat{\bz}_{j,s}^{t} - \frac{t}{t+1}(R_j^{t}-1)\hat{\bz}_{j,s}^{t} \right )^2 \right ] \\
    & =  \E\left [ \left ((\frac{\bA^{t+1}_{j,t+1} \cdot \bK^{t+1}_{j,t+1}}{(t+1) \bw^{t+1}_{j} \cdot \bgamma^{t+1}} - \frac{1}{t+1})\hat{\bz}_{j,s}^{t} \right )^2 \right ] \\
    &= \frac{1}{(t+1)^2}\left ((\bw^{t+1}_{j} \cdot \bgamma^{t+1})(\frac{1}{\bw^{t+1}_{j} \cdot \bgamma^{t+1}}\hat{\bz}_{j,s}^{t} - \hat{\bz}_{j,s}^{t})^2 + (1-\bw^{t+1}_{j} \cdot \bgamma^{t+1})(\hat{\bz}_{j,s}^{t})^2 \right ) \\
    &\leq\frac{1}{(t+1)^2 \bw^{t+1}_{j} \cdot \bgamma^{t+1}}
\end{align*}
\textbf{Case 2:} $j \notin S_{t}$. This fixes $R_j^t = 0$. 
\begin{align*}
    & = \E\left [ \left ((\frac{\bA^{t+1}_{j,t+1} \cdot \bK^{t+1}_{j,t+1}}{(t+1) \bw^{t+1}_{j} \cdot \bgamma^{t+1}} - 1)\hat{\bz}_{j,s}^{t+1} - \frac{t}{t+1}(-1)\hat{\bz}_{j,s}^{t} \right )^2 \right ] \\
    & =  \bw_{j}^{t+1} \bgamma^{t+1}\left ((\frac{1}{(t+1) \bw^{t+1}_{j} \cdot \bgamma^{t+1}} - 1)g_{j,s}^{t+1} - \frac{t}{t+1}(-1)\hat{\bz}_{j,s}^{t} \right )^2 \\
    &+  \bw_{j}^{t+1}(1-\gamma^{t+1}) \left((-1)g_{j,s}^{t+1} - \frac{t}{t+1}(-1)\hat{\bz}_{j,s}^{t} \right )^2 \\
    &+  (1-\bw_{j}^{t+1}) \left ((-1)\hat{\bz}_{j,s}^{t+1} - \frac{t}{t+1}(-1)\hat{\bz}_{j,s}^{t} \right )^2
\end{align*}
Now compute the expected squares and observe that the following cases maximize variance, either when $\hat{\bz}_{j,s}^{t}=1$ and $\bg_{j,s}^{t+1}=0$ or vice versa. Consider just the $(\hat{\bz}_{j,s}^{t})^2$ terms for now, this yields
\begin{align*}
    &= \frac{\bw_{j}^{t+1}t^2}{(1-\bw_{j}^{t+1})(t+1)^2}(\hat{\bz}^{t}_j)^2 \\
    &\leq \frac{1}{(t+1)^2 \bw_{j}^{t+1} \bgamma^{t}}(\bz_{j,s}^{t+1})^2 
\end{align*}
Now consider just the $(g_{j,s}^{t+1})^2$ terms for now, which can be simplifed to
\begin{align*}
    & \leq \frac{1}{(t+1)^2 \bw_{j}^{t+1} \bgamma^{t+1}}(g_{j,s}^{t+1})^2 
    - (g_{j,s}^{t+1})^2  
    +  \frac{1}{(1-\bw_{j}^{t+1})}\left (\frac{1}{(t+1)^2}-\frac{2\bw_{j}^{t+1}}{(t+1)} + (\bw_{j}^{t+1})^2 \right) (g_{j,s}^{t+1})^2 \\  
\end{align*}
Now observe that $(\bw_{j}^{t+1})^2 \leq \frac{\bw_{j}^{t+1}}{(t+1)}$, which gives:
\begin{align*}
    &\leq \frac{1}{(t+1)^2 \bw_{j}^{t+1} \bgamma^{t+1}}(g_{j,s}^{t+1})^2 
    - (g_{j,s}^{t+1})^2  
    +  \frac{1}{(1-\bw_{j}^{t+1})}\left (\frac{1}{(t+1)^2}-\frac{\bw_{j}^{t+1}}{(t+1)} \right) (g_{j,s}^{t+1})^2\\
    &\leq \frac{1}{(t+1)^2 \bw_{j}^{t+1} \bgamma^{t+1}}(g_{j,s}^{t+1})^2 
\end{align*}
Hence we can bound the expectation in this case by  
\begin{align*}
        &\leq \frac{1}{(t+1)^2 \bw_{j}^{t+1} \bgamma^{t+1}}
\end{align*}
Now when $w_j \leq \frac{1}{t}$, we have $\bw^t_j = w_j$, so the contribution of $j$ to the variance is trivially bounded by
\begin{align*}
        &\leq d(x_i, x_j)^2  \frac{2}{(t+1)^2 w_{j} \bgamma^{t+1}}
\end{align*}
If not  ($w_j \geq \frac{1}{t}$), then we check the probability of not activating element $j$ by time $t$, which can be computed as  $\Pr[j \notin \bS^t ] = (1-w_j)^{\frac{1}{w_j}}\prod^{t}_{\ell=\frac{1}{w_j}}(1-\frac{1}{\ell}) \leq \frac{\frac{1}{w_j}-1}{t}$. Recall for $j \in \bS^t$, $\bw^t_j = w_j$. We now compute a bound on the overall variance
\begin{align*}
    &\sum_{j=1}^n d(x_i, x_j)^2  \Pr[j \notin \bS^t] \underset{\bA^{t+1}, \bK^{t+1}, \bw^{t+1}}\E\left [ \left (\bR_j^{t+1} -1 )\mathbf{\hat{\bz}}_{j,s}^{t+1} - \frac{t}{t+1}(R_j^{t+1}-1)\hat{\bz}_{j,s}^{t} \right )^2  \big | \ j \notin \bS^t \right ] \\
    &+  \sum_{j=1}^n d(x_i, x_j)^2  \Pr[j \in \bS^t] \underset{\bA^{t+1}, \bK^{t+1}, \bw^{t+1}}\E\left [ \left (\bR_j^{t+1} -1 )\mathbf{\hat{\bz}}_{j,s}^{t+1} - \frac{t}{t+1}(R_j^{t+1}-1)\hat{\bz}_{j,s}^{t} \right )^2  \big | \ j \in \bS^t \right ] \\
    &\leq \sum_{j=1}^n d(x_i, x_j)^2  \frac{\frac{1}{w_j}}{t} \frac{1}{(t+1)^2 \bw_{j}^{t+1} \cdot \bgamma^{t+1}}+  \frac{1}{(t+1)^2 w_{j} \cdot \bgamma^{t+1}}\\
    &\leq \sum_{j=1}^n d(x_i, x_j)^2  \frac{2}{t(t+1) w_{j} \bgamma^{t+1}}
\end{align*}
Thus the overall variance over a single time step by 
\begin{align*}
\leq \sum_{j=1}^n d(x_i, x_j)^2\frac{2}{t(t+1) w_{j}\bgamma^{t+1}}
\end{align*}
\end{proof}

We need one final property of the weights before we can tackle the proof of Lemma \ref{lemma:ficticous-approximation} 
\begin{claim}{\color{blue} replace with citation} \label{claim:importance-sample-weights}
For any $i,j$ we have \[\frac{||x_i - x_j||_2}{\sum_{l=1}^n ||x_i - x_\ell|| _2} \leq 8 \left ( \frac{1}{T} \sum_{l=1}^n ||x_j - x_\ell||_2 \right ) \] Where $T=\sum_{i=1}^n\sum_{j=1}^n d(x_i,x_j)$ 
\end{claim} }

%% file: turnstile.tex

\renewcommand{\vec}{\mathrm{vec}}
\newcommand{\epspr}{\eps_{\mathrm{pr}}}

\section{Dynamic Streaming Algorithm}

 In dynamic geometric streaming algorithms, a low-space streaming algorithm will receive an input set of points $X \subset [\Delta]^d$ presented as a stream. The stream may consist of both insertions and deletions of individual points $x \in [\Delta]^d$, and the algorithm maintains a memory state for points currently in the set. We associate (multi)-sets $X$ from $[\Delta]^d$ by vectors $\vec(X)$ of integer coordinates with $\Delta^d$ coordinates. Each coordinate corresponds to a point in $[\Delta]^d$, and counts the number of times that point appears in the multi-set. An update specifies a coordinate $x \in [\Delta]^d$ and a number $\delta \in \{-1, 1\}$, and the underlying vector is updated by increasing the $x$-th coordinate by $\delta$ (i.e., inserting or deleting a point). With this encoding, all of our algorithm will be \emph{linear sketches}, they are given by $\bS \cdot \vec(X)$ for a random matrix $\bS$. However, we will use the fact that the final vector $\vec(X) \in \{0,1\}^{[\Delta]^d}$ corresponds to a set (and there are no duplicate points).

\paragraph{Geometric Sketches and Geometric Sampling Sketches.} We state the two preliminary components. The first will be the geometric sampling sketches from~\cite{CJK23}, and the second an algorithm to estimate the geometric weights. These sketches play the role of (i) sampling the ``active and kept'' points, i.e., those $x_i \in X$ with $\bA_{i,t} \cdot \bK_{i,t} = 1$ (for some $t$) in Theorem~\ref{thm:main-structural}, and (ii) computing the metric-compatible weights as in Section~\ref{sec:mpc} and Section~\ref{sec:insertion-only}. 

\begin{lemma}[Geometric Sampling Sketches (Lemma 4.2~\cite{CJK23})]\label{lem:geo-samp}
There is a randomized streaming algorithm initialized with $\epspr \in (0,1)$, $\delta \in (0,1)$, $p \geq 1$, $\Delta, d \in \N$, $s \in \N$ uses $\poly(s \cdot d \log (\Delta / \delta) / \epspr)$ space and satisfies the following:
\begin{itemize}
\item It receives an input dataset $X \subset [\Delta]^d$ of at least two distinct points as a dynamic stream. 
\item Then, it outputs $s$ pairs $(\by^*_j,\bp^*_j)$ of random variables $\by^*_j \in X \cup \{ \bot \}$, together with a number $\bp^*_j \in [0,1]$. 
\end{itemize}
With probability at most $\delta$, the algorithm fails (in a manner which is undetectable). If it doesn't fail, then each sample $\by^*_j$ is independent. With probability at least $1-1/\poly(\delta^{-1} s \Delta^d)$, for every $x \in X$,
\begin{align*}
\Prx_{(\by_j^*, \bp_j^*)}\left[ \by^*_j = x \right ] \eqdef p_j(x) \geq \frac{1}{\sfD} \cdot \dfrac{\sum_{y \in X} d_{\ell_p}(x, y)}{\sum_{y \in X} \sum_{y' \in X} d_{\ell_p}(y, y')},
\end{align*}
for known $\sfD = \poly(d \log \Delta / \delta)$. Whenever $\by^*_j \neq \bot$, then $\bp^*_j$ satisfies
\begin{align*}
p_j(\by_j^*) \leq \bp^*_j \leq (1+\epspr) p_j(\by_j^*).
\end{align*}
\end{lemma}

Lemma~\ref{lem:geo-samp} gives an algorithm to sample from $X$ with probability according to distances within $X$. We will refer to the distribution of the $j$-th sample generated from a call to Lemma~\ref{lem:geo-samp} conditioned on a successful initialization as $\calG_j(X)$. In the analysis, we set the failure probability $\delta$ to $\eps$, and we will condition on a successful initialization of Lemma~\ref{lem:geo-samp}. The case of an unsuccessful initialization is considered a ``catastrophic'' failure (but are rare enough so as to not affect the algorithm significantly).

\begin{remark}[Additive Error in Lemma~4.2 of~\cite{CJK23}]\label{rem:additive-err} \emph{The statement of Lemma~\ref{lem:geo-samp} differs slightly from Lemma~4.2 in~\cite{CJK23}, but the changes are solely cosmetic. We allow $\epspr \in (0, 1)$ and always overestimate, bound the failure probability by $\delta$ (instead of $.01$ in Lemma 4.2 of~\cite{CJK23}) by paying a $1 / \delta$ factor in $\sfD$, and replace $\lambda$ with $\sfD$. Similarly, we bound the probability of an individual sample failing by $1-\poly(\delta/(s\Delta^d))$ (instead of $1-1/\poly(\Delta^d)$ in Lemma 4.2 of~\cite{CJK23}) by paying an additional $\poly(\log(s\delta^{-1}))$ bits of space. By setting $\delta \leq \eps$, the probability that any of the samples fail is at most $\eps$, which would contribute at most $\eps \sum_{i=1}^n \sum_{j=1}^n d(x_i,x_j)$ total distance to the expected cut value. We will henceforth assume that the sampling procedure succeeds. }
\end{remark}

\begin{lemma}[Weight Sketches]\label{lem:geo-sketch}
There is a randomized streaming algorithm initialized with $p \geq 1$, $\Delta, d \in \N$, $\delta \in (0,1)$, which uses $\poly(d \log (\Delta/\delta))$ words of space, and for a known parameter $\sfD = \poly(d\log \Delta)$, satisfies the following:
\begin{itemize}
\item It receives an input set $X = \{ x_1,\dots, x_n \} \subset [\Delta]^d$ of at least two distinct points as a dynamic stream. 
\item Given a point $x \in [\Delta]^d$, it outputs a number $\bw(x)$.
\end{itemize} 
With probability at least $1 - \delta$ over the randomness in construction, every $x \in X$ satisfies
\[ \dfrac{\sum_{i=1}^n d_{\ell_p}(x_i, x)}{\sum_{i=1}^n \sum_{j=1}^{n} d_{\ell_p}(x_i, x_j)} \leq \bw(x) \leq \sfD \cdot \dfrac{\sum_{i=1}^n d_{\ell_p}(x_i, x)}{\sum_{i=1}^n \sum_{j=1}^{n} d_{\ell_p}(x_i, x_j)} . \]
\end{lemma}

\begin{proof}
    We describe a sketch which will process the dataset $X = \{x_1,\dots, x_n\} \subset [\Delta]^d$ and will be able to produce, for every $x \in [\Delta]^d$, a $\poly(d\log \Delta)$-approximation to the quantity $\boldeta(x) = \sum_{i=1}^n d_{\ell_p}(x_i, x)$, which is the desired numerator of the quantity $\bw(x)$ above. The denominator, which is the sum of all distances, can also be estimated up to a factor of $2$ within the space bound (in particular, it the value of the max-cut output by~\cite{CJK23} would be one such estimate). In order to approximate $\boldeta(x)$, we combine two ingredients:
    \begin{itemize}
        \item The first is Fact~1 from~\cite{I04} (and Lemma~4.3 in~\cite{CJK23}) which shows an embedding into a probabilistic tree $\bT$ (equivalently, an embedding into $\ell_1$) such that $d_{\bT}(x, y) \geq d_{\ell_p}(x, y)$ and $\Ex_{\bT}\left[ d_{\bT}(x,y) \right] \leq O(d\log \Delta) d_{\ell_p}(x,y)$ for any $x,y \in [\Delta]^d$. 
        \item The second is the Median Cost Evaluation sketch (MediEval) from~\cite{I04}, which can process a dataset of points $X$ such that for a set $C$ of at most $k$ points, the $k$-median cost of clustering $X$ with the center set $C$ can be recovered. 
    \end{itemize}
    In~\cite{I04}, the MediEval sketch is described for two-dimensional point sets, but it follows immediately for evaluating $k$-median distances in a tree metric as well. Our sketch will consider $L = O(d \log \Delta)$ draws of probabilistic tree embeddings $\bT_1, \dots, \bT_{L}$ as well as MediEval sketches for the dataset points in each tree instantiated to succeed with probability at least $1 - \delta/(\Delta^d \cdot L)$, instantiated for $k = 1$. Note that $\sum_{i=1}^n d_{\ell_p}(x_i, x)$ is the $k$-median with center $x$, so each of the $L$ trees gives us an estimate of $\sum_{i=1}^n d_{\bT_{\ell}}(x_i, x)$, which is at least $\sum_{i=1}^n d_{\ell_p}(x_i, x)$ and has expectation at most $O(d\log \Delta) \sum_{i=1}^n d_{\ell_p}(x_i, x)$. Hence, using the minimum over all $\ell \in [L]$ of the estimates $\sum_{i=1}^n d_{\bT_{\ell}}(x_i, x)$ gives the desired approximation to $\boldeta(x)$.
\end{proof}

We will use Lemma~\ref{lem:geo-sketch} in order to generate the metric-compatible weights $w_1, \dots, w_n$ which we use in order to activate points. The main difference between the lemmas used in prior sections (namely, Lemma~\ref{lem:mpc-weights} and Lemma~\ref{lem:computing-weights}) and Lemma~\ref{lem:geo-sketch} is that the above lemma works in the dynamic geometric streaming model and obtains a multiplicative $\poly(d\log \Delta)$-approximation, as opposed to $(1+\eps)$-approximations for $p \in [1,2]$. We also state the following lemma from~\cite{CJK23}, which will be used to generate samples from $X$ distributed according to metric-compatible weights.

\subsection{Activation Timelines and Masks for Dynamic Streams}

In this section, we state the main structural theorem that is needed for controlling the cut quality of the dynamic streaming algorithm (analogous to Theorem~\ref{thm:main-structural}). The structural theorem will refer to weights generated according to Lemma~\ref{lem:geo-sketch} and the distribution generated from Lemma~\ref{lem:geo-samp}, so that it is tailored to what can be achieved algorithmically for dynamic streams. We first define an appropriate version of activation timelines and masks, adapting Definition~\ref{def:activate} and Definition~\ref{def:mask}, as well as a timeline-mask summary, adapting Definition~\ref{def:timeline-mask-sum}. Then, we define a new assignment procedure, and state the main theorem.

\begin{definition}[Correlated Timeline-Mask]\label{def:corr-timeline-mask}
Let $X = \{ x_1,\dots, x_n \} \subset [\Delta]^d$, $t_0, t_e \in \N$, $\gamma > 1$ and $\eps,\epspr \in (0, 1)$. A correlated timeline-mask is sampled according to the following procedure:
\begin{enumerate}
\item\label{en:mask-corr} Sample a single mask $\bK \sim \calK(t_0, \gamma)$ from Definition~\ref{def:mask}, let $s$ denote the number of entries of $\bK$ set to $1$.
\item\label{en:weights} Initialize an algorithm from Lemma~\ref{lem:geo-sketch} with $\delta = \eps$, and let $\bw_i = \bw(x_i) / (2\sfD)$ for all $i \in [n]$. Similarly, initialize an algorithm from Lemma~\ref{lem:geo-samp} with $s$ (coming from $\bK$ above) and $\delta = \eps$, and let $(\by_j^*, \bp^*_j)$ for $j \in [s]$.
\item\label{en:iter-t-cor-def} Iterate through $t=1, \dots, t_e$ to generate $\bA_{1,t}, \dots, \bA_{n,t}$:
\begin{enumerate}
\item\label{en:k-0} If $\bK_t = 0$, each $i \in [n]$ proceeds independently by letting
\[ \bw^t_i = \left\{ \begin{array}{cc} \min\{ \bw_i, 1/t \} & \forall \ell < t , \bA_{i, \ell} = 0 \\
								 \bw_i & \text{o.w.} \end{array} \right. , \]
and then setting $\bA_{i,t} \sim \Ber(\bw_i^t)$ independently. 
\item\label{en:k-1} If $\bK_t = 1$, use a sample $(\by_t^*, \bp_t^*)$ from Line~\ref{en:weights} (re-indexed, see Remark~\ref{rem:re-indexing}), with error parameter $\epspr$ and failure probability $\eps$. If $i^* \in [n]$ is such that $\by^*_t = x_{i^*}$, let
\[ \brho_{i^*}^t = \left\{ \begin{array}{cc} \min\{ \bp^*_t, 1/t\} & \forall \ell < t, \bA_{i^*, \ell} = 0 \\
                \bp^*_t & \text{o.w.} \end{array} \right. .\]
Let $\bA_{i,t} = 0$ if $i \neq i^*$ (or all $i$ in case $\by^*_t = \bot$). If $\bA_{i^*,\ell} = 0$ for all $\ell < t$, sample $\bb_t \sim \Ber(\min\{ 1/(t\bp_t^*), 1\})$ and set $\bA_{i^*, t} = \bb_t$.\footnote{It is not exactly the case $\bA_{i^*, t}$ is one with probability $\brho_{i^*}^t$ because $\bp_t^{*}$ is not exactly the probability that $\by_t^* = x_{i^*}$.} Otherwise, let $\bA_{i^*, t} = 1$.
\end{enumerate}
\end{enumerate}
As in Definition~\ref{def:activate}, the activation time $\bt_i$ for $\bA_i$ is the smallest $t$ with $\bA_{i,t} = 1$.
\end{definition}

\begin{remark}[Re-indexing the Geometric Samples $(\by_t^*, \bp^*_t)$]\label{rem:re-indexing}
Given a mask $\bK$ and $s=\|\bK\|_1$ geometric samples $\{(\by^*_j,\bp^*_j)\} \sim \calG_s(X)$ (where $j\in [s]$) via Lemma \ref{lem:geo-samp}. Each sample corresponds to a non-zero entry of $\bK$. In the algorithm and later the analysis, it will be convenient to reference these samples using the index of the corresponding non-zero-entry in the mask. Hence we can re-index the samples $(\by^*_j,\bp^*_j)\rightarrow (\by^*_t,\bp^*_t)$ where $j$-th sample gets indexed by time $t$ where $\bK_t$ is the $j$-th non-zero entry of $\bK$.

\end{remark}

The remainder of the section mirrors that of Section~\ref{sec:structural-1}. The correlated timeline-mask from Definition~\ref{def:corr-timeline-mask} serves the purpose of the activation timelines and masks from Definition~\ref{def:activate} and~\ref{def:mask}. Furthermore, the $n$ timelines $\bA_1,\dots, \bA_n$ sampled from Definition~\ref{def:corr-timeline-mask} define a seed (see Definition~\ref{def:seed}). We define an analogous version of Definition~\ref{def:timeline-mask-sum}, where the same mask $\bK$ is used for $\bK_j$ for all $j \in [n]$.

\begin{definition}[Correlated Timeline-Mask Summary]\label{def:corr-timeline-mask-sum}
Consider $X = \{x_1,\dots, x_n\} \subset [\Delta]^d$, $t_e, t_0 \in \N$ and $\gamma \geq 1$.
\begin{itemize}
\item Sample a correlated timeline-mask $\bA_1,\dots, \bA_n$ and $\bK$ from Definition~\ref{def:corr-timeline-mask}.
\item The summary of $\bA_1,\dots, \bA_n$ and $\bK$ is given by the set of tuples
\begin{align*}
\bP &= \textsc{Summ}(\bA_1,\dots, \bA_n, \bK) \\
      &\eqdef \left\{ (x_j, \ell; p_{\ell}, t_j)  : \begin{array}{l} \bA_{j,\ell} \cdot \bK_{\ell} = 1 \\ \text{$t_j$ is the activation time of $\bA_j$} \\ p_{\ell} = \bp^*_{\ell} \text{ for $(\by^*_{\ell}, \bp^*_{\ell})$ produced by Line~\ref{en:k-1} in Definition~\ref{def:corr-timeline-mask}} \end{array} \right\}  .
\end{align*}
\end{itemize}
\end{definition}

We present the assignment rule for the dynamic streaming algorithm, $\textsc{Assign}$, in Figure~\ref{fig:assign-dynamic}. The assignment rule is analogous to that of Figure~\ref{fig:assign} which we used for Theorem~\ref{thm:main-structural}. The only difference between Figure~\ref{fig:assign} and Figure~\ref{fig:assign-dynamic} is in the use of a correlated timeline-mask $\bA_1,\dots, \bA_n$ and $\bK$ from Definition~\ref{def:corr-timeline-mask} and the correlated timeline-mask summary $\bP$ from Definition~\ref{def:corr-timeline-mask-sum}.

\begin{figure}
\begin{framed}
\textbf{Assignment Rule} \textsc{Assign}$_{\sigma,\bP}(x_i, t_i)$. The correlated timeline-mask $\bA_1, \dots, \bA_n$ and $\bK$ from Definition~\ref{def:corr-timeline-mask} has been sampled. $\bP$ is the correlated timeline-mask summary from Definition~\ref{def:corr-timeline-mask-sum} and $\sigma$ is a seed for $\bA_1,\dots, \bA_n$ (see Definition~\ref{def:seed}). Since the timeline-mask summary $\bP$ is completely determined by $\bA_1,\dots,\bA_n$ and $\bK$, we refer to the rule as $\textsc{Assign}_{\sigma}(x_i, t_i)$ for convenience. \\

\textbf{Input}: The point $x_i$ to be assigned and the activation time $t_i$ for the timeline $\bA_i$.  

\textbf{Inductive Hypothesis}: We assume the algorithm has access to a partial cut $\bz = z(\sigma, \bP, t_i) \in \{0,1\}^{[n] \times \{0,1\}}$ which assigns every point $x_j$ which participates in a tuple of $\bP$ and whose activation time $t_j$ is smaller than $t_i$.

\textbf{Output}: The assignment, either $(1, 0)$ or  $(0,1)$ corresponding to the 0-side or 1-side of the cut, for the $i$-th point $x_i$. 

\begin{itemize}
\item \textbf{Base Case}: If $t_i \leq t_0$, the point $x_i$ appears in $\bP$ (and so does all of $\bS^{t_0}$). Hence, the index $\pi(x_i) \in [m]$ can be determined. Output $(1, 0)$ if $\sigma_{\pi(x_i)} = 0$ and $(0,1)$ if $\sigma_{\pi(x_i)}= 1$. 
\item \textbf{Incorporating Estimates}: Otherwise, $t_i > t_0$ and we proceed as follows:
\begin{enumerate}
\item As per Definition~\ref{def:corr-timeline-mask} for every $(x_j, \ell; p_{\ell}, t_j) \in \bP$, we let $\brho_{j}^{\ell} = \min\{ p_{\ell}, 1/\ell \}$ if $\ell = t_j$ and $p_{\ell}$ if $\ell > t_j$. As per Definition~\ref{def:mask}, we let $\gamma^{\ell} = \min\{ \gamma / \ell, 1\}$. We consider the subset of pairs 
\[ \bP_{t_i} = \left\{ (x_j, \ell) : (x_j, \ell; p_{\ell}, t_j) \in \bP \text{ and } \ell \leq t_j < t_{i}\right\} \]
which have been assigned by the inductive hypothesis (by recursively executing $\textsc{Assign}_{\sigma, \bP}(x_j, t_j)$). 
\item We compute the estimated contributions, from the $0$-side and $1$-side of the cut, at the time $t_i$:
\[ \bC_{0} = \sum_{(x_j, \ell) \in \bP_{t_i}} \dfrac{d(x_i, x_j)}{\brho_{j}^{\ell} \cdot \gamma^{\ell}} \cdot \bz_{j,0} \qquad \text{and}\qquad \bC_{1} = \sum_{(x_j,\ell) \in \bP_{t_i}} \dfrac{d(x_i, x_j)}{\brho_{j}^{\ell} \cdot \gamma^{\ell}} \cdot \bz_{j,1}. \]
Output $(0, 1)$ if $\bC_{0} > \bC_{1}$ (to assign $x_i$ to the $1$-side of the cut), and $(1, 0)$ otherwise (to assign $x_i$ to the $0$-side of the cut). 
\end{enumerate}
\end{itemize}
\end{framed}
\caption{The Assignment Procedure $\textsc{Assign}_{\sigma}(\cdot,\cdot)$ for Dynamic Streaming Algorithms} \label{fig:assign-dynamic}
\end{figure}

\begin{theorem}[Main Structural for Dynamic Streams]\label{thm:main-structural-dynamic}
For any set $X = \{ x_1, \dots, x_n \} \subset [\Delta]^d$ and $p \geq 1$, let $\sfD = \poly(d\log(\Delta)/\eps)$ from Lemma \ref{lem:geo-samp} and \ref{lem:geo-sketch}, then let
\[ t_e \geq \frac{n \cdot \sfD}{\eps}  \qquad \gamma \geq \dfrac{\left( \sfD \ln(t_e) + 1 \right)^2}{\eps^2} \qquad t_0 \geq \max\left\{ \frac{\sqrt{\gamma} \cdot \sfD}{\eps}, \frac{1}{\eps} \right\}, \qquad\text{and}\qquad \epspr \leq \frac{\eps}{\ln(t_e)} \]
Consider the following setup:
\begin{itemize}
\item Draw a correlated timeline-mask $\bA_1,\dots, \bA_n$ and $\bK$ from Definition~\ref{def:corr-timeline-mask}. 
\item For any seed $\sigma$ of $\bA_1,\dots, \bA_n$, let $\bz(\sigma)$ denote the following cut:
\begin{align*}
\bz(\sigma)_i = \left\{\begin{array}{cc} \textsc{Assign}_{\sigma}(x_i, \bt_i) & \bt_i \text{ is the activation time for }\bA_i \\
(1, 0) & \bA_{i,t} = 0 \text{ for all $t \in [t_e]$} \end{array} \right. ,
\end{align*}
where we use the algorithm $\textsc{Assign}_{\sigma}(x_i, \bt_i)$ from Figure~\ref{fig:assign-dynamic}.
\end{itemize}
Then,
\begin{align*}
\Ex_{\bA, \bK}\left[ \min_{\sigma} f(\bz(\sigma)) \right] \leq \min_{\substack{z \in \{0,1\}^{n\times\{0,1\}}\\ z_i \neq (0,0)}} f(z) + O(\eps) \sum_{i=1}^n \sum_{j=1}^n d_{\ell_p}(x_i, x_j).
\end{align*}
with $f$ from (\ref{eq:internal-f-def}) with $\ell_p$-metric.
\end{theorem}

\subsection{Dynamic Streaming Algorithm}\label{sec:dynamic-stream-alg}

Given Theorem~\ref{thm:main-structural-dynamic}, we will show how to design a dynamic streaming algorithm which can maintain a low-memory space while receiving an input set $X \subset [\Delta]^d$ as a dynamic stream; execute a preprocessing step which can build the correlated timeline-mask summary $\bP$ and provide query access to activation times, and select an approximately minimum seed $\sigma$; and finally, provide query access to the assignment specified by $\textsc{Assign}_{\sigma}(x_i, t_i)$ from Figure~\ref{fig:assign-dynamic}. 

\begin{theorem}[Euclidean Max-Cut in Dynamic Streams]\label{thm:dynamic}
For any $p \geq 1$, $d, \Delta \in \N$ and $\eps \in (0, 1)$, there exists a randomized streaming algorithm with the following guarantees:
\begin{itemize}
\item \emph{\textbf{Maintenance}}: The streaming algorithm receives an input set $X \subset [\Delta]^d$ as a dynamic stream, and uses a random oracle $\boldr$ and $\poly(d \log \Delta / \eps)$ words of space. 
\item \emph{\textbf{Query}}: For a point $x \in [\Delta]^d$, the sub-routine $\textsc{Assign}(x)$ outputs $(1, 0)$ or $(0,1)$ corresponding to the $0$-side or the $1$-side of the cut. 
\end{itemize}
If, after receiving the set $X = \{ x_1, \dots, x_n \} \subset [\Delta]^d$, we let $\bz \in \{0, 1\}^{n \times \{0,1\}}$ be given by
\[ \bz_i \leftarrow \textsc{Assign}(x_i),\]
then
\[ \Ex\left[ \boldf(z) \right] \leq \min_{\substack{z \in \{0,1\}^{n\times\{0,1\}} \\ z_i \neq (0,0)}} f(z) + \eps \sum_{i=1}^n \sum_{j=1}^n d_{\ell_p}(x_i, x_j),\]
using $f$ from (\ref{eq:internal-f-def}) with $\ell_p$-metric.
\end{theorem}

\paragraph{Maintenance and Updates.} The memory contents of the dynamic streaming algorithm, $\textsc{DS-E-Max-Cut}$, are simple to describe, and the non-trivial part of Theorem~\ref{thm:dynamic} is showing how to unpack the memory contents to apply Theorem~\ref{thm:main-structural-dynamic}. In particular, the algorithm is initialized with parameters $\Delta, d \in \N$, $p \geq 1$ and an accuracy parameter $\eps \in (0, 1)$. The algorithm will maintain:
\begin{itemize}
\item A single instance of the streaming algorithm from Lemma~\ref{lem:geo-sketch}, initialized with $\delta = \eps / \Delta^d$. 
\item An instance of the streaming algorithm from Lemma~\ref{lem:geo-samp}, initialized with parameters derived from the initialization, $\epspr \leq \frac{\eps}{\ln(t_e)}$, $\delta = \eps $, and $s = \poly(d \log \Delta / \eps)$ which is set large enough so as to be more than the number of entries in $\bK$ will be $1$.
\item A pointer to a random oracle $\boldr \colon \{0,1\}^* \to \{0,1\}$.
\end{itemize}
Whenever a point is inserted or removed from the stream, the corresponding update is performed on the streaming algorithm from Lemma~\ref{lem:geo-sketch}, as well as the instance of Lemma~\ref{lem:geo-samp}. The space complexity is trivially bounded by $\poly(d\log \Delta / \eps)$, since $s = \poly(d \log \Delta /\eps)$ instances of Lemma~\ref{lem:geo-sketch} and Lemma~\ref{lem:geo-samp}, each of which uses $\poly(d\log \Delta /\eps)$ words.

\paragraph{Preprocessing: Generating the Summary and Accessing Activation Times.} We now give a procedure which will unpack the memory contents of $\textsc{DS-E-Max-Cut}$ in order to ``implicitly'' sample a correlated timeline-mask $\bA_1,\dots, \bA_n$ and $\bK$ (see Definition~\ref{def:corr-timeline-mask}). Here, ``implicit'' refers to the fact that the algorithm does not need to output the $(n+1) \cdot t_e$ bits which specify $\bA_1,\dots, \bA_n$ and $\bK$. We will show that the algorithm may produce the correlated timeline-mask summary $\bP = \textsc{Summ}(\bA_1,\dots, \bA_n, \bK)$ (see Definition~\ref{def:corr-timeline-mask-sum}) and provide access to the activation time $\bt_i$ for the timeline $\bA_i$ when given the point $x_i$. This is all that the algorithm will need in order to execute $\textsc{Assign}_{\sigma, \bP}(x_i, t_i)$.

We will utilize the random oracle $\boldr \colon \{0,1\}^* \to \{0,1\}$ in a way similarly to the insertion-only algorithm (see Remark~\ref{rem:generating-randomness}). We store a pointer to the random oracle $\boldr$, as well as the parameters $t_e,t_0$ and $\gamma$. Using the random oracle $\boldr$, we reference random variables
\[ \bK_t^{(r)} \sim [0,1] \qquad \text{and} \qquad \bA_{x, t}^{(r)} \sim [0, 1] \qquad \text{for all $t \in [t_e]$ and $x \in [\Delta]^d$.}\]
Importantly, the algorithm can use the random oracle $\boldr$ to access $\bK_t^{(r)}$ or $\bA_{x, t}^{(r)}$ when given $x$ and $t$ without having to explicitly store them. As in Remark~\ref{rem:generating-randomness}, it will be clear from the way we use $\bK_t^{(r)}$ and $\bA_{x,t}^{(r)}$ that they may sampled and stored to finite precision. The algorithms are specified in Figure~\ref{fig:preprocess-tm-dynamic}.

\begin{figure}
\begin{framed}
\textbf{Preprocessing Subroutines for} \textsc{DS-E-Max-Cut}. We provide access to the memory contents of \textsc{DS-E-Max-Cut} after processing $X = \{ x_1,\dots, x_n \} \subset [\Delta]^d$.
\\

\textsc{Build-Summ}$()$. Outputs the correlated timeline-mask summary $\bP = \textsc{Summ}(\bA_1,\dots, \bA_n, \bK)$.
\begin{enumerate}
\item Initialize a set of tuples $\bP$, initially empty, as well as a counter $c=1$.
\item\label{ln:iter-t} Iterate through $t=1,\dots, t_e$:
\begin{enumerate}
\item\label{ln:iter-t-1} Set $\gamma^t$ to $1$ if $t \leq t_0$ and $\min\{ \gamma / t, 1\}$ if $t > t_0$. Let $\bK_t = \ind\{ \bK_t^{(r)} \leq \gamma^t \}$. If $\bK_t = 0$, skip the following lines and continue the loop of Line~\ref{ln:iter-t} for next $t$. If $\bK_t = 1$, continue.
\item\label{ln:iter-t-2} Let $(\by^*_t, \bp^*_t)$ be the $c$-th sketch from Lemma~\ref{lem:geo-samp}. Increment $c$; if $c$ reaches $s/2$, output ``fail.''
\item\label{ln:iter-t-3} Execute \textsc{Activation-Time}$(\by_t^*, t-1, \bP)$ (specified below).
\begin{enumerate}
\item If it outputs $\ell \leq t-1$, insert $(\by_t^*, t; \bp_t^*, \ell) \in \bP$ if $\by_t^* \neq \bot$. 
\item\label{en:det-if-kept} If it outputs ``activated after $t-1$'', let $\bb_t \sim \Ber(\min\{1/(t \bp_t^*), 1 \})$ and insert $(\by_t^*, t; \bp_t^*, t)$ if $\bb_t= 1$ and $\by_t^* \neq \bot$. 
\end{enumerate}
\end{enumerate}
\item Output $\bP$.
\end{enumerate}

\textsc{Activation-Time}$(x, t, \bP)$. Assume (for inductive hypothesis) that the summary $\bP$ contains all tuples $(x, \ell; p_{\ell}, \tilde{t})$ in $\textsc{Summ}(\bA_1,\dots, \bA_n, \bK)$ with $\ell \leq t$. On input $x_i \in X, t$ and $\bP$, it outputs the activation time $\bt_i$ of $\bA_i$ if it is at most $t$, or ``activated after $t$'' otherwise. 
\begin{enumerate}
\item If $x$ appears in a tuple $(x, \ell; p_{\ell}, \tilde{t}) \in \bP$, output $\tilde{t}$.
\item\label{ln:activated-second-check} If not, execute the sketch from Lemma~\ref{lem:geo-sketch} with input $x$, and let $\bw = \bw(x) / (2\sfD)$ be the output. Using the random oracle $\boldr$, find the smallest $\ell \leq t$ which satisfies $\bK_{\ell} = 0$ (as in Line~\ref{ln:iter-t-1} and $\smash{\bA^{(r)}_{x, \ell} \leq \min \{ \bw, \frac{1}{\ell} \}}$ if one exists, and output $\ell$. If no such $\ell$, output ``activated after $t$.''
\end{enumerate}
\end{framed}
\caption{Preprocessing sub-routines \textsc{Build-Summ} and \textsc{Activation-Time}.}\label{fig:preprocess-tm-dynamic}
\end{figure}

\begin{lemma}\label{lem:build-sum-act-time}
An execution of $\textsc{Build-Summ}()$ which does not output ``fail'' specifies a draw to a correlated timeline-mask $\bA_1,\dots ,\bA_n$ and $\bK$ from Definition~\ref{def:corr-timeline-mask}. In addition,
\begin{itemize}
\item \textsc{Build-Summ}$()$ outputs $\bP=\textsc{Summ}(\bA_1,\dots, \bA_n, \bK)$ from Definition~\ref{def:corr-timeline-mask-sum}.
\item \textsc{Activation-Time}$(x_i, t, \bP)$ outputs the activation time $\bt_i$ of $\bA_i$ if $\bt_i \leq t$, or ``activated after $t$'' otherwise. 
\end{itemize}
For $s = O\left( (t_0 + \gamma \log t_e)/\eps\right)$, the probability that \textsc{Build-Summ}$()$ outputs ``fail'' is at most $\eps$. 
\end{lemma}

\begin{proof}
\ignore{We describe the draw of the correlated timeline-mask $\bA_1,\dots, \bA_n$ and $\bK$ which is specified by the memory contents of $\textsc{DS-E-Max-Cut}$ and $\textsc{Build-Summ}()$ which does not output ``fail.'' }
First, the state of a streaming algorithm from Lemma~\ref{lem:geo-sketch} is used as in Line~\ref{en:weights} in Definition~\ref{def:corr-timeline-mask}, where we denote $\bw_i = \bw(x_i)$. Then, $\bK$ is distributed as $\calK(t_0,\gamma)$, since $\bK_t = \ind\{ \bK_{t}^{(r)} \leq \gamma^t\}$ so $\bK_t \sim \Ber(\gamma^t)$, as in Line~\ref{en:mask-corr} in Definition~\ref{def:corr-timeline-mask}.
We specify $\bA_1,\dots, \bA_n$ by iterating through $t = 1,\dots, t_e$ according to Line~\ref{en:iter-t-cor-def} in Definition~\ref{def:corr-timeline-mask}. 
Assume for the sake of induction that all $\ell \leq t-1$ have $\bA_{1,\ell}, \dots, \bA_{n,\ell}$ generated according to Definition~\ref{def:corr-timeline-mask} (where the base case of $t=0$ is trivial). In step $t$,
\begin{itemize}
\item If $\bK_t = 0$, we let $\bA_{i,t} = \ind\{ \bA_{x_i, t}^{(r)} \leq \min\{ \bw_i, 1/t\}\}$ if $\bA_{i,\ell} = 0$ for all $\ell < t$ and $\bA_{i,t} = \ind\{ \bA_{x_i, t}^{(r)} \leq \bw_i\}$ if there exists $\ell < t$ with $\bA_{i,\ell} = 1$. By inductive hypothesis, $\bA_{i,t}$ is distributed as a Bernoulli random variable with desired parameter, as in Line~\ref{en:k-0} of Definition~\ref{def:corr-timeline-mask}.
\item If $\bK_t = 1$, we use the sketch from Lemma~\ref{lem:geo-samp} in \textsc{DS-E-Max-Cut} to extract the next sample $(\by_t^*, \bp_t^*)$. (Note that the failure condition occurs whenever the number of non-zero entries in $\bK$ exceeds $s$). Whenever $\by_t^* \neq x_i$, we let $\bA_{i, t} = 0$ as desired in Line~\ref{en:k-1} of Definition~\ref{def:corr-timeline-mask}. On the other hand, if $\by_t^* = x_i$, we let $\bA_{i, t} = 1$ iff $\bA_{i,\ell} = 1$ for some $\ell < t$; if $\bA_{i,\ell} = 0$ for all $\ell < t$, $\bA_{i,t} = \bb_t$ for $\bb_t \sim \Ber(\min\{1/(t \bp_t^*), 1 \}))$ sampled in Line~\ref{en:det-if-kept} of \textsc{Build-Summ}$()$, as desired in Line~\ref{en:k-1} of Definition~\ref{def:corr-timeline-mask}.
\end{itemize}
This completes the proof that the sampled $\bA_1,\dots, \bA_n$ and $\bK$ is a correlated timeline-mask. To prove that \textsc{Build-Summ}$()$ outputs $\bP = \textsc{Summ}(\bA_1,\dots ,\bA_n, \bK)$, we proceed by induction on $t$, and show that after the first $t$ iterations of Line~\ref{ln:iter-t} of \textsc{Build-Summ}$()$, all tuples with $\bA_{i,\ell} \cdot \bK_{\ell} = 1$ with $\ell \leq t$ are in $\bP$. The base case of $t=0$ is trivial since $\bP$ is empty. For the $t$-th iteration of Line~\ref{ln:iter-t}, Line~\ref{ln:iter-t-1} correctly sets $\bK_t$, and if $\bK_t = 0$, then there are no tuples with $\bA_{i,t} \cdot \bK_t = 1$, so the algorithm may skip the lines. When $\bK_t= 1$, there is at most one tuple with $\bA_{i,t} \cdot \bK_t = 1$, since $\bA_{i,t}$ can only be one for $x_i = \by_{t}^*$. 

As an intermediate claim, note that \textsc{Activation-Time}$(x_i, t-1, \bP)$ correctly outputs the activation time $\bt_i$ as specified in $\bA_i$ if it is at most $t-1$, and outputs ``activated after $t-1$'' otherwise. This is because if $x_i$ participated in a tuple in $\bP$, it's activation time was recorded in the tuple (by induction that $\bP$ is correct up to time $t-1$). If it did not participate in a tuple of $\bP$, then $\bA_{i, \ell} = 0$ for any $\ell \leq t-1$ with $\bK_{\ell} = 1$; then, Line~\ref{ln:activated-second-check} of \textsc{Activation-Time}$(x_i, t-1, \bP)$ checks whether $\bA_{i,\ell} = 1$ for any $\ell \leq t-1$ with $\bK_{\ell} = 0$. 

From the above paragraph, Line~\ref{ln:iter-t-3} in \textsc{Build-Summ}$()$ correctly determines the activation time $\bt_i$ of $x_i = \by_{t}^*$ if it occurred before $t$, or outputs ``not activated by $t-1$'' otherwise. Since $\by_t^* = x_i$, if $\textsc{Activation-Time}(\by_{t}^*, t-1, \bP)$ outputs $\ell \leq t-1$, then $\bA_{i, t} = 1$ and we insert the tuple in $\bP$ with the prior activation time $\ell$. If not, then all $\bA_{i,\ell} = 0$ for $\ell < t$ and $\bA_{i,t} = \bb_t$ for $\bb_t \sim \Ber(\min\{1/(t \bp_t^*), 1 \})$; if $\bb_t = 1$, then we insert the tuple into $\bP$. This concludes the proof that $\bP = \textsc{Summ}(\bA_1,\dots, \bA_n, \bK)$, and that \textsc{Activation-Time}$(x_i ,t_e, \bP)$ provides access to the activation time.

Finally, we upper bound the probability that $\textsc{Build-Summ}()$ outputs ``fail.'' Note that this occurs if and only if the algorithm \textsc{Build-Summ}$()$ must use more than $s$ samples of Lemma~\ref{lem:geo-samp}. Let $\| \bK \|_1 = \sum_{t=1}^{t_e} \bK_t$ denote the total number of samples $(\by_t^*, \bp_t^*)$ that the algorithm needs. The expectation of $\|\bK\|_1$ at most $O(t_0 + \gamma \log t_e)$. By Markov's inequality, we obtain the bound on the failure probability of $\textsc{Build-Summ}()$ for $s = O((t_0 + \gamma \log t_e) / \eps)$.  
\end{proof}

\paragraph{Preprocessing: Selecting the Seed.} From Lemma~\ref{lem:build-sum-act-time}, the dynamic streaming algorithm may sample a correlated timeline-mask $\bA_1,\dots, \bA_n$ and $\bK$ from Definition~\ref{def:corr-timeline-mask}, build the summary $\bP = \textsc{Summ}(\bA_1,\dots, \bA_n, \bK)$, and access the activation times $\bt_i$ for each point $x_i \in X$. The algorithm can extract from $\bP$ the set $\bS^{t_0}$ of points in $X$ whose activation time occurred up to time $t_0$, and determine $m = |\bS^{t_0}|$. For any seed $\sigma \in \{0,1\}^m$ for $\bA_1,\dots, \bA_n$ (Definition~\ref{def:seed}), and any $x_i \in X$, the algorithm may execute $\textsc{Assign}_{\sigma,\bP}(x_i,\bt_i)$ in Figure~\ref{fig:assign-dynamic} for activation time $\bt_i$ of $\bA_i$ as follows:
\begin{enumerate}
\item First, the algorithm executes $\textsc{Activation-Time}(x_i, t_e, \bP)$, and obtains as output the activation time $\bt_i$ of $\bA_i$, or ``not activated by $t_e$.''
\item Then, it either outputs $\textsc{Assign}_{\sigma,\bP}(x_i, \bt_i)$, or $(1, 0)$ if ``not activated by $t_e$.''
\end{enumerate} 
The seeds define a family of $2^m$ cuts $z \in \{0,1\}^{n \times \{0,1\}}$ obtained by executing the above procedure for $z_i$. For any fixed cut $z$, and one may approximately estimate $f(z)$ up to multiplicative error $(1+\eps)$ and additive error $\eps \sum_{i=1}^n \sum_{j=1}^n d_{\ell_p}(x_i, x_j)$; as in Lemma~\ref{lem:seed-check}, we union bound over all $2^m$ cuts.\footnote{The reason we do not invoke Lemma~\ref{lem:seed-check} directly is because we do not have access to a random set $\bC$ as generated in that lemma, and because we will only have approximate access to the sampling probabilities.} Since \textsc{Build-Summ}$()$ (if it does not fail) uses at most $s/2$ geometric sampling sketches (Lemma~\ref{lem:geo-samp}) in Figure~\ref{fig:preprocess-tm-dynamic}, the algorithm still has $s/2$ samples. It partitions samples into pairs; we denote the $h$-th successful pair when the samples $(\by_{h,1}^*, \bp_{h,1}^*)$ and $(\by_{h,2}^*, \bp_{h,2}^*)$ are both not $\bot$ (which occurs for any individual pair with probability at least $1/\sfD^2$). Let $\bi_{h,1}$ and $\bi_{h,2}$ be so $\by_{h,1}^* = x_{\bi_{h,1}}$ and $\by_{h,2}^* = x_{\bi_{h,2}}$. Letting $\bH$ be the total number of successful pairs, $\bH$ is stochastically dominated by $\Bin(s/4, 1/\sfD^2)$. Conditioned on any $\bH = H$ and it being realized by a fixed set $G \subset [s/4]$ of successful pair of sketches of size $|G| = H$, 
\begin{align*}
\frac{1}{(1+\epspr)^2} \cdot f(z) \leq \Ex\left[ \frac{1}{2 \cdot H} \sum_{h=1}^{H} \dfrac{d_{\ell_p}(x_{\bi_{h,1}}, x_{\bi_{h,2}})}{\bp_{h,1}^* \cdot \bp_{h,2}^*} \left(z_{\bi_{h_1}, 0} z_{\bi_{h_2}, 0} + z_{\bi_{h_1},1} z_{\bi_{h_2}, 1} \right) \right] \leq f(z),
\end{align*}
where the expectation is over the realization of the samples $(\by_{\cdot}^*, \bp_{\cdot}^*)$ from Lemma~\ref{lem:geo-samp}, and the approximate guarantee on the expectation is because $\bp_{h,1}^*$ and $\bp_{h,2}^*$ are $(1+\epspr)$-approximations of the probability that $\by_{h,1}^*$ and $\by_{h,2}^*$ are sampled. For any fixed $H$, above estimator is an average of independent random variables, with the variance of the $h$-th summand at most
\begin{align*}
\Varx\left[\dfrac{d_{\ell_p}(x_{\bi_{h,1}}, x_{\bi_{h_2}})}{\bp_{h,1}^* \cdot \bp_{h,2}^*} \left(z_{\bi_{h_1}, 0} z_{\bi_{h_2}, 0} + z_{\bi_{h_1},1} z_{\bi_{h_2}, 1} \right) \right] \leq \sfD^2 \sum_{i=1}^n \sum_{j=1}^n \dfrac{d_{\ell_p}(x_i,x_j)^2}{w_i \cdot w_j},
\end{align*}
where $w_i$ is proportional to the sum of all distances in $X$ to $x_i$. The same argument as in the proof of Lemma~\ref{lem:weight-to-compatible} implies that the above variance is at most
\[ 16 \left(\sum_{i=1}^n \sum_{j=1}^n d_{\ell_p}(x_i, x_j) \right)^2,\]
and thus, whenever $\bH$ is larger than $O(m / \eps^2)$, we may obtain the desired approximation to $f(z)$ for all $z \in \{0,1\}^m$. By selecting the minimizing $\bsigma^*$, we can select the best seed in $\{0,1\}^{m}$. 

\paragraph{Proof of Theorem~\ref{thm:dynamic}: Putting Everything Together.} For the correctness of algorithm, it suffices to check that we can apply Theorem~\ref{thm:main-structural-dynamic}. First, our algorithm initializes the sketch of Lemma~\ref{lem:geo-sketch} and the $s$ samples of Lemma~\ref{lem:geo-samp}. Both Lemma~\ref{lem:build-sum-act-time} and Lemma~\ref{lem:geo-samp}, fail with probability at most $\eps$ (either by outputting ``fail'', or in a manner undetectable to the algorithm); either way, we consider these catastrophic failure (and we upper bound its contribution later). If \textsc{Build-Summ}$()$ does not output ``fail,'' then it specifies a draw to a correlated timeline-mask $\bA_1,\dots, \bA_n$ and $\bK$ from Definition~\ref{def:corr-timeline-mask} and outputs $\bP = \textsc{Summ}(\bA_1,\dots, \bA_n, \bK)$; furthermore, \textsc{Activation-Time}$(x_i, t, \bP)$ provides query access to the activation time $\bt_i$ of $\bA_i$. We apply the algorithm from above to select the optimal seed $\bsigma^*$, and we let $\textsc{Assign}(x)$ output $(1,0)$ or $(0,1)$ according to the following process:
\begin{itemize}
\item Use \textsc{Activation-Time}$(x_i, t_e, \bP)$ to obtain the activation time $\bt_i$ for $\bA_i$ if its activated by time $t_e$, or output $(1, 0)$ as a default if it outputs ``not activated by $t_e$'' if it does not get activated.
\item If the algorithm produced an activation time $\bt_i$, then output $\textsc{Assign}_{\bsigma^*, \bP}(x_i, \bt_i)$
\end{itemize}
Since $\bsigma^*$ approximately minimizes $f(\bz(\sigma))$, we have that
\begin{align*}
\Ex\left[ f(\bz(\bsigma^*))  \right] &\leq (1+\eps) \Ex\left[ \min_{\sigma \in \{0,1\}^m} f(\bz(\sigma)) \right] + O(\eps) \sum_{i=1}^n \sum_{j=1}^n d_{\ell_p}(x_i, x_j) \\
						&\leq (1+\eps) \min_{\substack{z \in \{0,1\}^{n\times\{0,1\}} \\ z_i \neq (0,0)}} f(z) + O(\eps) \sum_{i=1}^n \sum_{j=1}^n d_{\ell_p}(x_i, x_j), 
\end{align*}
where the first line uses the fact that $\bsigma^*$ is chosen to minimize $f(\bz(\sigma))$, and the second inequality applies Theorem~\ref{thm:main-structural-dynamic}.

%% file: dynamic-analysis.tex
\section{Proof of Theorem~\ref{thm:main-structural-dynamic}}\label{sec:dynamic-greedy-process}

To prove Theorem~\ref{thm:main-structural-dynamic}, we follow the same proof structure as of Theorem \ref{thm:main-structural} incorporating the changes to the timeline and access to estimates of the sampling probability. As before, we first describe a randomized algorithm, which receives as input access to the metric space, certain parameters which describe the random process, and description of a minimizing cut $z^*$ of $f$. However, there are two important differences which add some complication:
\begin{itemize}
    \item We will sample ``activated'' and ``kept'' points using Lemma~\ref{lem:geo-samp}, which will produce $s$ samples $\by^*_j$ ($j\in [s]$) from a probability distribution which is $\poly(d\log \Delta /\eps)$-metric compatible; furthermore, it outputs an approximation $\bp^*_j$ to the probability the point was sampled. This probability can be used for importance sampling, but it is only an estimate of the true probability. We must ensure that the errors incurred from this estimate do not accumulate. 
    \item Perhaps more apparent, the procedure no longer behaves independently for each point. There is a single mask $\bK$, and whenever $\bK_t = 1$, at most a single point will have $\bA_{x, t} = 1$ (i.e., that which is sampled from Lemma~\ref{lem:geo-samp}). 
\end{itemize}

\paragraph{Process.} We call the modified random process $\textsc{CorrelatedGreedyProcess}$ with the following input/output behavior (designed to produce a correlated Timeline-Mask from Definition~\ref{def:corr-timeline-mask}):
\begin{itemize}
\item \textbf{Input}: access to the metric space $(X = \{ x_1,\dots, x_n \}, d)$. We consider a fixed successful instance of Lemma~\ref{lem:geo-sketch}, and we let $w \in (0, 1/2]^n$ denote the vector where $w_i = \bw(x_i) / (2\sfD)$. The times $t_0, t_e \in \N$, a parameter $\gamma \geq 1$, and a matrix $z^* \in \{0,1\}^{n \times \{0,1\}}$ which will be a minimizer of $f$ encoding cuts:
\[ z^* = \argmin\left\{ f(z) : z \in \{0,1\}^{n \times \{0,1\}} \text{ where every $i \in [n]$ has $z_i \in \{ (1, 0) , (0, 1) \}$} \right\}. \]
\item \textbf{Output}: the process outputs a random (partial) cut $\bz \in \{0,1\}^{n\times \{0,1\}}$ where $\bz_i \in \{ (1, 0), (0, 1), (0, 0)\}$ for all $i\in[n]$. In addition, a number $\bm \in \N$, and a string $\bsigma^* \in \{0,1\}^{\bm}$.
\end{itemize}
The algorithm is described below.
\paragraph{Description.} \textsc{CorrelatedGreedyProcess}$((X = \{ x_1,\dots, x_n\}, d), w, \gamma, t_0, t_e, z^*)$ proceeds in rounds $t=1, 2, \dots, t_e$, and will 
generate the following sequences of random variables:
\begin{align*}
    &\bz^t \text{: a matrix $\{0,1\}^{n\times \{0,1\}}$ which encodes a partial cut at time $t$.} \\
    &\bw^t \text{: a vector of $n$ weights at time $t$.} \\
    &\bP^{t} \text{: tuples $(x_j, \ell; p_{\ell}, t_j)$ for indices $\ell \leq t$, values $p_{\ell}$, and $t_{j} \leq \ell$ where $\bA_{j,\ell}^{t} \cdot \bK_{\ell}^t = 1$.} \\
    &\bS^t \text{: a set of assigned points at time $t$.}\\
    &\bg^t \text{: a binary $n\times \{0,1\}$ matrix which encodes the  greedy assignment at time $t$.}\\
    &\bA^t \text{: an $n \times t$ binary matrix.} 
\end{align*}
as well as a random binary vector $\bK \in \{0,1\}^{t_e}$ which determines how the process behaves at each time step. At the start, the process generates $\bK$, by 
iterating through each time $1 \leq t \leq t_e$  and setting $\gamma^t = 1$ when $t\leq t_0$ and $\gamma^t = \min(1,\gamma/t)$ otherwise. Then, sample $\bK_t \sim \Ber(\gamma^t)$. 
Let $s$ be the number of nonzero-entries in $\bK$, e.g. $s=\|\bK\|_1$. The algorithm then uses Lemma $\ref{lem:geo-samp}$ to drawn $s$ many samples from $\calG_s(X)$. We re-index these samples, as described by Remark \ref{rem:re-indexing}.

We let $\bz^0$ be the all zero's matrix, $\bA^0$ be the empty $n \times 0$ matrices, $\bP^0 = \emptyset$, and $\bS^0 = \emptyset$. For time $1 \leq t \leq t_e$, we instantiate $\bA^t$:
\begin{enumerate}
\item\label{en:dynamic-greedy-ln1} For each $i \in [n]$, if $i \notin \bS^{t-1}$ we set $\bw_i^t$ to $\min\{ w_i, 1/t\}$. Otherwise we set  $\bw_i^t$ to $w_i$. For each $(x_j, \ell; p_{\ell}, t_j) \in \bP^{t-1}$, we let $\brho_{j}^{\ell} = \min\{ p_{\ell}, 1/\ell \}$ if $\ell = t_j$ and $p_{\ell}$ if $\ell > t_j$.

\item\label{en:dynamic-greedy-ln2} For each $i \in [n]$, if $t \leq t_0$ let $ \bg^t_i = z^*_i$,
else for $t > t_0$, we let 
\begin{align*}
\tilde{\bc}_{i, 0}^{t} &= \sum_{j=1}^n d(x_i, x_j) \left( \frac{1}{t-1} \sum_{\ell=1}^{t-1} \dfrac{\bA_{j,\ell}^{t-1} \cdot \bK_{\ell}^{t-1}}{\brho_{j}^{\ell} \cdot \gamma^\ell} \right) \bz_{j, 0}^{t-1} \\
\tilde{\bc}_{i,1}^t &= \sum_{j=1}^n d(x_i, x_j) \left( \frac{1}{t-1} \sum_{\ell=1}^{t-1} \dfrac{\bA_{j,\ell}^{t-1} \cdot \bK_{\ell}^{t-1}}{\brho_{j}^{\ell} \cdot \gamma^\ell} \right) \bz_{j, 1}^{t-1};
\end{align*}
even though $\brho_j^{\ell}$ is only defined for $(x_j, \ell; p_{\ell}, t_j) \in \bP^{t-1}$, $\bA_{j,\ell}^{t-1} \cdot \bK_{\ell}^{t-1}$ is non-zero only for these. We let $\bg^t_i$ be set to $(1, 0)$ if $\tilde{\bc}_{i, 0}^{t} < \tilde{\bc}_{i, 1}^t$ or $(0, 1)$ otherwise. 
\item\label{en:dynamic-greedy-ln3} If $\bK_t = 1$, select the corresponding sample $(\by^*_t, \bp^*_t)$, and let $i^* \in [n]$ such that $\by^*_t = x_{i^*}$. 
\begin{itemize}
    \item If $i^* \notin \bS^{t-1}$, sample $\bb_t \sim \Ber(\min\{ 1/(t\bp_t^*), 1\})$ and set $\bA_{i^*, t}^t = \bb_t$. Update $\bP^t \leftarrow \bP^{t-1}$ and add $(x_{i^*}, t; \bp_t^*, t)$ if $\bb_t=1$. For $j \neq i^*$, $\bA_{j, t}^{t} = 0$, and $\bA_{j,\ell}^{t} = \bA_{j,\ell}^{t-1}$ for $\ell < t$.
    \item If $i^* \in \bS^{t-1}$, find the activation time $t_{i^*}$. Update $\bP^{t} \leftarrow \bP^{t-1} \cup \{ (x_{i^*}, t; \bp_t^*, t_{i^*}) \}$. Set $\bA_{j, t}^{t} = \ind\{ j = i^* \}$ and $\bA_{j,\ell}^{t} = \bA_{j,\ell}^{t-1}$ for $j < t$.
\end{itemize}
Finally, update $\bS^t$ to activated points $\bS^{t-1} \cup \{ i \in [n] : \bA_{i,t}^t = 1\}$ by time $t$.
\item\label{en:dynamic-greedy-ln4}  For each $i \in [n]$, define \begin{align*}
    \bz^t_i = \begin{cases}
        \bz^{t-1}_i &  i \in \bS^{t-1} \\
        \bg_i^t &  i \in \bS^{t} \land i \notin \bS^{t-1} \\
        (0,0) &  \text{else} \\
    \end{cases}
\end{align*}
\item Let $\bm = |\bS^{t_0}|$, and consider a natural ordering $\pi$ where $\pi_j$ is the $j$-th element of $\bS^{t_0}$. Let
\begin{align*}
    \bsigma_{i}^* = \by^*_{\pi_i}
\end{align*}
\end{enumerate}
When $t > t_e$, the process completes and we output $(\bz, \bm, \bsigma^*)$, and we will refer to $\bA^{t_e}$ and $\bP^{t_e}$ as $\bA$ and $\bP$ respectively.

One complication is that geometric samples $(\by_t^*, \bp_t^*)$ only provide approximate access to the true probabilities from which they are sampled, and only for the point which is sampled. We define the following variables which help us keep track of the quantities that the algorithm knows or does not know.
\begin{definition}\label{def:probs-seen-unseen}
Consider an execution of $\textsc{CorrelatedGreedyProcess}$ with the parameters $((\{x_1,\dots, x_n\},d), w, \gamma, t_0, t_e, z^*)$. Then, for any time $t$ where $\bK_t = 1$, consider the set of samples $\{(\by_t^*, \bp_t^*)\}$ from Lemma~\ref{lem:geo-samp}. With probability $1-\delta$ the initialization succeeds and we define the following:
    \begin{itemize}
        \item For $i \in [n]$, $p_i^t$ is the probability over $(\by_t^*, \bp_t^*)$ that $\by_t^* = x_i$, which by Lemma~\ref{lem:geo-samp} (see Remark~\ref{rem:additive-err}) satisfies
        \[ p_i^t = \Prx\left[ \by_t^* = x_i \right] \geq \frac{1}{\sfD} \cdot \dfrac{\sum_{j=1}^n d(x_i, x_j)}{\sum_{j=1}^n \sum_{j'=1}^n d(x_{j}, x_{j'})}. \]
        The algorithm does not observe $p_i^t$, but for $\by_t^* = x_{i^*}$, $p_{i^*}^t \leq \bp_t^* \leq (1+\epspr) p_{i^*}^t$.
        \item For $i^* \in [n]$ with $\by_t^* = x_{i^*}$, $\brho_{i^*}^t$ is set to $\min\{ \bp_t^*, 1/t\}$ if $i^* \notin \bS^{t-1}$ (i.e., first activated at $t$), and $\bp_t^*$ otherwise. The algorithm does know $\brho_{i^*}^t$.
        \item For $i^* \in [n]$ with $\by_t^* = x_{i^*}$, $\tilde{\bp}_{i^*}^t$ is set to $p_{i^*}^t \cdot \min\{ 1/(t \bp_t^*), 1\}$ if $i^* \notin \bS^{t-1}$, and $p_{i^*}^t$ otherwise.
    \end{itemize}
    For an execution of the first $t-1$ steps and $\bK_t = 1$, the probability that $\bA_{i,t} = 1$ is $\tilde{\bp}_i^t$, which is at most $1/t$ since $\bp_t^* \geq p_i^t$ whenever $\by_t^* = x_i$.
\end{definition}

\paragraph{Connecting \textsc{CorrelatedGreedyProcess} and $\textsc{Assign}_{\bsigma^*,\bP}(\cdot,\cdot)$.} Similarly to $\textsc{GreedyProcess}$, we have designed \textsc{CorrelatedGreedyProcess} so that $\bS^{t_0}$ contains all activated points by $t_0$ (which are all in $\bP$, and $\bsigma^*$ specifies the optimal cut assignment for $\bS^{t_0}$ (analogous to Observation~\ref{obs:1} and~\ref{obs:2}). The only minor difference is Observation~\ref{obs:3}, which incorporates the geometric samples.
\ignore{\subsection{From \textsc{CorrelatedGreedyProcess} to \textsc{Assign}}

We collect the following observations, which allows us to connect the analysis of \textsc{CorrelatedGreedyProcess} to the assignment rule from Theorem~\ref{thm:main-structural-dynamic}. In particular, we will show that the distribution over matrices $(\bA, \bK)$ at the end of the above greedy process is exactly that which is defined for the statement of Theorem~\ref{thm:main-structural-dynamic}. This allows us to prove Theorem~\ref{thm:main-structural-dynamic} by showing that the partial cut produced by the greedy process is approximately optimal. 

\begin{observation}\label{obs:dynamic-1}
Since $\gamma^{\ell} = 1$ for all $\ell \leq t_0$, $\bK_{\ell} = 1$ for all $\ell \leq t_0$ and $i \in [n]$. Therefore,
\[ \bS^{t_0} = \left\{ j\in [n] : \exists \ell \leq t_0 \text{ s.t }\bA_{j,\ell} = 1 \right\} = \left\{ j \in [n] : \exists \ell \leq t_0 \text{ s.t } \bA_{j,\ell} \cdot \bK_{\ell} = 1\right\}.\]
Therefore, by having access to pairs $(j, \ell)$ with $\bA_{j,\ell} \cdot \bK_{\ell} = 1$ (which is given through $\bP$) and knowledge of $t_0$, one can construct the set $\bS^{t_0}$, identify the parameter $\bm$, and determine the ordering $\pi$ so that $\pi_j$ denotes the $j$-th element of $\bS^{t_0}$.  
\end{observation}

\begin{observation}\label{obs:dynamic-2}
Since the row $\bg_i^t = z_i^*$ for every $t \leq t_0$, letting $\pi$ denote the natural ordering where $\pi_j$ is the $j$-th element of $\bS^{t_0}$, for every $j \in [\bm]$,
\[ \bz_{\pi(j)} = z_{\pi(j)}^* = \left\{\begin{array}{cc} (1, 0) & \bsigma^*_{j} = 0 \\
						 			(0, 1) & \bsigma^*_j = 1 \end{array} \right. .\]
\end{observation}}

\begin{observation}\label{obs:dynamic-3}
For every $i \in [n]$, let $\bt_i$ denote the smallest number where $\bA_{i, \bt_i} = 1$. Then, $\bz_i = \bg_i^{\bt_i}$, and $\bg_i^{\bt_i}$ depends on the values $\tilde{\bc}_{i,0}^{\bt_i}$ and $\tilde{\bc}_{i, 1}^{\bt_i}$, which may be computed from:
\begin{itemize}
\item The point $x_i$ and $\bt_i$,
\item The tuples $(x_j, \ell;\bp_{\ell}^*, \bt_j) \in \bP$ for $\ell \leq \bt_i-1$ which contain the point $x_j$, the time $\ell$ where $\bA_{j,\ell}\cdot \bK_{\ell}=1$, the sampling weight $\bp^*_{\ell}$, and the smallest $\bt_j$ where $\bA_{j, \bt_j} = 1$, and $\bz_j$ (which may be recursively computed from $\bsigma^*$ and $\bP$).
\end{itemize}
Similarly to $\textsc{GreedyProcess}$, the final vector $\bz$ is set by $\textsc{Assign}_{\bsigma^*,\bP}(x_i, \bt_i)$.
\end{observation}

\subsection{Proof of Theorem~\ref{thm:main-structural-dynamic} assuming Lemma~\ref{lemma:dynamic-main}}


\begin{lemma}[Main Lemma] \label{lemma:dynamic-main}
Let $(\bz, \bm, \bsigma^*)$ be the output of \textsc{CorrelatedGreedyProcess}$((X, d),$ $ w, \gamma, t_0, t_e, z^*)$, where the weights $w$ are $\sfD$-compatible. For any $\eps > 0$, if 
\[  \gamma \geq \dfrac{(\sfD \ln(t_e) + 1)^2}{\eps^2} \qquad \text{and}\qquad t_0 \geq \max\left\{ \frac{\sqrt{\gamma} \cdot \sfD}{\eps}, \frac{1}{\eps} \right\} \qquad \text{and}\qquad \epspr \leq \frac{\eps}{\ln(t_e)}. \] 
Then, 
\[ \Ex[f(\bz) - f(z^*)] \lsim \eps \sum_{i=1}^n \sum_{j=1}^n d(x_i, x_j).  \]
\end{lemma}


\begin{proof}[Proof of Theorem~\ref{thm:main-structural-dynamic} assuming Lemma~\ref{lemma:dynamic-main}] Notice that, by construction $\bK_{t}$ is sampled from $\Ber(\gamma^t)$ which is $1$ if $t \leq t_0$ and $\min\{ \gamma /t, 1\}$ otherwise; hence it is distributed exactly as $\calK(t_0, \gamma)$. This now imposes that each column of $\bA$ has it's entries either behave independently and distributed solely based on $w_i$ and $t$, or dependently, such that there is at most one activation, based on a sample $(\by^*_t,\bp^*_t)$ from $\calG_t(X)$ and $t$. 
In particular, note that when $\bK_t=0$, $\bA_{i, t}$ is always sampled from $\Ber(\bw_i^{t})$, where the setting of $\bw_i^t$ depends solely on whether $i$ is or is not in $\bS^{t-1}$; if it is not in $\bS^{t-1}$, it is set to $\min \{ w_i, 1/t\}$ and otherwise $w_i$. 
 $\bK_{i,t}=1$, $\bA_{i, t}$ is depends on a sample from $(\by^*_t,\bp^*_t) \sim\calG_t(X)$ from Lemma~\ref{lem:geo-samp}. Notice that if $i \notin \bS^{t-1}$, then $\bA_{i,t} = 1$ with probability $\tilde{\bp}_i^t = p_i^t \cdot \min \left\{ 1/(t\bp^*_t), 1 \right\}$, where $\bp^*_t$ is within $(1+\epspr)$ of $p_i^t$; if $i \in \bS^{t-1}$, then $\bA_{i,t} = 1$ with probability $p_i^t$. 
Note that $i \in \bS^{t-1}$ if and only if there exists $\ell \leq t-1$ with $\bA_{i,\ell} = 1$, matching Definition \ref{def:corr-timeline-mask}. Observe that the set $\bP$ depends solely on $\bA$ and $\bK$, and since those two distributions match, the distribution of the sets $\bP$ are identical, matching Definition \ref{def:corr-timeline-mask-sum}. If we let $\bz$ denote the cut which was output by greedy process, we partition the coordinates into
\[ \bG = \left\{ i \in [n] : \exists t \leq t_e \text{ s.t } \bA_{i,t} = 1\right\} \qquad\text{and}\qquad \bB = [n] \setminus \bG.  \]
Importantly, we have defined \textsc{CorrelatedGreedyProcess} and \textsc{Assign} so that, whenever $i \in \bG$, if we let $\bt_i$ denote the activation time of point $x_i$, we have $\bz(\bsigma^*)_i$, which is the output of $\textsc{Assign}_{\bsigma^*,\bP}(x_i, \bt_i)$ is the same as $\bz_i$. Therefore, we have
\begin{align}
\Ex_{\bA, \bK}\left[ \min_{\sigma} f(\bz(\sigma)) \right] \leq \Ex_{\bA, \bK}\left[ f(\bz(\bsigma^*)) \right] &\leq \Ex_{\bA, \bK}\left[ f(\bz)\right] + \sum_{i =1}^n \left( \sum_{j=1}^n d(x_i, x_j)\right) \cdot \Prx\left[ i \in \bB \right]. \label{eq:dynamic-almost}
\end{align}
The final bound follows from Lemma~\ref{lemma:dynamic-prob-activation}, where we upper bound
\[ \Prx[i \in \bB] = \Prx[i \notin \bS^{t_e}] \leq \frac{\sfD}{w_i \cdot t_e} \cdot e^2,\]
as long as $\epspr \leq 1/\ln(t_e)$. As in the proof of Lemma~\ref{lemma:main}, recall that $1/w_i \geq 1/(2n)$ (recall, Lemma~\ref{lem:weight-to-compatible}), so that we may upper bound the above probability by $O(\sfD n / t_e)$. Letting $t_e$ be a large enough constant factor of $O(\sfD n / \eps)$ and $\epspr \leq 1 / \ln(t_e)$ gives the desired bound. 
%
\end{proof}

\begin{lemma}\label{lemma:dynamic-prob-activation} 
Consider any execution of \textsc{CorrelatedGreedyProcess}, as well as a fixed point $x_i \in X$, and any time $t$, as long as $\epspr \leq 1/\ln t$,
    \begin{align*}
        \Prx\left[ i \notin \bS^t \right] \leq \frac{\sfD}{w_i \cdot t} \cdot e^2
    \end{align*}
\end{lemma}

\begin{proof}
For any fixed $\bK$, every time $\ell$ where $\bK_{\ell} = 0$ and $i \notin \bS^{\ell-1}$, the probability that $i \notin \bS^{\ell}$ is exactly $1 - \min\{ w_i, 1/\ell \}$; if $\bK_{\ell} = 1$ and $i \notin \bS^{\ell-1}$, the probability that $i \notin \bS^{\ell}$ is $1 - p_i^{\ell} \cdot \min \{ 1/(\bp^*_{\ell} \cdot \ell), 1\}$ which is at most $1 - \min\{ p_i^{\ell}, 1/(\ell (1+\epspr))\}$ (see Definition~\ref{def:probs-seen-unseen}). 
\begin{align*}
\Prx\left[ i \notin \bS^{t} \right] \leq \prod_{\substack{\ell=1 \\ \bK_{\ell}=0}}^{t} (1 - \min\{ w_i, 1/\ell\}) \prod_{\substack{\ell=1 \\ \bK_t = 1}}^{t} \left(1 - \min\left\{ p_i^{\ell}, \frac{1}{\ell(1+\epspr)} \right\} \right)
\end{align*}
For every $i \in [n]$, let $p_i$ denote the minimum of $w_i$ and $p_i^{\ell}$ for all $\ell$ with $\bK_t = 1$. 
Thus, we have
\begin{align*}
    \Prx\left[ i \notin \bS^t \right] \leq \prod_{\ell=\lceil 1/p_i\rceil}^{t} \left(1 - \frac{1}{\ell} \right) \cdot  \prod_{\substack{\ell=\lceil 1/p_i\rceil \\ \bK_{\ell} = 1}}^{t} \left(\dfrac{1 - \frac{1}{(1+\epspr)\ell}}{1 - \frac{1}{\ell}} \right) &\leq \dfrac{\lceil 1/p_i \rceil - 1}{t} \cdot \exp\left( 2\epspr \sum_{\ell=1}^t 1 /\ell \right) \\
            &\leq \dfrac{\lceil 1/p_i \rceil - 1}{t} \cdot e^{2}
\end{align*}
whenever $\epspr \leq 1/\ln(t)$. The final bound comes from the fact $p_i \geq w_i / \sfD$.
\ignore{\begin{align*}
\Prx\left[ i \notin \bS^t \right]
&\leq \Ex_{\bK} \left [\prod_{\substack{\ell=1 \\ \bK_{\ell}=0} }^{t} \left( 1 - \min\left\{ w_i , 1/ \ell\right\} \right) \prod_{\substack{\ell=1 \\ \bK_{\ell}=1} }^{t} \left( 1 - \min\left\{p_i , p_i/(p_i^* \ell)\right\} \right) \right ]\\
&\leq \Ex_{\bK} \left [\prod_{\substack{\ell=\max(\lceil \frac{1}{w_i} \rceil,\lceil\frac{1}{p_i} \rceil) \\ \bK_t=0} }^{t} \left( 1 - \frac{1}{\ell} \right) \prod_{\substack{\ell=\max(\lceil \frac{1}{w_i} \rceil,\lceil\frac{1}{p_i} \rceil) \\ \bK_{\ell}=1} }^{t} \left( 1 -\frac{1}{(1+\eps')t} \right) \right ]\\
&\leq \Ex_{\bK} \left [\prod_{\substack{t=\max(\lceil \frac{1}{w_i} \rceil,\lceil\frac{1}{p_i} \rceil)} }^{\ell} \left(  1 - \frac{1}{\ell} \right) \prod_{\substack{t=\max(\lceil \frac{1}{w_i} \rceil,\lceil\frac{1}{p_i} \rceil) \\ \bK_{\ell}=1} }^{\ell} \left( \frac{1-\frac{1}{(1+\eps')\ell}}{1-\frac{1}{\ell}} \right) \right ]\\
&\leq \frac{\max(\lceil \frac{1}{w_i} \rceil,\lceil\frac{1}{p_i} \rceil)-1}{t} \Ex_{\bK} \left [ \prod_{\substack{\ell=\max(\lceil \frac{1}{w_i} \rceil,\lceil\frac{1}{p_i} \rceil) \\ \bK_{\ell}=1} }^{t} \left( 1 + \frac{2\eps'}{\ell} \right) \right ]\\
\end{align*}
where for the last step, we use the fact that $w_i \leq 1/2$ and  $p_i \leq 1/2$ ensures $t\geq 2$. Now we upper bound the expectation with $\bgamma^{\ell} \leq \frac{\gamma}{\ell}$ to get
\begin{align*}
\Ex_{\bK} \left [ \prod_{\substack{\ell=\max(\lceil \frac{1}{w_i} \rceil,\lceil\frac{1}{p_i} \rceil) \\ \bK_{\ell}=1} }^{t} \left( 1 + \frac{2\eps'}{\ell} \right) \right ] &\leq \prod_{\substack{\ell=\max(\lceil \frac{1}{w_i} \rceil,\lceil\frac{1}{p_i} \rceil)}}^{t}  \left( (1-\frac{\gamma}{\ell}) + \frac{\gamma}{\ell} \left( 1 + \frac{2\eps'}{\ell} \right)  \right)\\
&\leq \prod_{\substack{\ell=\max(\lceil \frac{1}{w_i} \rceil,\lceil\frac{1}{p_i} \rceil)}}^{t} \left( 1 + \frac{2\eps'\gamma}{\ell^2} \right) \\
&\leq \exp \left (\sum_{\substack{\ell=\max(\lceil \frac{1}{w_i} \rceil,\lceil\frac{1}{p_i} \rceil)}}^{t}  \left(\frac{2\eps'\gamma}{\ell^2} \right) \right )\\
&\leq e^4\\
\end{align*}
where the last step requires $\eps' \leq \frac{1}{\gamma}$, which causes the series to converge to $\frac{\pi^2}{3}$, and we can bound the expectation by $e^4$. Putting everything together we have 
\begin{align*}
\Prx\left[ i \notin \bS^t \right] \leq \frac{e^4\left (\max(\lceil \frac{1}{w_i} \rceil,\lceil\frac{1}{p_i} \rceil)-1\right)}{t} 
\end{align*}}
\end{proof}

\section{Proof of Lemma~\ref{lemma:dynamic-main}} 
\label{sec:dynamic-greedy-analysis}



In the analysis to follow, it is important to establish the following definitions:
\begin{itemize}
    \item We let $w \in (0,1/2]^n$ denote the vector specified by $w_i = \bw(x_i)$ produced by a single execution of Lemma~\ref{lem:geo-sketch} which succeeds. We shall therefore consider a fixed instance of Lemma~\ref{lem:geo-sketch} which succeeds, so the weights $w_1,\dots, w_n \in (0, 1/2]^n$ are fixed (and not random, so unbolded).
    \item In an execution of \textsc{CorrelatedGreedyProcess}, whenever $\bK_t = 1$, we utilize an instance of Lemma~\ref{lem:geo-samp} to generate the corresponding geometric sample. For $t$ with $\bK_t = 1$, we define $p_i^t$ as the probability that the instance of Lemma~\ref{lem:geo-samp} used at time $t$ generates $(\by_t^*, \bp^*_t)$ with $\by_t^* = x_i$. 
    \item For $t$ with $\bK_t = 1$, we define $\tilde{\bp}_i^t$ as the probability that $\bA_{i,t} = 1$ (Definition~\ref{def:probs-seen-unseen}). Note $\tilde{\bp}_{i}^t$ depends on whether or not $i \in \bS^{t-1}$. When $i \in \bS^{t-1}$, $\tilde{\bp}_i^t = p_i^t$, since $\bA_{i,t} = 1$ when $(\by_t^*,\bp^*)$ satisfies $\by_t^* = x_i$. However, if $i \notin \bS^{t-1}$, then $\tilde{\bp}_{i}^t$ is given by the expectation of $p_i^t \cdot \min\{ 1/(t \cdot \bp^*_t), 1\}$ over $(\by^*_t, \bp^*_t)$ from Lemma~\ref{lem:geo-samp} with $\by_t^* = x_i$. Importantly, Lemma~\ref{lem:geo-samp} implies that anytime $\by_t^* = x_i$, we will have $p_i^t \leq \bp^*_t \leq (1+\epspr) p_i^t$, and this implies $\tilde{\bp}_{i}^t \leq 1/t$. 
\end{itemize}

We define an analogous fictitious cut, where we are careful to consider the various settings of $\bK_t$ in order to update in the case we do not activate at a particular time. 

\begin{definition}[Fictitious cut]\label{def:dynamic-fictitious}
The fictitious cut is specified by a sequence of random variables $\hat{\bz}^t \in [0, 1]^{n \times \{0,1\}}$ for each $t \in \{0, \dots, t_e\}$. We let:
\begin{itemize}
\item For $t = 0$, we define $\hat{\bz}^0 = z^*$. 
\item For $t > 0$, we denote for each $i \in [n]$, the vector $\hat{\bz}_i^t \in [0,1]^{\{0,1\}}$ by\footnote{As in Section~\ref{sec:greedy-analysis}, it is not a prior clear that $\hat{\bz}_{i,b}^t \in [0, 1]$ but we will establish this fact.}
\begin{align}
\hat{\bz}_i^t &= \left\{ \begin{array}{cc} \bz_i^t & i \in \bS^t \\
							\dfrac{1}{1 - \tilde{\bp}_i^{t}} \left(\dfrac{t-1}{t} \cdot \hat{\bz}_i^{t-1} + \frac{1}{t} \cdot \bg_i^t -  \tilde{\bp}_i^{t} \cdot \bg_i^t \right) & \text{ $\bK_t=1$}\\
       					\dfrac{1}{1 - \bw_i^t} \left(\dfrac{t-1}{t} \cdot \hat{\bz}_i^{t-1} + \frac{1}{t} \cdot \bg_i^t -  \bw_i^t \cdot \bg_i^t \right) & \text{ $\bK_t=0$}
       \end{array} \right. \label{eq:dynamic-shadow-def}
\end{align}
\end{itemize}
\end{definition}

As in Section~\ref{sec:greedy-analysis}, it is immediate to check the analogous version to Claim~\ref{cl:simple-fictitious}. Namely, that $\hat{\bz}^t = z^*$ for $t \in \{0,\dots, t_0\}$ and that $\hat{\bz}^{t_e}_{i,b} \geq \bz_{i,b}$ for $b \in \{0,1\}$. Furthermore, we similarly can re-write $f(\hat{\bz}^t) - f(\hat{\bz}^{t-1})$ into the various types of terms (\ref{eq:linear-term}), (\ref{eq:non-linear-1}) and (\ref{eq:non-linear-2}). We note that the first statement of Lemma~\ref{lem:non-linear}, that for any execution of the first $t-1$ steps, 
\begin{align}
    \Ex\left[ \hat{\bz}^t_{i,b} - \hat{z}_{i,b}^{t-1} \right] = \frac{1}{t} \cdot \left( g_{i,b}^t - \hat{z}_{i,b}^{t-1}\right)  \label{eq:dynamic-step-t-exp}
\end{align}
remains true (by considering the two cases of $\bK_t = 0$ or $\bK_t = 1$. However, the second statement of Lemma~\ref{lem:non-linear} still holds, albeit with a different argument and constant factor loss, since the proof of Lemma~\ref{lem:non-linear} crucially relies on separate points being independent. Even though we no longer have independence, we state a lemma which will help us handle this.

\begin{lemma}\label{lem:neg-corr} For any $0\leq p_1,p_2 \leq 1$ and $\alpha_1,\alpha_2\, \beta_1,\beta_2 \in \R$, then consider the following joint distribution over random variables $(\bX_1, \bX_2)$: 
\begin{align*}
\bX_1,\bX_2\sim\begin{cases}
   \bX_1= \alpha_1, \bX_2= \beta_2 &  \text{ w.p. } p_1 \\
   \bX_1= \beta_1, \bX_2= \alpha_2 &  \text{ w.p. } p_2 \\
   \bX_1 = \beta_1, \bX_2= \beta_2 &  \text{ w.p. } 1-p_1-p_2 
\end{cases} .
\end{align*}
Then, $\Ex[\bX_1 \bX_2] =  \Ex[\bX_1]\Ex[\bX_2] + p_1p_2(\alpha_1\beta_2 +\alpha_2\beta_1 - \alpha_1\alpha_2 - \beta_1\beta_2)$. 

\end{lemma}
\begin{proof}
    The desired property follows from a simple calculation of $\Ex[\bX_1 \bX_2]$ and $ \Ex[\bX_1]\Ex[\bX_2]$
    \begin{align*}
        \Ex[\bX_1 \bX_2] &= p_1\alpha_1\beta_2 + p_2\beta_1\alpha_2 + (1-p_1-p_2)\beta_1\beta_2 \\
            &= p_1p_2(\alpha_1\beta_2+\beta_1\alpha_2) + p_1(1-p_2)\alpha_1\beta_2 + p_2(1-p_1)\beta_1\alpha_2 + (1-p_1-p_2)\beta_1\beta_2 \\
            &=   p_1p_2\alpha_1\alpha_2 + p_1(1-p_2)\alpha_1\beta_2 + p_2(1-p_1)\beta_1\alpha_2 + (1-p_1)(1-p_2)\beta_1\beta_2 \\ &\qquad+ p_1p_2(\alpha_1\beta_2+\beta_1\alpha_2 - \alpha_1\alpha_2-\beta_1\beta_2)\\
          &=  \Ex[\bX_1]\Ex[\bX_2]+ p_1p_2(\alpha_1\beta_2+\beta_1\alpha_2 - \alpha_1\alpha_2-\beta_1\beta_2)
\end{align*}
\end{proof}

\begin{lemma}\label{lem:dynamic-non-linear}
    We have that for any $t \geq t_0$ and any $b \in \{0,1\}$,
    \begin{align*}
        \Ex\left[ \frac{1}{2} \sum_{i=1}^n \sum_{j=1}^n d(x_i, x_j) \left(\hat{\bz}_{i,b}^{t} - \hat{\bz}_{i,b}^{t-1} \right) \left( \hat{\bz}_{j,b}^{t} - \hat{\bz}_{j,b}^{t-1} \right)\right] \leq \frac{3}{t^2} \sum_{i=1}^n \sum_{j=1}^n d(x_i, x_j). 
    \end{align*}
\end{lemma}

\ignore{\begin{claim}\label{cl:dynamic-simple-fictitious}
The following are simple facts about the fictitious cut, which follow from Definition~\ref{def:dynamic-fictitious}:
\begin{itemize}
\item For every $t \in \{0, \dots, t_0\}$, we have $\hat{\bz}^t = z^*$.
\item For every $i \in [n]$, we have $\hat{\bz}^{t_e}_{i, 1} \geq \bz_{i,1}$ and $\hat{\bz}^{t_e}_{i,0} \geq \bz_{i, 0}$.
\end{itemize}
\end{claim}

\begin{proof}
The first item follows from the definition of \textsc{CorrelatedGreedyProcess}, as well as a simple inductive argument. For $t = 0$, set $\hat{\bz}^0 = z^*$ by Definition~\ref{def:dynamic-fictitious}, so assume such is the case for $t \in \{0 ,\dots, \tilde{t}\}$ for $\tilde{t} < t_0$, and we will establish it for $t = \tilde{t} + 1$.  Recall that, for every $t \in \{0, \dots, t_0\}$, line~\ref{en:dynamic-greedy-ln2}, sets $\bg_i^t = z^*_i$. So consider the three cases of \ref{eq:dynamic-shadow-def}; if $i \in \bS^t$, then $\hat{\bz}_i^{t} = \bz_i^t$, which is equal to $\bz_i = z^*_i$ for $i \in \bS_{t} \subset \bS^{t_0}$ (by Observation~\ref{obs:dynamic-2}). If $i \notin \bS^t$, then (\ref{eq:dynamic-shadow-def}) gives
\[ \hat{\bz}_i^{t} = \frac{1}{1 - \bw_i^t} \left( \frac{t-1}{t} \cdot z_i^* + \frac{1}{t} \cdot z_i^* - \bw_i^t \cdot z_i^* \right) = z_i^*. \]
or $\bK_t=1$ and
\[ \hat{\bz}_i^{t} = \frac{1}{1 - \bp_i^t} \left( \frac{t-1}{t} \cdot z_i^* + \frac{1}{t} \cdot z_i^* - \bp_i^t \cdot z_i^* \right) = z_i^*. \]

For the second item, notice that $\hat{\bz}_i^{t_e} = \bz_i$ if $i \in \bS^t$, and otherwise lies in $[0,1]^{\{0,1\}}$, whereas $\bz_i = (0, 0)$.
\end{proof}

Recall that, for the proof of Lemma~\ref{lemma:dynamic-main}, we seek to upper bound the expected difference between $f(\bz) - f(z^*)$. Thus, the definition of the fictitious cut, as well as Claim~\ref{cl:dynamic-simple-fictitious},
\begin{align}
\Ex\left[ f(\bz) - f(z^*) \right] &\leq \Ex\left[ f(\hat{\bz}^{t_e}) - f(\hat{\bz}^0)\right] = \sum_{t = t_0+1}^{t_e-1} \Ex\left[ f(\hat{\bz}^{t}) - f(\hat{\bz}^{t-1}) \right],\label{eq:dynamic-goal-1}
\end{align} 
where in the first inequality, the fact $f(\bz) \leq f(\hat{\bz}^{t_e})$ follows from definition of $f$, the fact distances are non-negative, and Claim~\ref{cl:dynamic-simple-fictitious}. We will focus on showing an upper bound for each $t \in \{ t_0+1, \dots, t_e\}$, and by symmetry of $d(\cdot, \cdot)$ and definition of $f$, we may write
\begin{align}
f(\hat{\bz}^{t}) - f(\hat{\bz}^t) &= \sum_{i=1}^n \left(\hat{\bz}_{i,0}^{t} - \hat{\bz}_{i,0}^{t-1}  \right) \sum_{j=1}^n d(x_i, x_j) \hat{\bz}_{j,0}^t +  \sum_{i=1}^n \left(\hat{\bz}_{i,1}^{t} - \hat{\bz}_{i,1}^{t-1}  \right) \sum_{j=1}^n d(x_i, x_j) \hat{\bz}_{j,1}^t \label{eq:dynamic-linear-term}\\
		&\qquad + \frac{1}{2} \sum_{i=1}^n \sum_{j=1}^n d(x_i, x_j) \left(\hat{\bz}_{i,0}^{t} - \hat{\bz}_{i,0}^t \right) \left( \hat{\bz}_{j,0}^{t} - \hat{\bz}_{j,0}^{t-1} \right) \label{eq:dynamic-non-linear-1}\\
		&\qquad + \frac{1}{2} \sum_{i=1}^n \sum_{j=1}^n d(x_i, x_j) \left(\hat{\bz}_{i,1}^{t} - \hat{\bz}_{i,1}^t \right) \left( \hat{\bz}_{j,1}^{t} - \hat{\bz}_{j,1}^{t-1} \right), \label{eq:dynamic-non-linear-2}
\end{align}}


\begin{proof}
Consider any fixed execution of the first $t-1$ steps of \textsc{CorrelatedGreedyProcess} which fixes the setting of $\hat{z}_{i,b}^{t-1}$ and $g_{i,b}^{t}$ the setting of $w_{i}^{t}$ (in Line~\ref{en:dynamic-greedy-ln1}), and $g_{i,b}^{t}$ (in Line~\ref{en:dynamic-greedy-ln2}), as these depend only on the randomness which occurred up to step $t-1$). We now consider what happens at step $t$, where we apply linearity of expectation to consider fixed $i, j$, and it suffices to only consider $i,j \notin S^{t-1}$, since $\hat{\bz}_{i,b}^{t} = \hat{z}_{i,b}^{t-1}$ if $i \in S^{t-1}$. With probability $1-\gamma^t$, we have $\bK_{t}=0$, and in this case $\hat{\bz}_{i,b}^{t}$ and $\hat{\bz}_{j,b}^{t}$ are independent. Thus, conditioned on $\bK_t = 0$, we have (as in Lemma~\ref{lem:non-linear})
\begin{align*}
&\Ex\left[ \frac{1}{2} \sum_{i=1}^n \sum_{j=1}^n d(x_i, x_j) \left(\hat{\bz}_{i,b}^t - \hat{z}_{i,b}^{t-1} \right)\left(\hat{\bz}_{j,b}^t - \hat{z}_{j,b}^{t-1} \right) \mid \bK_t = 0\right] \\
&\qquad = \frac{1}{2 t^2} \sum_{i \notin S^{t-1}} \sum_{j\notin S^{t-1}} d(x_i, x_j) \left( g_{i,b}^{t} - \hat{z}_{i,b}^{t-1}\right) \left( g_{j,b}^{t} - \hat{z}_{j,b}^{t-1}\right) \leq \frac{1}{2t^2} \sum_{i=1}^n \sum_{j=1}^n d(x_i, x_j),
\end{align*}
where we first use (\ref{eq:dynamic-step-t-exp}) and then the fact $g_{i,b}^t, \hat{z}_{i,b}^{t-1} \in [0,1]$. It remains to consider the case $\bK_t = 1$, where the challenge is that $\hat{\bz}_{i,b}^{t}$ and $\hat{\bz}_{j,b}^t$ are no longer independent. However, here we use Lemma~\ref{lem:neg-corr} with $\bX_1$ set to $\hat{\bz}_{i,b}^{t} - \hat{z}_{i,b}^{t-1}$ and $\bX_2$ set to $\hat{\bz}_{j,b}^{t} - \hat{z}_{j,b}^{t-1}$ for $i, j \notin S^{t-1}$. Using the notation of Lemma~\ref{lem:neg-corr}, we associate $\alpha_1$ with the event that $\hat{\bz}_{i,b}^{t} = g_{i,b}^{t}$ when $(\by_t^*, \bp^*_t) \sim \calG_t(X)$ happens to have $\by_t^* = x_i$, and similarly $\beta_1$ with the event that $(\by_t^*, \bp^*_t) \sim \calG_t(X)$ has $\by^*_t = x_j$. Note $\alpha_1, \alpha_2, \beta_1, \beta_2 \in [-1,1]$ (because $\hat{\bz}_{i,b}^t, \hat{z}_{i,b}^{t-1},\hat{\bz}_{j,b}^t, \hat{z}_{j,b}^{t-1} \in [0, 1]$), and we have always required $\alpha_1 = \tilde{\bp}_i^t \leq 1/t$ and $\beta_1 = \tilde{\bp}_{j}^t \leq 1/t$. By Lemma~\ref{lem:neg-corr}, the cases $\bK_t = 1$ satisfy
\begin{align*}
    \Ex\left[ \frac{1}{2} \sum_{i=1}^n \sum_{j=1}^n d(x_i, x_j) \left(\hat{\bz}_{i,b}^{t} - \hat{z}_{i,b}^{t-1} \right)\left(\hat{\bz}_{j,b}^{t} - \hat{z}_{j,b}^{t-1} \right) \mid \bK_t = 1 \right] \leq \frac{1}{2} \sum_{i=1}^n \sum_{j=1}^n d(x_i, x_j) \left( \frac{1}{t^2} + \frac{4}{t^2}\right) .
\end{align*}
\ignore{
thus 
\begin{itemize}
\item With probability $\bw_i^{t}$, we have $i \in \bS^{t}$, and 
\[ \hat{\bz}^{t}_{i,b} - \hat{z}^{t-1}_{i,b} = g_{i,b}^{t} - \hat{z}^{t-1}_{i,b}. \]
\item With probability $1 - \bw_i^{t}$, we have $i \notin \bS^{t}$, and 
\[ \hat{\bz}^{t}_{i,b} - \hat{z}^{t-1}_{i,b} = \frac{1}{1 - \bw_{i}^{t}} \left( \frac{t-1}{t} \cdot \hat{z}_{i,b}^{t-1} + \frac{1}{t} \cdot g_{i,b}^{t} - \bw_i^{t}  g_{i,b}^{t}\right) - \hat{z}^{t-1}_{i,b}. \]
\end{itemize}
else with probability $\bgamma^t$, we have $\bK_{t}=1$
\begin{itemize}
\item With probability $\bp_i^{t}$, we have $x_i = x \sim G(x)$, $\Ber(\min(\frac{1}{(p_i^*)t},1))=1$ and thus $i \in \bS^{t}$, and 
\[ \hat{\bz}^{t}_{i,b} - \hat{z}^{t-1}_{i,b} = g_{i,b}^{t} - \hat{z}^{t-1}_{i,b}. \]
\item With probability $1 - \bp_i^{t}$, we have $x_i \neq x \sim G(x)$ and $i \notin \bS^{t}$, and 
\[ \hat{\bz}^{t}_{i,b} - \hat{z}^{t-1}_{i,b} = \frac{1}{1 - \bp_{i}^{t}} \left( \frac{t-1}{t} \cdot \hat{z}_{i,b}^{t-1} + \frac{1}{t} \cdot g_{i,b}^{t} - \bp_i^{t}  g_{i,b}^{t}\right) - \hat{z}^{t-1}_{i,b}. \]
\end{itemize}
Thus, taking expectation with respect to the randomness in this set, we have (after some cancelations):
\begin{align*}
\Ex\left[ \hat{\bz}^{t}_{i,b} - \hat{z}^{t-1}_{i,b}\right] &= \frac{1}{t} \cdot \left(g_{i,b}^{t} - \hat{z}_{i,b}^{t-1}\right).
\end{align*}
This gives the first part of the lemma. For the second part, we note that for any execution of the first $t-1$ steps, any $i \in S^{t-1}$, will have $i \in \bS^{t}$ in Line~\ref{en:greedy-ln3}, and therefore $\hat{\bz}^{t}_i = \hat{z}_i^{t-1} = z_i^{t-1}$. In other words, we may re-write the expression (\ref{eq:dynamic-non-linear-1}) and (\ref{eq:dynamic-non-linear-2}) as
\begin{align*}
\frac{1}{2} \sum_{i=1}^n \sum_{j=1}^n d(x_i, x_j) \left( \hat{\bz}^{t}_{i,b} - \hat{z}^{t-1}_{i,b} \right) \left( \hat{\bz}_{j,b}^{t} - \hat{z}_{j,b}^{t-1} \right) = \frac{1}{2} \sum_{i \notin S^t} \sum_{j \notin S^{t-1}} d(x_i, x_j) \left( \hat{\bz}^{t}_{i,b} - \hat{z}^{t-1}_{i,b} \right) \left( \hat{\bz}_{j,b}^{t} - \hat{z}_{j,b}^{t-1} \right).
\end{align*}
Since $d(x_i, x_j) = 0$ if $i = j$, we have, that if $\bK_{t}=0$, then by independence when $i \neq j$, we have that the expectation over the step $t$ is 
\begin{align*}
\Ex_{\bA|\bK_{t}=0}\left[\left( \hat{\bz}_{i,b}^{t} - \hat{z}_{i,b}^{t-1}\right)\left( \hat{\bz}_{j,b}^{t} - \hat{z}_{j,b}^{t-1}\right)\right] &= \dfrac{g_{i,b}^{t} - \hat{z}_{i,b}^{t}}{t} \cdot \frac{g_{j,b}^{t} - \hat{z}_{j,b}}{t} \leq \frac{1}{t^2},
\end{align*}
since $g_{\cdot,b}^{t}, \hat{z}_{\cdot, b}^{t-1} \in [0,1]$.

Now consider when $\bK_{t}=1$. We can apply lemma \ref{lem:neg-corr} and then compute the cost of the additional term $p_1p_2(\alpha_1\beta_2 +\alpha_2\beta_1 - \alpha_1\alpha_2 - \beta_1\beta_2)$, where $\alpha_1$ is the value the first random variable takes on when activated (e.g. with proability $p_1$)and $\beta_1$ otherwise, likewise for $\alpha_2$ $\beta_2$ with the second random variable.
In our application, $p_1=\bp_i^{t}$, $p_1=\bp_j^{t}$, 
$\alpha_1 = g_{i,b}^{t} - \hat{z}^{t-1}_{i,b}$, 
$\alpha_2 = g_{j,b}^{t} - \hat{z}^{t-1}_{j,b}$ 
Now observe that since $g_{\cdot,b}^{t}, \hat{z}_{\cdot, b}^{t-1} \in [0,1]$, we have $-1 \leq \alpha_1,\alpha_2 \leq 1$  With slightly more effort, one can show $0 \leq \beta_1,\beta_2 \leq \frac{2}{t}$. Consider simplifying the expression for \[\beta_1 = \frac{1}{1 - \bp_{i}^{t}} \left(\frac{1}{t} \cdot g_{i,b}^{t} - \bp_i^{t}  g_{i,b}^{t}\right) +  \frac{(t-1)-t(1- \bp_i^{t})}{(1- \bp_i^{t})t} \cdot \hat{z}^{t-1}_{i,b}.\] 
Now since $g_{i,b}^{t} \in [0,1]$, $0\leq\bp_i^{t} \leq \frac{1}{t}$  we can simplify the first term to \[0\leq \frac{1}{1 - \bp_{i}^{t}} \left(\frac{1}{t} \cdot g_{i,b}^{t} - \bp_i^{t}  g_{i,b}^{t}\right) \leq \frac{1}{t}\]
The second term can be simplified using $ \hat{z}^{t-1}_{i,b} \in [0,1]$ and $0\leq\bp_i^{t} \leq \frac{1}{t}$ to get  \[0\leq \frac{(t-1)-(t)(1- \bp_i^{t})}{(1- \bp_i^{t})t} \cdot \hat{z}^{t-1}_{i,b} \leq \frac{1}{t}\]
Thus $0 \leq \beta_1 \leq \frac{2}{t}$. The same argument applies to $\beta_2$. Now we can apply these bounds alongside $\bp_i^{t}, \bp_j^{t} \leq \frac{1}{t}$ to get \[p_1p_2(\alpha_1\beta_2 +\alpha_2\beta_1 - \alpha_1\alpha_2 - \beta_1\beta_2) \leq \frac{1}{t^2}(\frac{2}{t}+\frac{2}{t}+1 - 0) \leq \frac{3}{t^2}\]

Now we can apply Lemma \ref{lem:neg-corr} alongside the above bound to get that 

\begin{align*}
\Ex_{\bA|\bK_t=1}\left[\left( \hat{\bz}_{i,b}^{t} - \hat{z}_{i,b}^{t-1}\right)\left( \hat{\bz}_{j,b}^{t} - \hat{z}_{j,b}^{t-1}\right)\right]
&\leq \Ex_{\bA|\bK_{t}=1}\left[\left( \hat{\bz}_{i,b}^{t} - \hat{z}_{i,b}^{t-1}\right)\right]
\Ex_{\bA|\bK_{t}=1}\left[\left( \hat{\bz}_{j,b}^{t} - \hat{z}_{j,b}^{t-1}\right)\right] + \frac{3}{t^2}\\
&= \dfrac{g_{i,b}^{t} - \hat{z}_{i,b}^{t-1}}{t} \cdot \frac{g_{j,b}^{t} - \hat{z}_{j,b}^{t-1}}{t}  +\frac{3}{t^2}\leq \frac{4}{t^2},
\end{align*}
Now taking the expectation over $\bK_{t}$, we conclude that 
\begin{align*}
\Ex_{\bA \bK_{t}}\left[\left( \hat{\bz}_{i,b}^{t} - \hat{z}_{i,b}^{t-1}\right)\left( \hat{\bz}_{j,b}^{t} - \hat{z}_{j,b}^{t-1}\right)\right]
&= \bgamma^{t} \Ex_{\bA|\bK_{t}=1}\left[\left( \hat{\bz}_{i,b}^{t} - \hat{z}_{i,b}^{t-1}\right)\left( \hat{\bz}_{j,b}^{t} - \hat{z}_{j,b}^{t-1}\right)\right] \\
&+ (1-\bgamma^t)\Ex_{\bA|\bK_{t}=0}\left[\left( \hat{\bz}_{i,b}^{t} - \hat{z}_{i,b}^{t-1}\right)\left( \hat{\bz}_{j,b}^{t} - \hat{z}_{j,b}^{t-1}\right)\right] \\
&\leq \frac{4}{t^2},
\end{align*}}
\end{proof}

As in Section~\ref{sec:greedy-analysis}, Lemma~\ref{lem:dynamic-non-linear} is used to upper bound terms analogous to (\ref{eq:non-linear-1}) and (\ref{eq:non-linear-2}), it remains to upper bound terms analogous to (\ref{eq:linear-term}). 
\begin{definition}[Contribution Matrix, Ideal Estimated Contribution, and (Non-Ideal) Estimated Contribution]\label{def:ideal-real-est}
For time $t \in \{1,\dots, t_e\}$, the contribution matrix $\bc^t \in \R_{\geq 0}^{n \times \{0,1\}}$ is given by letting, 
\[ \bc^t_{i, 0} = \sum_{j=1}^n d(x_i, x_j) \cdot \hat{\bz}^{t-1}_{j,0} \qquad\text{and}\qquad \bc^t_{i,1} = \sum_{j=1}^n d(x_i, x_j) \cdot \hat{\bz}^{t-1}_{j, 1}.  \]
For a time $t \in \{ 1, \dots, t_e\}$ and an index $j \in [n]$, we define $\bR_j^0 =1$ and for $t \geq 2$
\[ \hat{\bR}_{j}^{t-1} = \frac{1}{t-1} \sum_{\ell=1}^{t-1} \dfrac{\bA_{j,\ell}^{t-1} \cdot \bK_{\ell}}{\tilde{\bp}_j^{\ell} \cdot \gamma^{\ell}}  \qquad\text{and}\qquad \bR_{j}^{t-1} = \frac{1}{t-1} \sum_{\ell=1}^{t-1} \dfrac{\bA_{j,\ell}^{t-1} \cdot \bK_{\ell}}{\brho_{j}^{\ell} \cdot \gamma^{\ell}} ,\]
and ideal estimated contribution matrix  $\hat{\bc}^t \in \R_{\geq 0}^{n \times \{0,1\}}$ be given by
\[ \hat{\bc}^t_{i,0} = \sum_{j=1}^n d(x_i, x_j) \cdot \hat{\bR}_j^{t-1} \cdot \hat{\bz}_{j,0}^{t-1} \qquad \text{and}\qquad  \hat{\bc}^t_{i,1} = \sum_{j=1}^n d(x_i, x_j) \cdot \hat{\bR}_j^{t-1} \cdot \hat{\bz}_{j,1}^{t-1}, \]

and the (non-ideal) estimated contribution matrix  $\tilde{\bc}^t \in \R_{\geq 0}^{n \times \{0,1\}}$ be given by
\[ \tilde{\bc}^t_{i,0} = \sum_{j=1}^n d(x_i, x_j) \cdot \bR_j^{t-1} \cdot \hat{\bz}_{j,0}^{t-1} \qquad \text{and}\qquad  \tilde{\bc}^t_{i,1} = \sum_{j=1}^n d(x_i, x_j) \cdot \bR_j^{t-1} \cdot \hat{\bz}_{j,1}^{t-1}, \]
and notice that, since $\hat{\bz}_{j,0}^{t} = \bz_{j,0}^{t}$ and $\hat{\bz}_{j,1}^t = \bz_{j,1}^{t}$ when $j \in \bS^{t}$ (and equivalently, when $\bR_{j}^t$ is non-zero), we defined $\tilde{\bc}_{i,0}^t$ and $\tilde{\bc}_{i,1}^t$ to align with Line~\ref{en:dynamic-greedy-ln2} in $\textsc{CorrelatedGreedyProcess}$.
\end{definition}


Recall that similarly to Section~\ref{sec:greedy-analysis}, we upper bound for every $t \in \{t_0 + 1, \dots, t_e\}$,
\begin{align}
    \Ex\left[f(\hat{\bz}^{t}) - f(\hat{\bz}^{t-1}) \right] &\leq \sum_{b \in \{0,1\}}\Ex\left[ \sum_{i=1}^n \left(\hat{\bz}^{t}_{i,b} - \hat{\bz}^{t-1}_{i,b} \right) \sum_{j=1}^n d(x_i, x_j) \hat{\bz}_{j,b}^{t-1} \right]  \label{eq:hah12}\\
        &\qquad\qquad + \frac{6}{t^2} \sum_{i=1}^n \sum_{j=1}^n d(x_i, x_j), \label{eq:hah13}
\end{align}
using Lemma~\ref{lem:dynamic-non-linear} to upper bound (\ref{eq:hah13}). Hence, it suffices to upper bound the expectation in (\ref{eq:hah12}). We re-write the terms inside the expectation as a sum of three terms and use (\ref{eq:dynamic-step-t-exp}) to consider the expectation solely over the randomness in the $t$-th step:
\begin{align}
&\Ex_{\bA_{\cdot, t}, \bK_t}\left[ \sum_{b \in \{0,1\}} \sum_{i=1}^n \left(\hat{\bz}^{t}_{i,b} - \hat{\bz}^{t-1}_{i,b} \right) \sum_{j=1}^n d(x_i, x_j) \hat{\bz}_{j,b}^{t-1}\right] = \frac{1}{t} \sum_{b\in\{0,1\}}\sum_{i=1}^n \left( \bg_{i,b}^{t} - \hat{\bz}_{i,b}^{t-1} \right) \bc_{i,b}^t  \nonumber \\
&\qquad\leq \frac{1}{t} \sum_{b \in \{0,1\}}\sum_{i \notin \bS^{t-1}} \left| \bc_{i,b}^t - \hat{\bc}_{i, b}^{t} \right| + \frac{1}{t}\sum_{b\in\{0,1\}} \sum_{i \notin \bS^{t-1}} \left| \hat{\bc}_{i,b}^t - \tilde{\bc}_{i, b}^{t} \right| + \frac{1}{t} \sum_{b\in\{0,1\}}\sum_{i \notin \bS^{t-1}} \left(\bg_{i,b}^{t} - \hat{\bz}^{t-1}_{i,b} \right) \cdot \tilde{\bc}_{i,b}^t. \label{eq:three-terms-dynamic}
\end{align}
Where we use that $\hat{\bz}_i^{t} = \hat{\bz}_i^{t-1}$ when $i \in \bS^{t-1}$, and the fact that $\bg_{i,0}^t, \hat{\bz}_i^{t-1} \in [0,1]$. As in Section~\ref{sec:greedy-analysis}, we upper bound the three terms in (\ref{eq:three-terms-dynamic}) individually, using the following three lemmas.

\begin{lemma}\label{lem:dynamic-greedy-choice-opt}
For any execution of \textsc{CorrelatedGreedyProcess}, any $t \in \{t_0+1, \dots, t_e\}$ and any $i \notin \bS^{t-1}$, 
\begin{align*}
\left( \bg_{i,0}^{t} - \hat{\bz}^{t-1}_{i,0}\right) \tilde{\bc}_{i,0}^{t} + \left( \bg_{i,1}^{t} - \hat{\bz}^{t-1}_{i,1}\right) \tilde{\bc}_{i,1}^t \leq 0. 
\end{align*}
\end{lemma}

The proof of the lemma follows in the same way as Lemma~\ref{lem:greedy-choice-opt}, where we first show $\hat{\bz}_{i,0}^t + \hat{\bz}_{i,1}^t = 1$ and $|\hat{\bz}_{i,0}^{t} - \hat{\bz}_{i,1}^t| \leq 1$, and the only change is that one must consider the case $\bK_t = 1$ and $\bK_t = 0$ separately, where each $w_i^t, \tilde{\bp}_i^t \leq 1/t$. This allows us to upper bound the last term in (\ref{eq:three-terms-dynamic}). The middle term is bounded by the following lemma, simply using the fact that $\hat{\bR}_j^{t-1}$ and $\bR_{j}^{t-1}$ will be multiplicatively close to each other.

\begin{lemma}\label{lem:dynamic-non-ideal-est-error-bound}
For any $t \in \{t_0+1, \dots, t_e\}$, any $i \notin \bS^{t-1}$, and any $b \in \{0,1\}$,
\begin{align*}
\Ex\left[ \sum_{i \notin \bS^t} |\hat{\bc}_{i,b}^{t} - \tilde{\bc}_{i,b}^t| \right] \leq \epspr \sum_{i=1}^n\sum_{j=1}^n d(x_i,x_j) \end{align*}
\end{lemma}

\begin{proof}
From Definition~\ref{def:ideal-real-est}, we have that $\hat{\bz}_{j,b}^{t-1} \in [0,1]$. Recall, furthermore that whenever $\bK_{\ell} = 1$ and we use a sample $(\by_{\ell}^*, \bp_{\ell}^*)$ from Lemma~\ref{lem:geo-samp}, a setting of $\bA_{j,\ell}^{t-1} = 1$ means $\by_{\ell}^* = x_j$. In this case $\brho_j^{\ell} = \min\{ \bp^*_{\ell}, 1/\ell\}$ and $\tilde{\bp}_{j}^{\ell} = p_j^{\ell} \cdot \min\{ 1/(\ell \bp^*_{\ell}), 1\}$. This means, in particular, that $|\tilde{\bp}_j^* / \brho_j^{\ell} -1| \leq \epspr$. Thus, we have
\begin{align*}
    \left| \hat{\bc}_{i,0}^{t} - \tilde{\bc}_{i,0}^{t} \right| = \sum_{j=1}^n d(x_i, x_j) \left| \hat{\bR}_j^{t-1} - \bR_{j}^{t-1}\right| \leq \frac{1}{t-1} \sum_{j=1}^n d(x_i, x_j) \sum_{\ell=1}^{t-1}\left(\dfrac{\bA_{j,\ell}^{t-1} \cdot \bK_{\ell}}{\tilde{\bp}_j^{\ell} \cdot \gamma^{\ell}} \right) \cdot \left| 1 - \frac{\tilde{\bp}_j^{\ell}}{\brho_j^{\ell}} \right|,
\end{align*}
which simplifies to $\epspr\sum_{j=1}^n d(x_i, x_j)$ times $1/(t-1) \sum_{\ell=1}^{t-1} \bA_{j,\ell}^{t-1} \cdot \bK_{\ell} / (\tilde{\bp}_{j}^{\ell} \gamma^{\ell})$, and once we take expectation over $\bK$ and $\bA_{j,\ell}$, we obtain the desired bound.  
\ignore{
    We aim to upper bound the difference between the ideal-estimator, which has access to the true sampling probabilities $p$ and the non-ideal estimator, which only has access to an estimate of the sampling probabilities $p^*$.
\begin{align*}
\Ex\left[ \sum_{i \notin \bS^t} |\hat{\bc}_{i,0}^{t} - \tilde{\bc}_{i,0}^t| \right] &\leq \sum_{i=1}^n \Ex\left[ |\hat{\bc}_{i,0}^{t} - \tilde{\bc}_{i,0}^t| \right] \\
&\leq \sum_{i=1}^n \Ex\left[ | \sum_{j=1}^n d(x_i, x_j) \cdot \hat{\bR}_j^{t} \cdot \hat{\bz}_{j,0}^t  -  \sum_{j=1}^n d(x_i, x_j) \cdot \bR_j^{t} \cdot \hat{\bz}_{j,0}^t  | \right] \\ 
&\leq \frac{1}{t} \sum_{i=1}^n \sum_{j=1}^n d(x_i, x_j) \Ex\left[ \sum_{\ell=1}^{t} \left ( | \dfrac{\bA_{j,\ell}^{t} \cdot \bK_{\ell}^{t}}{\bp_j^{\ell} \cdot \gamma^{\ell}} - \dfrac{\bA_{j,\ell}^{t} \cdot \bK_{\ell}^{t}}{\bp_j^{*\ell} \cdot \gamma^{\ell}} | \right ) \cdot |\hat{\bz}_{j,0}^t| \right] \\
&\leq \frac{1}{t} \sum_{i=1}^n \sum_{j=1}^n d(x_i, x_j)  \sum_{\ell=1}^{t} \Ex\left[ \left ( | \dfrac{\bA_{j,\ell}^{t} \cdot \bK_{\ell}^{t}}{\bp_j^{\ell} \cdot \gamma^{\ell}} - \dfrac{\bA_{j,\ell}^{t} \cdot \bK_{\ell}^{t}}{\bp_j^{*\ell} \cdot \gamma^{\ell}} | \right ) \right] \\
\end{align*}
Since $\hat{\bz}_{j,0}^t \in [0,1]$.
Observe that $\bp_j^t \leq \bp_{*j}^t$, hence $\dfrac{\bA_{j,\ell}^{t} \cdot \bK_{\ell}^{t}}{\bp_j^{\ell} \cdot \gamma^{\ell}} - \dfrac{\bA_{j,\ell}^{t} \cdot \bK_{\ell}^{t}}{\bp_j^{*\ell} \cdot \gamma^{\ell}} \geq 0$. Computing the expectation over $\ell$-th time step gives 
\begin{align*}
     \Ex\left[ | \dfrac{\bA_{j,\ell}^{t} \cdot \bK_{\ell}^{t}}{\bp_j^{\ell} \cdot \gamma^{\ell}} - \dfrac{\bA_{j,\ell}^{t} \cdot \bK_{\ell}^{t}}{\bp_j^{*\ell} \cdot \gamma^{\ell}} | \right ] 
     &\leq \Ex\left[ \dfrac{\bA_{j,\ell}^{t} \cdot \bK_{\ell}^{t}}{\bp_j^{\ell} \cdot \gamma^{\ell}} - \dfrac{\bA_{j,\ell}^{t} \cdot \bK_{\ell}^{t}}{\bp_j^{*\ell} \cdot \gamma^{\ell}} \right ] \\
     &= \bp_j^{\ell}\gamma^{\ell} \left ( \frac{1}{\bp_j^{\ell}\gamma^{\ell}} - \frac{1}{\bp_j^{*\ell}\gamma^{\ell}} \right )\\
     &\leq \bp_j^{\ell}\gamma^{\ell} \left ( \frac{1}{\bp_j^{\ell}\gamma^{\ell}} - \frac{1}{(1+\eps')\bp_j^{\ell}\gamma^{\ell}} \right )\\
     &\leq \eps'
\end{align*}
Which when plugged back in concludes the proof. 
\begin{align*}
    \frac{1}{t} \sum_{i=1}^n \sum_{j=1}^n d(x_i, x_j)  \sum_{\ell=1}^{t} \Ex\left[ \left ( | \dfrac{\bA_{j,\ell}^{t} \cdot \bK_{\ell}^{t}}{\bp_j^{\ell} \cdot \gamma^{\ell}} - \dfrac{\bA_{j,\ell}^{t} \cdot \bK_{\ell}^{t}}{\bp_j^{*\ell} \cdot \gamma^{\ell}} | \right ) \right] 
    &\leq \eps' \sum_{i=1}^n \sum_{j=1}^n d(x_i, x_j)
\end{align*}}
\end{proof}

The final lemma remaining is that which bounds the first term of (\ref{eq:three-terms-dynamic}). This is the lemma analogous to Lemma~\ref{lem:bound-on-error} from Section~\ref{sec:greedy-analysis}.

\begin{lemma}\label{lem:dynamic-error-bound}
For any $t \in \{2, \dots, t_e\}$ and any $b \in \{0,1\}$, whenever $\|p^*\|_{\infty} \leq 1/2$,
\begin{align*}
\Ex\left[ \sum_{i \notin \bS^{t-1}} \left| \bc_{i,b}^t - \hat{\bc}_{i,b}^{t}\right| \right] \leq 10 e \sqrt{\sfD} \left(\frac{2\sqrt{\gamma}}{t-1} + \frac{1}{\sqrt{\gamma}} \right) \sum_{i=1}^n \left( \sum_{j=1}^n \frac{d(x_i, x_j)^2}{w_j} \right)^{1/2}.
\end{align*}
\end{lemma}

\begin{proof}[Proof of Lemma~\ref{lemma:dynamic-main} assuming Lemmas~\ref{lem:dynamic-greedy-choice-opt},~\ref{lem:dynamic-non-ideal-est-error-bound} and~\ref{lem:dynamic-error-bound}]
The proof follows analogously to that of Lemma~\ref{lemma:main} in Section~\ref{sec:greedy-analysis}. In particular, we bound
\begin{align*}
    \Ex\left[ f(\bz) - f(z^*)\right] &\leq \left( O(1/t_0) + O(\epspr) \right) \cdot \left( \sum_{i=1}^n \sum_{j=1}^n d(x_i, x_j)\right) \\
                    &\qquad\qquad + \sum_{t=t_0+1}^{t_e} \frac{10 e\sqrt{\sfD}}{t} \left(\frac{2\sqrt{\gamma}}{t-1} + \frac{1}{\sqrt{\gamma}} \right) \sum_{i=1}^n \left(\sum_{j=1}^n \dfrac{d(x_i, x_j)}{w_j} \right)^{1/2} .
\end{align*}
The second term above simplifies to
\[ O(\sfD) \left( \frac{\sqrt{\gamma}}{t_0} + \dfrac{\ln t_e}{\sqrt{\gamma}}\right)\left( \sum_{i=1}^n \sum_{j=1}^n d(x_i, x_j)\right), \]
which gives the desired bound once $\gamma = (\sfD\ln t_e + 1)^2 / \eps^2$, and $t_0 = O(\sfD \sqrt{\gamma}/\eps)$ and $\epspr \leq \eps / \ln t_e$.
\ignore{

We put all the elements together, where we seek to upper bound:
\begin{align*}
&\Ex\left[ f(\bz) - f(z^*) \right] \leq \sum_{t=t_0+1}^{t_e} \Ex\left[ f(\hat{\bz}^{t}) - f(\hat{\bz}^{t-1}) \right] \\
	&\qquad \leq \sum_{t=t_0+1}^{t_e} \Ex\left[ \sum_{i=1}^n \left( \hat{\bz}_{i,0}^{t} - \hat{\bz}_{i,0}^{t-1} \right) \bc_{i,0}^t + \sum_{i=1}^n \left( \hat{\bz}_{i,1}^{t} - \hat{\bz}_{i,1}^{t-1} \right) \bc_{i,1}^t \right]  + \left(\sum_{i=1}^n \sum_{j=1}^n d(x_i, x_j)  \right) \sum_{t=t_0+1}^{t_e} \frac{4}{t^2} ,
\end{align*}
where we first followed (\ref{eq:dynamic-goal-1}), and then simplified (\ref{eq:dynamic-non-linear-1}) and (\ref{eq:dynamic-non-linear-2}) using the second part of Lemma~\ref{lem:dynamic-non-linear}. Then, using (\ref{eq:dynamic-goal-2}) and Lemma~\ref{lem:dynamic-greedy-choice-opt}, we obtain 
\begin{align*}
\Ex\left[ f(\bz) - f(z^*) \right] &\leq \sum_{t=t_0+1}^{t_e}\left( \frac{1}{t} \Ex\left[ \sum_{i \notin \bS^{t-1}} |\bc_{i,0}^t - \hat{\bc}_{i,0}| \right] + \frac{1}{t} \Ex\left[ \sum_{i \notin \bS^{t-1}} |\bc_{i,1}^t - \hat{\bc}_{i,1}| \right] \right) \\
		&\qquad\qquad + \sum_{t=t_0+1}^{t_e}\left( \frac{1}{t} \Ex\left[ \sum_{i \notin \bS^{t-1}} |\hat{\bc}_{i,0}^t - \tilde{\bc}_{i,0}| \right] + \frac{1}{t} \Ex\left[ \sum_{i \notin \bS^{t-1}} |\hat{\bc}_{i,1}^t - \tilde{\bc}_{i,1}| \right] \right) \\
        &\qquad \qquad + O(1/t_0) \left( \sum_{i=1}^n \sum_{j=1}^n d(x_i, x_j) \right) \\
		&\lsim \frac{\ln(t_e) + 1}{\sqrt{\gamma}} \sum_{i=1}^n \left( \sum_{j=1}^n \frac{d(x_i, x_j)^2}{\min(w_j, p_j)} \right)^{1/2} + \dfrac{\sqrt{ \gamma}}{t_0} \sum_{i=1}^n \left( \sum_{j=1}^n \frac{d(x_i, x_j)^2}{\min(w_j, p_j)}\right)^{1/2} \\
  	&\qquad\qquad + \eps'\ln(t_e)  \sum_{i=1}^n  \sum_{j=1}^n            d(x_i,x_j) \\
		&\qquad \qquad + \frac{1}{t_0}  \left( \sum_{i=1}^n \sum_{j=1}^n d(x_i, x_j) \right).
\end{align*}
We note that the assumptions on our weight vector imply that
\begin{align*}
\sum_{i=1}^n \left(\sum_{j=1}^n \frac{d(x_i, x_j)^2}{\min(w_j, p_j)}\right)^{1/2} \leq \sqrt{\sfD} \sum_{i=1}^n \sum_{j=1}^n d(x_i, x_j).
\end{align*}
so the above error bound simplifies to
\[ \left( \frac{(\ln(t_e) + 1) \sqrt{\sfD} }{\sqrt{\gamma}} + \frac{\sqrt{\gamma \cdot \sfD}}{t_0} + \frac{\sfD}{t_0} + \eps'\ln(t_e)\right) \sum_{i=1}^n \sum_{j=1}^n d(x_i, x_j). \]
that when 
\[ \gamma \geq \frac{(\ln(t_e) + 1)^2 \cdot \sfD}{\eps^2} \qquad\text{and}\qquad t_0 \geq \max\left\{ \frac{\sqrt{\gamma \cdot \sfD}}{\eps}, \frac{\sfD}{\eps} \right\} \qquad\text{and}\qquad \eps' \leq \frac{1}{\gamma} \leq \frac{\eps}{\ln(t_e)}, \]
we obtain our desired bound.}
\end{proof}

\ignore{\subsection{Proof of Lemma~\ref{lem:dynamic-greedy-choice-opt}}

The proof of Lemma~\ref{lem:dynamic-greedy-choice-opt} proceeds by first showing that the fictitious cut $\hat{\bz}^t \in [0,1]^{n \times \{0,1\}}$. The proof is a straight-forward induction. This allows us to say that the greedy choices $\bg^t$ are chosen so as to make the expression in Lemma~\ref{lem:dynamic-greedy-choice-opt} negative.

\begin{lemma}\label{lemma:dynamic-fic-cut-bounded} For any $t \in \{0, \dots t_e\}$ and any $i \in [n]$, we have
\begin{align*}
    \hat{\bz}_{i,0}^t + \hat{\bz}_{i,1}^t = 1 &&     \left |\hat{\bz}_{i,0}^t - \hat{\bz}_{i,1}^t \right | \leq 1 ,
\end{align*}
which implies $\hat{\bz}^t \in [0,1]^{n \times \{0,1\}}$.
\end{lemma}

\begin{proof}
The proof proceeds by induction on $t$. We note that $\hat{\bz}^0 = z^* \in \{0,1\}^{n\times \{0,1\}}$, so the condition trivially holds. We assume for induction that the condition holds for all $t \leq \tilde{t}$ and show the condition for $t=\tilde{t}+1$. First, note that if $i \in \bS^{t}$, then $\hat{\bz}_i^{t} = \bz_{i}^t \in \{0,1\}^{\{0,1\}}$ and satisfies the condition. So, suppose $i \notin \bS^{t}$. Now assume $\bK_{t}=0$, then 
\begin{align*}
\hat{\bz}_{i,0}^{t} + \hat{\bz}_{i,1}^{t} &= \dfrac{1}{1 - \bw_{i}^{t}} \left(\dfrac{t-1}{t} \cdot \left(\hat{\bz}_{i,0}^{t-1} + \bz_{i,1}^{t-1}\right) + \left(\frac{1}{t} - \bw_{i}^t \right) \cdot \left(\bg_{i,0}^{t} + \bg_{i,1}^{t} \right) \right) \\
				&= \dfrac{1}{1 - \bw_i^t}\left( \frac{t-1}{t} + \frac{1}{t} - \bw_i^t \right) = 1,
\end{align*}
and when $\bK_t=1$
\begin{align*}
\hat{\bz}_{i,0}^{t} + \hat{\bz}_{i,1}^{t} &= \dfrac{1}{1 - \bp_{i}^{t}} \left(\dfrac{t-1}{t} \cdot \left(\hat{\bz}_{i,0}^{t-1} + \bz_{i,1}^{t-1}\right) + \left(\frac{1}{t} - \bp_{i}^t \right) \cdot \left(\bg_{i,0}^{t} + \bg_{i,1}^{t} \right) \right) \\
				&= \dfrac{1}{1 - \bp_i^t}\left( \frac{t-1}{t} + \frac{1}{t} - \bp_i^t \right) = 1,
\end{align*}
where we used the inductive hypothesis and the fact $\bg_{i,0}^t + \bg_{i,1}^t = 1$. For the difference, we have:
\begin{align*}
\left| \hat{\bz}^{t}_{i,0} - \hat{\bz}^t_{i,1}\right| &= \left|\dfrac{1}{1 - \bw_{i}^{t}} \left(\dfrac{t-1}{t} \cdot \left(\hat{\bz}_{i,0}^{t-1} - \bz_{i,1}^{t-1}\right) + \left(\frac{1}{t} - \bw_{i}^t \right) \cdot \left(\bg_{i,0}^{t} - \bg_{i,1}^{t} \right) \right) \right| \\
			&\leq \left| \frac{1}{1 - \bw_i^{t}}\right| \cdot \left( \frac{t-1}{t} + \left| \frac{1}{t} - \bw_i^t \right| \right).
\end{align*}
and 
\begin{align*}
\left| \hat{\bz}^{t}_{i,0} - \hat{\bz}^t_{i,1}\right| &= \left|\dfrac{1}{1 - \bp_{i}^{t}} \left(\dfrac{t-1}{t} \cdot \left(\hat{\bz}_{i,0}^{t-1} - \bz_{i,1}^{t-1}\right) + \left(\frac{1}{t} - \bp_{i}^t \right) \cdot \left(\bg_{i,0}^{t} - \bg_{i,1}^{t} \right) \right) \right| \\
			&\leq \left| \frac{1}{1 - \bp_i^{t}}\right| \cdot \left( \frac{t-1}{t} + \left| \frac{1}{t} - \bp_i^t \right| \right).
\end{align*}
Finally, we note that \textsc{CorrelatedGreedyProcess} enforces that $\bw_i^t \leq 1/t$ when $i \notin \bS^{t-1}$, so $(t-1) / t + |1/t - \bw_i^t| = 1 - \bw_i^t$ and $\bp_i^t \leq 1/t$ when $i \notin \bS^{t-1}$, so $(t-1) / t + |1/t - \bp_i^t| = 1 - \bp_i^t$.. Giving the upper bound on the difference.
\end{proof}    

To conclude the proof of Lemma~\ref{lem:dynamic-greedy-choice-opt}, recall that $\bg_{i,0}^{t}$ and $\bg_{i,1}^{t}$ are chosen so as to minimize
\begin{align*}
\bg_{i, 0}^{t} \cdot \tilde{\bc}_{i,0}^{t} + \bg_{i,1}^{t} \cdot \tilde{\bc}_{i,1}^{t} = \min_{\substack{\alpha, \beta \in [0,1] \\ \alpha + \beta = 1}} \alpha \cdot \tilde{\bc}_{i,0}^{t} + \beta \cdot \tilde{\bc}_{i,1}^{t} \leq \hat{\bz}_{i,0}^{t-1} \cdot \tilde{\bc}_{i,0} + \hat{\bz}_{i,1}^{t-1} \cdot \tilde{\bc}_{i,1}^t,
\end{align*}
which gives the desired bound.}

\paragraph{Proof of Lemma~\ref{lem:dynamic-error-bound}}

The proof will be similar to Lemma \ref{lem:error-bound} in Section~\ref{sec:greedy-analysis}. We will similarly re-index $t$ so that the definition of $\bc_{i,b}^{t+1}$ depends solely on the randomness up to step $t$, i.e., those from $\bA^{t}_{j,\ell}, \bK_{\ell}$ for all $j \neq i$ and $\ell \leq t$. We let
\begin{align*}
\berr_{i,b}^{t+1}(j) &\eqdef d(x_i, x_j) \left(1 - \hat{\bR}_j^t \right) \cdot \hat{\bz}_{j,b}^t =  d(x_i, x_j) \left(1 - \frac{\bA_{j,t}^t \cdot \bK_t}{t \cdot \tilde{\bp}_j^{t} \gamma^t} - \frac{t-1}{t} \cdot \hat{\bR}_j^{t-1} \right) \cdot \hat{\bz}^t_{j,b} \\ 
\berr_{i,b}^{t+1} &\eqdef \sum_{j=1}^n \berr_{i,b}^{t+1}(j) = \bc_{i,b}^{t+1} - \hat{\bc}_{i,b}^{t+1}.
\end{align*}
By linearity of expectation, Jensen's inequality, we seek to upper bound
\begin{align}
\Ex\left[ \sum_{i \notin \bS^t} |\bc_{i,b}^{t+1} - \hat{\bc}_{i,b}^{t+1}| \right]  
\leq \sum_{i=1}^n \left( \Ex\left[(\berr_{i,b}^{t+1})^2 \right] \right)^{1/2}. \label{eq:dynamic-expression}
\end{align}

The following lemma is analogous to Lemma~\ref{lem:martingale} in Section~\ref{sec:greedy-analysis}. As in Section~\ref{sec:greedy-analysis}, we have defined the fictitious cut so as to satisfy the conclusion of Lemma~\ref{lem:martingale}. The only minor subtlety is that one must consider the case $\bK_t = 1$ and $\bK_t = 0$ separately, since the fictitious cut and the probability that $\bA_{i,t} = 1$ differs in each case (when $\bK_t = 0$, $\bA_{i,t}$ is $1$ with probability $w_i^t$; when $\bK_t =1$, $\bA_{i,t}$ is $1$ with probability $\tilde{\bp}_i^t$). Considering both cases, we obtain the following lemma.

\begin{lemma}\label{lem:dynamic-martingale}
For any execution of the first $t-1$ steps of \textsc{CorrelatedGreedyProcess}, fixing the random variables $A^{t-1}$ and any setting of the first $t$ values $K_{\ell}$ for $\ell \leq t$, and $\hat{z}^{t-1}$. Then, considering only the $t$-th step of \textsc{CorrelatedGreedyProcess}, 
\[ \Ex_{\bA_{\cdot,t}}\left[ \berr_{i,b}^{t+1}(j) \right] =  \frac{t-1}{t} \cdot \err_{i,b}^t(j) ,\]
where the expectation is taken only with respected to the randomness in the $t$-th step. 
\end{lemma}

\ignore{ \begin{proof}
We simply go through the calculation, where we write $\berr_{i,b}^{t}(j)$ in terms of $R_j^{t-1}$ as 
\begin{align*}
\berr_{i,b}^{t+1}(j) = d(x_i, x_j) \left(1 - \frac{1}{t} \cdot \dfrac{\bA_{j,t}^{t} \cdot \bK_{t}^t}{p_j^{t} \cdot \gamma^{t}}  - \frac{1}{t} \cdot R_j^{t-1} \right) \hat{\bz}^t_{j,b}.
\end{align*}
We consider two cases, depending on whether or not $j \in S^{t-1}$. In the case that $j \in S^{t-1}$, the fictitious cut $\hat{\bz}^t_{j,b}$ will always be set to $\hat{z}^{t-1}_{j,b}$ (as it remains fixed). Therefore, we consider the 3 cases
\begin{itemize}
\item With probability $\gamma^t p_j^t$, we have $\bA_{j,t}^t\cdot \bK_{t}^t = 1$. Thus we get  
\begin{align*}
d(x_i, x_j) \left(1 - \frac{1}{t} \cdot \dfrac{1}{p_j^{t} \cdot \gamma^{t}}  - \frac{1}{t} \cdot R_j^{t-1} \right) \hat{\bz}^{t-1}_{j,b}
\end{align*}
\item With probability $\gamma^t (1-p_j^{t})$, we have $\bA_{j,t}^t=0$ and $\bK_{t}^t = 1$. Thus we get 
\begin{align*}
d(x_i, x_j) \left(1 - \frac{1}{t} \cdot R_j^{t-1} \right) \hat{\bz}^{t-1}_{j,b}
\end{align*}
\item With probability $1 - \gamma^t$, we have $\bK_{j,t}^t = 0$, in which case:
\begin{align*}
d(x_i, x_j) \left(1 - \frac{1}{t} \cdot R_j^{t-1} \right) \hat{\bz}^{t-1}_{j,b}
\end{align*}
\end{itemize}
Now combining these terms, we can compute the expectation:
\begin{align*}
\Ex\left[ d(x_i, x_j) \left(1 - \frac{1}{t} \cdot \dfrac{\bA_{j,t}^{t} \cdot \bK_{t}^t}{p_j^{t} \cdot \gamma^{t}}  - \frac{1}{t} \cdot R_j^{t-1} \right) \hat{\bz}^t_{j,b} \right] 
&= d(x_i, x_j) \left(1 - \frac{p_j^{t}}{p_j^{t} t} - \frac{1}{t} \cdot R_j^{t-1} \right) \hat{z}^{t-1}_{j,b} \\
&\leq \frac{t-1}{t} \cdot d(x_i,x_j) \left(1 - \frac{1}{t-1} \cdot R_j^{t-1} \right) \hat{z}_{j,b}^{t-1}\\
&\leq \frac{t-1}{t} \cdot \err_{i,b}^{t}(j).
\end{align*}
On the other hand, if $j \notin S^{t-1}$, then $R_j^{t-1} = 0$. The expectation calculation, however, is more complicated, since the setting of $\hat{\bz}^t_{j,b}$ depends on the realization of $\bA_{j,t}^t$ and $\bK_{t}^t$. We go through the four cases:
\begin{itemize}
\item With probability $\gamma^t p_j^t$, we have $\bA_{j,t}^t \cdot \bK_{t}^t = 1$. This results in a setting of $\hat{\bz}^t_{j,b} = g_{j,b}^t$, so we have that $\berr_{i,b}^{t+1}(j)$ becomes
\begin{align*}
d(x_i, x_j) \left(1 - \frac{1}{t} \cdot \frac{1}{p_j^{t}\gamma^t} \right) \cdot g_{j,b}^t.
\end{align*}
\item With probability $\gamma (1 - p_j^t)$, we have $\bK_{t}^t = 1$ and $\bA_{j,t}^t = 0$, in which case $\hat{\bz}_{j,b}^t$ is updated according to the second case of Definition~\ref{def:dynamic-fictitious}, and $\berr_{i,b}^{t+1}(j)$ becomes:
\begin{align*}
d(x_i,x_j) \cdot \frac{1}{1 - p_j^t} \left(\frac{t-1}{t} \cdot \hat{z}_{j,b}^{t-1} + \frac{1}{t} \cdot g_{j,b}^t - p_j^t \cdot g_{j,b}^t \right).
\end{align*}

\item With probability $(1-\gamma) w_j$, we have $\bK_{t}^t = 0$ and $\bA_{j,t}^t = 1$, in which case $\hat{\bz}^t_{j,b} = g_{j,b}^t$, and $\berr_{i,b}^{t+1}(j)$ becomes:
\begin{align*}
d(x_i, x_j) \cdot g_{j,b}^t.
\end{align*}
\item With probability $(1-\gamma) (1-w_j)$, we have $\bK_{t}^t = 0$ and $\bA_{j,t}^t = 0$, in which case $\hat{\bz}_{j,b}^t$ is updated according to the third case of Definition~\ref{def:dynamic-fictitious}, and $\berr_{i,b}^{t+1}(j)$ becomes:
\begin{align*}
d(x_i,x_j) \cdot \frac{1}{1 - w_j^t} \left(\frac{t-1}{t} \cdot \hat{z}_{j,b}^{t-1} + \frac{1}{t} \cdot g_{j,b}^t - w_j^t \cdot g_{j,b}^t \right).
\end{align*}
\end{itemize}

Computing the expectation over all 4 cases gives us an expectation of
\begin{align*}
\Ex\left[ \berr_{i,b}^{t+1}(j) \right] &\leq \frac{t-1}{t} \cdot d(x_i, x_j) \cdot \hat{z}_{j,b}^{t-1} \\
	&\leq \frac{t-1}{t} \cdot \err_{i,b}^{t}(j),
\end{align*}
concluding the proof.
\end{proof}}

\begin{lemma}\label{lem:dynamic-reduce-to-error}
For any $t  \in \{2, \dots, t_e-1\}$, 
\begin{align*}
\Ex\left[\sum_{i \notin \bS^t} \left|\bc_{i,b}^{t+1} - \hat{\bc}_{i,b}^{t+1} \right| \right] 
&\leq \frac{1}{t} \sum_{i=1}^n \left( \sum_{\ell=1}^{t} \ell^2 \sum_{j=1}^n \Ex\left[\left( \berr_{i,b}^{\ell+1}(j) - \frac{\ell-1}{\ell} \cdot \berr_{i,b}^{\ell}(j)\right)^2 \right] \right)^{1/2}.
\end{align*}
\end{lemma}

The proof of the above claim also follows similarly to Lemma~\ref{lem:reduce-to-error} in Section~\ref{sec:greedy-analysis}. The one thing to check that the final step in which Lemma~\ref{lem:reduce-to-error}, which uses independence of $\berr_{i,b}^{t+1}(j)$ and $\berr_{i,b}^{t+1}(j')$ may still be done even though these are not independent. In particular, Lemma~\ref{lem:reduce-to-error} analogously establishes
\begin{align*}
    \Ex\left[ \sum_{i \notin \bS^t} \left| \bc_{i,b}^{t+1} - \hat{\bc}_{i,b}^{t+1} \right|\right] \leq \frac{1}{t} \sum_{i=1}^n \left(\sum_{\ell=1}^t \ell^2 \Ex\left[\left(\berr_{i,b}^{\ell+1} - \frac{\ell-1}{\ell} \cdot \berr_{i,b}^{\ell} \right)^2 \right] \right)^{1/2},
\end{align*}
but we must verify that, even though differing values of $j$ are not independent,
\[ \Ex\left[\left( \berr_{i,b}^{\ell+1} - \frac{\ell-1}{\ell}\cdot \berr_{i,b}^{\ell} \right)^2\right] \leq \sum_{j=1}^n \Ex\left[\left(\berr_{i,b}^{\ell+1}(j) - \frac{\ell-1}{\ell} \cdot \berr_{i,b}^{\ell}(j)\right)^2 \right].\]
In particular, we have the following claim.
\begin{claim}
    For any fixed setting of the first $\ell-1$ time steps, any two differing $j, j' \in [n]$ satisfy
    \begin{align*}
        \Ex\left[ \left(\berr_{i,b}^{\ell+1}(j) - \frac{\ell-1}{\ell}\cdot \err_{i,b}^{\ell}(j) \right)\left(\berr_{i,b}^{\ell+1}(j') - \frac{\ell-1}{\ell} \cdot \err_{i,b}^{\ell}(j') \right)\right] \leq 0.
    \end{align*}
\end{claim}

\begin{proof}
\ignore{We will upper bound (\ref{eq:dynamic-expression}), where we first prove by induction on $t \in \{0, \dots, t_e-1\}$ that we always satisfy
\begin{align*}
\Ex\left[ \left( \berr_{i,b}^{t+1} \right)^2\right] = \sum_{\ell=1}^{t} \frac{\ell^2}{t^2} \Ex\left[ \left(\berr_{i,b}^{\ell+1} - \frac{\ell-1}{\ell} \cdot \berr_{i,b}^{\ell} \right)^2\right].
\end{align*}
Note that the base case of $t = 0$ corresponds to $\berr_{i,b}^1$, which is deterministically $0$ by $\bR_j^0 = 1$. Thus, assume for the sake of induction that
\begin{align*}
\Ex\left[ \left( \berr_{i,b}^{t+1} \right)^2\right] = \sum_{\ell=1}^{t} \frac{\ell^2}{t^2} \Ex\left[ \left(\berr_{i,b}^{\ell+1} - \frac{\ell-1}{\ell} \cdot \berr_{i,b}^{\ell} \right)^2\right].
\end{align*}
Consider proving the above by induction on $t$. The base case, $t = 0$, corresponds to $\berr_{i,b}^1$, which is deterministically $0$ by $\bR_j^0 = 1$. 
\begin{align*}
\Ex\left[ \left(\berr_{i,b}^{t+1}\right)^2\right] 
&= \Ex\left[ \left(\berr_{i,b}^{t+1} - \frac{t-1}{t} \cdot \berr_{i,b}^{t} + \frac{t-1}{t} \cdot \berr_{i,b}^{t}\right)^2\right] \\
&= \Ex\left[ \left( \berr_{i,b}^{t+1} - \frac{t-1}{t} \cdot \berr_{i,b}^{t} \right)^2 \right] + \left(\frac{t-1}{t} \right)^2 \Ex\left[ \left(\berr_{i,b}^{t}\right)^2 \right]. 
\end{align*}
Applying the inductive hypothesis on $t-1$ and simplifying the expression (since the value of $(t-1)^2$ will ``cancel out'') gives the first part of the bound.

Now we aim to prove that 
\begin{align*}
     \Ex\left[ \left(\berr_{i,b}^{\ell+1} - \frac{\ell-1}{\ell} \cdot \berr_{i,b}^{\ell}\right)^2 \right]
     &\leq \sum_{j=1}^n \Ex\left[\left( \berr_{i,b}^{\ell+1}(j) - \frac{\ell-1}{\ell} \cdot \berr_{i,b}^{\ell}(j)\right)^2 \right]
\end{align*}
We start by using linearity of expectation to break the error term into two pieces, the square terms and the cross terms. 
\begin{align*}
     \Ex\left[ \left(\berr_{i,b}^{\ell+1} - \frac{\ell-1}{\ell} \cdot \berr_{i,b}^{\ell}\right)^2 \right]
     &= \Ex\left[\left(\sum_{j=1}^n \berr_{i,b}^{\ell+1}(j) - \frac{\ell-1}{\ell} \cdot \berr_{i,b}^{\ell-1}(j)\right)^2 \right] \\ 
     &= \sum_{j=1}^n  \Ex\left[ \left( \berr_{i,b}^{\ell}(j) - \frac{\ell-1}{\ell} \cdot \berr_{i,b}^{\ell}(j)\right)^2 \right ]\\
     &+ \sum_{j=1}^n \sum_{\substack{k=1\\ k\neq j}}^n \Ex\left[ \left( \berr_{i,b}^{\ell+1}(j) - \frac{\ell-1}{\ell} \cdot \berr_{i,b}^{\ell}(j)\right) \left( \berr_{i,b}^{\ell+1}(k) - \frac{\ell-1}{\ell} \cdot \berr_{i,b}^{\ell}(k)\right) \right]\\
     &\leq \sum_{j=1}^n  \Ex\left[ \left( \berr_{i,b}^{\ell+1}(j) - \frac{\ell-1}{\ell} \cdot \berr_{i,b}^{\ell}(j)\right)^2 \right ]\\
\end{align*}
Now we simply need to show that for any $j,k\in [n]$ s.t. $j \neq k$, the expectation is smaller. We we will prove this by analyzing the terms of the expectation to show that the probability of non-positive terms increases compared to the product of the expectation, which by Lemma \ref{lem:neg-corr} gives us the bound.  
\begin{align*}
    &\Ex_{}\left[\left( \berr_{i,b}^{\ell+1}(j) - \frac{\ell-1}{\ell} \cdot \berr_{i,b}^{\ell}(j)\right) \left( \berr_{i,b}^{\ell+1}(k) - \frac{\ell-1}{\ell} \cdot \berr_{i,b}^{\ell}(k)\right) \right] \\ 
    &\qquad \leq \Ex\left[\left( \berr_{i,b}^{\ell+1}(j) - \frac{\ell-1}{\ell} \cdot \berr_{i,b}^{\ell}(j)\right) \right] \Ex\left[\left( \berr_{i,b}^{\ell+1}(k) - \frac{\ell-1}{\ell} \cdot \berr_{i,b}^{\ell}(k)\right) \right] \leq 0
\end{align*}}
For any setting of the first $\ell-1$ time steps, for any of the settings of $\bK_{\ell}$, the expectation of $\berr_{i,b}^{\ell+1}(j) - (\ell-1)/\ell \cdot \err_{i,b}^{\ell}(j)$ is zero from Lemma~\ref{lem:dynamic-martingale}. Whenever $\bK_{\ell}=0$, the two terms are independent and we obtain our desired bound. When $\bK_{\ell} = 1$, the terms are non-independent. We let
\[ \bX_1 = \berr_{i,b}^{\ell+1}(j) - \frac{\ell-1}{\ell} \cdot \err_{i,b}^{\ell}(j) \qquad\text{and}\qquad \bX_2 = \berr_{i,b}^{\ell+1}(j') - \frac{\ell-1}{\ell} \cdot \err_{i,b}^{\ell}(j'). \]
We claim that, whenever $\bA_{j,\ell} = 1$ (i.e., the $j$-th point is activated), we will have $\bX_1 < 0$ (and similarly, $\bA_{j',\ell} = 1$ implies $\bX_2 < 0$). This suffices to prove the claim for the following reason: once we condition on $\bK_{\ell} = 1$, we use Lemma~\ref{lem:neg-corr} with $\alpha_1$ be the setting of $\bX_1$ when $\bA_{j,\ell} = 1$, and we let $\alpha_2$ be the setting of $\bX_2$ when $\bA_{j', \ell} = 1$. $\bX_1$ and $\bX_2$ are both zero centered and supported on just two outcomes, so if $\alpha_1,\alpha_2 \leq 0$, then $\beta_1, \beta_2 \geq 0$. By Lemma~\ref{lem:neg-corr}, we have for any $p_1, p_2 \in [0,1]$,
\begin{align*}
    \Ex[\bX_1 \bX_2] = p_1 p_2 \left( \alpha_1 \beta_2 + \alpha_2 \beta_1 - \alpha_1 \alpha_2 - \beta_1 \beta_2\right) \leq 0.
\end{align*}
We verify that $\bX_1 < 0$ when $\bA_{j,\ell} = 1$ (the case of $\bX_2$ is analogous), and we do so by consider the two cases:
\begin{itemize}
    \item If $j \in S^{\ell-1}$, then we know $\hat{\bz}^{\ell}_{j,b} = \hat{z}^{\ell-1}_{j,b}$. Thus, if $\bA_{j,\ell}\cdot \bK_{\ell} = 1$, we have that 
    \[ \bX_1 = \berr_{i,b}^{\ell+1}(j) - \frac{\ell-1}{\ell} \cdot \err_{i,b}^{\ell}(j) = \frac{d(x_i, x_j)}{\ell} \left(1 - \frac{1}{\tilde{\bp}_{j}^{\ell} \gamma^{\ell}} \right) \hat{z}^{\ell-1}_{j,b} \leq 0 \]
    because $\tilde{\bp}_{j}^{\ell}$ and $\gamma^{\ell}$ are both at most $1$, and $\hat{z}_{j,b}^{\ell-1} \geq 0$.
    \item If $j \notin S^{\ell-1}$, then we know $\hat{R}_{j}^{\ell-1} = 0$. Thus, if $\bA_{j,\ell}\cdot \bK_{\ell} =1$, we have that $\hat{\bz}^{\ell}_{j,b} = g_{j,b}^{\ell} \geq 0$, and 
    \[ \bX_1 = \berr_{i,b}^{\ell+1}(j) - \frac{\ell-1}{\ell} \cdot \err_{i,b}^{\ell}(j) = d(x_i, x_j) \left( \left( 1 - \frac{1}{\ell \cdot \tilde{\bp}_{j}^{\ell} \cdot \gamma^{\ell}} \right) g_{j,b}^{\ell} - \frac{\ell-1}{\ell} \cdot \hat{z}^{\ell-1}_{j,b} \right) \leq 0 \]
    because $\tilde{\bp}_{j}^{\ell} \leq 1/\ell$, $\gamma^{\ell} \leq 1$, and $\hat{z}_{j,b}^{t-1} \geq 0$. 
\end{itemize}
\ignore{ 
First consider first fixing the randomness of $\bK_{\ell}$, and observe that with probability $1-\bgamma^{\ell}$, $\bK_{\ell}=0$ which means $\berr_{i,b}^{\ell+1}(j) - \frac{\ell-1}{\ell} \cdot \berr_{i,b}^{\ell}(j)$ and $\berr_{i,b}^{\ell+1}(k) - \frac{\ell-1}{\ell} \cdot \berr_{i,b}^{\ell}(k)$ independent from one another. Since for independent random variables, the expectation of the product and the product of the expectation are equal, these terms would cancel out. On the other hand when $\bK_{\ell}=1$, the terms are no longer independent. We first consider any fixed execution of the first $\ell-1$ steps, then take the expectation over the randomness in the $\ell$-th step.
Observe that since both $\Ex[\berr_{i,b}^{\ell+1}(j)-\frac{\ell-1}{\ell}\berr_{i,b}^{\ell}(j)]=0$ and $\Ex[\berr_{i,b}^{\ell+1}(k)-\frac{\ell-1}{\ell}\berr_{i,b}^{\ell}(k)]=0$, then the extra term in lemma \ref{lem:neg-corr}, $(\alpha_1\beta_2 +\alpha_2\beta_1 - \alpha_1\alpha_2 - \beta_1\beta_2)\leq 0$. This can also be verified by calculating of the signs of the extra terms.

We begin by considering the case where both $j\in\bS^t$ and $k\in\bS^t$ and so the fictitious cut $\hat{\bz}^t_{j,b}$ will always be set to $\hat{z}^{t-1}_{j,b}$ and $\hat{\bz}^t_{k,b}$ will always be set to $\hat{z}^{t-1}_{k,b}$.We have the following sub-cases:
\begin{itemize}
\item With probability $0$, we have $\bA_{j,t}^t = 1$ and $\bA_{k,t}^t = 1$. The value would be
\begin{align*}
d(x_i, x_j) \left(\frac{1}{t} - \frac{1}{t} \cdot \dfrac{1}{p_j^{t} \cdot \gamma^{t}} \right) \hat{\bz}^{t-1}_{j,b} \cdot d(x_i, x_k) \left(\frac{1}{t} - \frac{1}{t} \cdot \dfrac{1}{p_k^{t} \cdot \gamma^{t}} \right) \hat{\bz}^{t-1}_{k,b}  \geq 0
\end{align*}
since only $\left(\frac{1}{t} - \frac{1}{t} \cdot \dfrac{1}{p_j^{t} \cdot \gamma^{t}} \right)$ and $\left(\frac{1}{t} - \frac{1}{t} \cdot \dfrac{1}{p_k^{t} \cdot \gamma^{t}} \right)$ are less than zero, hence cancel making the expression is positive.
\item With probability $p_j^t$, we have $\bA_{j,t}^t = 1$ and $\bA_{k,t}^t = 0$. Thus we get  
\begin{align*}
d(x_i, x_j) \left(\frac{1}{t} - \frac{1}{t} \cdot \dfrac{1}{p_j^{t} \cdot \gamma^{t}} \right) \hat{\bz}^{t-1}_{j,b} \cdot d(x_i, x_k) \left(\frac{1}{t}\right) \hat{\bz}^{t-1}_{k,b} \leq 0
\end{align*}
since only $\left(\frac{1}{t} - \frac{1}{t} \cdot \dfrac{1}{p_j^{t} \cdot \gamma^{t}} \right) \leq 0$
\item With probability $ p_k^t$, we have $\bA_{j,t}^t = 0$ and $\bA_{k,t}^t = 1$. Thus we get  
\begin{align*}
d(x_i, x_j) \left(\frac{1}{t} \right) \hat{\bz}^{t-1}_{j,b} \cdot d(x_i, x_k) \left(\frac{1}{t}- \frac{1}{t} \cdot \dfrac{1}{p_k^{t} \cdot \gamma^{t}}\right) \hat{\bz}^{t-1}_{k,b} 
\end{align*}
since only $\left(\frac{1}{t} - \frac{1}{t} \cdot \dfrac{1}{p_k^{t} \cdot \gamma^{t}} \right) \leq 0$

\item With probability $(1-p_j^t-p_k^t)$, we have $\bA_{j,t}^t=0$ and $\bK_{t}^t = 1$. Thus we get 
\begin{align*}
d(x_i, x_j) \left(\frac{1}{t}\right) \hat{\bz}^{t-1}_{j,b} \cdot d(x_i, x_k) \left(\frac{1}{t}\right) \hat{\bz}^{t-1}_{k,b} \geq 0
\end{align*}
since all terms are non-negative.
\end{itemize}

Now we consider the case where  $j\in\bS^t$ and $k\notin\bS^t$. We have that $\hat{\bz}^t_{j,b}=\hat{z}^{t-1}_{j,b}$, while  $R_k^{t-1} = 0$. Note the case where $j\notin\bS^t$ and $k\in\bS^t$ is identical. 
We have the following sub-cases:
\begin{itemize}
\item With probability $0$, we have $\bA_{j,t}^t = 1$ and $\bA_{k,t}^t = 1$. The value would be
\begin{align*}
d(x_i, x_j) \left(\frac{1}{t} - \frac{1}{t} \cdot \dfrac{1}{p_j^{t} \cdot \gamma^{t}} \right) \hat{\bz}^{t-1}_{j,b} \cdot d(x_i, x_k) \left(\frac{1}{t} - \frac{1}{t} \cdot \dfrac{1}{p_k^{t} \cdot \gamma^{t}} \right)g_{k,b}^t \geq 0
\end{align*}
since only $\left(\frac{1}{t} - \frac{1}{t} \cdot \dfrac{1}{p_j^{t} \cdot \gamma^{t}} \right)$ and $\left(\frac{1}{t} - \frac{1}{t} \cdot \dfrac{1}{p_k^{t} \cdot \gamma^{t}} \right)$ are less than zero, hence cancel making the expression is positive.

\item With probability $p_j^t$, we have $\bA_{j,t}^t = 1$ and $\bA_{k,t}^t = 0$. Thus we get  
\begin{align*}
d(x_i, x_j) \left(\frac{1}{t} - \frac{1}{t} \cdot \dfrac{1}{p_j^{t} \cdot \gamma^{t}} \right) \hat{\bz}^{t-1}_{j,b} \cdot d(x_i,x_k) \cdot \left ( \frac{1}{1 - p_k^t} \left(\frac{t-1}{t} \cdot \hat{z}_{k,b}^{t-1} + \frac{1}{t} \cdot g_{k,b}^t - p_k^t \cdot g_{k,b}^t \right) -\frac{t-1}{t}\hat{z}_{k,b}^{t-1} \right ) \leq 0
\end{align*}
since $\frac{1}{1 - p_k^t}\frac{t-1}{t} \cdot \hat{z}_{k,b}^{t-1} - \frac{t-1}{t} \cdot \hat{z}_{k,b}^{t-1}\geq 0$, $\frac{1}{t} \cdot g_{k,b}^t - p_k^t \cdot g_{k,b}^t\geq 0$, and $\left(\frac{1}{t} - \frac{1}{t} \cdot \dfrac{1}{p_j^{t} \cdot \gamma^{t}} \right) \leq 0$
\item With probability $p_k^t$, we have $\bA_{j,t}^t = 0$ and $\bA_{k,t}^t = 1$. Thus we get  
\begin{align*}
d(x_i, x_j) \left(\frac{1}{t} \right) \hat{\bz}^{t-1}_{j,b} \cdot  d(x_i, x_k) \left(1 - \frac{1}{t} \cdot \frac{1}{p_k^{t}\gamma^t} \right) \cdot g_{k,b}^t \leq 0
\end{align*}
since only $\left(\frac{1}{t} - \frac{1}{t} \cdot \dfrac{1}{p_k^{t} \cdot \gamma^{t}} \right) \leq 0$
\item With probability $(1-p_j^t-p_k^t)$, we have $\bA_{j,t}^t=0$ and $\bK_{t}^t = 1$. Thus we get 
\begin{align*}
d(x_i, x_j) \left(\frac{1}{t}\right) \hat{\bz}^{t-1}_{j,b}  \cdot d(x_i,x_k) \cdot \left ( \frac{1}{1 - p_k^t} \left(\frac{t-1}{t} \cdot \hat{z}_{k,b}^{t-1} + \frac{1}{t} \cdot g_{k,b}^t - p_k^t \cdot g_{k,b}^t \right) -\frac{t-1}{t}\hat{z}_{k,b}^{t-1} \right ) \geq 0
\end{align*}
since all terms are non-negative. Specifically $\frac{1}{1 - p_k^t}\frac{t-1}{t} \cdot \hat{z}_{k,b}^{t-1} - \frac{t-1}{t} \cdot \hat{z}_{k,b}^{t-1}\geq 0$ and $\frac{1}{t} \cdot g_{k,b}^t - p_k^t \cdot g_{k,b}^t\geq 0$
\end{itemize}

Lastly, we consider the case where  $j\notin\bS^t$ and $k\notin\bS^t$, so $R_j^{t-1} = 0$ and $R_k^{t-1} = 0$. We have the following 3 sub-cases:
\begin{itemize}
\item With probability $0$, we have $\bA_{j,t}^t = 1$ and $\bA_{k,t}^t = 1$. The value would be 
\begin{align*}
d(x_i, x_j) \left(1 - \frac{1}{t} \cdot \frac{1}{p_j^{t} \gamma^t} \right) \cdot g_{j,b}^t \cdot d(x_i,x_k)  \left(1 - \frac{1}{t} \cdot \frac{1}{p_k^{t} \gamma^t} \right) \cdot g_{k,b}^t \geq 0
\end{align*}
since only $\left(\frac{1}{t} - \frac{1}{t} \cdot \dfrac{1}{p_j^{t} \cdot \gamma^{t}} \right)$ and $\left(\frac{1}{t} - \frac{1}{t} \cdot \dfrac{1}{p_k^{t} \cdot \gamma^{t}} \right)$ are less than zero, hence cancel making the expression is positive.

\item With probability $p_j^t$, we have $\bA_{j,t}^t = 1$ and $\bA_{k,t}^t = 0$. Thus we get  
\begin{align*}
d(x_i, x_j) \left(1 - \frac{1}{t} \cdot \frac{1}{p_j^{t}\gamma^t} \right) \cdot g_{j,b}^t \cdot d(x_i,x_k) \cdot \left ( \frac{1}{1 - p_k^{t}} \left(\frac{t-1}{t} \cdot \hat{z}_{k,b}^{t-1} + \frac{1}{t} \cdot g_{k,b}^t - p_k^t \cdot g_{k,b}^t \right) -\frac{t-1}{t}\hat{z}_{k,b}^{t-1} \right ) \leq 0
\end{align*}
since $\frac{1}{1 - p_k^t}\frac{t-1}{t} \cdot \hat{z}_{k,b}^{t-1} - \frac{t-1}{t} \cdot \hat{z}_{k,b}^{t-1}\geq 0$, $\frac{1}{t} \cdot g_{k,b}^t - p_k^t \cdot g_{k,b}^t\geq 0$, and only $\left(\frac{1}{t} - \frac{1}{t} \cdot \dfrac{1}{p_j^{t} \cdot \gamma^{t}} \right) \leq 0$
\item With probability $ p_k^t$, we have $\bA_{j,t}^t = 0$ and $\bA_{k,t}^t = 1$. Thus we get  
\begin{align*}
d(x_i,x_j) \cdot \left ( \frac{1}{1 - p_j^t} \left(\frac{t-1}{t} \cdot \hat{z}_{j,b}^{t-1} + \frac{1}{t} \cdot g_{j,b}^t - p_j^t \cdot g_{j,b}^t \right) -\frac{t-1}{t}\hat{z}_{j,b}^{t-1} \right ) \cdot  d(x_i, x_k) \left(1 - \frac{1}{t} \cdot \frac{1}{p_k^{t}\gamma^t} \right) \cdot g_{k,b}^t \leq 0
\end{align*}
since $\frac{1}{1 - p_j^t}\frac{t-1}{t} \cdot \hat{z}_{j,b}^{t-1} - \frac{t-1}{t} \cdot \hat{z}_{j,b}^{t-1}\geq 0$, $\frac{1}{t} \cdot g_{j,b}^t - p_j^t \cdot g_{j,b}^t\geq 0$, and only $\left(\frac{1}{t} - \frac{1}{t} \cdot \dfrac{1}{p_k^{t} \cdot \gamma^{t}} \right) \leq 0$

\item With probability $(1-p_j^t-p_k^t)$, we have $\bA_{j,t}^t=0$ and $\bK_{t}^t = 1$. Thus we get 
\begin{align*}
&d(x_i,x_j) \cdot \left ( \frac{1}{1 - p_j^t} \left(\frac{t-1}{t} \cdot \hat{z}_{j,b}^{t-1} + \frac{1}{t} \cdot g_{j,b}^t - p_j^t \cdot g_{j,b}^t \right) -\frac{t-1}{t}\hat{z}_{j,b}^{t-1} \right ) \\
&\times d(x_i,x_k) \cdot \left ( \frac{1}{1 - p_k^t} \left(\frac{t-1}{t} \cdot \hat{z}_{k,b}^{t-1} + \frac{1}{t} \cdot g_{k,b}^t - p_k^t \cdot g_{k,b}^t \right) -\frac{t-1}{t}\hat{z}_{k,b}^{t-1} \right ) \geq 0
\end{align*}
as all terms reduce to being non-negative, which follows from the prior inequalities. 
\end{itemize}
We have shown that the events were independent, the probability of the random variables realizing to positive values is increased, hence the expectation would be larger. More formally, the signs computed above show that in Lemma \ref{lem:neg-corr}, the term $(\alpha_1\beta_2 +\alpha_2\beta_1 - \alpha_1\alpha_2 - \beta_1\beta_2)\leq 0$, and thus we can apply it alongside Lemma $\ref{lem:dynamic-martingale}$ to conclude that for $j,k \in [n] / i$ such that $j\neq k$ 
\begin{align*}
    &\Ex_{}\left[\left( \berr_{i,b}^{\ell}(j) - \frac{\ell-1}{\ell} \cdot \berr_{i,b}^{\ell-1}(j)\right) \left( \berr_{i,b}^{\ell}(k) - \frac{\ell-1}{\ell} \cdot \berr_{i,b}^{\ell-1}(k)\right) \right] \\ &\qquad \leq \Ex\left[\left( \berr_{i,b}^{\ell}(j) - \frac{\ell-1}{\ell} \cdot \berr_{i,b}^{\ell-1}(j)\right) \right] \Ex\left[\left( \berr_{i,b}^{\ell}(k) - \frac{\ell-1}{\ell} \cdot \berr_{i,b}^{\ell-1}(k)\right) \right] \leq 0
\end{align*}
}
\end{proof}

\begin{lemma}\label{lem:dynamic-bound-on-error}
Fix any $\ell > 1$ and any $i, j \in [n]$ with $w_j \leq 1/2$,
\begin{align*}
\Ex\left[\left( \berr_{i,b}^{\ell+1}(j) - \frac{\ell-1}{\ell} \cdot \berr_{i,b}^{\ell}(j) \right)^2 \right] &\leq 100 \cdot \sfD \cdot e^2 \cdot d(x_i, x_j)^2 \cdot \dfrac{1}{\ell^2 \cdot w_j \cdot \gamma^{\ell}}.
\end{align*}
\end{lemma}

\begin{proof}
The proof of the lemmas analogous to Lemma~\ref{lem:bound-on-error} in Section~\ref{sec:greedy-analysis}, (note that when $\bK_t = 1$, we use the geometric sample $(\by_t^*, \bp_t^*)$ which satisfies $\by_t^* = x_j$ with probability $p_j^t$),
\begin{align*}
    \Ex\left[ \left(\berr_{i,b}^{\ell+1}(j) - \frac{\ell-1}{\ell} \cdot \err_{i,b}^{\ell}(j) \right)^2 \right] \leq \dfrac{d(x_i,x_j)^2}{\ell^2 \cdot p_j^{\ell} \cdot \gamma^{\ell}} + \ind\left\{ j \notin S^{\ell-1} \right\} \cdot \frac{30 \cdot d(x_i, x_j)^2}{\ell^2 \cdot \min\{ p_{j}^{\ell},1/\ell\} \cdot \gamma^{\ell}}.
\end{align*}
and recall that we also have $p_i^{\ell} \geq w_i / \sfD$.
As in Lemma~\ref{lem:bound-on-error}, when we average over times up to $\ell-1$, we invoke Lemma~\ref{lemma:dynamic-prob-activation} to upper bound 
\begin{align*}
    \Prx\left[ j \notin \bS^{\ell-1} \right] \leq \frac{\sfD \cdot e^2}{w_i \cdot (\ell-1)},
\end{align*}
and this gives us the final bound.
\ignore{

where we go through the two cases of $j \in S^{\ell-1}$ or not. In the case $j \in S^{\ell-1}$ and $j$ has already been activated before, we have that $\hat{\bz}^{\ell}_{j,b} = \hat{z}^{\ell-1}_{j,b}$. Notice that if $\bA_{j,\ell} \cdot \bK_{\ell} \neq 1$, then
\[ \left(\berr_{i,b}^{\ell+1}(j) - \frac{\ell-1}{\ell} \cdot \err_{i,b}^{\ell}(j)\right)^2 \leq \frac{d(x_i, x_j)^2}{\ell^2}, \]
and $\bA_{j,\ell} \cdot \bK_{\ell} = 1$ occurs with probability $\gamma^t \cdot p_{j}^{\ell}$, since $\tilde{\bp}_{j}^{\ell} = p_{j}^{\ell}$. In this case, 
\begin{align*}
\left(\berr_{i,b}^{\ell+1}(j) - \frac{\ell-1}{\ell} \cdot \err_{i,b}^{\ell}(j)\right)^2 \leq \dfrac{d(x_i, x_j)^2}{\ell^2 \cdot p_j^{\ell} \cdot \gamma^{\ell}},
\end{align*}
When $j \notin S^{\ell-1}$. We also have:
\begin{itemize}
\item With probability $p_{j}^{\ell} \cdot \gamma^{\ell}$, both $\bA_{j,\ell}^{\ell}$ and $\bK_{\ell}^{\ell}$ are both set to $1$, and $\hat{\bz}_{j,b}^{\ell}$ is set to $g_{j,b}^{\ell}$. So,
\begin{align*}
\left( \berr_{i,b}^{\ell+1}(j) - \frac{\ell-1}{\ell} \cdot \err_{i,b}^{\ell}(j) \right)^2 &= d(x_i, x_j)^2 \cdot \left( \left(1 - \frac{1}{\ell \cdot p_j^{\ell} \cdot \gamma^{\ell}} \right) g_{j,b}^{\ell} - \frac{\ell-1}{\ell} \cdot \hat{z}_{j,b}^{\ell-1} \right)^2 \\
	&\leq d(x_i, x_j)^2 \cdot \left( \frac{3}{\ell \cdot p_j^{\ell} \cdot \gamma^{\ell}} \right)^2,
\end{align*}
where we used the fact $p_j^{\ell} \leq 1/\ell$, as well as the fact $g_{j,b}^{\ell}, \hat{z}_{j,b}^{\ell-1} \in [0,1]$.

\item With probability $w_j^{\ell} \cdot (1 - \gamma^{\ell})$, $\bA_{j,\ell}^{\ell}$ is set to $1$, but $\bK_{\ell}^{\ell}$ is set to $0$. In this case, $\hat{\bz}_{j,b}^{\ell}$ is set to $g_{j,b}^{\ell}$, and we obtain
\begin{align*}
\left( \berr_{i,b}^{\ell+1}(j) - \frac{\ell-1}{\ell} \cdot \err_{i,b}^{\ell}(j)\right)^2 &= d(x_i, x_j)^2 \cdot \left( g_{j,b}^{\ell} - \frac{\ell-1}{\ell} \cdot \hat{z}_{j,b}^{\ell-1} \right)^2 \leq 4 \cdot d(x_i, x_j)^2
\end{align*}

\item With probability $\gamma^{\ell} \cdot (1 - p_j^{\ell})$, we set $\bA_{j,\ell}^{\ell}$ to $0$, despite $\bK_{\ell}^{\ell}=1$, in which case $\hat{\bz}_{j,b}^{\ell}$ is updated according to the second case in Definition~\ref{def:dynamic-fictitious}. Thus, we obtain
\begin{align*}
\left( \berr_{i,b}^{\ell+1}(j) - \frac{\ell-1}{\ell} \cdot \err_{i,b}^{\ell}(j)\right)^2 &= d(x_i, x_j)^2 \cdot \left( \frac{p_j^{\ell}}{1 - p_{j}^{\ell}} \cdot \frac{\ell-1}{\ell} \cdot \hat{z}^{\ell-1}_{j,b} + \frac{1}{\ell} \cdot g_{j,b}^{\ell} - p_j^{\ell} \cdot g_{j,b}^{\ell} \right)^2 \\
		&\leq d(x_i, x_j)^2 \cdot \left( \dfrac{4}{\ell} \right)^2,
\end{align*}
where we use the fact $p_j^{\ell} \leq \min\{ 1/\ell, 1/2\}$ to obtain our desired bound.

\item With probability $(1 - \gamma^{\ell})(1 - w_j^{\ell})$, we have $\bA_{j,\ell}^{\ell}$ and $\bK_{\ell}^{\ell}$ both set to $0$, in which case $\hat{\bz}_{j,b}^{\ell}$ is updated according to the third case in Definition~\ref{def:dynamic-fictitious}. Thus, we obtain
\begin{align*}
\left( \berr_{i,b}^{\ell+1}(j) - \frac{\ell-1}{\ell} \cdot \err_{i,b}^{\ell}(j)\right)^2 &= d(x_i, x_j)^2 \cdot \left( \frac{w_j^{\ell}}{1 - w_{j}^{\ell}} \cdot \frac{\ell-1}{\ell} \cdot \hat{z}^{\ell-1}_{j,b} + \frac{1}{\ell} \cdot g_{j,b}^{\ell} - w_j^{\ell} \cdot g_{j,b}^{\ell} \right)^2 \\
		&\leq d(x_i, x_j)^2 \cdot \left( \dfrac{4}{\ell} \right)^2,
\end{align*}
where we use the fact $w_j^{\ell} \leq \min\{ 1/\ell, 1/2\}$ to obtain our desired bound.
\end{itemize}
Once we combine the above four cases (while multiplying by their respective probabilities), we have that
\[ \frac{1}{\ell^2 \cdot p_j^{\ell} \cdot \gamma^{\ell}} \geq \frac{1}{\ell \cdot \gamma^{\ell}},  \]
which is larger than $w_j^{\ell}$, $p_j^{\ell}$ as well as $1/\ell^2$, which means that the first term dominates the calculation. In particular, we have 
\begin{align*}
\Ex\left[ \ind\{ j \notin S^{\ell-1} \} \left( \berr_{i,b}^{\ell+1}(j) - \frac{\ell-1}{\ell} \cdot \err_{i,b}^{\ell}(j) \right)^2 \right] &\leq \ind\{ j \notin S^{\ell-1} \} \cdot \dfrac{30 \cdot d(x_i, x_j)^2}{\ell^2 \cdot p_j^{\ell} \cdot \gamma^{\ell}}.
\end{align*}
The above computations of the expectations averaged over the randomness solely over step $\ell$. We now take the expectation over the randomness up to time $\ell-1$, which allows us to exploit the fact $j \notin \bS^{\ell-1}$ with probability which decays for large $\ell$. In particular, whenever $\ell-1 \geq 1/p_j$, we can apply Lemma \ref{lemma:dynamic-prob-activation} to compute $\Prx[j \notin \bS^{\ell-1]}$. Applying this to the expectation over the time steps up to $\ell-1$, we have
\begin{align*}
\Ex\left[ \left(\berr_{i,b}^{\ell+1}(j) - \frac{\ell-1}{\ell} \cdot \berr_{i,b}^{\ell}(j) \right)^{2}\right] &\leq \frac{d(x_i,x_j)^2}{\ell^2 \cdot p_j \cdot \gamma^{\ell}} + \Prx\left[ j \notin \bS^{\ell-1} \right] \cdot \dfrac{30 \cdot d(x_i, x_j)^2}{\ell^2 \cdot p_j^{\ell} \cdot \gamma^{\ell}},\end{align*}
and this value is at most 
\[ \frac{31 \cdot d(x_i, x_j)^2}{\ell^2 \cdot p_j \cdot \gamma^{\ell}} \]
whenever $p_j^{\ell} = \min\{ p_i/(p_i^*\ell), p_j \} = p_j$. Otherwise, we have $p_j^{\ell} = p_i/(p_i^*\ell)$ and this occurs when $\ell > 1$ since $p_j \leq 1/2$ and $\eps'< 1$. The above expectation becomes at most
\begin{align*}
\frac{d(x_i, x_j)^2}{\ell \cdot p_j \cdot \gamma^{\ell}} + \dfrac{e^4}{(\ell - 1) \cdot \min(w_j, p_j)} \cdot \frac{30 (1+\eps') \cdot d(x_i, x_j)^2}{\ell \cdot \gamma^{\ell}} \leq \frac{100e^4 \cdot d(x_i, x_j)^2}{\ell^2 \cdot \min(w_j, p_j) \cdot \gamma^{\ell}}.
\end{align*}}

\end{proof}

We complete the proof as in Lemma~\ref{lem:error-bound}, where we have:
\begin{align*}
    \Ex\left[\sum_{i \notin \bS^t} \left| \bc_{i,b}^{t} - \hat{\bc}_{i,b}^t\right| \right] &\leq \frac{1}{t} \sum_{i=1}^n \left(100 \sfD \cdot e^2 \sum_{\ell=1}^t \sum_{j=1}^n \dfrac{d(x_i, x_j)^2}{w_j \cdot \gamma^{\ell}} \right)^{1/2} \\
    &\leq 10 e\sqrt{\sfD} \left( \frac{2\sqrt{\gamma}}{t} + \frac{1}{\sqrt{\gamma}}\right) \sum_{i=1}^n \left(\dfrac{d(x_i, x_j)^2}{w_j} \right)^{1/2}.
\end{align*}
\ignore{

Combining Lemma~\ref{lem:reduce-to-error} and~\ref{lem:bound-on-error}, we obtain:
\begin{align*}
\Ex\left[ \sum_{i \notin \bS^t} \left| \bc_{i,b}^{t} - \hat{\bc}_{i,b}^t \right| \right] &\leq  \frac{1}{t} \sum_{i=1}^n \left(100e^4 \sum_{\ell=1}^{t} \sum_{j=1}^n \dfrac{d(x_i, x_j)^2}{\min(w_j, p_j) \cdot \gamma^{\ell}}  \right)^{1/2}  \leq \frac{10e^2}{t} \sum_{i=1}^n \left( \sum_{j=1}^n \frac{d(x_i, x_j)^2}{\min(w_j, p_j)} \right)^{1/2} \left(\sum_{\ell=1}^t \frac{1}{\gamma^{\ell}} \right)^{1/2}
\end{align*}
and we have,
\begin{align*}
\left( \sum_{\ell=1}^{t} \frac{1}{\gamma^{\ell}}\right)^{1/2} \leq \left( 1 + \gamma + \frac{1}{\gamma} \sum_{\ell=1}^t \ell \right)^{1/2} \leq 2\sqrt{\gamma}  + \frac{t}{\sqrt{\gamma}} ,
\end{align*}
which means that plugging in to the above bound gives us our desired bound of
\begin{align*}
10e^2 \left(\frac{2\sqrt{\gamma}}{t} + \frac{1}{\sqrt{\gamma}} \right) \sum_{i=1}^n \left( \sum_{j=1}^n \frac{d(x_i, x_j)^2}{\min(w_j, p_j)} \right)^{1/2}.
\end{align*}

}

%% file: algorithmic-prelims.tex

\section{Algorithmic Preliminaries}


\subsection{For Computing Metric-Compatible Weights}

Our first ingredients are the sketching tools to compute metric compatible weights. The idea will be to apply the turnstile streaming algorithm for cascaded norms~\cite{JW09}, and is useful for obtaining a set of weights $w = (w_1, \dots, w_n)$ for instances which are subsets of $\ell_p$ for $p \in [1, 2]$. When the input $X = \{ x_1,\dots, x_n \}$ consist of vectors in $\ell_p^d$ for $p \in [1, 2]$, we will use linear sketching in order to recover weights $w_1,\dots, w_n$ which satisfy
\begin{align}
\frac{\sum_{j=1}^n d_{\ell_p}(x_i, x_j)}{ \sum_{k=1}^n \sum_{j=1}^n d_{\ell_p}(x_k, x_j)} \leq w_i \leq (1+\eps) \cdot \frac{\sum_{j=1}^n d_{\ell_p}(x_i, x_j)}{ \sum_{k=1}^n \sum_{j=1}^n d_{\ell_p}(x_k, x_j)}.\label{eq:weight-approx}
\end{align}
We first state the lemma describing the linear sketch for $\ell_k(\ell_p)$, and then show that if we find a setting of weights satisfying~(\ref{eq:weight-approx}), these are $O(1)$-compatible weights (which follows from an argument in~\cite{FK01}). 

\renewcommand{\vec}{\mathrm{vec}}

\begin{theorem}[Sketching $\ell_k(\ell_p)$~\cite{JW09}]\label{lem:cascaded}
For any fixed $n, d \in \N$, any $1 \leq k \leq p \leq 2$, and any $\eps, \delta\in(0,1)$, let $s = \poly(\log(nd/\delta)/\eps)$. There exists a distribution $\calD$ over $s \times nd$ matrices such that, for any $n \times d$ matrix $A$:
\begin{itemize}
\item Draw $\bS \sim \calD$ and let $\sk(A) = \bS \cdot \vec(A)$, where $\vec(A) \in \R^{nd}$ is the flattened matrix.
\item There is a deterministic algorithm which receives as input $\sk(A)$ and outputs a number $\boldeta \in \R_{\geq 0}$ which satisfies
\[ \|A\|_{\ell_k(\ell_p)} \eqdef \left( \sum_{i=1}^n \|A_{i}\|_p^k \right)^{1/k} \leq \boldeta \leq (1+\eps) \cdot \| A\|_{\ell_k(\ell_p)}, \]
where $A_i \in \R^{d}$ is the $i$-th row of $A$, with probability at least $1 - \delta$ over $\bS \sim \calD$. 
\end{itemize}
Furthermore, the draw $\bS \sim \calD$ may be generated using $\poly(s)$ random bits. 
\end{theorem}

\ignore{We state the following consequence of the above lemma which we use in our algorithms to obtain a set of metric compatible weights.
\begin{lemma}\label{lem:cascaded}
For any fixed $n, d \in \N$, $p \in [1, 2]$ and $\delta \in (0, 1)$, let $s = \polylog(nd/\delta))$. There exists a distribution $\calD$ supported on $s \times nd$ matrices which satisfies the following guarantees:
\begin{itemize}
\item For any $x_1, \dots, x_n \in \R^d$, let $x^{\circ} \in \R^{nd}$ denote the vector which stacks all $n$ vectors $x_1,\dots, x_n$, and for a draw $\bS \sim \calD$, let $\sk(x^{\circ}) = \bS x^{\circ}$.
\item There is an algorithm which receives as input $\sk(x^{\circ})$, a vector $x_i$ and the index $i$, as well as the matrix $\bS$, and outputs a number $\bw_i$ which satisfies 
\[ \frac{\sum_{j=1}^n d_{\ell_p}(x_i, x_j)}{ \sum_{k=1}^n \sum_{j=1}^n d_{\ell_p}(x_k, x_j)} \leq \bw_i \leq 4 \cdot \frac{\sum_{j=1}^n d_{\ell_p}(x_i, x_j)}{ \sum_{k=1}^n \sum_{j=1}^n d_{\ell_p}(x_k, x_j)} \]
with probability at least $1- \delta$ over the draw of $\bS \sim \calD$.
\end{itemize}
Furthermore, a draw of $\bS \sim \calD$ may be generated using $\polylog(nd/\delta)$ bits of randomness. 
\end{lemma}

\begin{proof}
We will make use of the streaming algorithm for $\ell_k(\ell_p)$ matrix norms from~\cite{JW09} which uses $\polylog(nd/\delta)$ space and can compute a $2$-approximation to the $\ell_k(\ell_p)$ norm of an $n \times d$ matrix in the turnstile streaming model, whenever $1 \leq k \leq p \leq 2$.\footnote{In fact, it can compute an arbitrary $(1+\eps)$-approximation, by paying polynomially in $1/\eps$.} In particular, recall that given an $n \times d$ matrix $A$, the $\ell_k(\ell_p)$-norm of the matrix is given by 
\[ \left( \sum_{i=1}^n \| A_{i, \cdot}\|_p^k \right)^{1/k},\]
where $A_{i, \cdot} \in \R^d$ is the $i$-th row of the matrix. We set the parameters $k = 1$ and $p$ as indicated by the lemma statement. The streaming algorithm of~\cite{JW09} is implemented via a linear sketch, given by a matrix multiplication of an $nd$ vector of its rows stacked. The protocol proceeds as follows, and uses the $\ell_1(\ell_p)$-norm streaming algorithm twice (to obtain approximations of the numerator and denominator for the weight). In the first part, we obtain a factor-$2$ approximation to $\sum_{j=1}^n d_{\ell_p}(x_i, x_j)$.
\begin{itemize}
\item \textbf{Encoding}. The $n$ points $x_1,\dots, x_n \in \R^d$ are stacked into an $n \times d$ matrix. The first portion of the sketch utilizes the $\ell_1(\ell_p)$-norm streaming algorithm for this matrix to obtain a sketch $\sk_1$ of this matrix.
\item \textbf{Decoding}. Given the vector $\sk_1$, the vector $x_i$, the index $i \in [n]$, and access to the matrix $\bS$ (which may be reconstructed given the $\polylog(nd/\delta)$ random bits used by the streaming algorithm), we consider subtracting $x_i$ from every row of the sketched matrix (using the turnstile nature of the streaming algorithm). Then, we run the estimation procedure for the streaming algorithm, which obtains a $2$-approximation to $\sum_{j=1}^n d_{\ell_p}(x_i, x_j)$. 
\end{itemize}
The second part of the sketch obtains a sketch for $\sum_{k=1}^n \sum_{j=1}^n d_{\ell_p}(x_k, x_j)$ by considering the streaming algorithm which implicitly consider the $n^2 \times d$ matrix, where each row of the matrix is associated with a pair $(k,j)$, and we apply the linear transformation on $x^{\circ}$ which sets the row $(k,j)$ to $x_k - x_j$. This part of the sketch produces the second part of the sketch $\sk_2$, which gives a $2$-approximation to $\sum_{k=1}^n \sum_{j=1}^n d_{\ell_p}(x_k, x_j)$ (this second part is a trivial decoding procedure, as it does not depend on $x_i$ or $i$). Taking the product gives the desired $4$-approximation.
\end{proof}}

\begin{lemma}[Lemma 7 in~\cite{FK01} and Lemma~3.5 in~\cite{CJK23}]\label{lem:weight-to-compatible}
Let $X = \{ x_1,\dots, x_n\}$ be any set of points in a metric $d$, let $w_1,\dots, w_n \in [0, 1]$ satisfy 
\[ \frac{\sum_{j=1}^n d(x_i, x_j)}{\sum_{k=1}^n \sum_{j=1}^n d(x_k, x_j)} \leq w_i \leq \sfD \cdot \frac{\sum_{j=1}^n d(x_i, x_j)}{\sum_{k=1}^n \sum_{j=1}^n d(x_k, x_j)}. \]
Then, every $w_i \geq 1/(2n)$ and any $c$ in the metric satisfies:
\begin{align*}
\dfrac{d(c, x_j)}{w_j'} \leq 4 \sum_{k=1}^n d(c, x_{k}).
\end{align*}
Letting $w_1',\dots, w_n' \in (0, 1/2]$ be $w_j' = w_j / 2$, we satisfy $\|w'\|_1 \leq \sfD/2$ and $w'$ are $8$-compatible.
\end{lemma}

\begin{proof}
We let $\ell \in [n]$ denote the minimizer of $\sum_{k=1}^n d(x_{\ell}, x_{k})$, and we note that
\begin{align}
\sum_{k=1}^n \sum_{t=1}^n d(x_{k}, x_{t}) \leq \sum_{k=1}^n \sum_{t=1}^n\left( d(x_{k},  x_{\ell}) + d(x_{\ell}, x_{t}) \right) \leq 2n  \sum_{k=1}^n d(x_{k}, x_{\ell}), \label{eq:he}
\end{align}
and furthermore, every $i \in [n]$ satisfies
\[ w_i \geq \frac{\sum_{k=1}^n d(x_k, x_i)}{\sum_{k=1}^n \sum_{t=1}^n d(x_k, x_t)} \geq \dfrac{\sum_{k=1}^n d(x_k, x_{\ell})}{\sum_{k=1}^n \sum_{t=1}^n d(x_k, x_t)} \geq \frac{1}{2n}. \]
Therefore, for any $c$, 
\begin{align}
\frac{d(c, x_j)}{w_j'} = 2 \cdot \frac{d(c, x_j)}{w_j} &\leq 2 \cdot \frac{d(c, x_j)}{\sum_{k=1}^n d(x_j, x_k)} \left( \sum_{k=1}^n \sum_{t=1}^n d(x_{k}, x_{t}) \right) \nonumber \\
		&\leq 2 \cdot \left( \dfrac{\Ex_{\bk \sim [n]}\left[ d(c, x_{\bk}) \right]}{\sum_{k=1}^n d(x_j, x_k)} + \frac{1}{n} \right) \left( \sum_{k=1}^n \sum_{t=1}^n d(x_k, x_t)\right), \label{eq:hehs}
\end{align}
where we have used the triangle inequality to upper bound
\[ d(c, x_j) \leq \Ex_{\bk\sim[n]} \left[d(c, x_{\bk}) \right] + \Ex_{\bk \sim [n]}\left[ d(x_{\bk}, x_j) \right]. \]
By the choice of $\ell$ and (\ref{eq:he}), we can upper bound 
\begin{align*}
\dfrac{\Ex_{\bk \sim [n]}\left[ d(c, x_{\bk}) \right]}{\sum_{k=1}^n d(x_j, x_k)} \leq \frac{1}{n} \cdot \dfrac{\sum_{k=1}^n d(c, x_k)}{\sum_{k=1}^n d(x_{\ell}, x_k)} \leq 2 \cdot \frac{\sum_{k=1}^n d(c, x_k) }{\sum_{k=1}^n \sum_{t=1}^n d(x_k, x_t)}.
\end{align*}
Thus, the right-hand side of (\ref{eq:hehs}) is at most
\begin{align*}
4 \sum_{k=1}^n d(c, x_k) + \frac{2}{n} \sum_{k=1}^n \sum_{t=1}^n d(x_k, x_t) &\leq 4 \sum_{k=1}^n d(c, x_k) + \frac{2}{n} \sum_{k=1}^n \sum_{t=1}^n \left( d(x_k, c) + d(x_t, c) \right) \\
		&\leq 8 \sum_{k=1}^n d(c, x_k). 
\end{align*}
Finally, the fact $\|w'\|_1 \leq \sfD / 2$ follows from the fact $\|w\|_1 \leq \sfD$. 
\end{proof}

\ignore{\subsection{Geometric Sampling}

In certain contexts, one downside of Lemma~\ref{lem:cascaded} may be that the sketch requires one to know the size of the dataset $n$, and each point $x_i$ in the dataset should know its index $i$. This is because the sketch relies on being able to (implicitly) generate the vector in $\R^{nd}$ which stacks the points $x_1,\dots, x_n$. As an alternative when points would not know their indices, we may also employ the ``geometric sampling'' primitive, which was recently introduced in~\cite{CJK23}. 

\begin{lemma}[Lemma 4.2 in~\cite{CJK23}]
For any $\eps \in (0, 1)$, $p \geq 1$, and integers $\Delta, d \geq 1$, there is a randomized streaming algorithm (in fact, a linear sketch) which receives a dataset $X = \{ x_1,\dots, x_n \} \subset [\Delta]^d$ as a dynamic geometric stream which uses space $\poly(d \log \Delta / \eps)$ and has the following guarantees:
\begin{itemize}
\item After processing the stream, it outputs $\bz^* \in X \cup \{\bot\}$ and a number $\bp^* \in [0, 1]$ such that, with probability at least $0.99$, it does not output $\bot$. 
\item Furthermore, conditioned on it not outputting $\bz^* = \bot$, for any $x_i \in X$, 
\begin{align*}
\Prx\left[ \bz^* = x_i \right] \geq \frac{1}{\poly(d \log \Delta)} \cdot \frac{\sum_{j=1}^n d_{\ell_p}(x_i, x_j)}{\sum_{k=1}^n \sum_{\ell=1}^n d_{\ell_p}(x_k, x_{\ell})},
\end{align*}
and $\bz^* = x_i$ implies $\Prx\left[ \bz^* = x_i \right] \leq \bp^* \leq (1+\eps) \Prx\left[ \bz^* = x_i\right]$ except with probability $1 - 1/\poly(\Delta^d)$.
\end{itemize}
\end{lemma}

Furthermore, given that the streaming algorithms is dynamic, it is straight-forward to apply the above geometric sampling primitive in order to obtain $\poly(d\log \Delta)$-compatible weights in a fashion similar to that of~Lemma~\ref{lem:cascaded}, without knowledge of the indices $i$ of the point $x_i$. Since the randomized streaming algorithm is a linear sketch, we state the subsequent lemma to draw parallel with that of Lemma~\ref{lem:cascaded}.

\begin{lemma}
For any fixed $\Delta, d \in \N$, $p \geq 1$ and $\delta \in (0, 1)$, let $s = \poly(d \log \Delta \log(1/\delta))$. There exists a distribution $\calD$ supported on $s \times \Delta^d$ matrices which satisfy the following guarantees:
\begin{itemize}
\item For any $X = \{ x_1,\dots, x_n \} \subset [\Delta]^d$, let $x^{\circ} \in \R^{[\Delta]^d}$ denote the vector of counts of the (multi-) dataset $X \subset [\Delta]^d$, and for a draw of $\bS \sim \calD$, let $\sk(x^{\circ}) = \bS x^{\circ}$. 
\item There is an algorithm which receives as input $\sk(x^{\circ})$, a vector $x_i \in X$, as well as the matrix $\bS$, and outputs a number $\bw_i$ which satisfies
\begin{align*}
\frac{\sum_{j=1}^{n} d_{\ell_p}(x_i, x_j)}{\sum_{k=1}^n \sum_{j=1}^n d_{\ell_p}(x_k, x_j)} \leq \bw_i \leq \poly(d\log \Delta) \cdot \frac{\sum_{j=1}^{n} d_{\ell_p}(x_i, x_j)}{\sum_{k=1}^n \sum_{j=1}^n d_{\ell_p}(x_k, x_j)} 
\end{align*}
with probability at least $1 - \delta$ over the draw $\bS \sim \calD$. 
\end{itemize}
Furthermore, the matrix $\bS \sim \calD$ may be generated using $\poly(d \log \Delta \log(1/\delta))$ bits of randomness. 
\end{lemma}

\begin{proof}
...
\end{proof}}

\subsection{For Choosing the Best Seed}

The final algorithmic component will provide a means to determine the minimizing seed $\sigma$ to use for Theorem~\ref{thm:main-structural}. Note that Theorem~\ref{thm:main-structural} argues about the expected minimum over a family of cuts. Namely, for $\bA_1, \dots, \bA_n$, we let $\bm$ denote the number of indices $i \in [n]$ which contain a non-zero entry within the first $t_0$ entries. Then, Theorem~\ref{thm:main-structural} provides a family of at most $2^{\bm}$ cuts, one for each possible $\sigma \in \{0,1\}^{m}$, and the minimizing $\sigma$ provides an approximate max-cut. Here, we state a lemma which enables us to estimate the quality of the cut given local access. Algorithmically, we use this primitive in order to ``agree'' on a setting of $\sigma$. 

\begin{lemma}[Seed Check]\label{lem:seed-check}
For any $\eps, \delta \in (0, 1)$, fix a metric $(X = \{ x_1, \dots, x_n\}, d)$ and let $w = (w_1, \dots, w_n) \in \R_{\geq 0}^n$ be $\lambda$-metric compatible weights with $\lambda \geq 1$ and $\|w\|_1 = O(1)$. For a family of cut assignments $\calF$, let $\xi \geq \poly(\lambda \log(n|\calF| / \delta)/\eps)$. There exists a randomized algorithm that given  
\begin{enumerate}
    \item A random subset $\bC \subset X$, where $x_i \in \bC$ independently with probability $\min\{\xi \cdot w_i, 1\}$.
    \item For each $x_i \in \bC$, access to the sampling probability $\min\{ \xi \cdot w_i, 1\}$.
     \item A family $\calF$ of cut assignments $z \in \calF$ where $z \in \{(0,1), (0, 1)\}^{n}$, and query access to $z_i$ for all $x_i \in \bC$. 
\end{enumerate}
The algorithm outputs with probability $1-\delta$ over $\bC$, an $\hat{\bz}\in \calF$ such that 
\begin{align*}
f(\hat{\bz}) \leq  \min_{z \in \calF}  f(z) + \eps \sum_{i=1}^n \sum_{j=1}^n d(x_i, x_j),
\end{align*}
for $f$ in (\ref{eq:internal-f-def}).
\end{lemma}

\begin{proof}
The algorithm proceeds by iterating over $z \in \calF$, and using the sample in $\bC$ in order to estimate $f(z)$. Then, we will take a union bound over all $\calF$ and choose the minimum. For each $i \in [n]$, we let $\tilde{w}_i = \min \{ \xi \cdot w_i, 1\}$, $\bC_i$ denote the indicator random variable that $x_i \in \bC$, $B_{ij} = d(x_i, x_j) \cdot (z_{i,0} z_{j,0} + z_{i,1} z_{j,1})/2$, $T_i = \sum_{j=1}^n d(x_i, x_j)$ and $T = \sum_{i=1}^n T_i$. Notice then, that $f(z)$ is the sum of all $B_{ij}$. It suffices to show that for any $z \in \calF$, 
\begin{align*}
\left| \sum_{i=1}^n \sum_{j=1}^n \dfrac{\bC_i \cdot \bC_j}{\tilde{w}_i \cdot \tilde{w}_j} \cdot B_{ij}  - f(z) \right| \leq \eps T
\end{align*}
with probability $1 - \delta / |\calF|$. Standard concentration bounds (namely, Bernstein's inequality) do not directly apply because the summands $\bC_i \bC_j / (\tilde{w}_i \cdot \tilde{w}_j) \cdot B_{ij}$ are not all independent. Let $\bC_i' = \bC_i - \tilde{w}_i$, so that expanding the definitions and symmetry in $B_{ij}$ gives the upper bound
\begin{align}
\Prx\left[ \left| \sum_{i=1}^n \sum_{j=1}^n \dfrac{\bC_i \cdot \bC_j}{\tilde{w}_i \cdot \tilde{w}_j} \cdot B_{ij}  - f(z) \right| \geq \eps T\right] &\leq \Prx\left[ \left| \sum_{i=1}^n \sum_{j=1}^n \dfrac{\bC_i' \cdot \bC_j'}{\tilde{w}_i \cdot \tilde{w}_j } \cdot B_{ij} \right| \geq \frac{\eps T}{2} \right] \label{eq:prob-1} \\
&\qquad + \Prx\left[ \left| \sum_{i=1}^n \frac{\bC_i}{\tilde{w}_i} \sum_{j=1}^n B_{ij} - f(z) \right| \geq \frac{\eps T}{4}\right] \label{eq:prob-2}.
\end{align}
The probability in (\ref{eq:prob-2}) is simple to bound via Bernstein's inequality: first, we note that
\[ \sum_{i=1}^n \Ex\left[ \left( \frac{\bC_i' \cdot T_i}{\tilde{w}_i}\right)^2 \right] \leq \sum_{\substack{i\in[n] \\ \tilde{w_i} < 1}} \left(\frac{T_i^2}{\xi w_i} + T_i^2\right) \leq 2\sum_{i=1}^n \dfrac{T_i^2}{\xi w_i} \leq \frac{2\lambda}{\xi} \cdot T^2,\qquad\text{and}\qquad \max_{i:\tilde{w_i}<1} \frac{T_i}{\tilde{w_i}} \leq \frac{\lambda}{\xi} \cdot T,   \]
hence, Bernstein's inequality implies the probability in (\ref{eq:prob-2}) is at most
\begin{align*}
2\exp\left( - \dfrac{\eps^2 \cdot \xi}{32 \cdot \lambda \cdot (1 + \eps / 3)} \right) \leq \frac{\delta}{2|\calF|}
\end{align*}
whenever $\xi$ is a large enough constant times $\lambda \log(|\calF|/\delta) / \eps^2$. It then remains to upper bound (\ref{eq:prob-1}), which we do using the following argument. Note that we can always upper bound
\begin{align*}
\left| \sum_{i=1}^n \sum_{j=1}^n \frac{\bC_i' \cdot \bC_j'}{\tilde{w}_i \cdot \tilde{w}_j} \cdot B_{ij} \right| &\leq \left| \sum_{i \in \bC} \frac{1- \tilde{w}_i}{\tilde{w}_i} \sum_{j=1}^n \frac{\bC_j'}{\tilde{w}_j} \cdot B_{ij} \right| + \left| \sum_{i \notin \bC} \sum_{j=1}^n \frac{\bC_j'}{\tilde{w}_j} \cdot B_{ij}\right| \\
		&\leq \underbrace{|\bC| \cdot \max_{i} \frac{1}{\xi w_i} \left| \sum_{j=1}^n \frac{\bC_j'}{\tilde{w}_j} \cdot B_{ij} \right|}_{A} + \underbrace{\left| \sum_{i \notin \bC} \sum_{j=1}^n \frac{\bC_j'}{\tilde{w}_j} \cdot B_{ij}\right|}_{B}.
\end{align*}
Thus, whenever (\ref{eq:prob-1}) occurs, at least of the two summands, $A$ or $B$, is greater than $\eps T/ 4$. Suppose that (\ref{eq:prob-1}) occurs with probability at least $\delta / (2|\calF|)$, and that among those events, half of them have term $A$ being at least $\eps T/4$. Then, we note that since $|\bC|$ is a sum of independent random variables with expectation $O(\xi)$, a simple Chernoff bound shows that $|\bC|$ is larger than $O(\xi \log(|\calF|/\delta))$ except with probability $o(\delta / |\calF|)$, and therefore, with probability at least $\delta / (8|\calF|)$, 
\begin{align*}
\max_{i} \frac{1}{\xi w_i} \left| \sum_{j=1}^n \frac{\bC_j'}{\tilde{w}_j} \cdot B_{ij} \right| \geq \Omega\left(\frac{\eps}{\xi \log(|\calF|/\delta)}\right) \cdot T,
\end{align*} 
However, we always have
\begin{align*}
&\Prx\left[ \exists i \in [n] : \left| \sum_{j=1}^n \frac{\bC_j'}{\tilde{w}_j} \cdot B_{ij} \right| \geq \Omega\left(\frac{\eps}{\log(|\calF|/\delta)}\right) \cdot T \cdot w_i \right] \\
&\qquad\qquad \leq 2 n \cdot \exp\left( -\Omega(1) \cdot \dfrac{\eps^2 \cdot \xi}{\log^2(|\calF|/\delta) \cdot (\lambda^3 + \eps \lambda /3)} \right),
\end{align*}
which is smaller than $\delta / (8 |\calF|)$ whenever $\xi$ is a large enough constant factor of $\lambda^3 \log^3(|\calF|/\delta) \log n / \eps^2$. This would give a contradiction, so summand $B$ must be at least $\eps T/4$ with probability at least $\delta / (4|\calF|)$. This, too, gives a contradiction, since whenever this occurs, there must be some $i$ where the magnitude of $\sum_{j=1}^n \bC_j' / \tilde{w_j} \cdot B_{ij}$ exceeds $\eps T_i/4$, and this probability is similarly bounded by
\[ 2n \cdot \exp\left( - \Omega(1) \cdot \dfrac{\eps^2 \xi}{\lambda (1 + \eps/3)} \right).\]
\end{proof}

%% file: simple-greedy.tex

\section{A greedy extension fails}\label{sec:simple}

As discussed in Section \ref{sec:introduction}, the algorithm of~\cite{CJK23} builds on a prior work of random sampling for max-CSPs. One of the main result shows that, for any collection of dataset points $X = \{ x_1, \dots, x_n \}$ and distance function $d$, suppose we consider a random sample of $s = O(1/\eps^4)$ points $x_{\bi_1}, \dots, x_{\bi_s}$ from a distribution over $\bi$ where every $j=1, \dots, n$ satisfies
\begin{align*}
    w_{j} \eqdef \Prx_{\bi}\left[ \bi = j \right] \propto \sum_{\ell=1}^n d(x_j, x_{\ell}),
\end{align*}
Then, with probability at least $0.99$ over the draw of the sample, a (rescaled) max-cut of the sample gives a good approximation to the optimal cut on the entire dataset:
\begin{align*}
    \left| \max_{A \subset [s]} \sum_{\ell_1=1}^s \sum_{\ell_2=1}^s \dfrac{\ind\left\{ \begin{array}{c} \ell_1 \in A \\ 
                                                    \ell_2 \notin A \end{array} \right\}}{w_{\bi_{\ell_1}} \cdot w_{\bi_{\ell_2}}} \cdot d(x_{\bi_{\ell_1}}, x_{\bi_{\ell_2}}) - \max_{B \subset [n]} \sum_{i_1=1}^n \sum_{i_2 =1}^n \ind\left\{ \begin{array}{c} i_1 \in B \\ i_2 \notin B \end{array} \right\} \cdot d(x_{i_1}, x_{i_2}) \right| \leq \eps T,
\end{align*}
where $T = \sum_{i_1=1}^n \sum_{i_2=1}^n d(x_{i_1}, x_{i_2})$. The random sampling experiment above shows that the \textit{value} of the maximum cut can be approximated by finding the maximum cut of a small sample weighted based on their fraction of the total distance~\cite{CJK23}. However, it is generally not a good idea to naively extend the cut greedily. In particular, suppose $\bS = \{ x_{\bi_1},\dots, x_{\bi_s}\}$ is the sample, and let $\bA \subset [s]$ denote maximum weighted cut of the sample (i.e., that which maximizes the right-most term in the absolute value above). 

There is a natural \emph{greedy} extension which assigns points based on the weighted contribution from the random sample. Once the the maximum weighted cut of the sample, $\bA$ has been found, for each point $x_j\in X \setminus \bS$, compute the (scaled) contribution of $\bA$ and $\ol{\bA}$ within $\bS$, and assign $x_j$ to the side maximizing this contribution.

\begin{framed} \label{alg:naive-assign}
\textbf{Algorithm}
$\textsc{Naive-Assign}$.
This algorithm is given a dataset $X$ and random sample $\bS = \{ x_{\bi_1}, \dots, x_{\bi_s} \}$, which it then uses to compute a cut of $X$.

\textbf{Setup}. A metric $(X = \{ x_1,\dots, x_n\}, d)$, weights $w = (w_1, \dots ,w_n) \in (0, 1/2]^n$, and a random sample of $s = O(1/\eps^4)$ points $x_{\bi_1}, \dots, x_{\bi_s}$. 

The algorithm begins by enumerating over all cuts of $\bS$ and selecting the maximum, $\bA \subset [s]$ 
\[\bA = \argmax_{A \subset [s]} \sum_{\ell_1=1}^s \sum_{\ell_2=1}^s \dfrac{\ind\left\{ \begin{array}{c} \ell_1 \in A \\
\ell_2 \notin A \end{array} \right\}}{w_{\bi_{\ell_1}} \cdot w_{\bi_{\ell_2}}} \cdot d(x_{\bi_{\ell_1}}, x_{\bi_{\ell_2}}).\]
then computing for each point $x_j \in X \setminus \bS$ the expected contribution to the sampled cut
\begin{align*}
    c_{j,0} = \sum_{\ell_1=1}^s \dfrac{\ind\left\{\ell_1 \in \bA \right\}}{w_{\bi_{\ell_1}} } \cdot d(x_{\bi_{\ell_1}}, x_{j}) \qquad \text{and}\qquad 
    c_{j,1} = \sum_{\ell_1=1}^s \dfrac{\ind\left\{\ell_1 \notin \bA \right\}}{w_{\bi_{\ell_1}} } \cdot d(x_{\bi_{\ell_1}}, x_{j}). 
\end{align*}
If $c_{j,0} \leq c_{j,1}$, then include $j$ in the output cut $\bB$. 

\textbf{Output}: The cut $\bB \cup\bA$, which is now a cut of $X$. 
\end{framed}

Roughly speaking, the greedy assignment produced by this algorithm is highly sensitive to the initial sample and has no way of ``breaking symmetry.'' Consider the following example:
\begin{itemize}
    \item We consider a dataset $X$ which consists of a $n/2$ vectors randomly drawn from the unit sphere, and $n/2$ copies of the same unit vector $\bz$. 
    \item The true max-cut is the natural partition and has cut value around $\sqrt{2} \cdot (n/2)^2$. Considers a small sample $\bS$ of points, and let $\bA \subset \bS$ be the vectors from the sample which are all copies of $\bz$.
\end{itemize}
Using the greedy extension, every vector equal to $\bz$ will be correctly assigned to the correct side (i.e, that which agrees with $\bA$), since the contribution to $\bA$ would be zero. However, if one considers one among the random unit vectors, the contribution to either side of the cut is roughly the same. If the estimate to either side of the cut has some noise, one runs the risk of assigning these vectors to different sides. This is not a formal argument, since the same random sample would imply that the but suggest cuts which assign

. Moreover this algorithm has no way of breaking symmetry, meaning if a dataset has a large number of points with equal contributions, all such points will be assigned to one side, producing an imbalanced cut. Even if such ``symmetric" points are assigned randomly, the ratio of such ``symmetric'' points inside and outside the max-cut can be adjusted to causing the output cut to be imbalanced.

{\color{red} Is this roughly the argument we want to make?}

For example, consider the $n$-standard simplex with $n-1$ additional points clustered tightly around one of the unit vectors. The maximum cut separates the cluster $Y$ from the remaining $n$ points on the simplex $Z$. For any of the cluster's non-sampled point, the greedy assignment will always assign them to a single side. However for the non-sampled points in $Z$, the sum of distances between the cluster and the other points in $Z$ is the same, hence the contribution is expected to be the same, and the assignment varies greatly with respect to the sample. If it places the non-sampled points in $Z$ to the same side as $Y$, the cut output would have extremely low value compared to the true max-cut.

\ignore{\begin{example} \label{ex:planted-cut} Consider the following set of points $X = y_1,\dots, y_{n},z_1,\dots,z_{n} \in \R^{n+3}$ under the $\ell_2$ norm where each $y_i = \be_i + \bu_i \cdot \be_{n+2}$ and $z_i = \be_{n+1} + \bu_{n+i} \cdot \be_{n+2} $ for $\bu_1,\dots,\bu_{2n} \sim \calN(0, \frac{\eps}{10n^2})$, for some $0<\eps < 1$. Recall that $\be_i$ refers to the i-th standard basis vector. Essentially $X$ consists of some noise added to $n+1$ unit vectors with one vector duplicated $n$ times.
\end{example}
}